%% file: E8Applv2.tex
\documentclass[reqno, 11pt]{amsart}

\usepackage[parfill]{parskip}    
\usepackage{amsrefs, amsmath}
\usepackage{amsfonts, mathrsfs}
\usepackage{amssymb}
\usepackage{amscd}
\usepackage{cleveref} 
\usepackage{fullpage, color}
\usepackage[labelformat=empty]{subfig}
\usepackage{booktabs}
\usepackage{xypic}
\usepackage[dvips]{graphicx,epsfig}
\usepackage{enumerate}
\usepackage{makecell}

\numberwithin{equation}{section}

\BibSpec{article}{%
+{}{\sc \PrintAuthors} {author}
+{,}{ \textrm} {title}
+{,}{ \textit} {journal}
+{,}{ } {volume}
+{}{ \IfEmptyBibField{journal}{}{\PrintDate{date}}} {transition}
+{,}{ } {pages}
+{,}{ } {status}
+{.}{} {transition}
+{}{ Preprint available at \tt} {eprint}
+{.}{} {transition}
}

\BibSpec{book}{%
+{}{\sc \PrintAuthors} {author}
+{,}{ \textit} {title}
+{.}{ } {series}
+{,}{ } {volume}
+{.}{ } {publisher}
+{,}{ } {place}
+{,}{ } {date}
+{.}{} {transition}
}

\BibSpec{collection.article}{
+{} {\sc \PrintAuthors} {author}
+{,} { \textrm } {title}
+{,} { in \it \PrintConference} {conference}
+{,} { pp.~} {pages}
+{.} {\PrintBook} {book}
+{,} { \PrintDateB} {date}
+{.} {} {transition}
}

\BibSpec{innerbook}{
+{.} { \emph} {title}
+{.} { } {part}
+{:} { \emph} {subtitle}
+{.} { } {series}
+{,} { } {volume}
+{.} { Edited by \PrintNameList} {editor}
+{.} { Translated by \PrintNameList}{translator}
+{.} { \PrintContributions} {contribution}
+{.} { } {publisher}
+{.} { } {organization}
+{,} { } {address}
+{,} { \PrintEdition} {edition}
+{,} { \PrintDateB} {date}
+{.} { } {note}
}

\allowdisplaybreaks[3]

\crefname{equation}{Eq.}{Eqs.}
\crefname{eqnarray}{Eq.}{Eqs.}
\crefname{conj}{Conjecture}{Conjectures}
\crefname{defn}{Definition}{Definitions}
\crefname{lem}{Lemma}{Lemmas}
\crefname{thm}{Theorem}{Theorems}
\crefname{claim}{Claim}{Claims}
\crefname{rmk}{Remark}{Remarks}
\crefname{prop}{Proposition}{Propositions}
\crefname{section}{Section}{Sections}
\crefname{appendix}{Appendix}{Appendices}
\crefname{cor}{Corollary}{Corollaries}
\crefname{figure}{Figure}{Figures}
\crefname{table}{Table}{Tables}
\crefname{example}{Example}{Examples}

\newcommand{\de}{{\partial}}

\newcommand{\rd}{\mathrm{d}}
\newcommand{\ri}{\mathrm{i}}
\newcommand{\re}{\mathrm{e}}

\newcommand{\bbV}{\mathbb{V}}
\newcommand{\bbI}{\mathbb{I}}
\newcommand{\bbA}{\mathbb{A}}

\newcommand{\bbN}{\mathbb{N}}

\newcommand{\bbL}{\mathbb{L}}
\newcommand{\bbZ}{\mathbb{Z}}
\newcommand{\bbR}{\mathbb{R}}
\newcommand{\bbC}{\mathbb{C}}
\newcommand{\bbP}{\mathbb{P}}

\newcommand{\bbQ}{\mathbb{Q}}
\newcommand{\bbT}{\mathbb{T}}

\def\bary{\begin{array}} 
\def\eary{\end{array}} 
\def\ben{\begin{enumerate}} 
\def\een{\end{enumerate}}
\def\bit{\begin{itemize}} 
\def\eit{\end{itemize}}
\def\nn{\nonumber} 

\newcommand{\cY}{\mathcal{Y}}
\newcommand{\cZ}{\mathcal{Z}}
\newcommand{\cO}{\mathcal{O}}
\newcommand{\cT}{\mathcal{T}}
\newcommand{\cE}{\mathcal{E}}
\newcommand{\cP}{\mathcal{P}}
\newcommand{\cC}{\mathcal{C}}
\newcommand{\DD}{\mathcal{D}}
\newcommand{\LL}{\mathcal{L}}
\newcommand{\cS}{\mathcal{S}}
\newcommand{\cK}{\mathcal{K}}
\newcommand{\cN}{\mathcal{N}}
\newcommand{\cW}{\mathcal{W}}
\newcommand{\cG}{\mathcal{G}}
\newcommand{\cA}{\mathcal{A}}
\newcommand{\HH}{\mathcal{H}}
\newcommand{\cB}{\mathcal{B}}
\newcommand{\cF}{\mathcal{F}}

\newcommand{\cX}{\mathcal{X}}

\newcommand{\cM}{\mathcal M}

\newcommand{\cU}{\mathcal U}

\def\beq{\begin{equation}}                     %  
\def\eeq{\end{equation}}                       % 
\def\bea{\begin{eqnarray}}                     %         % 
\def\eea{\end{eqnarray}}
\def\bary{\begin{array}} 
\def\eary{\end{array}} 
\def\ben{\begin{enumerate}} 
\def\een{\end{enumerate}}
\def\bit{\begin{itemize}} 
\def\eit{\end{itemize}}
\def\nn{\nonumber} 
\def\de {\partial}

\def\a{\alpha}
\def\b{\beta}

\def\e{\epsilon}

\def\Res{\mathrm{Res}}

\theoremstyle{plain}

\newtheorem{thm}{Theorem}[section]
\newtheorem{lem}[thm]{Lemma}

\newtheorem{prop}[thm]{Proposition}
\newtheorem{conj}[thm]{Conjecture}
\newtheorem{claim}[thm]{Claim}
\newtheorem*{conj*}{Conjecture}
\newtheorem{cor}[thm]{Corollary}
\newtheorem*{cor*}{Corollary}

\newtheorem{defn}{Definition}[section]

\theoremstyle{definition}

\newtheorem{rmk}[thm]{Remark}

\newcommand{\GIT}[1]{/\!\!/_{\kern-.2em #1 \kern0.1em}}

\renewcommand{\l}{\left}
\renewcommand{\r}{\right}
\newcommand{\bra}{\left\langle}
\newcommand{\ket}{\right\rangle}

\def\bred{\begin{color}{red}}
\def\ered{\end{color}}

\def\bes{\begin{subequations}}
\def\ees{\end{subequations}}

\begin{document}

\title{$\mathrm{E}_8$ spectral curves}

\author{Andrea Brini}
\address{Department of Mathematics,
Imperial College London, 180 Queen's Gate, SW7 2AZ, London, United Kingdom
}
\address{School of Mathematics and Statistics, University of Sheffield,
  Hounsfield Road, Sheffield S3 7RH, United Kingdom}
\address{On leave from IMAG, Univ.~Montpellier, CNRS, Montpellier, France}

\email{a.brini@sheffield.ac.uk}

\begin{abstract}

I provide an explicit construction of spectral curves for the
affine $\mathrm{E}_8$ relativistic Toda chain.
%Specifically, the closed-form expression for the spectral
%curve of the relativistic generalisation of the affine $\mathrm{E}_8$-Toda chain
%(Coxeter--Bogoyavlensky--Toda lattice) is obtained; 
Their closed form expression is obtained by
determining the full
set of character relations in the representation ring of $\mathrm{E}_8$ for the exterior algebra of the adjoint
representation; this is in turn employed to provide
an explicit construction of both integrals of motion  and
the
action-angle map for the resulting integrable system.\\
I consider two main areas of applications of these constructions. On the one
hand, I consider the resulting family of spectral curves in the context of the correspondences
between Toda systems, 5d Seiberg--Witten theory,
%solution of $\mathcal{N}=1$
%five-dimensional gauge theory, 
Gromov--Witten theory of
orbifolds of the resolved conifold, and Chern--Simons theory
% on spherical
%3-manifolds 
to establish a version of the B-model Gopakumar--Vafa correspondence for the
$\mathrm{sl}_N$ L\^e--Murakami--Ohtsuki invariant of the Poincar\'e integral homology sphere to all orders in $1/N$.
%with an application to the strongly coupled
%(maximally Argyres--Douglas) $\mathrm{E}_8$-superconformal point; a
%Hori--Iqbal--Vafa-type mirror geometry for the local Calabi--Yau geometry given by
%the $\mathbb{I}_{120}$ fibrewise quotient of the resolved conifold,
%
%  whose
%large $N$ duality with $\mathrm{U}(N)$ Chern--Simons theory on  is established  and a one-dimensional
On the other, I consider a degenerate version of the spectral curves and prove a
1-dimensional Landau--Ginzburg mirror theorem for the Frobenius manifold
structure on the space of orbits of the
extended affine Weyl group of type $\mathrm{E}_8$ introduced by Dubrovin--Zhang (equivalently, the orbifold quantum cohomology of the type-$\mathrm{E}_8$
polynomial $\bbC P^1$ orbifold). This leads to closed-form expressions for the
flat co-ordinates of the Saito metric, the prepotential, and a higher genus
mirror theorem based on the Chekhov--Eynard--Orantin recursion. I will also 
show how the constructions of the paper %should 
lead to a generalisation of a conjecture of Norbury--Scott to ADE
$\bbP^1$-orbifolds, and a mirror of the Dubrovin--Zhang construction for all Weyl groups
and choices of marked roots.

\end{abstract}

\maketitle
\tableofcontents

\section{Introduction}

Spectral curves have been the subject of considerable study in a variety of
contexts. These are moduli spaces $\mathscr{S}$
%$\cM_g(\Gamma, \rd \sigma,
%\rd \tau, \Lambda)$ 
of complex projective 
%genus-$g$ 
curves 
%$\Gamma$ 
endowed with a distinguished pair
of meromorphic abelian differentials 
%$(\rd \sigma, \rd \tau)$ 
and a
marked symplectic subring 
%$\Lambda \subset H_1(\Gamma,\bbZ)$ 
of their
first homology group; such data define (one or more) families of flat connections on the
tangent bundle of the smooth
part of moduli space. In particular, a
Frobenius manifold structure on the base of the family, a
dispersionless integrable hierarchy on its loop space, and 
%$\cM_g(\Gamma, \rd \sigma,
%\rd \tau, \Lambda)$, 
the genus zero part of a semi-simple CohFT 
%$\cM_g(\Gamma_g, \rd \sigma,\rd \tau, \Lambda)$ 
are then naturally defined in terms of periods of the aforementioned
differentials over the marked cycles; a canonical reconstruction of the
dispersive deformation (resp. the higher genera of the CohFT) is furthermore determined by
$\mathscr{S}$ through the topological recursion of \cite{Eynard:2007kz}.

The one-line summary of this paper is that I offer two constructions (related
to Points (II) and (IV) below) and two isomorphisms (related to Points (III), (V)
and (VI)) in the context of spectral curves with exceptional gauge symmetry of
type $\mathrm{E}_8$. 
%I describe the context below, and an outline of the paper in \cref{sec:about}.

\subsection{Context}

Spectral curves are abundant in several problems in enumerative
geometry and mathematical physics. In particular:
\ben[(I)]
\item in the spectral theory of finite-gap solutions of the KP/Toda hierarchy,
  spectral curves arise as the (normalised, compactified)
  affine curve in $\bbC^2$ given by the vanishing locus of the Burchnall--Chaundy
  polynomial ensuring commutativity of the operators generating two
  distinguished flows of the hierarchy; the marked abelian differentials here are just the
  differentials of the two coordinate projections 
%$(\sigma, \tau)$ 
onto the plane.  In this case,
  to each smooth point 
%$u\in \cS$ 
in moduli space with
  fibre a smooth Riemann surface $\Gamma$
%$\Gamma_u$ 
there corresponds a canonical theta-function solution of the
  hierarchy depending on $g(\Gamma)$ times, and the associated dynamics 
  is encoded into a  linear flow on the Jacobian of the curve;
%$\mathrm{Pic}^{(0)}(\Gamma_u)$;
\item in many important cases, this type of linear flow on a Jacobian (or, more
  generally, a principally polarised Abelian subvariety thereof, singled out by the marked basis of 1-cycles on the curve) is a
  manifestation of the Liouville--Arnold dynamics of an auxiliary, finite-dimensional
  integrable system. Coordinates in moduli space correspond to Cauchy data --
  i.e., initial values of involutive Hamiltonians/action variables --  and
  flow parameters are given by linear coordinates on the associated torus;
\item all the action has hitherto taken place at a fixed fibre over a point in moduli
  space; however additional structures emerge once moduli are varied by
  considering secular (adiabatic)
  deformations of the integrals of motions via the Whitham averaging method. This defines
a dynamics on moduli space which is itself integrable and admits a
$\tau$-function; remarkably, the logarithm of the $\tau$-function satisfies the big
phase-space version of WDVV equations, and its restriction to initial data/small
phase space defines an almost-Frobenius manifold structure on the moduli space;
\item from the point of view of four dimensional supersymmetric gauge theories with eight
  supercharges, the appearance of WDVV equations for the Whitham $\tau$-function is
  equivalent to the constraints of rigid special K\"ahler
  geometry on the effective prepotential; such equivalence is indeed realised by presenting the Coulomb branch of the theory as a
  moduli space of spectral curves, the marked differentials giving rise to
% $\sigma \rd \tau$
  %as 
the Seiberg--Witten 1-form,
%differential, 
the BPS central charge as the period
  mapping on the marked homology sublattice,
% spanned by $\Lambda$, 
and the prepotential as the logarithm of
  the Whitham $\tau$-function;
\item in several cases, the Picard--Fuchs equations satisfied by the periods of
%  $\sigma\rd\tau$ 
the SW differential
are a reduction of the GKZ hypergeometric system for a toric
  Calabi--Yau variety, whose quantum cohomology is then isomorphic to the
  Frobenius manifold structure on the moduli of spectral curves.
% This is
%  one of the simplest instances of mirror symmetry for local Calabi--Yau
%  threefolds. 
What is more, 
%the special one-dimensional
%nature of this setting opens the way to a direct, and
%  striking, generalisation 
spectral curve mirrors open the way to include higher genus
  Gromov--Witten invariants in the picture through the Chekhov--Eynard--Orantin
  topological recursion: a universal calculus of residues on the fibres of the
  family $\mathscr{S}$, which is recursively determined by the spectral data.
This provides simultaneously a
  definition of a higher genus topological B-model on a curve, a higher genus
  version of local mirror symmetry, and a dispersive deformation of the
  quasi-linear hierarchy obtained by the averaging procedure;
\item in some cases, spectral curves 
%arising in this context
%  above 
may also be related to
 multi-matrix models and topological gauge theories (particularly
  Chern--Simons theory) in a formal $1/N$
  expansion: for fixed 't~Hooft parameters, 
%and generalising the Wigner
%  semi-circle law for the Gaussian unitary ensemble, 
the generating function of
  single-trace insertion of the gauge field in the planar limit cuts out a
  plane curve in $\bbC^2$.
%, with a fixed choice of one-dimensional paths with
%  symplectic intersection matrix, 
%and thus a particular spectral curve via the
%  two coordinate projections. 
The asymptotic analysis of the matrix model/gauge
  theory then falls squarely within the above setup: the formal solution of
  the Ward identities of the model dictates that the planar free energy is
  calculated by
  the special K\"ahler geometry relations for the associated spectral curve,
  and the full $1/N$ expansion of connected multi-trace correlators is
  computed by the topological recursion.
\een

A paradigmatic example is given by the spectral curves arising as the
vanishing locus for the characteristic polynomial of the Lax matrix for the periodic
Toda chain with $N+1$ particles. In this case (I) coincides with the theory of
$N$-gap solutions of
the Toda hierarchy, which has a counterpart (II) in the Mumford--van Moerbeke algebro-geometric integration of
the Toda chain by way of a flow on the Jacobian of the curves. In turn, this
gives a Landau--Ginzburg picture for an (almost) Frobenius manifold structure (III),
which is associated to the Seiberg--Witten solution of $\mathcal{N}=2$ pure
${\rm SU}(N+1)$ gauge theory (IV). The relativistic deformation of the system relates
the Frobenius manifold above to the quantum cohomology (V) of a family of toric
Calabi--Yau threefolds (for $N=1$, this is $K_{\bbP^1 \times \bbP^1}$), which
  encodes the planar limit of ${\rm SU}(M)$ Chern--Simons--Witten invariants on lens
  spaces $L(N+1,1)$ in (VI). 

\subsection{What this paper is about}
\label{sec:about}

A wide body of literature has been devoted in the
  last two decades to further generalising at least part of this web of relations to a wider
  arena (e.g. quiver gauge theories). A somewhat orthogonal direction, and one
  where the whole of (I)-(VI) have a concrete generalisation, is to
  consider the Lie-algebraic extension of the Toda hierarchy and its
  relativistic counterpart to
  arbitrary root systems $R$ associated to semi-simple Lie algebras, the
  standard case corresponding to $R={\rm A}_{N}$. Constructions
  and proofs of the relations above have been available for quite a while for
  (II)-(IV) and more recently for (V)-(VI), in complete generality except for
  one egregious example: $R={\rm E}_8$, whose
  complexity has put it out of reach of previous treatments in the
  literature. This paper
  %grows out of the
  %author's stubborness to
  fills the gap in this exceptional case and provides, as an upshot, a series of novel applications of Toda spectral
  curves which may be of interest for geometers and mathematical physicists
  alike. As was mentioned, the aim of the paper is to provide two main constructions, and
prove two isomorphisms, 
%in the context of Toda curves with $\mathrm{E}_8$
%symmetry. More in detail, this goes 
as follows.

\begin{description}
\item[Construction 1] The first
construction gives a closed-form
expression for arbitrary moduli of the family of curves associated to the relativistic
Toda chain of type $\widehat{\mathrm{E}}_8$ for its sole quasi-minuscule
representation -- the adjoint. This is achieved in two steps: by determining the dependence
of the regular fundamental characters of the Lax matrix on the spectral
parameter, and by subsequently computing the polynomial character
relations in the representation ring of $\mathrm{E}_8$ (viewed as a polynomial ring
over the fundamental characters) corresponding to the exterior
powers of the adjoint representation.
%, $\wedge^k \mathfrak{e}_8$. 
The last
step, which is of independent
representation theoretic interest, is of significant computational complexity
and is solved by a reduction to an equivalent large-size linear problem which
is amenable to an efficient solution by distributed computation.
This is beyond the scope of this paper and will find
a detailed description in \cite{E8comp}: I herein
limit myself to announce and condense the ideas of \cite{E8comp}
into the 2-page summary given in
\cref{sec:charE8}, and accompany this paper with a {\it Mathematica}
package\footnote{This is available at {\tt http://tiny.cc/E8SpecCurve}. Part of
  the complexity is reflected in the size of the compressed data containing
  the final solution ($\sim$ 180Mb -- should the reader wish to have a closer
  look at this, they should be aware that this unpacks to binary files and a
  {\it Mathematica} notebook that are collectively almost 1GB of data.).} containing the
    solution thus achieved. As an immediate spin-off I obtain
the generating function of the integrable model (in the language of
\cite{Fock:2014ifa}) as a function of the basic involutive Hamiltonians
attached to the fundamental weights, and a family of spectral curves as its
vanishing locus. In the process, this yields  
a canonical set of integrals of motion in involution in cluster variables and
in Darboux co-ordinates for the integrable system on a special double
Bruhat cell
%, corresponding to the Coxeter
%element of the affine Weyl group $\cW_{\widehat{\mathrm{E}_8}}\times W_{\widehat{\mathrm{E}_8}}$,
of the coextended Poisson--Lie loop group $\widehat{\mathrm{E}_8}^\#$, which, by analogy with
the case of $\widehat{{\rm A}}$-series, I call ``the relativistic
$\widehat{\mathrm{E}_8}$ Toda chain'', and whose dynamics is solved completely by the
preceding construction. 
\item[Construction 2] The previous construction gives the first element in the
  %4-tuple $\cM=(C_g, \rd \sigma, \rd \tau, \Gamma)$ 
description of the spectral curve -- a family of plane complex algebraic curves, which
  are themselves integrals of motion. The next step determines the
  three remaining characters in the play, namely the two marked Abelian
  differentials and the distinguished sublattice of the first homology of the curves; this goes hand in hand with the
  construction of appropriate action--angle
  variables for the system.  I identify the phase space of the Toda system with a fibration over the
  Cartan torus of $\mathrm{E}_8$ (times $\bbC^\star$) by Abelian varieties, which are
  Prym--Tyurin sub-tori of the spectral curve Jacobian. These are selected by the curve geometry itself,
due to an argument going back to
  Kanev \cite{MR1013158}, and 
%(and correspondingly, a distinguished   subring  $\Gamma \subset H_1(C_g,
%\bbZ)$) 
the Liouville--Arnold flows linearise on them. The Hamiltonian structure
  inherited from the embedding of the system into a Poisson--Lie--Bruhat cell
  translates into a canonical choice of symplectic form on the universal
  family of Prym--Tyurins, and it pins down (up to canonical transformation) a
  marked pair of Abelian third kind differentials on the curves. 

Altogether,
  the family of curves, the marked 1-forms, and the choice of preferred cycles
  lead to the assignment of a set of {\it Dubrovin--Krichever data} (\cref{defn:dk}) to the
  family of spectral curves. Armed with this, I turn to some of the uses of
  Toda spectral curves in the 
  context of Fig.~\ref{fig:dualities}. 

\begin{figure}[!h]
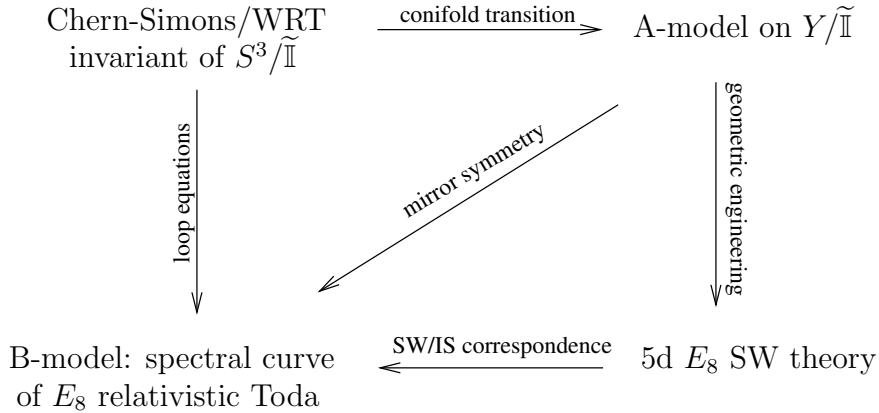
\caption{Duality web for the B-model on Toda spectral curves}
\label{fig:dualities}
\end{figure}
%schematically, for $\sigma :
%  \cM \to \mathrm{Sym}^{8}(\mathcal{C}_g)$ a section of the
%  $8^{\rm th}$-symmetric power of the tautological family of curves over
%  $\cM$, this is given by $\rd
%  \sigma(\sigma) \wedge \rd \tau(\sigma)$, .

\item[Isomorphism 1]  

Toda spectral curves have long been proposed to encode the
  Seiberg--Witten solution 
of $\cN=2$ pure gluodynamics in four dimensional
  Minkowski space   \cite{Gorsky:1995zq, Martinec:1995by}, as well as of its higher dimensional
  $\cN=1$ parent theory on $\bbR^4 \times S^1$
  \cite{Nekrasov:1996cz} in the relativistic case. From the physics point of
  view, Constructions~1-2 provide the
  Seiberg--Witten solution for minimal, five-dimensional supersymmetric $\mathrm{E}_8$
  Yang--Mills theory on $\bbR^4\times S^1$; and as the latter should be related to
  (twisted) curve counts on an orbifold of the resolved conifold
  $Y=\cO_{\bbP^1}(-1)\oplus \cO_{\bbP^1}(-1)$ by the
  action of the binary icosahedral group $\tilde{\mathbb{I}}$, the same
  construction provides a conjectural 1-dimensional mirror
  construction for the orbifold Gromov--Witten theory of these targets, as well as to its large $N$ Chern--Simons dual theory on the
  Poincar\'e sphere $S^3/\tilde{\mathbb{I}}\simeq \Sigma(2,3,5)$
  \cite{Gopakumar:1998ki, Ooguri:1999bv, Aganagic:2002wv, Borot:2015fxa}. 
I do not pursue here the proof of either the bottom horizontal (SW/integrable
systems correspondence) or the diagonal (mirror symmetry) arrow in the diagram
of Fig.~\ref{fig:dualities}, although it is highlighted in the text how having access to the global solution on its Coulomb
  branch allows to study particular degeneration limits of the solution
  corresponding to superconformal (maximally Argyres--Douglas) points where
  mutually non-local dyons pop up in the massless spectrum, and 
  limiting versions of mirror symmetry for the Toda curves in Isomorphism~2
  below are also considered. What I do prove instead is a version of the vertical arrow:
%  completing results in a previous joint paper with Borot
  %  \cite{Borot:2015fxa}:
  namely,
  that the Chern--Simons/Reshetikhin--Turaev--Witten invariant of $\Sigma(2,3,5)$ restricted to the
  trivial flat connection (the L\^e--Murakami--Ohtsuki invariant), as
  well as the quantum invariants of fibre knots therein in the same limit and
  for arbitrary colourings, are
  computed to all orders in $1/N$ from the Chekhov--Eynard--Orantin topological
  recursion on a suitable subfamily of $\widehat{\mathrm{E}_8}$ relativistic
  Toda spectral curves.
  %As in
  %\cite{Borot:2015fxa},
  The strategy resorts to studying the trigonometric
  eigenvalue model associated to the LMO invariant of the Poincar\'e sphere at
  large $N$ and to prove that the planar resolvent is one of the meromorphic coordinate
  projections of a plane curve in $(\bbC^\star)^2$, which is in turn shown to
  be the affine part of the spectral curve of
  the $\widehat{\mathrm{E}_8}$ relativistic Toda chain.

\item[Isomorphism 2] I further consider two meaningful operations that can be performed
  on the spectral curve setup of Constructions~1-2. The first is to take a
  degeneration limit to the leaf where the natural Casimir function of the
  affine Toda chain goes to zero; this corresponds to the restriction to
  degree-zero orbifold invariants on the top-right corner of
  Fig.~\ref{fig:dualities}, and to the perturbative limit of the 5D
  prepotentials of the bottom-right corner. The second is to replace one of the
  marked Abelian integrals with their exponential; this is a version of 
  Dubrovin's notion of (almost)-duality of Frobenius manifolds
  \cite{MR2070050}. 

I conjecture and prove that the resulting spectral curve provides a
  1-dimensional Landau--Ginzburg mirror for the Frobenius manifold structure
  constructed on orbits of the extended affine Weyl group of type $\mathrm{E}_8$ by
  Dubrovin and Zhang \cite{MR1606165}. Their construction depends on a choice
  of simple root, and the canonical choice they take matches with the Frobenius
  manifold structure on the Hurwitz space determined by our global spectral
  curve. This
%  extends to the first (and most) exceptional case the LG mirror
  %  theorems of \cite{Dubrovin:2015wdx} for the classical series; and it
  opens the way to formulate a precise conjecture for how the general
  case, encompassing general choices of simple roots in the Dubrovin--Zhang
  construction, should receive an analogous description in terms of Toda
  spectral curves for the corresponding Poisson--Lie group and twists thereof
  by the action of a Type I symmetry of WDVV (in the language of
  \cite{Dubrovin:1994hc}). Restricting to simply-laced Lie algebras, this gives a
  mirror theorem for the quantum cohomology of ADE orbifolds of $\bbP^1$:
our genus zero mirror statement then lifts to an all-genus statement by virtue
of the equivalence of the topological recursion with Givental's
quantisation for R-calibrated Frobenius manifolds. This provides a version, for
the ADE series, of statements by Norbury--Scott 
\cite{MR3268770,DuninBarkowski:2012bw,MR3654104} for the Gromov--Witten theory
of $\bbP^1$.

\end{description}

\subsection{Structure of the paper and relation to other work} The two constructions and two isomorphisms above will find their place in
\cref{sec:E8chain}, \ref{sec:actangl}, \ref{sec:applI} and \ref{sec:applII}
respectively. The main novel results of the paper are structured in the
following logical progression:
\bit
\item \cref{claim:E8} (which is Theorem~2.4 in the companion paper
  \cite{E8comp}) and \cref{lem:uispec} provide the explicit form of
  relativistic $\widehat{\mathrm{E}_8}$ Toda spectral curves of {\bf
    Construction 1}.
\item \cref{thm:pt,thm:kp,cor:actangl} establish the linearisation of the
  flows on the canonical Prym--Tyurin fibration over the family of Toda
  spectral curves, as well as their Hamiltonian nature, completing {\bf
    Construction 2}.
\item \cref{thm:gv} proves the weak B-model Gopakumar--Vafa
  correspondence for the Poincar\'e sphere in {\bf Isomorphism 1}.
\item \cref{conj:mirror,thm:frob} provide respectively a uniform construction
  of Landau--Ginzburg mirrors of the Dubrovin--Zhang Frobenius manifolds
  associated to orbits of extended affine Weyl groups in all cases, and a
  proof for the type $\mathrm{E}_8$ group and the canonical marked node, which
  is {\bf Isomorphism 2}.
\eit
%
%While the body of results above is
Some facets of the problems addressed here have surfaced with a different
angle in previous works in the literature, and in order to make the text
self-contained we review as necessary the
links with their methodology at the beginning of each Section. The input datum
of our {\bf Construction 1} is the Lax
formalism with spectral parameter of Fock--Marshakov in \cite{Fock:2014ifa},
which is the starting point of our reduction of the computation of spectral curves to
a problem in Lie theory. {\bf Construction 2}, while new for relativistic
systems of type other than $\widehat{\mathrm{A}_n}$, owes an intellectual debt to the
classical ideology of \cite{MR1013158, MR533894, MR815768, MR1401779,
  MR1397059} in the non-relativistic case, and to the construction of
algebro-geometric symplectic forms of \cite{Krichever:1997sq, DHoker:2002kfd},
both of which are shown in this paper to be
adaptable to the relativistic setting at hand. {\bf Isomorphism 1} concludes a
program initiated in my joint work with Borot \cite{Borot:2015fxa} to prove
the B-model Gopakumar--Vafa correspondence for Clifford--Klein 3-manifolds by
treating the central missing case of
the Poincar\'e sphere, and furthermore completes it to the full higher genus
theory by
proving that the Chern--Simons planar two-point function agrees with the
symmetrised Bergmann recursion kernel on the Toda curves, thereby establishing
the equality of initial data for the Chekhov--Eynard--Orantin recursion on the
two sides of the correspondence. The previous state-of-the-art in the
construction of mirrors for Dubrovin--Zhang Frobenius manifolds in type other
than $\mathrm{A}_n$ was given by
\cite{Dubrovin:2015wdx}, where a version of {\bf Isomorphism~2} is given by an entirely
different route for extended affine Weyl groups associated to $\mathrm{Spin(n, \bbC)}$ and $\mathrm{Sp(n, \bbC)}$ groups. Our construction instead provides a
general method which is applicable uniformly to all simple, simply-connected Lie groups, including exceptional
cases and all choices of marked roots,  recovers as a particular case \cite{MR1606165, Dubrovin:2015wdx} by restricting to Dynkin types A, B, C, and D, and is shown in
particular to yield the correct mirror for the most exceptional case of $\mathrm{E}_8$. More details for the other exceptional groups will appear in \cite{karophd}.

I have tried to give a self-contained exposition of the material
in each of \cref{sec:E8chain,sec:actangl,sec:applI,sec:applII}, and to a good extent the reader interested in a particular
angle of the story may read them independently (in particular
\cref{sec:applI,sec:applII}). 

\subsection*{Acknowledgements} I would like to thank G.~Bonelli, G.~Borot, A.~D'Andrea,
B.~Dubrovin, N.~Orantin, N.~Pagani, P.~Rossi, A.~Tanzini, Y.~Zhang for discussions and correspondence
on some of the topics touched upon in this paper, and H.~Braden for bringing
\cite{MR1182413,MR1401779,MR1668594} to my attention during a talk at SISSA in
2015. For the calculations of \cref{sec:charE8} and
\cite{E8comp}, I have availed myself of cluster computing facilities at the
Universit\'e de Montpellier
({\tt Omega} departmental cluster at IMAG, and the HPC@LR centre {\tt Thau/Muse} inter-faculty cluster) and the compute cluster of the
Department of Mathematics at Imperial College London. I am grateful to B.~Chapuisat and especially A.~Thomas for
their continuous support and patience whilst these computations were carried
out. This research was partially supported by the ERC Consolidator Grant
no.~682603 (PI:~T.~Coates) and, for the revised version of the manuscript, the
EPSRC Fellowship grant no.~EP/S003657/2.

\section{The $\mathrm{E}_8$ and $\widehat{\mathrm{E}_8}$ relativistic Toda chain}
\label{sec:E8chain}

I will provide a succinct, but rather complete account of the construction of Lax pairs for the
relativistic Toda chain for both the finite and affine $\mathrm{E}_8$ root
system. This is mostly to fix notation and key concepts for the discussion to
follow, and there is virtually no new material here until
\cref{sec:speccurve}. I refer the reader to \cite{Fock:2014ifa,
  Williams:2012fz, MR1993935, Reyman:1979ru, Olshanetsky:1981dk} for more context,
references, and further discussion. I
will subsequently move to the explicit construction of spectral curves and the
action-angle map for the affine $\mathrm{E}_8$ chain in \cref{sec:speccurve,sec:actangl}.

\subsection{Notation}
\label{sec:notation}
I will start by fixing some basic notation for the foregoing discussion; in
doing so I will endeavour to avoid the uncontrolled profileration of subscripts ``8'' related
to $\mathrm{E}_8$ throughout the text, and stick to generic symbols instead (such as $\cG$ for the $\mathrm{E}_8$ Lie group,
$\mathfrak{g}$ for its Lie algebra, and so on). I wish to make clear from the outset
though that whilst many aspects of the discussion are general, the focus of
this section is on $\mathrm{E}_8$ alone; the attentive reader
will notice that some of its
properties, such as simply-lacedness, or triviality of the centre, are
implicitly assumed in the formulas to follow.

Let then $\mathfrak{g}\triangleq \mathfrak{e}_8$ denote the complex simple Lie algebra corresponding to
the Dynkin diagram of type $\mathrm{E}_8$ (Fig.~\ref{fig:dynke8}). I will write $\cG=\exp \mathfrak{g}$ for the
corresponding simply-connected complex Lie group, $\cT=\exp\mathfrak{h}$ for the
maximal torus (the exponential of the Cartan algebra $\mathfrak{h}\subset \mathfrak{g}$), and $\cW=N_\cT/\cT$ for the
Weyl group. I will also write $\Pi=\{\a_1, \dots, \a_8\}$ for the set of simple
roots (see e.g. \eqref{eq:roots}), and $\Delta$, $\Delta^*$, $\Delta^{(0)}$, $\Delta^\pm$ to indicate respectively the
full root system, the non-vanishing roots, the zero roots, and the negative/positive roots;
%I will fix once and for all a set 
%$\Pi=\{\a_1, \dots, \a_8\}$ of simple roots  and, 
the choice of splitting $\Delta^\pm$ determines accordingly Borel subgroups $\cB^\pm$
intersecting at $\cT$. Each Borel realises $\cG$ as a
disjoint union of double cosets $\cG=\cB^\pm \cW \cB^\pm=\coprod_{w\in \cW} \cB^\pm w
\cB^\pm=\coprod_{(w_+,w_-)\in \cW \times \cW} \l(\cB^+ w_+ \cB^+ \cap \cB^- w_- \cB^-\r)
=: \coprod_{(w_+,w_-)\in \cW \times \cW} \cC_{w_+, w_-}$,
the double Bruhat cells of $\cG$.  The Euclidean vector space $(\mathrm{span}_\bbR
\Pi; \bra, \ket)
\subset \mathfrak{h}^*$ is a vector subspace of $\mathfrak{h}^*$ with an
inner product structure $\bra \beta, \gamma\ket$ given by the dual of the Killing
form; in particular, $\bra \a_i, \a_j \ket \triangleq \mathscr{C}^{\mathfrak{g}}_{i,j}$ is the Cartan
matrix \eqref{eq:cartE8}. For a weight $\lambda$ in
the lattice
$\Lambda_w(\cG)\triangleq \{\lambda \in \mathfrak{h}^* | \bra \lambda, \alpha \ket \in \bbZ\}$, I will
write $\cW_\lambda=\mathrm{Stab}_\lambda \cW$ for the parabolic subgroup
stabilised by $\lambda$; the action of $\cW$ on weights is the restriction of
the coadjoint action on $\mathfrak{h}^*$; since $Z(\cG)=e$ in our case, the
weight lattice is isomorphic to the root lattice $\Lambda_r(\cG)=\bbZ\bra \Pi \ket \simeq \Lambda_w(\cG)$.
Corresponding to the choice of
$\Pi$, Chevalley generators
$\{(h_i\in \mathfrak{h}, e_{\pm i}\in {\rm Lie}(\cB^\pm)| i \in
\Pi \}$ for $\mathfrak{g}$ will be chosen satisfying
\bea
[h_i, h_j] &=& 0, \nn \\
\l[h_i,e_j\r] &=& \mathrm{sgn}(j) \delta_{i|j|} e_j, \nn \\
\l[e_i, e_{-i}\r] &=& \mathrm{sgn}(i) \mathscr{C}^{\mathfrak{g}}_{ij}h_j, \nn \\
(\mathrm{ad} e_i)^{1-\mathscr{C}^{\mathfrak{g}}_{ij}} e_j  &=& 0 \quad \mathrm{for} \quad i+j\neq 0.
\label{eq:chevalley}
\eea
Accordingly, the correponding time-$t$ flows on $\cG$ lead to Chevalley generators
$H_i(t)=\exp{t h_i }$, $E_i(t)=\exp{t e_i}$ for the Lie group.  Finally, I denote by $R(\cG)$ the
representation ring of $\cG$, namely the free abelian group of virtual
representations of $\cG$ (i.e. formal differences), with ring structure given by
the tensor product; this is a polynomial ring $\bbZ[\omega]$ over the integers with
generators given by the irreducible $\cG$-modules having $\omega_i\in \Lambda_w(G)$ as
their highest weights, where $\bra \omega_i, \alpha_j\ket =\delta_{ij}$.
\begin{figure}
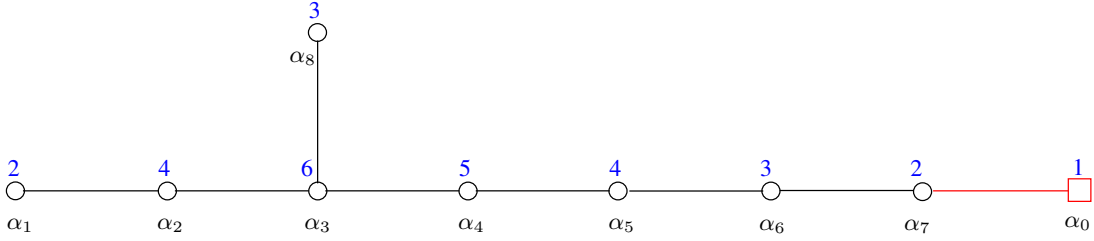
\caption{The Dynkin diagrams of type $\mathrm{E}_8$ and, superimposed in red, type
  $\mathrm{E}_8^{(1)}$; roots are labelled following Dynkin's convention
  (left-to-right, bottom-to-top). The numbers in blue are the Dynkin labels for each vertex --
  for the non-affine roots, these are the components of the highest root in
  the $\alpha$-basis.}
\label{fig:dynke8}
\end{figure}

Most of the notions (and notation) above carries through to the setting of
the Kac--Moody group\footnote{It should be noticed that, while in
  \eqref{eq:chevalley} passing from $h_i$
  to $h'_i=\sum \mathscr{C}^{\mathfrak{g}}_{ij} h_j$ is an isomorphism of Lie algebras, the same is not
  true in the affine setting as the Cartan matrix is then degenerate. Our
  discussion below sticks to the Lie algebra relations as written in
  \eqref{eq:chevalley}, rather than their more common dualised form; in the
  affine setting, this substantial difference leads to the
  centrally coextended loop group instead of the more
  familiar central extension in Kac--Moody theory. In  \cite{Fock:2014ifa},
  this is stressed by employing the notation $G^\#$ for the co-extended group;
  as I make clear from the outset in \eqref{eq:chevalley} what side of the duality I
  am sitting on, I somewhat abuse notation and denote $\widehat{\cG}$ the
  resulting Poisson--Lie group.} $\widehat{\cG}=\exp{\mathfrak{g}^{(1)}}$ where
$\mathfrak{g}^{(1)}\simeq \mathfrak{g} \otimes \bbC[\lambda,\lambda^{-1}] \oplus \bbC c$ is the (necessarily
untwisted, for $\mathfrak{g}\simeq \mathfrak{e}_8$) affine Lie algebra corresponding to $\mathfrak{e_8}$. In this case
we adjoin the highest (affine) root $\a_0$ as in \eqref{eq:affroot}, leading to
the Dynkin diagram and Cartan matrix in Fig.~\ref{fig:dynke8} and
\eqref{eq:cartE81}. Elements $g \in \widehat{\cG}$ are linear $q$-differential polynomials
in the spectral parameter $\lambda$; namely, $g=M(\lambda)
q^{\lambda \de \lambda}$, with the pointwise multiplication rule leading to
\beq
g_1 g_2 = M_1(\lambda) M_2(q_1 \lambda) ~(q_1 q_2)^{\lambda \de \lambda}.
\eeq
The Chevalley generators for the simple Lie group
$\cG$ are then lifted to $\widehat{H_i(q)}\triangleq H_i(q) q^{\mathfrak{d}_i \lambda \de \lambda}$,
with $\mathfrak{d}_i$ the Dynkin labels as in Fig.~\ref{fig:dynke8}, and extended to
include $(H_0, E_0, E_{\bar 0})$ where
\beq
H_0(q)= q^{\lambda \de \lambda}, \quad E_0=\exp(\lambda e_0), \quad E_{\bar
  0}=\exp(e_{\bar 0}/\lambda)
\eeq
with $e_0 \in \mathrm{Lie}(\cB^+)$ and $e_{\bar 0}\in \mathrm{Lie}(\cB^-)$ the Lie
algebra generators corresponding to the highest (lowest) roots -- i.e. the
only non-vanishing iterated commutators of order $h(\mathfrak{g})=30$ of $e_i$ ($e_{\bar
  i}$), $i=1, \dots, 8$.

\subsection{Kinematics}

Consider now the 16-dimensional symplectic algebraic torus $$\cP\simeq \l(
(\bbC_x^\star)^8 \times (\bbC_y^\star)^8, \{,\}_\cG\r)$$ with Poisson bracket
\beq
\{x_i, y_j\}_\cG= \mathscr{C}^{\mathfrak{g}}_{ij} x_i y_j.
\eeq
Semi-simplicity of $\cG$ amounts to the non-degeneracy of the bracket, so that $\cP$ is symplectic. 

There is an injective morphism from $\cP$ to a distinguished Bruhat cell of $\cG$, as
follows. Notice first that $\cG$ carries an adjoint action by the Cartan torus
$\cT$ which obviously preserves the Borels, and therefore, descends to an
action on the double cosets of the Bruhat decomposition. Consider now Weyl
group elements $w_+=w_-=\bar w$ where $\bar w$ is the ordered product of the
eight simple reflections in $\cW$. The corresponding cell $\cP^{\rm
  Toda}\triangleq \cC_{\bar w, \bar w}
\subset \cG/\cT$ has dimension 16 \cite{Fock:2014ifa}, and it inherits a
symplectic structure from $\cG$, as I now describe. Recall that the latter carries a
Poisson structure given by the canonical Belavin--Drinfeld--Olive--Turok 
solution of the classical Yang--Baxter equation \cite{MR674005,Olive:1983mw}:
\beq
\{g_1 \stackrel{\otimes}{,} g_2\}_{\rm PL} = \frac{1}{2}\l[r, g_1 g_2 \r],
\label{eq:pl}
\eeq
with $r\in \mathfrak{g}\otimes \mathfrak{g}$ given by 
\beq
r=\sum_{i\in \Pi} h_i \otimes h_i + \sum_{\a \in \Delta^+} e_\alpha \otimes
e_{-\a}.
\label{eq:rmat}
\eeq
Since $\cT$ is a trivial Poisson submanifold, $\cP^{\rm Toda}$ inherits a
Poisson structure from the parent Poisson--Lie group. Consider now the (Lax) map 
\beq
\bary{rcrclcl}
L_{x, y} & : & \cP & \to & \cP^{\rm Toda} & & \\
& & (x,y) & \to & \prod_{i=1}^8 H_i(x_i) E_i H_i(y_i) E_{-i}
\eary
\label{eq:lax}
\eeq
Then the following proposition holds.
\begin{prop}[Fock--Goncharov, \cite{MR2263192}]
$L$ is an algebraic Poisson embedding into an open subset of $\cP^{\rm Toda}$.
\label{prop:fg}
\end{prop}
Similar considerations apply to the affine case. In $(\bbC^\star)^{18}\simeq
(\bbC_x^\star)^{9}\times (\bbC_y^\star)^{9}$ with exponentiated linear
co-ordinates $(x_0, x_1, \dots, x_8; y_0, y_1, \dots, y_8)$ and log-constant
Poisson bracket
\beq
\{x_i, y_j\}_{\widehat{\cG}}= \mathscr{C}^{\mathfrak{g}^{(1)}}_{ij} x_i y_j,
\label{eq:ppaff}
\eeq
consider the
hypersurface $\widehat{\cP}\triangleq \mathbb{V}(\prod_{i=0}^8 (x_i y_i)^{\mathfrak{d}_i}-1)$, where $\{\mathfrak{d}_i\}_i$ are
the Dynkin labels of Fig.~\ref{fig:dynke8}. Since
$\mathrm{Ker}\mathscr{C}^{\mathfrak{g}^{(1)}}=1$, $\widehat{\cP}$ is not symplectic anymore, unlike the
simple Lie group case above; in
particular, the regular function
\beq
\cO(\widehat{\cP}) \ni  \aleph \triangleq \prod_{i=0}^8 x_i^{\mathfrak{d}_i}=\prod_{i=0}^8 y_i^{-\mathfrak{d}_i}
\label{eq:aleph}
\eeq
is a Casimir of the bracket \eqref{eq:ppaff}, and it foliates $\widehat{\cP}$ symplectically. As before,
there is a double coset decomposition of $\widehat{\cG}$ indexed by
pairs of elements of the affine Weyl group $\widehat{\cW}$, and a distinguished
cell $\cC_{\bar w, \bar w}$ labelled by the element $\bar w$ corresponding to the longest cyclically irreducible word in the generators of
$\widehat{\cW}$. Projecting to trivial central (co)extension
\beq
\widehat{G} \ni g=M(\lambda) q^{\lambda \de \lambda} \stackrel{\pi}{\to} M(\lambda) \in \mathrm{Loop}(\cG)
\eeq
induces a Poisson structure on the projections of the cells $\cC_{w_+,w_-}$ (and
in particular $\cC_{\bar w, \bar w}$), as well as their quotients
$\cC_{w_+,w_-}/\mathrm{Ad} \cT $ by the adjoint
action of the Cartan torus,  upon lifting to the loop group the Poisson--Lie structure of the
non-dynamical r-matrix \eqref{eq:pl}. I will write $\widehat{\cP}^{\rm Toda}
\triangleq \pi(\cC_{\bar w, \bar w})/\mathrm{Ad}\cT$ for the resulting Poisson
manifold; and
we have now that \cite{Fock:2014ifa} $$\mathrm{dim}_\bbC\cP^{\rm Toda}=2~ 
\mathrm{length}(\bar w) -1=2 \times 9 -1 = 17.$$ Consider now the morphism
\beq
\bary{rcrclcl}
\widehat{L_{x, y}}(\lambda) & : & \widehat{\cP} & \to & \widehat{\cP}^{\rm Toda} & & \\
& & (x,y) & \to & \prod_{i=0}^8 \widehat{H_i}(x_i) \widehat{E_i} \widehat{H_i}(y_i) \widehat{E_{-i}}.
\eary
\label{eq:laxaff}
\eeq
It is instructive to work out explicitly the form
of the loop group element corresponding to $\widehat{L_{x, y}}$; we have 
\bea
\widehat{L_{x, y}}(\lambda) &=& \prod_{i=0}^8 \widehat{H_i}(x_i) \widehat{E_i}
\widehat{H_i}(y_i) \widehat{E_{-i}}\nn \\
&=& E_0(\lambda/y_0) E_{\bar 0}(\lambda) \l[\prod_{i=0}^8 (x_i
  y_i)^{\mathfrak{d}_i}\r]^{\lambda \rd \lambda}\prod_{i=1}^8 H_i(x_i) E_i
H_i(y_i) E_{-i} \nn \\
&=& E_0(\lambda/y_0) E_{\bar 0}(\lambda) \prod_{i=1}^8 H_i(x_i) E_i
H_i(y_i) E_{-i}
\label{eq:laxaff2}
\eea
where in moving from the first to the second line we have expanded $g \in
\widehat{G}$ as a linear $q-$differential operator and grouped together all
the multiplicative $q-$shifts, and then used that $\prod_{i=0}^8 (x_i
  y_i)^{\mathfrak{d}_i}=1$ on $\widehat{\cP}$, which gives indeed an element with trivial
  co-extension. The same line of reasoning of \cref{prop:fg} shows that
  $\widehat{L}$ is a Poisson monomorphism.

\subsection{Dynamics}

For functions $H_1, H_2 \in \cO(\widehat{\cP}^{\rm Toda})$, the Poisson
bracket \eqref{eq:pl} reads, explicitly,
\beq
\{H_1, H_2\}_{\rm PL}=-\frac{1}{2}\sum_{\a\in\Delta^+}\l[ L_{e_\alpha} H_1 R_{e_{-\alpha}}
H_2 - (1 \leftrightarrow 2)\r]
\eeq
where $L_X$ (resp. $R_X$) denotes the left (resp. right) invariant vector
field generated by $X\in T_e G\simeq \mathfrak{g}$. Then a complete system of
involutive Hamiltonians for \eqref{eq:pl} on $\cG$, and any Poisson Ad-invariant
submanifold such as $\cP^{\rm Toda}$, is given by Ad-invariant functions on
the group -- or equivalently, Weyl-invariant functions on $\cT$. This is a
subring of $\cO(\cP^{\rm Toda})$ generated by the regular fundamental
characters
\beq
H_i(g)=\chi_{\rho_i}(g), \quad i=1,\dots, 8
\label{eq:hi}
\eeq
where $\rho_i$ is the irreducible representation having the $i^{\rm th}$
fundamental weight $\omega_i$ as its highest weight. In the affine case the
same statements hold, with the addition of the central Casimir $\aleph$ in
\eqref{eq:aleph}. The Lax maps \eqref{eq:lax}, \eqref{eq:laxaff} then
pull-back this integrable dynamics to the respective tori $\cP$ and
$\widehat{\cP}$. Fixing a faithful representation $\rho \in R(\cG)$ (say, the
adjoint), the same dynamics on $\cP^{\rm Toda}$ and $\widehat{\cP}^{\rm Toda}$ takes the form of isospectral flows
\cite[Sec.~3.2-3.3] {MR1995460}
\bea
\frac{\de \rho(L)}{\de t_i} = \l\{\rho(L), H_i(L)\r\}_{\rm PL}= \l[\rho(L), (P_i(\rho(L)))_+\r]\\
\frac{\de \rho(\widehat{L})}{\de t_i} = \l\{\widehat{\rho(L)}, H_i(\widehat{L})\r\}_{\rm PL}= \l[\widehat{\rho(L)}, (P_i(\widehat{\rho(L)}))_+\r]
\label{eq:laxeq}
\eea
where $P_i\in \bbC[x]$ is the expression of the Weyl-invariant Laurent polynomial
$\chi_{\omega_i}\in \cO(\cT)^\cW$ in terms of power sums of the eigenvalues of 
$\rho(g)$, and $()_+:\cG\to \cB^+$ denotes the projection to the positive Borel.

\subsection{The spectral curve}
\label{sec:speccurve}
We henceforth consider the affine case only. Since \eqref{eq:laxeq} is isospectral, all functions of the spectrum
$\sigma(\rho(\widehat{L}))$ of $\rho(\widehat{L})$ are integrals of motion. A central
role in our discussion will be played by the spectral invariants constructed
out of elementary symmetric polynomials in the eigenvalues of $\widehat{L}$, for the case in
which $\rho= \mathfrak{g}$ is the adjoint representation, that is, is the minimal-dimensional
non-trivial irreducible representation of $\cG$. I write
\beq
\Xi_{\mathfrak{g}}(\mu,\lambda) \triangleq \det_{\mathfrak{g}} \l(\widehat{L}(\lambda)-\mu
\mathbf{1}\r)
\label{eq:xi}
\eeq
for the characteristic polynomial of $\widehat{L}$ in the adjoint, thought of
as a 2-parameter family of maps $\Xi_{\mathfrak{g}}(\mu, \lambda): \widehat{\cP}\to \bbC$. It
is clear by \eqref{eq:laxeq} that $\Xi_{\mathfrak{g}}(\mu,\lambda)$ is an integral of motion for all $(\mu,
\lambda)$, and so is therefore the plane curve in $\mathbb{A}^2$ given by its
vanishing locus $\mathbb{V}(\Xi_{\mathfrak{g}})$. \\

We will be interested in expanding out the flow invariant \eqref{eq:xi} as an
explicit polynomial function of the basic integrals of motion \eqref{eq:hi}. I will do
so in two steps: by determining the dependence of \eqref{eq:hi} on the
spectral parameter when $g=\widehat{L}(\lambda)$ in \eqref{eq:laxaff2} and
\eqref{eq:hi}, and by computing the dependence of $\Xi_{\mathfrak{g}}(\mu,\lambda)$ on the
basic invariants \eqref{eq:hi}. We have first the following
\begin{lem}
$H_i(\widehat{L})$, $i=1,\dots, 8$ are Laurent polynomials in $\lambda$, which
  are constant except for $i=3$. In particular, there exist functions $u_i \in \cO(\widehat{\cP})$ such that 
\beq
H_i(\widehat{L})=u_i(x,y)-\delta_{i,3}\l(\lambda'+\frac{\aleph^2}{\lambda'}\r)
\label{eq:ui}
\eeq
with $\de_{x_0}u_i(x,y)=\de_{y_0}u_i(x,y)=0$ and $\lambda'=\lambda y_0 \aleph^2$.
\label{lem:uispec}
\end{lem}
\begin{proof}[Proof (sketch)]
The proof follows from a lengthy but straightforward calculation from
\eqref{eq:laxaff2}. Since we are looking at the adjoint representation, explicit matrix expressions for the Chevalley
generators \eqref{eq:chevalley} can be computed by systematically reading off
the structure constants in \eqref{eq:chevalley}, the full set of which for all
the $\dim \mathfrak{g}=248$ generators of the algebra is determined from the canonical assignment of
signs to so-called extra-special pairs of roots reflecting the ordering of simple
roots within $\Pi$ (see \cite{MR2047097} for details). The resulting
$248\times 248$ matrix in \eqref{eq:laxaff2}, with coefficients depending on
$(\lambda, x,y)$, is moderately sparse,
which allows to compute power sums of its eigenvalues efficiently. We can then
show
from a direct calculation that \eqref{eq:ui} holds for $i=3,\dots, 7$ the
relations in $R(\cG)$ 
\bea
\rho_{\omega_7} &=& \mathfrak{g}, \quad 
\rho_{\omega_6} = \wedge^2\mathfrak{g}\ominus \mathfrak{g},\quad
\rho_{\omega_5} = \wedge^3\mathfrak{g}\ominus \rho_{\omega_6} \oplus
\wedge^2\mathfrak{g}\ominus \mathfrak{g}\otimes \mathfrak{g} \nn \\
\rho_{\omega_4} &=& \wedge^4\mathfrak{g}\ominus \rho_{\omega_6} \otimes \mathfrak{g} \oplus
\rho_{\omega_6}, \quad
\rho_{\omega_3} = \wedge^5\mathfrak{g} \oplus \rho_{\omega_5}\otimes (\mathbf{1}
\ominus \mathfrak{g})\ominus \mathrm{Sym}^3 \mathfrak{g}
\label{eq:char37}
 \eea
which are an easy consequence of the decomposition into irreducibles of
$\wedge^n\mathfrak{g}$, $\mathrm{Sym}^n\mathfrak{g}$, and their tensor powers
for $n\leq 5$,\footnote{We used {\it Sage} for the decomposition of the
  plethysms, and {\it LieART} for that of the tensor powers.} and the use of
Newton identities relating power sum polynomials (i.e. traces of powers) to
elementary/complete symmetric polynomials (i.e. antisymmetric/symmetric
traces). A little more work is required to show that $u_i$ is constant for
$i=1,2,8$; this uses the more complicated character relations \eqref{eq:p6}-\eqref{eq:p7} of \cref{sec:charE8}. The final result is \eqref{eq:ui}.
%
%\beq
%\bary{rcl}
%H_7(g) & \triangleq & \mathrm{Tr}_\mathfrak{g}(g),  \quad  H_6(g) 
%=  \frac{1}{2} \l(\mathrm{Tr}(g)^2-\mathrm{Tr}(g^2)\r)-\mathrm{Tr}(g), \nn \\
%H_5(g) &=& \frac{1}{6} \l(\mathrm{Tr}_{\mathfrak{g}}(g)^3-6 \mathrm{Tr}_{\mathfrak{g}}(g)^2-3
%    (\mathrm{Tr}_{\mathfrak{g}}(g^2)-2) \mathrm{Tr}_{\mathfrak{g}}(g)+2 \mathrm{Tr}_{\mathfrak{g}}(g^3)\r)\nn \\
%H_4(g) &=& 
%1/24 (\mathrm{Tr}_{\mathfrak{g}}(g)^4-12 \mathrm{Tr}_{\mathfrak{g}}(g)^3-6 (\mathrm{Tr}_{\mathfrak{g}}(g^2)-6)
%      \mathrm{Tr}_{\mathfrak{g}}(g)^2+4 (3 \mathrm{Tr}_{\mathfrak{g}}(g^2)+2 \mathrm{Tr}_{\mathfrak{g}}(g^3)-6)
%            \mathrm{Tr}_{\mathfrak{g}}(g)\nn \\ &+& 3 (\mathrm{Tr}_{\mathfrak{g}}(g^2)^2-4 \mathrm{Tr}_{\mathfrak{g}}(g^2)-2
%                  \mathrm{Tr}_{\mathfrak{g}}(g^4))) \nn \\
%H_3(g) &=& 1/120 (\mathrm{Tr}_{\mathfrak{g}}(g)^5-20
%\mathrm{Tr}_{\mathfrak{g}}(g)^4-10 (\mathrm{Tr}_{\mathfrak{g}}(g^2)-12)
%\mathrm{Tr}_{\mathfrak{g}}(g)^3+20 (3
%\mathrm{Tr}_{\mathfrak{g}}(g^2)+\mathrm{Tr}_{\mathfrak{g}}(g^3)-12)
%\mathrm{Tr}_{\mathfrak{g}}(g)^2\nn \\ &+& 5 (3 \mathrm{Tr}_{\mathfrak{g}}(g^2)^2-24 \mathrm{Tr}_{\mathfrak{g}}(g^2)-8 \mathrm{Tr}_{\mathfrak{g}}(g^3)-6 \mathrm{Tr}_{\mathfrak{g}}(g^4)+24) \mathrm{Tr}_{\mathfrak{g}}(g)-20 \mathrm{Tr}_{\mathfrak{g}}(g^2) \mathrm{Tr}_{\mathfrak{g}}(g^3)+24 \mathrm{Tr}_{\mathfrak{g}}(g^5))
%\eary
%\eeq
%
\end{proof}
%Write $\cU$ for the image of $(H_1,\dots,H_8) : \widehat{\cP}^{\rm Toda} \to \bbC^8$
%This gives a family of algebraic curves
It is immediately seen from \eqref{eq:ui} that $u_i(x,y)$ are involutive, independent
integrals of motion; they are equal to the fundamental Hamiltonians
\eqref{eq:hi} for $i\neq 3$, and for $i=3$ they are a
$\bbC[\lambda,\lambda^{-1}]$ linear combination of $H_3$ and the Casimir
$\aleph$. Denote now by $\cU=\vec u(\widehat{\cP})\subset \bbC^8$ the image of
$\widehat{\cP}$ under the map $\vec u = (u_i)_i:\widehat{\cP}\to \bbC^8$ . It is clear
from \eqref{eq:xi} and \eqref{eq:ui} that $\Xi_{\mathfrak{g}} : \widehat{\cP} \to
\bbC[\lambda',\aleph^2 \lambda^{-1},\mu]$ factors through $\vec u$ and a map
$\mathfrak{p}=\sum_k(-)^k \mathfrak{p}_k : \cU \to \bbC[\lambda',\aleph^2 \lambda'^{-1},\mu] $ given by the decomposition
of the characteristic polynomial into fundamental characters:
\bea
\Xi_{\mathfrak{g}}(\lambda,\mu) &=& \sum_{k=0}^{248}(-\mu)^k \chi_{\wedge^{k}
  \mathfrak{g}}\widehat{L_{x, y}}(\lambda)  \nn \\
&=& \sum_{k=0}^{124} (-)^k \mathfrak{p}_k\l(u_1, u_2,
u_3+\l(\lambda'+\aleph^2/\lambda'\r), u_4, \dots, u_8\r) \l(\mu^k+\mu^{248-k}\r)
\label{eq:xi2}
\eea
where the reality of the adjoint representation has been used. Here $\mathfrak{p}_k$ is the polynomial relation of formal characters
\beq
\chi_{\wedge^k \mathfrak{g}} = \mathfrak{p}_k(\chi_{\omega_1}, \dots,
\chi_{\omega_8}) \in \bbZ[\chi_{\omega_1}, \dots,
\chi_{\omega_8}]\simeq R(\cG)
\label{eq:pk}
\eeq
evaluated at the group element $\widehat{L}$. For fixed $(u_i)_i\in \cU $ and
$\aleph\in \bbC$, the vanishing locus $\mathbb{V}(\Xi_\mathfrak{g})$ of the
characteristic polynomial is a complex algebraic curve in $\bbC^2$; I shall
write $\mathscr{B}_{\mathfrak{g}} \triangleq \cU \times \bbA^1$ for the variety of
parameters this polynomial will depend on. Even
though $\mathfrak{g}$ is irreducible, the curve $\mathbb{V}(\Xi_\mathfrak{g})$ is
reducible since $\Xi_\mathfrak{g}$ is. Indeed, conjugating $\widehat{L}$ to an
element $\exp{l} \in \cT$ in the Cartan torus, $l\in \mathfrak{h}$, we have
\bea
\Xi_{\mathfrak{g}}(\lambda, \mu) &=& \det_{\mathfrak{g}} \l(\widehat{L}-\mu
\mathbf{1}\r) = \prod_{\a \in \Delta}
\l(\exp(\a(l)) - \mu\r) \nn \\
&=& (\mu-1)^8
\prod_{\a \in \Delta_+}
\l(\exp(\a(l)) - \mu\r)\l(\exp(-\a(l)) - \mu\r)
\label{eq:xi3}
\eea
For a general representation $\rho$, we would obtain as many irreducible components
as the number of Weyl orbits in the weight system. When $\rho=\mathfrak{g}$,
and for this case alone, we have only one non-trivial orbit, as well as eight
trivial orbits corresponding to the zero roots. I will factor out the trivial
component corresponding to zero roots by writing $\Xi_{\mathfrak{g},\mathrm{red}}=\Xi_{\mathfrak{g}}/(\mu-1)^8$.

\begin{defn}
For $(u, \aleph) \in \mathscr{B}_{\mathfrak{g}}$, let $\Gamma_{(u, \aleph)}$ be
the normalisation of the projective closure of $\mathrm{Spec}
\frac{\bbC[\lambda,\mu]}{\bra \Xi_{\mathfrak{g},\mathrm{red}} \ket}$. 
We call the corresponding family of plane curves $\pi:\mathscr{S}_\mathfrak{g}
\to \cU\times \bbC$, 
\beq
\xymatrix{ \Gamma_{u, \aleph} \ar[d]  \ar@{^{(}->}[r]&
  \mathscr{S}_{\mathfrak{g}}\ar[d]^\pi\ar^{(\lambda,\mu)~}[r]  &  \bbP^1\times \bbP^1  \\
                     (u, \aleph)  \ar@{^{(}->}[r]^{\rm pt} 
  \ar@/^1pc/[u]^{P_{i}} &  
\mathscr{B}_{\mathfrak{g}} \ar@/^1pc/[u]^{\Sigma_i} & 
}
\label{eq:diagsc}
\eeq
the family of {\rm spectral curves} of
the $\widehat{\mathrm{E}}_8$ relativistic Toda chain in the adjoint
representation. In \eqref{eq:diagsc}, $P_i$ are the points added in the
compactification of $\bbV(\Xi_{\mathfrak{g},\mathrm{red}})$ (see
\cref{rmk:ptsinf,tab:ptsinf} below) and
$\Sigma_{i}$ are the sections marking them.
\label{def:sc}
\end{defn}

As is known in the more familiar setting of $\widehat{\cG}=\widehat{\rm sl}_N$, and as we
will discuss in \cref{sec:actangl}, spectral curves are a key ingredient in the integration of the Toda flows.
Knowledge of the spectral curves is encoded into knowledge of the 
character relations \eqref{eq:pk}, which grant access to the explicit form of
the polynomial $\Xi_{\mathfrak{g}, \rm red}$ to spectral curves for arbitrary moduli $(u, \aleph)$: the
description of the spectral curves is then reduced to the purely
representation-theoretic problem of determining these relations. 

In view of this, denote $\theta_\bullet \triangleq \chi_{\rho_{\omega_\bullet}}$, $\phi_\bullet
\triangleq \chi_{\wedge^\bullet \mathfrak{g}}$. What we are looking
for are explicit polynomials 
\beq
\mathfrak{p}_k(\theta) = \sum_{I\in M} n_{I,k} \prod_{j=1}^8 \theta_j^{d_j^{(I)}}
\label{eq:pkgen}
\eeq
where the index $I$ runs over a suitable finite set $M \ni (d_1^{(I)},
\dots,d_8^{(I)})$, $M \subset \bbN^8$, and $n_{I,k}\in \bbZ$. Since what we
are ultimately interested in is the reduced characteristic curve
$\Gamma_{u, \aleph}$, it suffices to compute $\{ n_{I,k}\}$ (and hence $\mathfrak{p}_k$) for $k\leq 120$.
\begin{claim}[\cite{E8comp}]
We determine  $\{ n_{I,k}\in \bbZ\}$ for all $I\in M$, $k\leq 120$.
\label{claim:E8}
\end{claim}
This is the result of a series of computer-assisted calculations, of
independent interest and whose
details will appear elsewhere \cite{E8comp}, but for which I provide a fairly
comprehensive summary in \cref{sec:charE8}. For the sake of example, we obtain for the first few values of $k$,
\bea
\mathfrak{p}_6 &=& 
\label{eq:p6}
\theta_7 \theta_1^2-\theta_1^3-\theta_6 \theta_1^2-\theta_1^2+2 \theta_7^2 \theta_1+2 \theta_2 \theta_1-\theta_4 \theta_1+\theta_5 \theta_1-\theta_6
\theta_1+\theta_6 \theta_7 \theta_1-2 \theta_7 \theta_1-\theta_8 \theta_1-\theta_6^2\nn \\
&+& \theta_6 \theta_7^2-\theta_7^2-\theta_3+\theta_2 \theta_6+\theta_5 \theta_6+\theta_2 \theta_7+\theta_4
   \theta_7-2 \theta_6 \theta_7+\theta_2 \theta_8-\theta_6 \theta_8-\theta_7 \theta_8, \\
\mathfrak{p}_7 &=& 
\label{eq:p7}
\theta_7^4+2 \theta_1 \theta_7^3-4 \theta_7^3+\theta_1^2 \theta_7^2-6 \theta_1 \theta_7^2+2 \theta_2 \theta_7^2+2 \theta_5 \theta_7^2-2
\theta_6 \theta_7^2+\theta_7^2-2 \theta_1^3 \theta_7-\theta_1^2 \theta_7+4 \theta_1 \theta_7\nn \\
&+& 4 \theta_1 \theta_2 \theta_7 -\theta_3 \theta_7+\theta_4
   \theta_7+2 \theta_1 \theta_5 \theta_7-4 \theta_5 \theta_7+\theta_1 \theta_6 \theta_7+4 \theta_6 \theta_7-\theta_1 \theta_8 \theta_7-\theta_6 \theta_8
   \theta_7 \nn \\
&+&\theta_1^3+2 \theta_1^2+\theta_2^2+\theta_5^2+\theta_1 \theta_6^2+\theta_6^2+\theta_1 \theta_8^2+\theta_1-\theta_1 \theta_2-2 \theta_1
   \theta_3-\theta_3+\theta_1 \theta_4+\theta_4-2 \theta_1^2 \theta_5 \nn \\ &+& 2 \theta_2 \theta_5-\theta_5+2 \theta_1^2 \theta_6+3 \theta_1
   \theta_6-\theta_2 \theta_6+\theta_4 \theta_6-2 \theta_5 \theta_6+\theta_6-\theta_1^2 \theta_8-\theta_1 \theta_8+\theta_2 \theta_8  \nn \\ 
&-& 2 \theta_2 \theta_7-\theta_8 \theta_7+2 \theta_7-\theta_4 \theta_8-\theta_1 \theta_6
   \theta_8-3 \theta_1 \theta_5\\
\mathfrak{p}_8 &=& 
\theta_8-\theta_1^4-\theta_6 \theta_1^3+2 \theta_7^2 \theta_1^2+3 \theta_2 \theta_1^2-\theta_4 \theta_1^2+\theta_5 \theta_1^2+\theta_6
\theta_1^2-\theta_7 \theta_1^2-2 \theta_7 \theta_8 \theta_1^2-2 \theta_7^3 \theta_1+\theta_6^2 \theta_1\nn \\ &+& 3 \theta_6 \theta_7^2 \theta_1-3 \theta_3 \theta_1+2 \theta_4
   \theta_1-\theta_5 \theta_1+2 \theta_2 \theta_6 \theta_1+\theta_5 \theta_6 \theta_1+2 \theta_6 \theta_1+2 \theta_2 \theta_7 \theta_1+\theta_4 \theta_7 \theta_1 \nn \\ &+& 2 \theta_6 \theta_7 \theta_1+5 \theta_7 \theta_1+\theta_7^2 \theta_8 \theta_1+\theta_2 \theta_8 \theta_1-2 \theta_5 \theta_8 \theta_1-3 \theta_7 \theta_8
   \theta_1+\theta_8 \theta_1-2 \theta_7^4+\theta_6 \theta_7^3+\theta_7^3\nn \\ &-& 2 \theta_5^2+\theta_6^2-\theta_2
   \theta_7^2+2 \theta_4 \theta_7^2-4 \theta_5 \theta_7^2+3 \theta_7^2-\theta_8^2+\theta_2+\theta_3+\theta_2 \theta_4-\theta_4-2 \theta_2 \theta_5\nn \\ &-& 2 \theta_3
   \theta_6+ \theta_4 \theta_6-\theta_5 \theta_6+\theta_6^2 \theta_7+\theta_2 \theta_7-3 \theta_3 \theta_7+\theta_5 \theta_7+2 \theta_2 \theta_6
   \theta_7+\theta_5 \theta_6 \theta_7+\theta_6 \theta_7-2 \theta_7\nn \\ &-& \theta_2 \theta_8-3 \theta_3 \theta_8+2 \theta_4 \theta_8-3 \theta_5 \theta_8+2 \theta_6 \theta_8+3 \theta_2
   \theta_7 \theta_8+2 \theta_6 \theta_7 \theta_8+4 \theta_7 \theta_8,  \nn
   \\ 
&-& \theta_7^2 \theta_1+2
   \theta_4 \theta_5+2 \theta_5 -3
   \theta_5 \theta_7 \theta_1+\theta_8^3-\theta_2^2-4 \theta_7^2 \theta_8\\
\mathfrak{p}_9 &=& 
2 \theta_1^2 \theta_7^3-2 \theta_7^4-7 \theta_1 \theta_7^3-3 \theta_6 \theta_7^3+4 \theta_7^3-2 \theta_1^3 \theta_7^2-\theta_1^2
\theta_7^2+\theta_6^2 \theta_7^2+\theta_8^2 \theta_7^2+10 \theta_1 \theta_7^2+4 \theta_1 \theta_2 \theta_7^2\nn
\\ &-& 2 \theta_3 \theta_7^2+\theta_4
   \theta_7^2+2 \theta_1 \theta_5 \theta_7^2-5 \theta_5 \theta_7^2+2 \theta_1 \theta_6 \theta_7^2+2 \theta_6 \theta_7^2+\theta_6 \theta_8
   \theta_7^2-\theta_8 \theta_7^2-2 \theta_7^2+2 \theta_1^3 \theta_7+\theta_1^2 \theta_7\nn \\ &+& \theta_2^2 \theta_7+2 \theta_1 \theta_6^2 \theta_7+6 \theta_6^2 \theta_7+\theta_1
   \theta_8^2 \theta_7-2 \theta_8^2 \theta_7-\theta_1 \theta_7-3 \theta_1 \theta_2 \theta_7+5 \theta_2 \theta_7-4 \theta_1 \theta_3 \theta_7+2 \theta_3
   \theta_7\nn \\
&+&2 \theta_1 \theta_4 \theta_7
-2 \theta_1^2 \theta_5 \theta_7-6 \theta_1 \theta_5 \theta_7+2 \theta_2 \theta_5 \theta_7+5 \theta_5 \theta_7+\theta_1^2
   \theta_6 \theta_7+5 \theta_1 \theta_6 \theta_7+2 \theta_2 \theta_6 \theta_7+3 \theta_4 \theta_6 \theta_7\nn \\
&-& 4 \theta_5 \theta_6 \theta_7+4 \theta_6 \theta_7-2 \theta_1 \theta_8 \theta_7+\theta_4 \theta_8 \theta_7-2 \theta_5 \theta_8 \theta_7-2 \theta_1 \theta_6 \theta_8 \theta_7-\theta_6 \theta_8 \theta_7+\theta_8 \theta_7-\theta_1^3+2
   \theta_6^3\nn \\
&+& 2 \theta_4^2-\theta_5^2+2 \theta_1 \theta_6^2+\theta_2 \theta_6^2+3 \theta_6^2-\theta_1 \theta_8^2+\theta_2 \theta_8^2+\theta_5
   \theta_8^2+2 \theta_1 \theta_2+\theta_2+\theta_1^2 \theta_3+\theta_1 \theta_3-2 \theta_2 \theta_3\nn \\
&-& \theta_1 \theta_4-2 \theta_2
   \theta_4-\theta_4+\theta_1^2 \theta_5+2 \theta_1 \theta_5-2 \theta_2 \theta_5-3 \theta_3 \theta_5-\theta_1^2 \theta_6+\theta_1 \theta_6+\theta_1 \theta_2
   \theta_6+4 \theta_2 \theta_6-2 \theta_3 \theta_6\nn \\
&+& \theta_1 \theta_4 \theta_6-\theta_4 \theta_6-2 \theta_1 \theta_5 \theta_6+\theta_5 \theta_6+\theta_6-\theta_1^3 \theta_8+2 \theta_1
   \theta_2 \theta_8-\theta_3 \theta_8-2 \theta_1 \theta_4 \theta_8+\theta_4 \theta_8+\theta_1 \theta_5 \theta_8\nn \\ &-& \theta_1^2 \theta_6 \theta_8+\theta_2 \theta_6 \theta_8-\theta_5 \theta_6 \theta_8-\theta_1^2+\theta_1^2 \theta_4-4 \theta_2 \theta_7^2-\theta_4 \theta_7.
\label{eq:p9}
\eea

\begin{figure}[t]
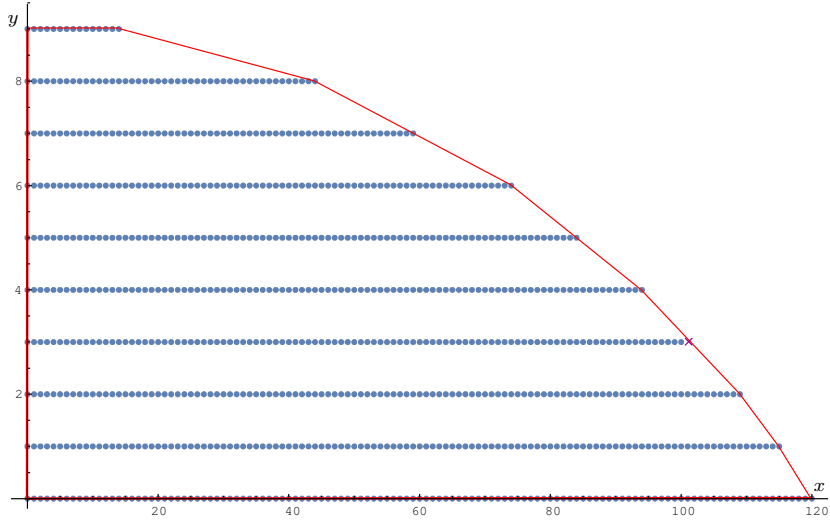
\caption{The Newton polygon of $\Xi''_{\mathfrak{g}, \rm  red}$ (in
  red); blue spots depict monomials in $\Xi''_{\mathfrak{g}, \rm
    red}$ with non-zero coefficients; the purple cross marks the vanishing of the coefficient of
  $x^{101} y^{3}$ on the boundary of the polygon.}
\label{fig:npol}
\end{figure}

\subsubsection{Genus, ramification points, and points at infinity}

\label{sec:curvdet}

The curves $\Gamma_{u, \aleph}$ have two obvious involutions, coming from
the $\bbZ_2 \times \bbZ_2$ symmetry \eqref{eq:xi2} of the reduced characteristic polynomial
$\Xi_{\mathfrak{g}, \rm red}$,
\beq
\Xi_{\mathfrak{g}, \rm red}(\lambda', \mu)=\Xi_{\mathfrak{g}, \rm
  red}(\aleph/\lambda', \mu), \quad \Xi_{\mathfrak{g}, \rm red}(\lambda',
\mu)=\mu^{240} \Xi_{\mathfrak{g}, \rm red}(\lambda', 1/\mu).
\eeq
This realises $\Gamma_{u, \aleph} \stackrel{y}{\to} \Gamma'_{u,
  \aleph} \stackrel{x}{\to} \Gamma''_{u,
  \aleph}$, where $x=\mu+\mu^{-1}$, $y=\lambda+\aleph \lambda^{-1}$, as a branched fourfold cover of a curve $\Gamma''_{u,
  \aleph} \triangleq \{\Xi''_{\mathfrak{g}, \rm red}(y,x)=0\}$,
 so that
\beq
\Xi_{\mathfrak{g}, \rm red}(\lambda',\mu)=:\mu^{120}  \Xi''_{\mathfrak{g}, \rm
  red}\l(\lambda'+\frac{\aleph}{\lambda'}, \mu+\frac{1}{\mu}\r)
\eeq
We see from \eqref{eq:xi2} and \eqref{eq:djdeg} that $\deg_y\Xi''_{\mathfrak{g}, \rm
  red}(y,x)=9$, $\deg_x \Xi''_{\mathfrak{g}, \rm  red}=120$. The
Newton polygon of $\Xi''_{\mathfrak{g}, \rm  red}$ is depicted in
\cref{fig:npol}. By way of example, some of the simplest coefficients on the
boundary are given by:
\bea
\label{eq:y9}
[y^9]\Xi''_{\mathfrak{g}, \rm
  red} &=& (x+1)^3 (x+2) \left(-1+x+x^2\right)^5,  \\
\l\{\l[x^{\deg_x \l[y^i\r] \Xi''_{\mathfrak{g}, \rm  red}}\r] \Xi''_{\mathfrak{g}, \rm
  red}\r\}_{i=0}^8 &=& \{1,-1,-1,-3 u_7-5,1,2,1,-2,1\}.
%\Xi''_{\mathfrak{g}, \rm  red} &=&  1 +10 x+9 x^2-233 x^3-748
%x^4+1931 x^5+6 (u_7+1987) x^6+(60 u_7-1936) x^7\nn \\
%&+& (90 u_7-98558) x^8-2 (519 u_7+46312) x^9+(10
%   u_1-3982 u_7+491274) x^{10}+\cO\l(x^{11}\r)\nn \\
\label{eq:npolbdy}
\eea
%
%\begin{prop}

Let us now compute the genus of $\Gamma''_{u}$, $\Gamma'_u$ and $\Gamma_{u,\aleph}$. 

\begin{prop}
We have, for generic $(u, \aleph) \in \mathscr{B}_{\mathfrak{g}}$, 
\beq
g(\Gamma''_{u})=61, \quad g(\Gamma'_u)=128, \quad g(\Gamma_{u,\aleph})=495.
\label{eq:genred}
\eeq
\end{prop}

\begin{proof}
Since \cref{lem:uispec,claim:E8} determine the polynomial $\Xi''_{\mathfrak{g}, \rm red}$
completely, the calculation of the genus can be turned
into an explicit calculation of discriminants of $\Xi''_{\mathfrak{g}, \rm
  red}$; and because $\deg_y
\Xi''_{\mathfrak{g}, \rm red} \ll \deg_x
\Xi''_{\mathfrak{g}, \rm red}$, it is much easier to  start from the
$y$-discriminant. This is computed to be
\beq
\mathrm{Discr}_y \Xi''_{\mathfrak{g}, \rm red}=  (x+2)^4
\Delta_1(x) \Delta_2 (x)^2 \Delta_3(x)^2
\label{eq:discry}
\eeq
where $\deg \Delta_1= 133$, $\deg \Delta_2= 215$ and $\deg \Delta_3=392$. Call
$r_i^k$, $i=1, 2, 3$, $k=1,\dots, \deg\Delta_i$ the roots of $\Delta_i$. We 
can verify directly by substitution into $\Xi''_{\mathfrak{g},
  \rm red}$ that the roots $x=r_2^{k}$ and $x=r_3^k$ correspond to images on the
$x$-line of exactly one point with $\de_y \Xi''_{\mathfrak{g}}=0$,
which is always an ordinary double point. Similarly, we get that the roots $x=-2$ and $x=r_1^{k}$ correspond
in all cases to degree 2 ramification points; there are four of them lying
over $x=-2$. On the desingularised projective
curve $\Gamma''_{u}$, the nodes are resolved into pairs
of unramified points; and Puiseux expansions of $\Xi''_{\mathfrak{g},
  \rm red}$ at infinity  show that we have one extra point with degree 2 ramification
above $x=\infty$ (see below). By Riemann--Hurwitz, this gives
\beq
g(\Gamma''_{u})=1-\deg_y \Xi''_{\mathfrak{g},
  \rm red} +\frac{1}{2}\sum_{P | \rd x(P)=0} e_x(P)=1-9+\frac{133+1+4}{2}=61.
%\label{eq:genred}
\eeq
The genera of the branched double covers $x:\Gamma'_{u}\to \Gamma''_{u}$,
$y: \Gamma_{u,\aleph}\to \Gamma'_{u}$ follow from an elementary Riemann--Hurwitz calculation.
%These are, for generic $\aleph \in \bbC$,
%
%\beq
%g(\Gamma'_{u})=128
%=1+ 2 \times (61-1) + \overbrace{\frac{5}{2}^{x=\mu+\mu^{-1}=-2}} +\overbrace{\frac{9}{2}^{x=\mu+\mu^{-1}=2}}
%, \quad g(\Gamma_{u,\aleph})=495.
%=1+ 2 \times (128-1) + \overbrace{\frac{240}{2}^{y=\lambda+\aleph \lambda^{-1}=2\sqrt{\aleph}}} +\overbrace{\frac{9}{2}^{x=\mu+\mu^{-1}=-2\sqrt{\aleph}}}
%\label{eq:gen2}
%\eeq
%
\end{proof}

\begin{rmk}
It can readily be deduced from \eqref{eq:y9} that the smooth completion $\Gamma''_{u}$ is obtained
topologically by adding 12 points at infinity $P''_i$; their relevant properties are shown in
\cref{tab:ptsinf}. Their pre-images in $\Gamma'_u$ and $\Gamma_{u,\aleph}$
will be labelled $P'_{k}$ and $P_{j}$ respectively, $k=1,\dots, 23$ (notice
that $P_1''$ is a branch point of $x:\Gamma'_u\to \Gamma''_u$), $j=1,\dots, 46$.
\label{rmk:ptsinf}
\end{rmk}

\begin{table}[!h]
\begin{tabular}{|c|c|c|c|}
\hline
$i$ & $x(P''_i)$ & $e_y(P''_i)$ & $-\mathrm{ord} y_{P''_i}$ \\
\hline
1 & $-2$ & 1 & 1 \\
\hline
2 & $-1$ & 1 & 3 \\
\hline
3 & $-\phi$ & 1 & 5 \\
\hline
4 & $\phi^{-1}$ & 1 & 5 \\
\hline
5 & $\infty$ & 1 & 5 \\
\hline
6 &  $\infty$ & 1 & 6 \\
\hline
7 &  $\infty$ & 1 & 10 \\
\hline
8 &  $\infty$ & 1 & 10 \\
\hline
9 &  $\infty$ & 2 & 15 \\
\hline
10 &  $\infty$ & 1 & 15 \\
\hline
11 &  $\infty$ & 1 & 15 \\
\hline
12 &  $\infty$ & 1 & 30 \\
\hline
\end{tabular}

\caption{Points at infinity in $\Gamma''_{u}$. I indicate the value of their $x$-projection, their degree of
  ramification in $y$, and the order of the poles of
  $y$ in the second, third, and fourth column respectively. Here
  $\phi=\frac{\sqrt{5}+1}{2}$ is the golden ratio.}
\label{tab:ptsinf}
\end{table}

\subsection{Spectral vs parabolic vs cameral cover}
\label{sec:cameral}

The construction of $\Gamma_{u,\aleph}$ as the
non-trivial irreducible component of the vanishing locus of  \eqref{eq:xi}-\eqref{eq:xi3} realises it as a ``curve
of eigenvalues'': it is a branched cover of the space of spectral parameters
$\lambda \in \bbP^1\setminus \{0,\infty\}$ of the Lax matrix $\widehat{L_{x, y}}(\lambda)$; the
fibre over a $\lambda$-unramified point is given by the eigenvalues $\mu_\a$ of
$\widehat{L_{x, y}}(\lambda)$ that are different from 1. By
\eqref{eq:xi3}, each sheet $\mu_\a$ is labeled by a non-trivial root $\a \in \Delta^*$, and there
is an action of the Weyl group $\cW$ on $\Gamma_{u,\aleph}$ given by the
interchange of sheets corresponding to the 
Coxeter action of $\cW$ on the root space $\Delta$. 

Away from the ramification locus, this structure can be understood as
follows. Let $$\cG^{\rm red}=\{g \in \cG | \dim_\bbC C_\cG(g)=\mathrm{rank~} \cG=8\}$$ be the Zariski open set of regular
elements of $\cG$; I'll similarly append a superscript $\cT^{\rm red}$ for the
regular elements of $\cT$. Then the projection 
\bea
\pi : \cG/\cT \times  \cT^{\rm red} & \to & \cG^{\rm red}\nn \\
(g\cT, t) & \to & \mathrm{Ad}_g t 
\label{eq:adjbun}
\eea
is a principal $\cW$-bundle on $\cG^{\rm red}$, the fibre over a regular
element $g'$ being $N_\cT/\cT\simeq \cW$. We can pull this back via $\widehat{L_{x, 
    y}}$ to a $\cW$-bundle 
$$\Theta_{x, y}\triangleq\widehat{L_{x, 
    y}}^*(\cG/\cT \times  \cT^{\rm red})$$
over $\bbP^1\setminus D$, where $D=\widehat{L_{x, y}}^{-1}(\cG
\setminus \cG^{\rm red})$. This is a regular $\cW$-cover and each weight
$\omega\in \Lambda_w(\cG)$ determines a subcover $\Theta^\omega_{x, y} \simeq
\Theta_{x, y}/\cW_\omega$, where we quotient by the action of the stabiliser of $\omega$
by deck transformations. Write $\overline{\Theta_{x, y}}$ and
$\overline{\Theta^\omega_{x, y}}$ for the pull-back to
$\bbC^{\star}\simeq \bbP^1\setminus \{0,\infty\}$ of the closure of
\eqref{eq:adjbun} in $\cG/\cT \times  \cT \to \cG$. As in \cite{MR1397059}, we call
$\overline{\Theta_{x, y}}$ (resp. $\overline{\Theta^\omega_{x,
    y}}$) the {\it cameral} (resp. the $\omega${\it-parabolic}) cover
associated to $\widehat{L}_{x, y}$. 

Notice that when
$\omega=\omega_7=\alpha_0$ is the highest weight of the adjoint
representation, i.e. the highest (affine) root $\alpha_0$,
$\cW/\cW_{\alpha_0}$ is set-theoretically the root system of $\mathfrak{g}$, minus
the set of zero roots; the residual $\cW$ action is just the restriction to
$\Delta$ of the Coxeter action on $\mathfrak{h}^*$. In particular, we have
that $\overline{\Theta^\omega_{x, y}}$ is a degree
$|W/\cW_{\alpha_0}|=|\mathrm{Weyl}(\mathfrak{e}_8)/\mathrm{Weyl}(\mathfrak{e}_7)|=\frac{696729600}{2903040}=240$ branched cover of $\bbP^1$, with sheets
labelled by non-zero roots $\alpha\in \Delta^*$.
\begin{prop}
There is a birational map $\iota: \Gamma_{u, \aleph}  \dashrightarrow 
\overline{\Theta^{\omega_7}_{x, y}}$ given by an isomorphism 
\bea
\iota: \Gamma_{u, \aleph}\setminus \{\rd \mu=0\} &
\stackrel{\sim}{\rightarrow} & \Theta^{\omega_7}_{x, y} \nn \\
(\lambda, \mu_\alpha(\lambda)) & \to &  (\lambda, \alpha) 
\eea
away from the ramification locus of the $\lambda$-projection.
\label{prop:isopar}
\end{prop}
\begin{proof}
The proof is nearly verbatim the same as that of \cite[Thm.~13]{MR1182413}.
\end{proof}

From the proposition, we learn that a possible source of ramification
$\lambda: \Gamma_{u, \aleph} \to \bbP^1$ comes from the spectral values
$\lambda$ such that $\widehat{L}_{x, y}(\lambda)$ is an irregular
element of $\cG$; and from \eqref{eq:xi3}, we see that this happens if and only if $\alpha(l)=0$
for some $\alpha \in \Delta$. 

\begin{prop}
For generic $(u, \aleph)$, there are exactly 18 values of $\lambda$, 
\beq
b_i^{\pm}\triangleq\lambda(Q_i^{\pm}), \quad i=1, \dots, 9,
\eeq
such that $\widehat{L}_{x, y}(\lambda)$ is irregular, i.e. $\alpha(\log
\widehat{L}_{x, y}(\lambda))=0$ for some $\alpha\in\Delta$. Furthermore,  $\alpha\in \Pi$ is a simple root in
each of these cases. 
\label{prop:br}
\end{prop}

\begin{proof}
To see this, look at the base curve
$\Gamma''_u$. It is obvious that $\Xi_{\mathfrak g, \rm red}$ has only double zeroes at
$x=2$, since $\Xi_{\mathfrak g}$ has only double zeroes at $\mu=1$ as roots
come in (positive/negative) pairs in \eqref{eq:xi3}. For each of the nine points
$$\{Q''_i\}_{i=1}^9 \triangleq x^{-1}(2) \subset \Gamma''_u,$$ we compute from \cref{lem:uispec,claim:E8} that
\beq
e_x(Q''_i)=28
\label{eq:exPi}
\eeq
for all $i$. Calling $\alpha_i \in \Delta^+$ the
positive root
such that $\alpha_i \cdot l(\lambda(Q_i))=0$, we see from
\eqref{eq:xi3} that
\beq
e_x(Q''_i)=\mathrm{card}\l\{\beta\in\Delta^+ |
\beta-\a_i\in \Delta^+\r\}.
\eeq
It can be immediately verified that the right hand side is less than or equal
to 28, with equality iff $\alpha_i$ is simple. It is also clear
that there are no other points of ramification in the affine part of the
curve\footnote{In principle, from \eqref{eq:xi3}, this would be the case if
  $\a(l(\lambda))=\b(l(\lambda))$ for $\a-\b \notin \Delta$, leading to a double
  zero at $\mu\neq 1$ in \eqref{eq:xi3}, which we can't {\it
    a priori} rule out without appealing to 
  \eqref{eq:genred} and \eqref{eq:exPi} as we do below.} ; indeed, from \cref{tab:ptsinf}, we have that $e_x(\infty)=120-12=108$, and from
\eqref{eq:genred} we see that
\beq
60=g(\Gamma''_{u})-1=-\deg_x \Xi''_{\mathfrak{g},
  \rm red} +\frac{1}{2}\sum_{\rd x(P)=0} e_x(P)=-120+\frac{9 \times 28+108}{2}.
\eeq
As the covering map $x:\Gamma'_{u}\to\Gamma''_{u}$ is ramified at $x=2$, and $y:\Gamma_{u,\aleph}\to
\Gamma'_{u}$ is generically unramified therein for generic
$\aleph$, we have two preimages $Q_{i,\pm}$ on $\Gamma_{  u,\aleph}$ for each $Q''_{i}\in \Gamma''_{u}$.
\end{proof}

\section{Action-angle variables and the preferred Prym--Tyurin}
\label{sec:actangl}

Since \eqref{eq:hi} are a complete set of Hamiltonians in involution on the
leaves of the foliation of $\widehat{\cP}$ by level sets of $\aleph$, 
the compact fibres of the map $(u, \aleph): \widehat{\cP}\to \bbC^9$ are
isomorphic to a $\mathrm{rank}(\mathfrak{g})=8$-dimensional torus by the (holomorphic)
Liouville--Arnold--Moser theorem. A central feature of integrable systems of
the form \eqref{eq:laxeq} is an algebraic characterisation of 
their Liouville--Arnold dynamics, the torus in question being an Abelian
sub-variety of the Jacobian of $\Gamma_{u, \aleph}$. 

I determine in this section the action-angle integration explicitly for the
$\widehat{\mathrm{E}}_8$ relativistic Toda chain, which results in endowing
$\mathscr{S}_g$ with extra data \cite{Dubrovin:1992eu,Krichever:1992qe}, as
per the following
\begin{defn}
\label{defn:dk}
We call {\rm Dubrovin--Krichever data} %$DK^{\rm Toda}_{\mathfrak{g}}$ 
a $n$-tuple
$(\mathscr{F}, \mathscr{B}, \cE_1, \cE_2, \DD, \Lambda, \Lambda^L)$, with 
\bit
\item $\pi: \mathscr{F}\to\mathscr{B}$ a family of (smooth, proper) curves over an $n$-dimensional variety $\mathscr{B}$;
\item $\DD$ a smooth normal crossing divisor intersecting the fibres of $\pi$
  transversally;
\item meromorphic sections $\mathscr{E}_i \in H^0(\mathscr{F},
  \omega_{\mathscr{F}/\mathscr{B}}(\log \DD))$
  of the relative canonical sheaf having logarithmic poles along $\DD$;
\item $(\Lambda^L, \Lambda)$ a locally-constant choice of a marked subring
  $\Lambda$ of  the first homology of the fibres, and a Lagrangian sublattice
  $\Lambda^L$ thereof.
\eit
\end{defn}
\cref{defn:dk} isolates the extra data attached to spectral curves that were
identified in  \cite{Dubrovin:1992eu,Krichever:1992qe} (see also
\cite{Dubrovin:1994hc,Krichever:1997sq}) to provide the basic ingredients for
the construction of a Frobenius manifold structure on
$\mathscr{B}$ and a dispersionless integrable dynamics on its loop space
given by the Whitham deformation of the isospectral flows
\eqref{eq:laxeq}; the logarithm of those $\tau$-functions respects the type of
constraints that arise in theory with eight global supersymmetries (rigid
special K\"ahler geometry). These will be key aspects of the story to be
discussed in \cref{sec:applI,sec:applII}; in the language of
\cite{Dubrovin:1992eu}, when the pull-back of $\cE_1$ to the fibres of the
family is exact, the associated potential is the superpotential of the
Frobenius manifold, and $\cE_2$ its associated primitive differential. Now, \cref{claim:E8} and \cref{def:sc} gave us
$\mathscr{F}=\mathscr{S}_{\mathfrak{g}}$, $\mathscr{B}=\mathscr{B}_{\mathfrak{g}}$
already. We'll see, following \cite{Krichever:1997sq}, how the remaining ingredients are determined by
the Hamiltonian dynamics of \eqref{eq:laxeq}: this will culminate with the content of
\cref{cor:actangl}. I wish to add from the outset that the process leading up to \cref{cor:actangl}
relies on both common lore and results in the literature that are established
and known to the {\it cognoscenti}
at least for the non-relativistic limit; the gist of this section is to unify
several of these scattered ideas and adapt them to the setting at
hand. For the sake of completeness, I strived to provide precise pointers to places in the literature where similar arguments
have been employed.

\subsection{Algebraic action-angle integration}

From now until the end of this section, I will be sitting at a generic point
$(x,y)\in \widehat{\cP}$, and correspondingly, smooth moduli
point $(u, \aleph) \in \mathscr{B}_{\mathfrak{g}}$. As is the case for the ordinary periodic Toda chain with $N$ particles, and
for initial data specified by $(u, \aleph)$, the compact orbits of \eqref{eq:laxeq} are geometrically encoded into a linear flow
on the Jacobian variety
$\mathrm{Pic}^{(0)}(\Gamma_{u,\aleph})$
\cite{MR815768,MR533894,MR597729,MR0413178}; I recall here why this is the
case. The eigenvalue problem\footnote{For ease of notation, and since we've
  fixed $\rho=\mathfrak{g}$ in the previous section, I am dropping here any
  reference to the representation $\rho$ of the Lax operator.}
 at time-$t$, 
\beq
\widehat{L_{x, y}}(\lambda) \Psi_{x, y} = \mu \Psi_{x, y}
\label{eq:laxpsi}
\eeq
with $x=x(\vec t)$,  $y=y(\vec t)$, endows the spectral
curve with an eigenvector line bundle $\LL_{x, y} \to \Gamma_{u,\aleph}$ and a section
$\Psi:\Gamma_{u, \aleph} \to \LL_{x, y}$ given as follows. We have an
eigenspace morphism 
\beq
\mathcal{E}_{x, y}:\Gamma_{u, \aleph} \to \bbP^{\dim \mathfrak{g}-1} = \bbP^{247}
\eeq
that, away from ramification points of the $\lambda:\Gamma_{u, \aleph}\to
\bbP^1$ projection, assigns to a point $(\lambda, \mu)\in \Gamma_{u, \aleph}$ the
(time-dependent) eigenline of \eqref{eq:laxpsi} with eigenvalue $\mu$; this
in fact extends to a locally free rank one sheaf on the whole of $\Gamma_{u, \aleph}$
\cite[Ch.~5, II Proposition on p.131]{MR1995460}.  We write 
\beq
\LL_{x,
  y} \triangleq \mathcal{E}_{x, y}^*\cO_{ \bbP^{247}}(1)\in
\mathrm{Pic}(\Gamma_{u,\aleph})
\label{eq:BAlb}
\eeq
for the pullback of the
hyperplane bundle on $\bbP^{\dim \mathfrak{g}-1}$ via the eigenline map
$\mathcal{E}_{x, y}$, and fix (non-canonically) a section of the
latter by
\beq
\Psi_j(\lambda, \mu_i(\lambda)) = \frac{\Delta_{j1}\l(\widehat{L_{x, y}}(\lambda)-\mu_i(\lambda)\r)}{\Delta_{11}\l(\widehat{L_{x, y}}(\lambda)-\mu_i(\lambda)\r)},
\label{eq:BApsi}
\eeq
where $\mu_i(\lambda)= \exp(\a_i(l(\lambda))$ (cfr. \eqref{eq:xi3}) and we
denoted by 
$\Delta_{ij}(M)$ the $(i,j)^{\rm th}$ minor of a matrix $M$. As $t$ and
$x(t)$, $y(t)$ vary, so
will $\LL_{x(t), y(t)}$, and
\beq
\cB_{x, y}(t)\triangleq \LL_{x, y}(t) \otimes \LL^*_{x, y}(0) \in \mathrm{Pic}^{(0)}(\Gamma_{u,
  \aleph})\simeq \frac{H^1(\Gamma_{u,
  \aleph},\cO)}{H^2(\Gamma_{u,
  \aleph},\bbZ)}
\label{eq:Bxy}
\eeq
 is a time-dependent degree zero line bundle on $\Gamma_{
  u, \aleph}$. 

The flows \eqref{eq:laxeq} thus determine a flow $t\to \cB_{x(t), y(t)}$ in the Jacobian of
$\Gamma_{u, \aleph}$, which is
actually
{\it linear} in Cartesian coordinates for the torus $\mathrm{Pic}^{(0)}(\Gamma_{u,
  \aleph})$. Indeed, let $\{\omega_k\}_k$ be a basis for the
$\bbC$-vector space of holomorphic differentials on $\Gamma_{u, \aleph}$,  $\bbC\bra \{\omega_k\}_k \ket = H^1(\Gamma_{u, \aleph},\cO)$,
and let
\bea
\psi:\mathrm{Sym}^{g}\Gamma_{u, \aleph} & \rightarrow & \mathrm{Pic}^{(0)}(\Gamma_{u,
  \aleph})\nn \\
(\gamma_1+ \dots+ \gamma_g) & \rightarrow & \sum_{i=1}^g \cA(p_i)
\eea
 be the surjective, degree one morphism from the $g^{\rm th}$-symmetric power of
$\Gamma_{u, \aleph}$ to its Jacobian, given by taking the Abel sums of
$g$ unordered points on $\Gamma_{u, \aleph}$; here 
\bea
\cA:\Gamma_{u, \aleph} & \rightarrow & \mathrm{Pic}^{(0)}(\Gamma_{u,
  \aleph})\nn \\
\gamma & \rightarrow & \l(\int^\gamma \rd\omega_1, \dots, \int^\gamma \rd\omega_g \r)
\eea
denotes the Abel map for some fixed choice of base point. Writing $$\mathrm{Sym}^{g}\ni \gamma(t)= (\gamma_1(t),\dots,
\gamma_g(t))=\psi^{-1}(\cB_{x(t), y(t)})$$ for the inverse of $\cB_{x(t), y(t)}$, which is unique for
generic time $t$ by Jacobi's theorem, we have that \cite[Thm.~4]{MR533894}
\beq
\Omega_{ik}\triangleq \frac{\de}{\de t_i} \sum_{j=1}^g \int^{\gamma_j(t)}\omega_k= \sum_{p\in
  \lambda^{-1}(0) \cup \lambda^{-1}(\infty)}
\mathrm{Res}_{p}\l[\omega_k P_i(\widehat{L_{x, y}}(\lambda)) \r]
%+\mathrm{Res}_{\lambda^{-1}(\infty)}\l(\omega_k
%\widehat{L_{x, y}}(\lambda))_{[\infty]}\r) \r]
\quad \forall~ k=1,\dots, g
\label{eq:flowjac}
\eeq
%
%where, for $f \in \Gamma(\bbP^1,\cM)$ and locally around $p\in \bbP^1$, we
%denoted by $f_{[p]}$ the projection to its polar part (the ``Laurent tail'')
%around $p$. 
The left hand side is the derivative of the flow on the Jacobian
(its angular frequencies) in the chart on $\mathrm{Pic}^{(0)}(\Gamma_{
  u,\aleph})$ determined by the linear coordinates $H^1(\Gamma_{
  u,\aleph},\cO)$ w.r.t the chosen basis $\{\omega_k\}_k$. The right hand side shows that this is
independent of time, and hence the flow is linear in these coordinates, since $\omega_k$ and $P_i(\widehat{L_{x, y}}(\lambda)))$ are:
the former since it only feels the dynamical phase space variables $\{x_i,
y_i\}_{i=0}^8$ in $\widehat{L_{x, y}}(\lambda)$ via $\Gamma_{
  u, \aleph}$, itself an integral of motion, and the latter by
\eqref{eq:laxeq}.\\

\subsection{The Kanev--McDaniel--Smolinsky correspondence}

The story above is common to a large variety of systems (the Zakharov--Shabat
systems with spectral-parameter-dependent Lax pairs), and the $\widehat{\mathrm{E}_8}$ relativistic
Toda fits entirely into this scheme. In particular, in the better known examples of the
periodic relativistic and non-relativistic Toda chain with $N$-particles
(i.e. $\mathfrak{g}=\mathrm{sl}_N; \rho=\square$ in \eqref{eq:laxeq}), where
the spectral curves have genus $g=N-1$, the action-angle map $\{x_i,y_i\} \to
(\Gamma_{u,\aleph}, \LL_{x,y})$ gives a
family of $\mathrm{rank}\mathfrak{g}=N-1$ commuting flows on their $N-1$-dimensional
Jacobian. A question that does {\it not} arise in these ordinary examples, however, is the following: in our case, we have way more angles than we have
actions, as the genus of the spectral curve is much higher than the rank of $\mathfrak{g}=\mathfrak{e}_8$. Indeed, the Jacobian is $495$-complex dimensional in our case by
\eqref{eq:genred}; but the (compact) orbits of \eqref{eq:genred} only span an
8-dimensional Abelian subvariety of the Jacobian. 

How do we single out this subvariety geometrically? In the non-relativistic
case, pinning down the dynamical subtorus from the
geometry of the spectral curve has been the subject of intense study since the
early studies of Adler and van Moerbeke \cite{MR597729} for
$\mathfrak{g}=\mathfrak{b}_n,\mathfrak{c}_n,\mathfrak{d}_n,\mathfrak{g}_2$,
and the fundamental works of Kanev \cite{MR1013158}, Donagi \cite{MR1397059} and
McDaniel--Smolinsky \cite{MR1401779,MR1668594} in greater generality. We
now work out how these ideas can be applied to our case as well.\\

Recall from \cref{prop:isopar} that we have a $\cW$-action on
$\Gamma_{u, \aleph}$ by deck transformations given by
\bea
\phi : \cW \times \Gamma_{u, \aleph} & \rightarrow & \Gamma_{
  u, \aleph} \nn \\
(w, \lambda, \mu_\alpha(\lambda)) & \rightarrow & (\lambda, \mu_{w(\alpha)}(\lambda))
\eea
which is just the residual action of the vertical transformations on the
cameral cover. Write $\phi_{w} \triangleq \phi(w, -)\in
\mathrm{Aut}(\Gamma_{u, \aleph})$ for the automorphism corresponding to
$w\in \cW$. Extending by linearity, $\phi_w$ induces an action on
$\mathrm{Div}(\Gamma_{u, \aleph})$ which obviously descends to give actions
on the Picard group $\mathrm{Pic}(\Gamma_{u, \aleph})$, the Jacobian $\mathrm{Pic}^{(0)}(\Gamma_{
  u, \aleph})\simeq \mathrm{Jac}(\Gamma_{u, \aleph})$ (since $\phi_w$
is compatible with degree and linear equivalence), and the $\bbC$-space of holomorphic 1-forms
$H^1(\Gamma_{u, \aleph},\cO)$. At the divisorial level we have furthermore an action of the
group ring 
\bea
\varphi : \bbZ[\cW] \times \mathrm{Div}(\Gamma_{u, \aleph}) & \rightarrow &
\mathrm{Div}(\Gamma_{u, \aleph}), \nn \\
\l(\sum_i a_i w_i, \sum_j b_j(\lambda_j, \mu_\alpha(\lambda_j))\r) & \rightarrow &
\sum_{i,j} a_i b_j (\lambda_j, \mu_{w_i(\alpha)}(\lambda_j)).
\label{eq:divact}
\eea
Recall from \cref{prop:isopar} that, since the group of deck transformations of the
cover $\Gamma_{u, \aleph}\setminus \{\rd \mu=0 \}$ is isomorphic to the
  Coxeter action of $\cW$ on the root space $\Delta\simeq \cW/\cW_{\alpha_0}$, 
%stabilises the subgroup $\cW_{\alpha_0}\triangleq
%\mathrm{Stab}_{\omega_7=\alpha_0}\cW\simeq \mathrm{Weyl}(\mathfrak{e}_7)$, where
%$\alpha_0=\omega_7$ is the highest root
the map \eqref{eq:divact} factors through the coset projection map $\cW\to \Delta $, i.e.
\beq
\varphi(w,-)= |\cW_{\alpha_0}| \sum_{\alpha \in \Delta} \tilde{a}_\alpha w_\alpha,
\label{eq:divact2}
\eeq
for some $\{\tilde a_\alpha\in \bbZ\}_{\a\in\Delta}$. Restrict now to elements
$\varphi(w,-) \in \bbZ[\cW]$ such that $\varphi(w,-) :\bbZ[\cW] \rightarrow
\bbZ[\mathrm{Aut}(\Gamma_{u,\aleph})]$ is a ring homomorphism. Then the
action \eqref{eq:divact} is the pull-back of an action of the maximal subgroup of
$\bbZ[\Delta]$ which respects the 
product structure induced from $\bbZ[\cW]$: this is the Hecke ring $H(\cW,\cW_{\alpha_0})
\simeq \bbZ[\cW_{\alpha_0} \backslash \cW/\cW_{\alpha_0}]\simeq \bbZ[\Delta]^{\cW_{\alpha_0}}$. Its additive
structure is given by the free $\bbZ$-module structure on the space of
double cosets of $\cW$ by $\cW_{\alpha_0}$, and its product is defined as the push-forward\footnote{That is, the image under the double-quotient projection of
the product of the pullback functions on $\cW$, which is well-defined on the double quotient even
when $\cW_{\alpha_0}$ is not normal, as in our case.} of
the product on $\bbZ[\cW]$. In practical terms, this forces the integers $a_\alpha$ in the
sum over roots in $\Delta^*$ (i.e. right cosets of $\cW/\cW_{\alpha_0}$) to be constant
over left cosets $\cW_{\alpha_0}\backslash \cW$ in \eqref{eq:divact2}.

The Weyl-symmetry action is the key to single out the Liouville-Arnold
algebraic torus that is home to the flows \eqref{eq:laxeq}. We first start
from the following 
\begin{defn}
Let $\mathcal{D} \in \mathrm{Div}(\Gamma \times \Gamma)$ be a
self-correspondence of a smooth projective irreducible curve $\Gamma$ and
let $\mathcal{C}\in \mathrm{End}(\Gamma)$ be the map 
\bea
\mathcal{C}: \mathrm{Jac}(\Gamma) & \to & \mathrm{Jac}(\Gamma) \nn \\
\gamma & \to & (p_2)_*(p_1^*(\gamma) \cdot \DD),
\eea
where $p_i$ denotes the projection to the $i^{\rm th}$ factor in $\Gamma
\times \Gamma$. The Abelian subvariety
\beq
\mathrm{PT}_\mathcal{C}(\Gamma)\triangleq \l(\mathrm{id}-\mathcal{C}\r)\mathrm{Jac}(\Gamma)
\eeq
is called a {\rm Prym--Tyurin} variety iff
\beq
(\mathrm{id}-\mathcal{C})(\mathrm{id}-\mathcal{C}-q_\cC)=0
\label{eq:quadeqC}
\eeq
for $q_\cC\in \bbZ$, $q_\cC \geq 2$.
\end{defn}
By \eqref{eq:quadeqC}, the tangent fibre at the identity
$T_e(\mathrm{Jac}(\Gamma))$ splits into eigenspaces $T_e(\mathrm{Jac}(\Gamma))=\mathfrak{t}_{\mathrm{PT}}
\oplus \mathfrak{t}_{\rm PT}^\vee$ of $\cC$ with
eigenvalues $1$ and $1-q_\cC$. Because $q_\cC\in \bbZ$, these exponentiate to subtori
$\cT_{\mathrm{PT}}= \exp \mathfrak{t}_{\mathrm{PT}}$, $\cT_{\mathrm{PT}}^\vee= \exp
\mathfrak{t}_{\mathrm{PT}}^\vee$, with
$\cT_{\mathrm{PT}}=\mathrm{PT}_\cC(\Gamma)$, such that $\mathrm{Jac}(\Gamma)= \cT_{\mathrm{PT}} \times \cT_{\mathrm{PT}}^\vee$. In particular, in terms of the linear
spaces $V_{\mathrm{PT}} \simeq \widetilde{\cT_{\mathrm{PT}}} $, $V_{\mathrm{PT}}^\vee \simeq \widetilde{\cT_{\mathrm{PT}}^\vee} $ which are the universal covers of the
two factor tori, we have
\beq
\mathrm{PT}_\cC(\Gamma) \simeq V_{\mathrm{PT}}/\Lambda_{\mathrm{PT}}
\eeq
where $\Lambda_{\mathrm{PT}} = H_1(\Gamma, \bbZ) \cap V_{\mathrm{PT}}$. Furthermore \cite{MR1013158}, there is a
natural principal polarisation $\Xi$ on $\mathrm{PT}_\cC(\Gamma)$ given by the restriction of
the Riemann form $\Theta$ on $H^1(\Gamma, \cO) \simeq V_{\mathrm{PT}} \oplus V_{\mathrm{PT}}^\vee$ to $V_{\mathrm{PT}}$;
we have $\Theta=q_\cC \Xi$, with $\Xi$ unimodular on $\Lambda_{\mathrm{PT}}$. In particular,
$\mathrm{id}-\cC$ acts as a projector on the space of 1-holomorphic differentials, and,
dually, 1-homology cycles on $\Gamma$, such that
\bit
\item the projection selects a symplectic vector space $V_{\mathrm{PT}}\subset H_1(\Gamma, \cO)$ and dual subring $\Lambda_{\mathrm{PT}}
  \in H_1(\Gamma, \bbZ)$; 1-forms in $V_{\mathrm{PT}}$ have zero periods on cycles in
  $\Lambda_{\mathrm{PT}}^\vee$;
\item bases $\{\omega_1, \dots, \omega_{\mathrm{dim} V_{\mathrm{PT}}}\}$,
  $\{(A_i,  B_i)\}_{i=1}^{\mathrm{dim} V_{\mathrm{PT}}}$ can be chosen such
  that the corresponding matrix minors
  of the period matrix of $\Gamma$ satisfy
\beq
\int_{A_j}\omega_i=q_\cC \delta_{ij}, \quad \int_{B_j} \omega_i = \tau_{ij}
\eeq
with $\tau_{ij}$ non-degenerate positive definite.
\eit
%

%{\bf Avoid $\bbZ[\Delta]^{\cW_{\alpha_0}}$ in the following. Make statements at
%  either the group ring $\bbZ[\cW]$ or the Hecke ring as appropriate.} 
There is a canonical element of $H(\cW,\cW_{\alpha_0})$ which has particular importance for
us, and which will eventually act as a projector on a distinguished Prym--Turin
subvariety of $\mathrm{Jac}(\Gamma_{u, \aleph})$. This is the Kanev--McDaniel--Smolinsky self-correspondence\footnote{This has also been considered
  in the gauge theory literature, implicitly in
  \cite{Martinec:1995by,Hollowood:1997pp} and more diffusely in \cite{Longhi:2016rjt}.} 
\cite{MR1013158,MR1401779,MR1668594} 
\beq
\mathscr{P}_{\mathfrak{g}} \triangleq \sum_{w \in \cW/\cW_{\alpha_0}} \bra w^{-1}(\alpha_0),\alpha_0  \ket w.
\label{eq:proj}
\eeq
I summarise here some of its key properties, some of which are easily
verifiable from the definition \eqref{eq:proj}, with others having been
worked out in meticulous detail in \cite[Sec.~3--5]{MR1401779}. Some further
explicit results that are relevant to our case, but that did not fit in the
discussion of \cite{MR1401779}, are presented below.
\begin{prop}
In the root space $(\mathfrak{h}^*, \bra, \ket)$ consider the hyperplanes 
\beq
\mathsf{H}_i= \{\beta \in \mathfrak{h}^* | \bra \beta, \alpha_0=i \ket \}.
\label{eq:hypkms}
\eeq
Then, set-theoretically, $\cW_{\alpha_0} \backslash \cW / \cW_{\alpha_0}
\simeq \{\delta_i\triangleq \mathsf{H}_i \cap \Delta\}_{i=-2}^2$. Letting $\cW
\stackrel{\pi_1}{\longrightarrow} \cW / \cW_{\alpha_0} \stackrel{\pi_2}{\longrightarrow} \cW_{\alpha_0} \backslash \cW / \cW_{\alpha_0}$ be the projection
to the double coset space and $\{s_i\}_{i=-2}^2=\pi_2(\Delta)$, we furthermore have
\beq
\mathscr{P}_{\mathfrak{g}} = \pi_2^* \l[\sum_{\delta_i \in \cW_{\alpha_0} \backslash \cW /
  \cW_{\alpha_0}} i s_i \in H(\cW,\cW_{\a_0})\r].
\label{eq:pgdh}
\eeq
\label{prop:pg}
\end{prop}
\begin{proof}
The fact that $\mathscr{P}_{\mathfrak{g}} \in \bbZ[\Delta]^{\cW_{\alpha_0}}=H(\cW,\cW_{\alpha_0})$ follows
immediately from its definition in \eqref{eq:proj} and the constancy of $\bra
w^{-1}(\alpha_0),\alpha_0\ket$ on left cosets. The rest of the proof follows
from explicit identification of the elements of $H(\cW,\cW_{\a_0})$ in terms
of the hyperplanes of \eqref{eq:hypkms}, and evaluation of \eqref{eq:proj} on
them. The proof is somewhat lengthy and the reader may find the details in \cref{sec:proofpg}.
\end{proof}
\begin{cor} $\mathscr{P}_{\mathfrak{g}}$ satisfies the quadratic equation in
  $H(\cW,\cW_{\alpha_0})$ with integral roots 
\beq
\mathscr{P}_{\mathfrak{g}}^2 = %\frac{|\cW| \bra \alpha_0, \alpha_0 \ket}{|\cW_{\alpha_0}| \mathrm{rank}
%  \mathfrak{g}} \mathscr{P}_{\mathfrak{g}} = 
q_{\mathfrak{g}} \mathscr{P}_{\mathfrak{g}}
\label{eq:projnorm}
\eeq
with
\beq
q_{\mathfrak{g}}=60.
\eeq
In particular, the correspondence $\mathcal{C}=1-\mathscr{P}_{\mathfrak{g}}$
defines a Prym--Tyurin variety $\mathrm{PT}_{1-\mathscr{P}_{\mathfrak{g}}}(\Gamma_{u, \aleph})\subset
\mathrm{Jac}(\Gamma_{u, \aleph})$.
\label{prop:proj}
\end{cor}

\begin{proof}
This is a straightforward calculation from \cref{eq:pgdh}.
%statement is \cite[Prop.~8]{MR1401779}, an easy consequence of the orthogonality relations for the
%characters of $\cW$.
\end{proof}

In the following, I will simply write
$\mathrm{PT}(\Gamma_{u, \aleph})$, dropping the $1-\mathscr{P}_{\mathfrak{g}}$ subscript
which will be implicitly assumed.

%
%
%\ben
%\item Let $S$ be a set carrying a $\cW$-action, so that $\mathscr{P}_\mathfrak{g}$ acts as a linear
%  (un-normalised) projector on the
%  permutation representation $\bbC[S]$. If the latter decomposes as
%%
%\beq
%\bbC[S] \simeq n_0 \mathfrak{h} \oplus_i n_i R_i
%\eeq
%
%with $R_i$ irreducible and $\mathfrak{h}$ the Coxeter representation, then
%$\dim_\bbC \mathrm{Im} \mathscr{P}_{\mathfrak{g}}=n_0$, with the projection consisting of a direct sum of one line
%for each Coxeter summand.
%%
%\item Cartan torus in the cameral Jacobian. Isomorphic to the spectral Prym--Tyurin.
%\item
%\een
%
The main statement about
${\rm PT}(\Gamma_{u,\aleph})$ is the subject of the next Theorem. Note that this bears a large intellectual debt to
previous work in \cite{MR1013158,MR1668594}; the modest contribution of this paper is a combination of the
results of this and the previous Section with  \cite{MR1013158,MR1668594} to prove that the Liouville--Arnold torus (the
image of the flows \eqref{eq:laxeq} on the Jacobian) is indeed isomorphic to
the full Kanev--McDaniel--Smolinsky Prym--Tyurin, rather than being just a
closed subvariety thereof. 
%{\bf
%  Question about isogeny?}
%
\begin{thm}
The flows \eqref{eq:laxeq}, \eqref{eq:flowjac} of the $\widehat{\mathrm{E}_8}$ relativistic Toda
chain linearise on the Prym--Tyurin variety
$\mathrm{PT}(\Gamma_{u, \aleph})$ and they fill it
for generic initial data $(u, \aleph)$.
\label{thm:pt}
\end{thm}
\begin{proof}
The linearisation of the flows on ${\rm PT}(\Gamma_{u,\aleph})$
amounts to say that
\beq
\sum_{p \in \lambda^{-1}\{0,\infty\}} \mathrm{Res}_p\l[\omega
  P_i(\widehat{L_{x, y}(\lambda)})\r] \neq 0 \quad
  \Rightarrow \quad \mathscr{P}_{\mathfrak{g}}^* \omega = \omega
\eeq
in \eqref{eq:flowjac}. 
This is essentially the content of
\cite[Theorem~8.5]{MR1013158} and especially \cite[Theorem~29]{MR1668594}, to
which the reader is referred. The latter paper greatly relaxes an assumption on the spectral dependence of $\widehat{L}_{
  x,y}(\lambda)$ \cite[Condition~8.4]{MR1013158} which renders
incompatible \cite[Theorem~8.5]{MR1013158}
with \eqref{eq:laxaff2}; this restriction is entirely lifted in
\cite[Theorem~29]{MR1668594}, where the fact that \eqref{eq:laxaff2} depends
rationally on $\lambda$ is sufficient for our purposes. While
\cite{MR1013158,MR1668594} deal with the non-relativistic counterpart of the
system \eqref{eq:laxeq}, it is easy to convince oneself that replacing
their Lie-algebraic setting with the Lie-group arena we are playing in in this
paper amounts to a purely notational redefinition of $\mathfrak{g}$ to $\cG$ in
the arguments leading up to \cite[Theorem~29]{MR1668594}.

Since the first part of the statement has been settled in \cite{MR1668594},
I now move on to prove that the Prym--Tyurin {\it is} the Liouville--Arnold torus. Denoting $\phi^{(i)}_t:\widehat{\cP}\to \widehat{\cP}$ be the time-$t$ flow of
\eqref{eq:laxeq}, and for fixed $(x, y) \in \widehat{\cP}$, the above
proves that 
\bea
\phi^{(1)}_{t_1} \cdot \dots \cdot \phi^{(8)}_{t_8} : \bbP^1 \times \dots \times \bbP^1 &
\to & \widehat{\cP} \nn \\
(x, y) & \to & (x(\vec t), y(\vec t))
\label{eq:flowtor}
\eea
surjects to an eight-dimensional subtorus
%\footnote{This should be
%  topologically a cylinder by
%  commutativity of the flows, and actually a torus by compactness of the
%  ambient space $\mathrm{PT}_{1-\mathscr{P}_{\mathfrak{g}}}$.} 
of
${\rm PT}(\Gamma_{u,\aleph})$. 
%Let us now show that this torus is actually the
%Prym--Tyurin itself. 
To see the resulting torus is the Prym--Tyurin, we use the dimension formula of
\cite[Theorem~17]{MR1401779}. Let $C_\star \triangleq \bbP^1\setminus \{b_i^\pm
\}_{i=1}^9$, $\mathfrak{M}:\pi_1(C_\star)
\to \cW$ be the Galois map of the spectral cover $\Gamma_{u, \aleph}$, and for
$P\in \Gamma_{u, \aleph}$ write $S^{(P)}$ for the stabiliser of $P$ in the
group of deck transformations of $\Gamma_{u, \aleph}$, and
$\mathfrak{h}^*_P$ for the fixed point eigenspace of $S^{(P)} \subset \cW$. Then \cite[Theorem~17]{MR1401779}
\beq
\dim_\bbC {\rm PT}(\Gamma_{u,\aleph}) =
\frac{1}{2}\sum_{\lambda(p)|\rd \mu(p)=0} \l(8-\dim_\bbC
\mathfrak{h}^*_p\r)- 8+ \bra \mathfrak{h}, \bbC[\cW/\mathfrak{M}(\pi_1(\bbP^1_\star))]\ket 
\eeq
where one representative $p$ is chosen in each fibre of $\lambda:\Gamma_{
  u, \aleph}\to \bbP^1$. In our case, $\mathfrak{M}(\pi_1(\bbP^1_\star))=\cW$ by
\cref{prop:isopar} and the fact that the $\alpha_0$-parabolic cover is
irreducible (hence a connected covering space of $\bbP^1$), so the last term
vanishes. Then
\beq
\dim_\bbC {\rm PT}(\Gamma_{u,\aleph}) =
\frac{1}{2}\sum_{i=1,\dots 9, j=\pm} \l(8-\dim_\bbC
\mathfrak{h}^*_{Q_{i,j}}\r)+\frac{1}{2}\sum_{j=\pm} \l(8-\dim_\bbC
\mathfrak{h}^*_{Q_{\infty,j}}\r) - 8
\eeq
where $Q_{i,j=\pm}$ are the ramification points of the $\lambda$-projection as
in \cref{prop:br}. Since $\alpha_{k(i)}\cdot \mu(Q_{i,\pm})=0$ for some permutation
$k:\{1,\dots, 8\}\to \{1, \dots,8\}$, the deck transformations in $S^{(Q_{i,\pm})}$ are simple reflections that stabilise the hyperplane
orthogonal to the root $\alpha_{k(i)}$, so that $\dim_\bbC
\mathfrak{h}^*_{Q_{i,j}}=7$. As far as $Q_{\infty,\pm}$ are
concerned, the deck transformation associated to a simple loop around them corresponds to the product of the Coxeter
element of $\cW$ times a simple root, as this is the lift under the projection
to the base curve of a loop around all branch points on the affine part of the
curve\footnote{The root in question is the one that is repeated in the
  sequence $\{k(i)\}_{i=1}^9$. There could be more of them in principle, but
  this would be in contrast with $\mathfrak{M}(\pi_1(\bbP^1_\star))=\cW$; equivalently,
  {\it a posteriori},
  this would lead to $\dim_\bbC {\rm PT}(\Gamma_{u,\aleph})<8$,
  contradicting the independence of the flows \eqref{eq:laxeq}, which in turn
  is a consequence of the algebraic independence of the fundamental characters $\theta_i$ in
  $R(\cG)$.}. Then $\dim_\bbC \mathfrak{h}^*_{Q_{\infty,j}}=1$,
$\dim_\bbC {\rm PT}(\Gamma_{u,\aleph})=8$, and the flows \eqref{eq:flowtor}
surject on the latter.
\end{proof}

An explicit construction of Kanev's Prym--Tyurin
${\rm PT}(\Gamma_{u,\aleph})$, after
\cite[Section~3]{Martinec:1995by}, can be given as follows. With reference to
\cref{fig:cuts}, let $\gamma^\pm_{i}$ be a simple counterclockwise loop around the branch point
$b^\pm_i$. I will similarly write $\gamma^-_0$ (resp. $\gamma^+_0$) for
analogous loops around $\lambda=0$ (resp. $\lambda=\infty$). For $\alpha \in
\Delta^*$ and $i=1, \dots, 8$, I define $C^\alpha_i, D^\alpha_i \in
C^1(\Gamma_{u, \aleph}, \bbZ)$ to be the lifts of the contours in red
(respectively in blue) to the cover $\Gamma_{u, \aleph}$, where we fix 
arbitrarily a base point $r\in \gamma^\pm_i$ and we look at the path in
$\Gamma_{u, \aleph}$ lying over $\gamma_i^\pm$ with starting point on the $\lambda$-preimage of $r$
labelled by $\alpha$. In other words,
\bea
C_i^\alpha & \triangleq &\lambda_{\sigma_i(\alpha)}^{-1}\l(\gamma^+_i\r) \cdot
 \lambda^{-1}_{\alpha}\l(\gamma^-_0\r), \nn \\
D_i^\alpha & \triangleq & \lambda_{\sigma_i(\alpha)}^{-1}\l(\gamma^+_i\r) \cdot
\lambda_\alpha^{-1}\l(\gamma^-_i\r).
\label{eq:CD}
\eea
Let now 
\beq
A_i  \triangleq  \frac{1}{q_\mathfrak{g}} \l(\mathscr{P}_{\mathfrak g}\r)_* C^{\alpha_0}_i,
\quad B_i  \triangleq  \frac{1}{2} \l(\mathscr{P}_{\mathfrak g}\r)_*
D^{\alpha_0}_i
\label{eq:AB}
\eeq
where the normalisation factor for $A_i$, $B_i$ will be justified momentarily. Notice that $A_i, B_i \in Z_1(\Gamma_{u, \aleph},\bbQ)$ are closed paths
on the cover: every summand $C_i^\alpha$ and $D_i^\alpha$ is indeed always
accompanied by a return path $C_i^{\sigma_i(\alpha)}$ and
$D_i^{\sigma_i(\alpha)}$, which has opposite weight in \eqref{eq:AB}. Denoting
by the same letters $A_i$, $B_i$ their conjugacy classes in homology, we identify
$H_1(\Gamma_{u, \aleph}, \bbQ)\supset \Lambda_{\mathrm{PT}}\triangleq \bbZ\bra \{A_i,
B_i\}_{i=1}^8\ket$. If $\{\omega_1, \dots, \omega_8\}$ is any choice of
1-holomorphic differentials such that $\mathrm{dim} \mathscr{P}^*_{\mathfrak g} \bbC\bra \omega_1, \dots, \omega_8\ket=8$, then 
\beq
\mathrm{PT}(\Gamma_{u, \aleph})= \frac{\mathscr{P}^*_{\mathfrak g} \bbC\bra \omega_1, \dots, \omega_8\ket}{\bbZ\bra \{A_i,
B_i\}_{i=1}^8\ket}
\eeq
by construction. It is instructive to compute the intersection index of the
cycles \eqref{eq:AB}: we have, from \eqref{eq:CD}, that
\bea
\label{eq:AB1}
(A_i, B_j) &=& \frac{1}{2q_\mathfrak{g}}\sum_{\b,\gamma \in \Delta^*} (C_i^\beta, D_j^\gamma)=
\frac{\delta_{ij}}{2 q_\mathfrak{g}} \sum_{\b, \gamma \in\Delta^*} \delta_{\b\gamma} \bra \alpha_0, \beta
\ket^2 =\delta_{ij},\\
(A_i, A_j) &=& (B_i, B_j) = 0,
\label{eq:AB2}
\eea
so that they are a symplectic basis for $\Lambda_{\mathrm{PT}}$; the normalisation factor
\eqref{eq:AB} has been chosen to ensure both that this is so and to render the period
integrals on $\{A_i, B_i\}$ compatible with
the usual form of special geometry relations.

\begin{figure}
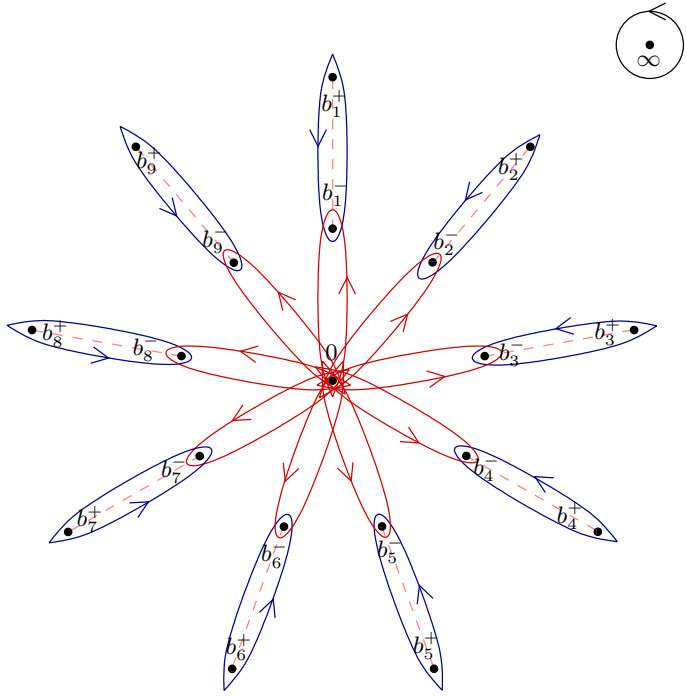
\caption{Contours on $\bbP^{1}_\star = \bbC^* \setminus\{b_i^\pm
  \}_{i=1}^9$. Projections of the $A$- and $B$-cycles are depicted in red and
  blue respectively.}
\label{fig:cuts}
\end{figure}

\subsection{Hamiltonian structure and the spectral curve differential}
The fact that the isospectral flows \eqref{eq:laxeq} turn into straight line
motions on $\mathrm{PT}(\Gamma_{u,
  \aleph})$ is the largest bit in the proof of the algebraic complete
integrability of the $\widehat{\mathrm{E}_8}$ relativistic Toda. We conclude it now by
working out in detail a choice of Darboux co-ordinates $\{S_i,
\vartheta^i\}_{i=1}^8$, with $\vartheta^i \in S^1$, such that the Hamiltonians
\eqref{eq:ui} are functions of $S_i$ alone. In the process, this will complete the
construction of the Dubrovin--Krichever data of \cref{defn:dk}.
%spectral curve triplet $(\Gamma_{u, \aleph},\Lambda,
%\rd \lambda )$ by endowing the spectral cover $\Gamma_{u, \aleph}$ with a
%canonical 1-differential. {\bf mention section relative dualising sheaf?}

Composing the surjection \eqref{eq:proj} with the Abel--Jacobi map gives an
Abel--Prym--Tyurin map
\bea
\cA_{\rm PT}: \Gamma_{u, \aleph} & \to & \mathrm{PT}(\Gamma_{u, \aleph}) \nn \\
p & \to & \mathscr{P}_{\mathfrak{g}}\cdot \cA(p).
\label{eq:apt}
\eea
Since $\mathrm{PT}(\Gamma_{u, \aleph})$ is principally
polarised, an analogue of the Jacobi theorem holds for $\cA_{\rm PT}$ \cite[Lemma~2.1]{jow2016principally}, and the
Abel--Prym--Tyrin map \eqref{eq:apt} is an embedding of $\Gamma_{u, \aleph}$ into $\mathrm{PT}(\Gamma_{u,
  \aleph})$ as a $q_\mathfrak{g}=60$-multiple of %By \eqref{eq:projnorm}, the homology class $[\cA_{\rm PT}(\Gamma_{u, \aleph})]$
%of the
%Abel--Prym--Tyurin curve is $q_\mathfrak{g}=60$ times 
its minimal curve $\frac{\Xi^7}{7!}$. % of $\mathrm{PT}(\Gamma_{u, \aleph})$. 
Then, taking Abel sums of 8 points on $\Gamma_{u, \aleph}$
and projecting their image to $\mathrm{PT}$,
\bea
\cA_{\rm PT}:\mathrm{Sym}^{8}\Gamma_{u, \aleph} & \rightarrow & \mathrm{PT}(\Gamma_{u,
  \aleph})\nn \\
(\gamma_1+ \dots+ \gamma_g) & \rightarrow & \l(\mathscr{P}_{\mathfrak{g}}\r)_*\sum_{i=1}^8 \cA(\gamma_i)
\eea
gives a finite,
degree $q_\mathfrak{g}^8=2^{16} 3^8 5^8$ surjective morphism\footnote{I
  slightly abuse
  notation here and call it with the same symbol of \eqref{eq:apt}.} from the $8$-fold symmetric product of $\Gamma_{u,
  \aleph}$ to $\mathrm{PT}(\Gamma_{u, \aleph})$ which maps
the fundamental class
$[\mathrm{Sym}^8(\Gamma_{u, \aleph})] \to q_\mathfrak{g}^8 [\mathrm{PT}(\Gamma_{u, \aleph})]$ to $q_\mathfrak{g}^8$ the fundamental class of the
Prym--Tyurin. The fibre $\cA_{\rm PT}^{-1}(\xi)$ of a point $\xi\in  \mathrm{PT}(\Gamma_{u, \aleph})$ is given by $q_\mathfrak{g}^8$ unordered 8-tuples of
points $\gamma_1 + \dots + \gamma_8$ on $\Gamma_{u, \aleph}$ satisfying
\bea
\xi & = & \cA_{\rm PT}\l(\sum_i \gamma_i\r) = (\mathscr{P}_{\mathfrak{g}})_* \sum_{i=1}^8 \cA(\gamma_i) =  (\mathscr{P}_{\mathfrak{g}})_*\sum_{i=1}^8
\l(\int^{\gamma_i} \rd \omega_1, \dots, \int^{\gamma_i} \rd \omega_{495}\r), \nn \\
& = & 
\sum_{i=1}^8
\l(\int^{\mathscr{P}_{\mathfrak{g}}(\gamma_i)} \rd \omega_1, \dots,
\int^{\mathscr{P}_{\mathfrak{g}}(\gamma_i)} \rd \omega_{495}\r),
\nn \\
&=& \sum_{i=1}^8
\l(\int^{\gamma_i}  \mathscr{P}_{\mathfrak{g}}^*\rd\omega_1, \dots,
\int^{\gamma_i} \mathscr{P}_{\mathfrak{g}}^*\rd \omega_{495}\r).
\eea
Let us now reconsider the action-angle map $\{x_i, y_i\}\to (\Gamma_{u,
  \aleph}, \cB_{x, y})$ of \eqref{eq:ui}, \eqref{eq:Bxy} and
\eqref{eq:flowjac} in light of \cref{thm:pt}. By the above reasoning, the
flows $(x(t), y(t))$ are encoded into the motion of $\cB_{x,
  y}(t)$, or equivalently, any of the pre-images $\cA_{\rm PT}^{-1} \cB(t)=
(\gamma_1(t)+ \dots +\gamma_8(t))$. I want to study the motion in terms of the
latter, and argue that the Cartesian projections of $\gamma_i$ provide
logarithmic Darboux
coordinates for \eqref{eq:pl}. I begin with the following
\begin{thm}
Write $\omega_{\rm PL}$ for the symplectic 2-form on an $\aleph$-leaf of
$\cP_{\rm Toda}$ and let $\delta:\Omega^\bullet(\widehat{\cP})\to
\Omega^{\bullet+1}(\widehat{\cP})$ denote exterior differentiation on $\widehat{\cP}$. Then
\beq
\omega_{\rm PL}= \mathscr{P}_{\mathfrak{g}}^*
\sum_{i=1}^8 
 \frac{\delta\mu(\gamma_i)}{\mu(\gamma_i)} \wedge \frac{\delta \lambda(\gamma_i)}{\lambda(\gamma_i)}.
\label{eq:opllm}
\eeq
\label{thm:kp}
\end{thm}
\begin{proof}
Recall that (see e.g.~\cite[Section~3.3]{MR1995460}) any Lax system of the
type \eqref{eq:laxeq} with rational spectral parameter and with $L(\lambda)\in
\mathfrak{g}$ can be interpreted as a
flow on a coadjoint orbit of $\mathfrak{g}^*$ which is Hamiltonian with
respect to the Kostant--Kirillov bracket. More in detail, the pull-back of the
Kostant--Kirillov symplectic 2-form reads \cite[Sections~3.3, 5.9, 14.2]{MR1995460}
\beq
\omega_{\rm KK}= \frac{1}{2 \dim \mathfrak{g}}\sum_{\lambda_k=0,\infty}
\mathrm{Res}_{\lambda_k}\mathrm{Tr}\l((A_k)_- g_k^{-1} \delta g_k \wedge
g_k^{-1} \delta g_k \r),
\label{eq:omegaKK}
\eeq
where we diagonalise\footnote{Note that the eigenvalue 1 of $\widehat{L_{x, 
      y}}(\lambda)$ has full geometric multiplicity 8, and the other eigenvalues
  are all distinct when $\lambda$ is in a punctured neighbourhood of $0$ or $\infty$.}
  $L(\lambda)=g_k^{-1} A_k g_k$ locally around the poles at $\lambda=0,
\infty$, we denote $M_-(\lambda_0)$  the projection to the Laurent tail
around $\lambda=\lambda_0$, and $\delta$ indicates exterior differentiation on $\widehat{\cP}$. This can be rewritten in terms of the
 Baker--Akhiezer eigenvector line bundle \eqref{eq:BAlb} and its marked
section \eqref{eq:BApsi} as an instance of the Krichever--Phong universal
symplectic form $\omega_{\rm KP}$ \cite{Krichever:1997sq,DHoker:2002kfd}. Let $\Psi=(\Psi_j)$ be
the $248\times 248$ matrix whose $j^{\rm th}$ column is the normalised
eigenvector \eqref{eq:BApsi}. Then  \cite[Section~5]{DHoker:2002kfd}
\beq
\omega_{\rm KK}= \omega_{\rm KP}^{(1)}\triangleq \frac{1}{\dim \mathfrak{g}}\sum_{\lambda_k=
  0,\infty}\mathrm{Res}_{\lambda_k}\mathrm{Tr}\l(\Psi_{x, y}^{-1}
\delta \widehat{L_{x, y}} \wedge \delta \Psi_{x, y} \r) \rd \lambda.
\eeq
where $\rd:\Omega(\bbP^1)\to \Omega(\bbP^1)$ is the exterior differential on
the spectral parameter space. 

This is pretty close to what need, and it would recover the results of
  obtained in \cite{Adams:1992ej} in a related context, but it actually requires two
extra tweaks to get the symplectic form we are after,
$\omega_{\rm PL}$. First off, as explained in \cite[Section~6.5]{MR1995460}, if we are interested in the $r$-matrix solution
\eqref{eq:rmat} for the Toda lattice, what we need to consider is rather a
version $\omega_{\rm KP}^{(1)}$ of the universal
symplectic form which is logarithmic in $\lambda$, i.e.
\beq
\omega_{\rm KP}^{(1)} \triangleq \frac{1}{\dim \mathfrak{g}}\sum_{\lambda_k=
  0,\infty}\mathrm{Res}_{\lambda_k}\mathrm{Tr}\l(\Psi_{x, y}^{-1}
\delta \widehat{L_{x, y}} \wedge \delta \Psi_{x, y} \r)
\frac{\rd \lambda}{\lambda}.
\label{eq:krichph1}
\eeq
Secondly and more importantly, since we are dealing with an integrable system
on a Poisson-Lie Kac--Moody {\it group}, rather than a Lie algebra, $\omega_{\rm PL}$ is given by a
different\footnote{But, non-trivially, compatible: the resulting system is
  then bihamiltonian.}
Poisson bracket, as explained in \cite[Section~5.3]{DHoker:2002kfd}. This
is the logarithmic Krichever--Phong Poisson bracket $\omega_{\rm KP}^{(2)}$
\beq
\omega_{\rm PL}=\omega_{\rm KP}^{(2)}\triangleq  \frac{1}{\dim \mathfrak{g}}\sum_{\lambda_k=
  0,\infty}\mathrm{Res}_{\lambda_k}\mathrm{Tr}\l(\Psi_{x, y}^{-1}
\widehat{L_{x, y}}^{-1} \delta \widehat{L_{x, y}} \wedge \delta \Psi_{x, y} \r)
\frac{\rd \lambda}{\lambda}.
\label{eq:krichph2}
\eeq
The calculation of the residues of \eqref{eq:krichph2} is straightforward (see
\cite[Section~5.9]{MR1995460} for a completely analogous calculation in the
context of the Kostant--Kirillov form \eqref{eq:omegaKK}). From the general theory of Baker--Akhiezer
functions\footnote{See for example the discussion in \cite[Section~2]{MR3184178}.} and \eqref{eq:flowjac},
$\ln \Psi_{x, y}$ has simple poles, with residue equal to the
identity, at a divisor $D(t)\in
\mathrm{Div}(\Gamma_{u, \aleph})$ such that
\beq
\frac{\cA (D(t))-\cA(D(0))}{t}=  \sum_{p\in
  \lambda^{-1}(0) \cup \lambda^{-1}(\infty)}
\mathrm{Res}_{p}\l[\omega_k P_i(\widehat{L_{x, y}}(\lambda))) \r],
\eeq
and by \cref{thm:pt}, the l.h.s. is actually in the Prym--Tyurin variety
$\mathrm{PT}(\Gamma_{u, \aleph})$. This means that $\Psi_{x(t), y(t)}$
has simple poles at $(\mathscr{P}_{\mathfrak g})_*(\gamma_1(t)+\dots+\gamma_8(t))$
for some $\gamma=\gamma_1(t)+\dots+\gamma_8(t) \in \mathrm{Sym}^8(\Gamma_{u,
  \aleph})$; different $\gamma$ have the same image under $(\mathscr{P}_{\mathfrak g})_*$. Write
\beq
\sum r_k \epsilon_k \triangleq \gamma= \sum_{i,\a}\bra w(\alpha_0),\alpha_0\ket w_* \gamma_i
\label{eq:sumeps}
\eeq
for some $r_k \in \bbZ$, $\epsilon_k \in \Gamma_{u, \aleph}$. Near $\epsilon_k$ we have then
\beq
\delta \Psi_{x, y}=\frac{\Psi \delta \lambda(\epsilon_k)}{\lambda-\lambda(\epsilon_k)}\l(1+\cO\l(\lambda-\lambda(\epsilon_k)\r)\r).
\eeq
It turns out that the rest of the expression \eqref{eq:krichph2} is regular at
$\epsilon_k$. Indeed, exterior differentiation of the eigenvalue equation \eqref{eq:laxpsi}
yields
\beq
\delta  \l[ \widehat{\ln L_{x, y}}-\ln\mu \Psi_{x, y}\r]= \l(\widehat{L_{x, y}}\r)^{-1}\delta
\widehat{L_{x, y}}-\frac{\delta \mu}{\mu} \Psi_{x, y}- \ln
\mu \delta \Psi_{x, y}=0
\eeq
Multiplying by $\l( \lambda \Psi_{x, 
  y}\r)^{-1}$ and exploiting the fact that $\Psi_{x, 
  y}^{-1}\l(\widehat{L_{x, y}}-\mu\r)=0$ for the dual section of
$\LL_{x, y}$,  we get
\bea
\mathrm{Res}_{\lambda(\epsilon_k)}\mathrm{Tr}\l(\Psi_{x, y}^{-1}
\widehat{L_{x, y}}^{-1} \delta \widehat{L_{x, y}} \wedge \delta \Psi_{x, y} \r)
\frac{\delta \lambda}{\lambda} &=& \mathrm{Res}_{\lambda(\epsilon_k)}\mathrm{Tr}\l(\Psi_{x, y}^{-1}
\widehat{L_{x, y}}^{-1} \delta \widehat{L_{x, y}} 
\Psi_{x, y} \r)\wedge\frac{\delta
  \lambda(\epsilon_k)}{\lambda-\lambda(\epsilon_k)}\frac{\rd \lambda}{\lambda},\nn \\
&=&
\mathrm{Tr}\l(\Psi_{x, y}^{-1}
\widehat{L_{x, y}}^{-1} \delta \widehat{L_{x, y}}
\Psi_{x, y} \r) \wedge\frac{\delta
  \lambda(\epsilon_k)}{\lambda(\epsilon_k)},\nn \\
&=& 
248 \delta \ln \mu(\epsilon_k) \wedge \delta \ln \lambda (\epsilon_k).
\eea
Swapping orientation in the contour giving the sum over residues
\eqref{eq:krichph2} amounts to picking up residues over the affine part of
$\Gamma_{u, \aleph}\setminus \lambda^{-1}(0)$. We have two possible
contributions here: one is the sum over residues at
the Baker--Akhiezer poles $\epsilon_k$ that we have just computed. Another is
given by the branch points of the $\lambda$ projection, since $\det \Psi_{x, y}(b_i^\pm)= \cO(\sqrt{\lambda-b_i^{\pm}})$: hence both $\Psi_{x, y}^{-1}$ and
  $\delta|_{\lambda=\rm const}  \Psi_{x, y}$ develop a (simple)
pole there. Whilst the residues are individually non-zero, their sum vanishes: it is a simple
observation that adding a contribution of the form 
\beq
\Delta_{\rm KP} = \sum_{\lambda_k=0,\infty}
\mathrm{Res}_{\lambda_k} \mathrm{Tr}\l(\Psi_{x, y}^{-1}
\delta \ln \mu
\Psi_{x, y} \r) \wedge\frac{\delta
  \lambda}{\lambda}
\eeq
to \eqref{eq:krichph2} exactly offsets the aforementioned non-vanishing
residues at the branch points, and it has opposite residues at
$\lambda=0$ and $\infty$. Taking into account sign changes and summing over poles, images of the Weyl
action as in \eqref{eq:sumeps}, and pre-images of the Abel--Prym--Tyurin map returns \eqref{eq:opllm}.

\end{proof}

We are now ready to write down explicitly the action-angle integration variables for the
system. Let $\pi^{\rm sym}: \mathscr{S}_{\mathfrak{g}}^{\rm sym} \to \mathscr{B}_{\mathfrak{g}}$ be the family of Abelian varieties obtained by
replacing $\Gamma_{u, \aleph}$ with its eightfold symmetric product in the top
left corner of \eqref{eq:diagsc}; this is a $q_\mathfrak{g}^8$-cover of
$\widehat{\cP}$. Let $\mathscr{D}_{\mathfrak{g}} \in \mathrm{Div}
\mathscr{S}_{\mathfrak{g}}$ be the sum of $\Sigma_i(\mathscr{B}_{\mathfrak{g}})$ in \eqref{eq:diagsc}.

On the open set where the Prym--Tyurin does not degenerate, 
\beq
\mathscr{B}_{\mathfrak{g}}^{\rm reg} \triangleq \{ (u, \aleph) \in \mathscr{B}_{\mathfrak{g}} | \dim_\bbC
\mathrm{PT}(\Gamma_{u, \aleph})=8\},
\eeq
introduce the
(vertical) 1-forms on
$\mathscr{S}_{\mathfrak{g}}^{\rm sym}$ written, on a bundle chart
$((u,\aleph); (\gamma_i)_i)$, as
\bea
\mathscr{L} & \triangleq & \sum_{i=1}^8 \frac{\rd \lambda(\gamma_i)}{\lambda(\gamma_i)} \in
\omega_{\mathscr{S}_{\mathfrak{g}}^{\rm sym}/\mathscr{B}_{\mathfrak{g}}}(\mathscr{B}_{\mathfrak{g}}^{\rm reg}), \nn \\
\mathscr{M} & \triangleq & \sum_{i=1}^8 \frac{\rd \mu(\gamma_i)}{\mu(\gamma_i)} \in
\omega_{\mathscr{S}_{\mathfrak{g}}^{\rm sym}/\mathscr{B}_{\mathfrak{g}}}(\mathscr{B}_{\mathfrak{g}}^{\rm reg}), \nn \\
\rd\cS & \triangleq & \sum_{i=1}^8 \rd \sigma(\gamma_i) \triangleq \sum_{i=1}^8 \log
\mu(\gamma_i) \frac{\rd \lambda(\gamma_i)}{\lambda(\gamma_i)} \in \omega_{\mathscr{S}_{\mathfrak{g}}^{\rm sym}/\mathscr{B}_{\mathfrak{g}}}(\mathscr{B}_{\mathfrak{g}}^{\rm reg}).
\label{eq:dS}
\eea
The same notation $\rd \cS$ and $\rd \sigma$ will indicate the pullbacks to fibres
$\mathrm{Sym}^8(\Gamma_{u, \aleph})$ and $\Gamma_{u, \aleph}$ of the
respective families; in \eqref{eq:dS}  
$\rd \sigma$ is an arbitrary choice of branch of the log-meromorphic differential $\log\mu\rd
\log\lambda$ on $\Gamma_{u, \aleph}$. Notice that
$$ 
\omega_{\rm PL}= \omega_{\rm KP}^{(1)} =
\mathscr{P}^*_{\mathfrak{g}}\mathscr{M}
\wedge \mathscr{L}.
$$
\begin{lem}
We have that
\beq
\frac{\de \rd \sigma}{\de u_i} \in H^0(\Gamma_{u, \aleph},\Omega^1) \quad
\forall i=1, \dots, 8,
\label{eq:dsdui}
\eeq
where the moduli derivative in \eqref{eq:dsdui} is taken at fixed
$\mu:\Gamma_{u, \aleph}\to \bbP^1$.
\label{lem:holder}
\end{lem}

\begin{proof}
By definition,
\beq
\rd \kappa_i \triangleq \frac{\de \rd \sigma}{\de u_i}= -\frac{\de_{u_i}\lambda \rd \mu}{\lambda \mu}=
-\frac{\lambda^9 \de_{u_i}\Xi_{\mathfrak{g},\rm red} \rd
  \mu}{\de_{\lambda}(\lambda^9 \Xi_{\mathfrak{g},\rm red})\lambda \mu}
\label{eq:dsdui2}
\eeq
Recall that, for a generic polynomial $P(x,y)$, the 1-forms
%\beq
$$\rd \omega_{ij} = \frac{x^{i-1} y^{j-1}\rd x}{\de_y P}$$
%\eeq
%
with $(i,j)$ in the strict interior of the Newton polygon of $P(x, y)$
are holomorphic 1-forms on $\overline{\{P(x,y)=0\}}$. The expression 
\eqref{eq:dsdui2} is a linear combination of terms that are exactly of this form: notice that the doubly logarithmic
form of $\rd \sigma$ in \eqref{eq:dS} is crucial to ensure the presence of the
product $\lambda \mu$ at the denominator which makes this statement true. However $\lambda^9
\Xi_{\mathfrak{g},\rm red}$ is highly non-generic, and by the way
$\Gamma_{u, \aleph}$ was introduced in \cref{def:sc} the 1-forms in
\eqref{eq:dsdui2} may still have simple poles with opposite residues at the
strict transform of the nodes in \eqref{eq:discry}. A direct computation
however shows that
\beq
\frac{\de \Xi_{\mathfrak{g},\rm red}}{\de u_i}(p)=0 \quad \mathrm{if}~ \rd\mu(p)=0,~
\mu(p)+\frac{1}{\mu(p)}=r_i^k, i=2,3.
\eeq
which entails the vanishing of the residues on the normalization, and thus
$\rd \kappa_j \in H^1(\Gamma_{u, \aleph},\cO)~\forall j=1, \dots, 8$.

\end{proof}

As a consequence, an algebraic Liouville--Arnold-type statement can be made as
follows. Locally on $\Gamma_{u, \aleph}$ and its $8$-fold symmetric
product, consider the Abelian integral
\beq
\sigma(p) =\int^p \rd \sigma
\label{eq:sigma}
\eeq
and correspondingly $\cS(p_i)=\sum_i \sigma
(p_i)$. Define the $A_i$-periods of $\rd\sigma$ as
\beq
\alpha_i \triangleq \oint_{A_i} \rd \sigma.
\eeq
By \cref{prop:proj,thm:kp}, these are phase-space areas (action variables) for
the angular motion on ${\rm PT}(\Gamma_{u,\aleph})$. Indeed,
\bea
\alpha_i &=& \oint_{A_i} \rd \sigma = \frac{1}{q_\mathfrak{g}}
\oint_{(\mathscr{P}_{\mathfrak{g}})_*(A_i)} \rd \sigma\nn \\
&=& \frac{1}{q_\mathfrak{g}} \oint_{A_i} \mathscr{P}_{\mathfrak{g}}^*\rd \sigma
=\frac{1}{q_\mathfrak{g}} \oint_{D^2| \de D^2=A_i} \omega_{\rm PL}.
\eea
Define the normalised basis of holomorphic 1-forms
\beq
\rd\vartheta_i= \sum_{j}\l(\frac{\de \alpha}{\de u}\r)^{-1}_{ij} \rd \kappa_j \in
H^1(\Gamma_{u, \aleph},\cO),
\eeq
so that $\oint_{A_j} \rd \vartheta_i=\delta_{ij}$. This is the normalised
$\bbC$-basis of the vector space $V_-$ of
$\mathscr{P}_{\mathfrak{g}}$-invariant forms on $\Gamma_{u, \aleph}$ with
respect to our choice of $A$ and $B$ cycles in \eqref{eq:AB}.

\begin{thm}
We have
\beq
\omega_{\rm PL} = \frac{1}{q_\mathfrak{g}} \sum_{i=1}^8 \rd u_i \wedge \rd \kappa_i =
\frac{1}{q_\mathfrak{g}} \sum_{i=1}^8 \rd
\alpha_i \wedge \rd \vartheta_i
\eeq
and in the action-angle coordinates $(\alpha_i, \vartheta_i)_{i=1}^8$ the flows
\eqref{eq:laxeq}, \eqref{eq:flowjac} are Hamiltonian with respect to
$\omega_{\rm PL}$, with Hamiltonian $u_i=u_i(\a)$. In particular, the angular
frequencies in \eqref{eq:flowjac} are given by the Jacobian
\beq
\Omega_{ij} = \l(\frac{\de \alpha}{\de u}\r)^{-1}_{ij}.
\label{eq:angfreq}
\eeq
\label{cor:actangl}
\end{thm}

\begin{proof}
The statement just follows from writing down the symplectic
change-of-variables given by looking at $\cS$ as a type II generating function of
canonical transformation, first in $u$ and then in $\alpha$,
\beq
\mu(\gamma_i)=\exp\l(\lambda(\gamma_i)\de_{\lambda(\gamma_i)} \cS\r), \quad
\kappa_i= \frac{\de \cS}{\de u_i}=\int \rd \kappa_i, \quad 
\vartheta_i= \frac{\de \cS}{\de \alpha_i}=\int \rd \theta_i,
\eeq
and use of \cref{lem:holder}.
\end{proof}

Keeping in mind the discussion below \cref{defn:dk}, the constructions of this section bestow on $\mathscr{S}_g$ a canonical choice
of Dubrovin--Krichever data as follows
\bea
\mathscr{F} & \longleftrightarrow & \mathscr{S}_{\mathfrak{g}} \nn \\
\mathscr{B} & \longleftrightarrow & \mathscr{B}_{\mathfrak{g}}  \nn \\
\mathscr{D} & \longleftrightarrow & \sum_i \Sigma_i(\mathscr{B}_{\mathfrak{g}})  \nn \\
\mathscr{E}_1 & \longleftrightarrow & \mathscr{L}  \longleftrightarrow
\frac{\rd \lambda}{\lambda} \nn \\
\mathscr{E}_2 & \longleftrightarrow & \mathscr{M} \longleftrightarrow
\frac{\rd \mu}{\mu}\nn \\
\Lambda & \longleftrightarrow & \Lambda_{\rm PT} 
\label{eq:dk}
\eea
which is complete but for the choice of
the Lagrangian sublattice $\Lambda^L$. The latter is left unspecified by the Toda
dynamics -- and, in the
applications of the next two Sections its choice will vary depending on the circumstances. 

\section{Application I: gauge theory and Toda}
\label{sec:applI}

I will now consider the first application of the constructions of the
previous two sections: this will culminate with a proof of a B-model version
of the Gopakumar--Vafa correspondence for the Poincar\'e homology sphere. For
the sake of completeness, I
will set the stage by recalling all the necessary ingredients of
\cref{fig:dualities}; the reader familiar with them may skip directly to \cref{sec:gvproof}.

\subsection{Seiberg--Witten, Gromov--Witten and Chern--Simons}

\subsubsection{Gauge theory}

From the physics point of view, the first object of interest for us is the minimal supersymmetric
five-dimensional gauge theory on the product $M_{5}=\bbR^4 \times S_{R}^1$ of
four-dimensional Minkowski space times a circle of radius $R$ with gauge group
the compact real form $E^{\bbR}_8$. On shell, and at $R\to \infty$, its gauge/matter content consists of
one $E^{\bbR}_8$ vector multiplet $(A, \lambda, \varphi)$ with real scalar $\varphi$, gluino $\lambda$, and gauge
field $A$; upon compactification this is enriched by an extra scalar $\vartheta$, which is
the Wilson loop around the fifth-dimensional $S^1_R$. The infrared dynamics of
the compactified theory is
characterised by a dynamical holomorphic scale in
four dimensions $\Lambda_4$ \cite{Seiberg:1994bz,Peskin:1997qi}, which is
(perturbatively) a renormalisation group invariant.
For generic vacua
parametrised by $R\in \bbR^+$ and the complexified scalar vev $\phi=\bra
\varphi+\ri\theta \ket$, and assuming that the latter is much
higher than the non-perturbative scale $|\phi|\gg \Lambda_4$, the massless modes are those 
of a theory of 
$\mathrm{rank}\mathfrak{e}_8=8$ weakly coupled photons in four dimensions, whose Wilsonian
effective action is given by integrating out both
perturbative (electric) and non-perturbative (dyonic) contributions of BPS
saturated Kaluza--Klein states. This is expressed (up to two derivatives in
the $U(1)$ gauge superfield strenghts $W_\alpha^i$, and in four-dimensional $\cN=1$ superspace
coordinates $(x, \theta)$) by the Wilsonian effective Lagrangian
\beq
\LL=\frac{1}{4\pi}\mathrm{Im}\l[\int \rd^4\theta \frac{\de F_0^{\rm SYM}}{\de A_i}\bar
  A^i +\frac{1}{2} \int \rd^2 \theta  \frac{\de^2 F_0^{\rm SYM}}{\de A_i \de A_j}W^i_\alpha\bar
  W^{\alpha,j}  \r],
\label{eq:effL}
\eeq
which is entirely encoded by the prepotential $F_0^{\rm SYM}$; in particular, the Hessian
$\frac{\de^2 F_0^{\rm SYM}}{\de A_i \de A_j}$ returns the gauge coupling
matrix for the low-energy photons. 

Mathematically, this gauge theory prepotential should coincide with the
equivariant limit of a suitable generating function of instanton numbers. Let
$\mathrm{Bun}_k(\cG)$ be the moduli space of principal $\mathrm{E}_8$-bundles on
the projective compactification 
$\bbP^2 \simeq \bbC^2 \cup \bbP^1_{\infty}$ of $\bbR^4 \simeq \bbC^2$ with
second Chern class $k$, which I assume to be positive in the following; here
``framed'' means that we fix a trivialisation of the projective line at
infinity. $\mathrm{Bun}_k(\cG)$ is an irreducible smooth quasi-affine variety
of dimension $2 h(\mathfrak{g}) k= 60 k$, and it admits an irreducible, affine
partial compactification given by the Uhlenbeck stratification
\beq
\cU_k(\cG)=\mathrm{Bun}_k(\cG) \sqcup \l(\mathrm{Bun}_{k-1}(\cG) \times \bbC^2
\r) \sqcup \l(\mathrm{Bun}_{k-2}(\cG) \times \mathrm{Sym}^2\bbC^2 \r) \dots
\sqcup \l( \mathrm{Sym}^k \bbC^2\r).
\eeq
There is an algebraic $\bbT \times (\bbC^\star)^2 \simeq (\bbC^\star)^{10}$
torus action on $\mathrm{Bun}_k(\cG)$, where the two factors act by scaling the
trivialisation at infinity and the linear coordinates of $\bbC^2$
respectively, and which extends to the whole of $\cU_k(\cG)$ \cite{MR2181803}, and
leads to a ten-dimensional torus action on the vector space $H^0(\cU_k(\cG),\cO)$ of
regular functions on  $\cU_k(\cG)$. Denoting the characters of $\bbT \times
(\bbC^\star)^2$ by $\mu$, the latter decomposes
\beq
H^0(\cU_k(\cG),\cO) = \bigoplus_{\mu} \l(H^0(\cU_k(\cG),\cO)\r)_\mu
\eeq
as a direct sum of weight spaces which are, non-trivially, finite-dimensional
over $\bbC$ (see e.g. \cite{Nakajima:2003pg}). Let $a_i=c_1(\bbL_i)$, $i=1, \dots, 8$ for the
first Chern class of the dual of the $i^{\rm th}$ tautological line bundle 
$\bbL_i \to B\cT$, and likewise write $\epsilon_{1,2}=c_1(\bbL_{1,2})$ for the
equivariant parameters of the right (spacetime) $(\bbC^\star)^2$ factor in
$\bbT \times (\bbC^\star)^2\circlearrowright \cU_k(\cG)$. The instanton partition function of the 5-dimensional
gauge theory\footnote{Physically this should be thought to be in an
  $\Omega$-background, with equivariant parameters $(\e_1, \e_2)$ corresponding
  to the torus weights of the $(\bbC^{\star})^2$ factor acting on $\bbC^2$
  above.} is then defined to be the Hilbert sum
\beq
\cZ^{\rm inst}_{\cG}(a_1, \dots, a_8; \epsilon_1, \epsilon_2, Q)= \sum_{n \in
  \bbZ^+}\l(Q R^{q_\mathfrak{g}} \re^{-q_\mathfrak{g} R (\e_1+\e_2)/2}\r)^n \sum_\mu \mathrm{dim}_\bbC
\l(H^0(\cU_k(\cG),\cO)\r)_\mu \exp{\l[\mu\cdot (\vec a, \vec \epsilon)\r]}
% \in  \bbC[[\re^{R a}, \re^{R\epsilon} ]]
.
\eeq
The prepotential \eqref{eq:effL} is recovered as the sum of the
non-equivariant limit of $\epsilon_1 \epsilon_2 \ln \cZ^{\rm inst}_{\cG}$,
which is well-defined \cite{MR2095899,MR2183121}, plus a classical+one-loop perturbative
contribution. Namely,
\beq
F^{\rm SYM}_0 \triangleq F_{\rm cl}+F_{\rm 1-loop}+ F_{\rm inst},
\label{eq:gprep}
\eeq
where
\bea
F_{\rm cl} &=& \frac{\tau_{ij} a_i a_j }{2}, \nn \\
F_{\rm 1-loop} &=& \sum_{\a \in \Delta^*}\l[ \frac{(\a \cdot a)^2 \log (R \Lambda_4)}{2}-(\a \cdot a)^3\frac{R }{12}
+\mathrm{Li}_3\l(\re^{-R \a \cdot a}\r)\r], \nn \\
F_{\rm inst} &=& \lim_{\e_i\to 0}\e_1 \e_2 \log \cZ^{\rm inst}_{\cG}(a_1, \dots, a_8; \epsilon_1, \epsilon_2, \Lambda_4),
\label{eq:Fgauge}
\eea
where $\tau$ is the bare gauge coupling matrix. In the
following, I am going to measure energies in units of $\Lambda_4$ and thus set
$\Lambda_4=1$; the dependence on $\Lambda_4$ can be restored by appropriate
rescaling of the
dimensionful quantities $a_i$ and $1/R$.

\subsubsection{Topological strings}

It has long been argued that the prepotential \eqref{eq:Fgauge} might also be regarded as the generating function of genus zero
Gromov--Witten invariants on a suitable non-compact Calabi--Yau threefold
\cite{Katz:1996fh,Lawrence:1997jr}. Let 
\beq
X= \mathrm{Tot}\l(\cO^{\oplus 2}_{\bbP^1}(-1)\r) = \l\{ (A, v) \in \mathrm{Mat}(2,\bbC)
\times \bbP^1 | A v =0\r\}
\label{eq:rescon}
\eeq
be the minimal toric resolution of the singular quadric $\det A=0$ in
$\mathbb{A}^4$; columns of $A$ give trivialisations of $\cO(-1)$ over the North/South affine
patches of $\bbP^1$. Any $\Gamma \lhd \mathrm{SL}_2(\bbC)$,
$|\Gamma|<\infty$ gives an action $\Gamma \circlearrowright X$
by left multiplication, $(\gamma \in \Gamma, A) \to \gamma \cdot
A$, which is trivial on the canonical bundle of $X$, and covers
the trivial action on the base $\bbP^1$. The quotient space is thus
a singular Calabi--Yau fibration over $\bbP^1$ by surface singuarities of the
same 
%$\bbC^2/\Gamma$. By the McKay correspondence, the finite groups $\Gamma$ are
%classified by simply-laced Dynkin diagrams of a simple Lie algebra of type
ADE type of $\Gamma$ \cite{MR1886756}; and type $\mathrm{E}_8$ corresponds to taking $\Gamma\simeq \tilde{\mathbb{I}}$
the binary icosahedral group (see \cref{sec:I120}).

There are two distinguished chambers in the stringy K\"ahler moduli
space of $\cX\triangleq [X/\tilde{\mathbb{I}}]$ that are of importance in our discussion. One is
the large radius chamber: in this case we take the minimal crepant resolution
$\pi:Y\to X/\tilde{\mathbb{I}}$, which corresponds to fibrewise resolving the
$\mathrm{E}_8$ singularities on a chain of rational curves whose intersection matrix equates
$-\mathscr{C}^{\mathfrak{g}}_{ij}$ \cite{MR1886756}. In particular, $H_2(Y, \bbZ)= \bbZ\bra H; E_1, \dots, E_8 \ket$
where $H$ (resp. $E_i$) is the pull-back to $Y$ from the base $\bbP^1$
(resp. the blow-up $\widehat{\bbC^2/\tilde{\mathbb{I}}}$ of the fibre
singularity) of the fundamental class of the base curve (resp. of the class of
the $i^{\rm th}$ exceptional curve). The Gromov--Witten potential
of $Y$ is the generating sum
\beq
F^{\rm GW}(Y;\epsilon; t_{\rm B}, t_1, \dots, t_8)= \sum_{d \in H_2(Y,\bbZ),g\in
  \bbZ^+}\epsilon^{2g-2} \re^{-d \cdot t} N^Y_{g,d}=  \sum_{g\in
  \bbZ^+}\epsilon^{2g-2} F^{\rm GW}_g(Y;t_{\rm B}, t_1, \dots, t_8)
\label{eq:gwpot}
\eeq
where 
\beq
N_{g,d}^Y= \int_{[\cM_{g}(Y,d)]^{\rm vir}}1
\label{eq:gwinv}
\eeq
is the genus-$g$, degree $d$ Gromov--Witten invariant of $Y$, and we write 
$d=d_B [H] + d_1[E_1]+\dots +d_8[E_8]$ for the degrees of stable maps from
$\bbP^1$ to $Y$. Owing to the non-compactness of $Y$, what is really meant by
$N_{g,d}^Y$ is a sum of degrees of the localised virtual fundamental classes
at the fixed loci w.r.t. a suitable $\bbC^\star$ action: notice 
that $Y$ supports a rank two torus action given by a diagonal scaling the
fibres (which commutes with $\Gamma$, and w.r.t. which the resolution map
$\pi$ is equivariant) and a $\bbC^\star$ rotation of the base $\bbP^1$. In
particular we can always cut out a 1-torus action which is trivial on $K_Y$ (the equivariantly
Calabi--Yau case) by tuning the weights of the two factors appropriately, and
this is the choice that is picked\footnote{This is the choice that is picked,
  for toric targets, by the topological vertex; this is consistent with the
  fact that, by the equivariant CY condition, this is a rational
number (rather than an element of $H_{B\bbC^\star}(\rm pt, \bbQ)$).} in \eqref{eq:gwinv}.
 Furthermore, natural Lagrangian A-branes $L \hookrightarrow Y$ and a counting theory of open stable maps can be constructed (at least
operatively) via localisation \cite{Katz:2001vm,MR2861610,Borot:2015fxa}; in a vein similar to
\eqref{eq:gwpot}-\eqref{eq:gwinv}, one defines
\beq
W_{g,h}^{\rm GW}(Y,L; \lambda; t_B, t_1, \dots, t_8, \lambda)= \sum_{d \in
  H_2(Y,L,\bbZ)}\sum_{w_1, \dots, w_h \in H_1(L,\bbZ)\}}\re^{-d \cdot t} \prod_i \lambda_i^{w_i} N^{Y,L}_{g,h, d, w}
\label{eq:ogwpot}
\eeq
where 
\beq
N_{g,h,d,w}^{Y,L}= \int_{[\cM_{g,h}(Y,L,d,w)]^{\rm virt}}1
\label{eq:ogwinv}
\eeq
is the genus-$g$, $h$-holed open Gromov--Witten invariant of $(X,L)$ of
relative degree $d$ and winding numbers $\{w_i\}_{i=1}^h$, $\de d=\sum w_i$.

The relation of these curve-counting generating functions and the instanton
prepotentials of the previous section is given by the 
so-called {\it geometric engineering of gauge theories}, a (partial)
statement of which can be given as follows: 
\begin{claim}[\cite{Katz:1996fh,Lawrence:1997jr}]
The genus zero Gromov--Witten potential of $Y$ equates the five-dimensional gauge
theory prepotential/instanton generating function
\beq
F^{\rm SYM}_0(R, a_1, \dots, a_8) = F^Y_0(t_{\rm B}, t_1, \dots, t_8) + \mathrm{cubic}
\eeq
under the identification
\beq
a_i = t_i R, \quad R=\re^{-t_{\rm B}/4}.
\label{eq:geoeng}
\eeq
\label{cl:geoeng}
\end{claim}

\cref{cl:geoeng} has an extension to higher genera wherein gravitational
corrections to $F_0^{\rm SYM}$ are considered \cite{Bershadsky:1993cx, Antoniadis:1993ze}, or
equivalently, the gauge theory is placed in the $\Omega$-background
(without taking the limit \eqref{eq:Fgauge}) and one restricts to the
self-dual background $\e_1=-\e_2=\epsilon$ \cite{Nekrasov:2002qd}. The open
string potentials \eqref{eq:ogwpot} have similarly a counterpart in terms of
surface operators in the gauge theory \cite{Alday:2009fs,Kozcaz:2010af}.

The second chamber is the orbifold chamber: here we consider the
stack quotient $\cX=[\cO(-1)^{\oplus 2}/\tilde{\mathbb{I}}]$, which has a $\bbP^1$
worth of $\tilde{\mathbb{I}}$-stacky points. Open and closed
Gromov--Witten invariants of $\cX$ can be defined, if only computationally, along the
same lines as before by virtual localisation on moduli of twisted stable maps \cite{MR2861610}; I
refer the reader to \cite[Sec.~3.3-3.4]{Borot:2015fxa} where this is more amply discussed.

\subsubsection{Chern--Simons theory}
\label{sec:cs}

The previous Calabi--Yau geometry has been argued in \cite{Borot:2015fxa},
following the earlier work \cite{Aganagic:2002wv}, to be related to
the large $N$ limit of $\mathrm{U}(N)$ Chern--Simons theory on the Poincar\'e
sphere. This is a real three-manifold $\Sigma$ obtained from $S^2 \times
S^1$ after rational surgery with exponents $1/2$, $1/3$ and $1/5$ on a
3-component unlink wrapping the fibre direction of $S^1\times S^2\to S^2$, and
it is the only $\bbZ$-homology sphere, other than $S^3$, to have a finite
fundamental group. Equivalently, it can be realised as the quotient
$S^3/\tilde{\mathbb{I}}\simeq \bbR\bbP^3/\mathbb{I}$ of the three-sphere by
the left-action of the binary icosahedral group \cite{MR0217740}.

I will very succintly present the statement we are after, referring the
reader to the beautiful review \cite{Marino:2005sj} or the presentation of
\cite{Borot:2015fxa} for more details. Let $k\in \mathbb{Z}^+$, $\cA$ a smooth
gauge connection on the trivial $\mathrm{U}(N)$ bundle on $\Sigma$. The $\mathrm{U}(N)$ Chern--Simons
partition function of $\Sigma$ at level $k$ is the functional integral
\bea
\label{eq:ZCS}
Z_{\rm CS}(\Sigma, k,N) &=& \bra 1 \ket_{\rm CS} = \int_{\mathscr{A}/\mathscr{G}} [\DD\cA]
\exp\l(\frac{\ri k}{2\pi} \mathrm{CS}[\cA]\r), \\
\mathrm{CS}[\cA] &=& \int_{\Sigma} \mathrm{Tr}_\square \l(\mathcal{A} \wedge \rd
\mathcal{A}+\frac{2}{3} \mathcal{A}^3\r),
\label{eq:CS}
\eea
where \eqref{eq:CS} is the Chern--Simons action. For $\cK \hookrightarrow
\Sigma$ a link in $\Sigma$ and $\rho \in R(U(N))$, we will also consider the
expectation value under the measure \eqref{eq:ZCS}--\eqref{eq:CS} of the
$\rho$-character of the holonomy around $\cK$,
\beq
W_{\rm CS}(\Sigma, \cK, k,N, \rho) = \frac{\bra \mathrm{Tr}_\rho \mathrm{Hol}_\cK(A)
\ket_{\rm CS}}{Z_{\rm CS}} = Z_{\rm CS}^{-1}\int_{\mathscr{A}/\mathscr{G}} [\DD\cA]
\exp\l(\frac{\ri k}{2\pi} \mathrm{CS}[\cA]\r) \mathrm{Tr}_\rho \mathrm{Hol}_\cK(A).
\label{eq:WCS}
\eeq
\eqref{eq:ZCS}-\eqref{eq:WCS} were proposed by Witten \cite{Witten:1988hf} to
be smooth\footnote{More precisely, $Z_{\rm CS}$
  is only invariant under diffeomorphisms of $\Sigma$ that preserve a given
  framing of its tangent bundle, and changes in a definite way under
  change-of-framing; the same applies for $W_{\rm CS}$ and a choice of framing
  on $\cK$. In the following I implicitly work in canonical framing for both
  $\Sigma$ and $\cK$; also the change of framing won't affect the large $N$
  behavious of $\mathcal{F}_{\rm CS}$ but for a constant in $t$, $\cO(N^0)$
  (unstable) term.} invariants of $\Sigma$ and
$(\Sigma, \cK)$, reflecting the near metric independence of \eqref{eq:ZCS} at
the quantum level \cite{Reshetikhin:1991tc}; when $\Sigma$ is replaced by
$S^3$, \eqref{eq:WCS} is the HOMFLY polynomial of $\cK$ coloured in the
representation $\rho$.

We will be looking at \eqref{eq:ZCS} in two ways, which are both essentially
disentangled with the question of giving a rigorous treatment of 
the path integral \eqref{eq:ZCS}. One is in Gaussian perturbation theory at
large $N$, where we take \eqref{eq:ZCS} as a formal expansion in ribbon graphs \cite{'tHooft:1973jz, Marino:2005sj}. Writing 
\beq
g_{\rm YM}=\frac{2\pi \ri }{k+N}, \quad t= g_{\rm YM} N,
\eeq
the perturbative free energy takes the form
\bea
\mathcal{F}^{\rm CS}(\Sigma,g_{{\rm YM}},t) &=& \ln Z^{\rm CS}(\Sigma,
k,N)\nn \\
&=& \sum_{g\geq 0} \mathcal{F}^{\rm CS}_{g}(\Sigma,t) g_{{\rm YM}}^{2g-2} \in
g_{{\rm YM}}^{-2}\mathbb{Q}[[t,g_{\rm YM}^2]].
\label{eq:fcs}
\eea
Similarly, for $h>0$, $l\in \bbN^h$ and $\cK \in
\Sigma$ a link in $\Sigma$, we get for the connected Chern--Simons average of
a Wilson loop around $\cK$ that
\bea
W^{(h)}_{\rm CS}(\Sigma, \cK, k,N, \{\lambda_i\}_{i=1}^h) & \triangleq & \sum_{l\in
  \bbN^h}\prod_i \lambda_i^{l_i} \frac{1}{|l|!}\frac{\de^{|l|} \log \bra \re^{\sum_i q_i
    \mathrm{Tr} \l(\mathrm{Hol}_\cK(A)\r)^i} \ket_{\rm CS}}{\de q_1^{l_1} \dots \de
  q_s^{l_s}}\Bigg|_{q_i=0}, \nn \\
&=& \sum_{g\geq 0} W_{g,h}(\Sigma, \cK,t,\{\lambda_i\}) g_{{\rm YM}}^{2g-2+h} \in g_{{\rm YM}}^{h-2}\mathbb{Q}[[t,g_{\rm YM}^2]].
\label{eq:wcs}
\eea
The second way of looking at \eqref{eq:fcs} and \eqref{eq:wcs} comes from
their independent mathematical life
as the $U_{\mathfrak{q}}(\mathrm{sl}_N)$ Reshetikhin--Turaev--Witten invariants of $\Sigma$ and $\cK\hookrightarrow
\Sigma$ respectively \cite{Reshetikhin:1991tc}. Recall that $\Sigma$ has a Hopf-like realisation as a circle bundle over 
the orbifold projective line $\bbP^1_{2,3,5}$ with three orbifold points with
isotropy group $\bbZ_{{\mathsf{s}}(n)}$, with ${\mathsf{s}}(1)=2$,
${\mathsf{s}}(2)=3$, ${\mathsf{s}}(3)=5$. I will write
$\mathsf{s}=\prod_{i}\mathsf{s}(i)=30$, and $\cK_n \simeq S^1$ for the knots wrapping the exceptional fibre labelled by $n$. Then the RTW invariants of $\Sigma$
and $(\Sigma, \cK_n)$ can be computed explicitly from a rational surgery formula
\cite{MR1875611} (or
equivalently, Witten's surgery prescription for Chern--Simons vevs
\cite{Witten:1988hf}), leading to
closed-form expressions for \eqref{eq:fcs} and \eqref{eq:wcs} alike in terms
of Weyl-group sums \cite{Marino:2002fk}. Denote by $\cF_{l}$ the set of dominant
weights $\omega$ of $\mathrm{SU}(N)$ such that, if $\omega=\sum a_i \omega_i$ in terms
of the fundamental weights $\omega_i$, then
$\sum_{i}a_i<l$. Then,
\bea
Z_{\rm CS}(\Sigma,k,N) &=& \cN(\Sigma) \sum_{\b \in  \cF_{k+N}}
\frac{1}{\prod_{\a>0}\sin\l(\frac{\pi \b\cdot \a}{k+N} \r)}
\prod_{i=1}^3 \sum_{f_i \in \Lambda_r/{\mathsf{s}}(i)\Lambda_r} \sum_{w_i \in  S_N}\epsilon(w_i) \times \nn \\
& \times & \exp\l\{\frac{\ri \pi}{(k+N) {\mathsf{s}}(i)}\l(- \beta^2
-2\beta\l((k+N)f_i+w(\rho)\r)+\l((k+N) f_i+w(\rho)\r)^2\r) \r\},\nn \\
\label{eq:fcs2}
\eea
where $\rho$ is the Weyl vector of $\mathfrak{sl}_{N}$, $\rho=
\sum_{i=1}^{N-1} \omega_i$,  $\Lambda_r$ is the $\mathfrak{sl}_N$ root
lattice, and $\cN(\Sigma)$ is an explicit multiplicative factor involving the
surgery data and the Casson--Walker--Lescop invariant of $\Sigma$. A similar expression holds for the (un-normalised)
Chern--Simons vevs of the Wilson loops around fibre knots: this is obtained by
replacing $\rho\to \rho+\Lambda$ for $\Lambda$ a dominant weight in
\cref{eq:fcs}, after which \eqref{eq:wcs} can be recovered by expressing the
representation-basis colouring by the connected power sum colouring of
\eqref{eq:wcs}, and powers multiple of ${\mathsf{s}}_i$ for $i=1,2,3$ single out the
holonomies around the $i^{\rm th}$ exceptional fibre (see
\cite{Borot:2014kda,Borot:2015fxa} and the discussion of \cref{sec:LMOmm}
below). 

Two remarks are in order about \eqref{eq:fcs2}. Firstly, unlike \eqref{eq:fcs}-\eqref{eq:wcs}, \eqref{eq:fcs2} is an {\it
  exact} expression at finite $N$; among its virtues however, as first
emphasised in \cite{Marino:2002fk}, is the possibility to express it as a
matrix-like integral, and thus use standard asymptotic methods in random
matrix theory to study its large $N$, finite $t$ regime: this fact will be
used extensively in the next Section. Secondly, as pointed out in \cite{Marino:2002fk} and further confirmed in
\cite{Beasley:2005vf,Blau:2013oha} by a functional integral analysis, 
the sum over $f_i$ in \eqref{eq:fcs2} may be interpreted as a sum over
critical points of the Chern--Simons functional \eqref{eq:CS},
\beq
\mathrm{Crit}_{\rm CS}^N=\l\{ A\in \mathrm{Conn}(\Sigma, U(N))| F_A
=0\r\}=\mathrm{Hom}(\pi_1(\Sigma), U(N)) /{\rm U}(N)\,
\label{eq:critCS}
\eeq
namely, flat $\mathrm{U}(N)$-connections on $\Sigma$; this is a finite set at
finite $N$ since
$|\pi_1(\Sigma)|<\infty$.
% $|\mathrm{Crit}_{\rm
%  CS}^N|<\infty$ since $|\pi_1(\Sigma)|<\infty$. 
In the monodromy representation of
the latter equality in \eqref{eq:critCS}, these can be labelled by integers $(f_0,f_1, \dots, f_8)$ satisfying
\beq
\sum_i \mathfrak{d}_i f_i=N
\eeq
where $\mathfrak{d}_i$, $i=0, \dots, 8$ is the dimension of the $i^{\rm th}$ irreducible
representation of $\pi_1(\Sigma)=\tilde{\mathbb{I}}$ (see \cref{tab:ctI120}; equivalently, the $i^{\rm
  th}$ Dynkin label in \cref{fig:dynke8}), the trivial connection contribution
to \eqref{eq:fcs2} being given by $f_i=0$, $i>0$. The latter is the exponentially
dominant summand in the limit $g_{\rm
  YM}\to 0$, as the classical Chern--Simons functional attains there its minimum
value (equal to zero), and it leads to a quantum invariant of
3-manifolds in its own right: this is the {\it L\^e--Murakami--Ohtsuki} (LMO)
invariant, which is a derivation of the universal Vassiliev--Kontsevich
invariant by
taking its Kirby-move-invariant part \cite{MR1604883}.

In landmark papers by Gopakumar and Vafa
\cite{Gopakumar:1998ki} and Ooguri--Vafa \cite{Ooguri:1999bv}, it was proposed
that the large $N$ expansion of the Chern--Simons invariants of $S^3$ and
$\cK=\bigcirc$ yield the genus expansion of the topological A-model on the
resolved conifold $X=\cO(-1)\oplus \cO(-1)\to \bbP^1$. Following in this
direction and that of its generalisation of \cite{Aganagic:2002wv} for lens spaces\footnote{Some of these arguments require extra care when one
  considers non-${\rm SU}(2)$ quotients of the three-sphere; see
  e.g. \cite{iolpq}.}, it was proposed amongst other
things in \cite{Borot:2015fxa} that the large $N$ limit of the connected
averages \eqref{eq:fcs} and \eqref{eq:wcs} should be interpreted (respectively) 
as the generating functions of closed and open Gromov--Witten invariants
\eqref{eq:gwpot} and \eqref{eq:ogwpot} of the orbifold-of-the-conifold
$\cX$, as I now recall. Let $\mathrm{Crit}_{\rm CS}^\infty = \lim_{N \rightarrow \infty}
\mathrm{Crit}_{\rm CS}^N$ be the direct limit of the finite critical point sets \eqref{eq:critCS}
w.r.t. the composition of morphisms of sets induced by ${\rm U}(N)
\hookrightarrow {\rm U}(N + 1)$. A point in $\mathrm{Crit}_{\rm CS}^\infty$
consists of a flat background $[A]_{\mathtt{t}}$ parameterised by $\mathtt{t}_i \triangleq
N_i g_{\rm YM}$ for $i \in \{0,\dots, 8\}$. Write now
$\mathcal{F}_g(\Sigma,\mathtt{t})$ and
$\mathcal{W}_{g,n}(\Sigma,\mathtt{t};\vec{x})$ for the contribution of
each $[A]_{\mathtt{t}}$ at large $N$ to the perturbative free
energies and correlators of ${\rm U}(N)$ Chern--Simons theory \eqref{eq:fcs}
and \eqref{eq:wcs}.

\begin{claim}[\cite{Borot:2015fxa}]
There is an affine linear change-of-variables $(\tau, t_{\rm B})=\mathscr{L}(\mathtt{t})$ such that
\beq
F^{\rm CS}_g(\Sigma,\mathtt{t}) = F^{\rm
  GW}_g(\cX,\mathscr{L}(\mathtt{t})),\qquad W^{\rm CS}_{g,h}(\Sigma,\cK;
\mathtt{t}, \lambda_1, \dots, \lambda_h) = W^{\rm
  GW}_{g,h}(\cX,\LL,\mathscr{L}(\mathtt{t}); \lambda_1, \dots, \lambda_h).
\label{eq:gvs}
\eeq
Consequently, the LMO contribution to the Chern--Simons free energy ($f_i=0$
for $i>0$) is
obtained as the corresponding restriction of GW potentials:
\bea
F^{\rm CS}_g(\Sigma,\mathtt{t}_i=\mathtt{t_0}\delta_{i0}) &=& F^{\rm
  GW}_g(\cX,\mathscr{L}(\mathtt{t}))\Big|_{\mathtt{t}_i=\mathtt{t_0}\delta_{i0}},\nn
\\
W^{\rm CS}_{g,h}(\Sigma,\mathtt{t}_i=\mathtt{t_0}\delta_{i0}, \lambda_1, \dots, \lambda_h) &=& W^{\rm
  GW}_{g,h}(\cX,\LL,\mathscr{L}(\mathtt{t}); \lambda_1, \dots, \lambda_h)\Big|_{\mathtt{t}_i=\mathtt{t_0}\delta_{i0}}.
\label{eq:gvw}
\eea
\label{cl:gva}
\end{claim}

I will refer to \eqref{eq:gvs} and \eqref{eq:gvw} as, respectively, the {\it
  strong} and {\it weak A-model Gopakumar--Vafa correspondence for~$\Sigma$}.

\subsubsection{Toda spectral curves and the topological recursion}
\label{sec:toprec}

A major point of the foregoing discussion is
to argue that there exist completions of the Dubrovin--Krichever data
\eqref{eq:dk} of the $\mathrm{E}_8$ relativistic Toda spectral curves in the
form Lagrangian sublattices $\Lambda^L_{\rm PT}\subset \Lambda_{\rm PT}$
leading to the existence of genus zero prepotentials $F_0^{\rm Toda}$ from rigid special K\"ahler
geometry relations\footnote{This type of relations, which condense the
 fact there exists a prepotential for the periods on the mirror curve, have
 different names and tasks in different communities: in gauge theory, they are
 a manifestation of $\cN=2$ super-Ward identities; and in Whitham theory, they codify the existence of a
  $\tau$-structure for the underlying hierarchy.} on $\mathscr{B}_{\mathfrak{g}}$, as
well as higher genus open/closed potentials $F_g^{\rm Toda}$, $W_{g,h}^{\rm
  Toda}$ from the Chekhov--Eynard--Orantin topological recursion
\cite{Eynard:2007kz}, which are purported to be the all-genus solutions of the
open/closed topological B-model with $\mathscr{S}_{\mathfrak{g}}$ as its target
geometry \cite{Bouchard:2007ys}. Following completely analogous statements \cite{Nekrasov:1996cz,
  Bouchard:2007ys,Eynard:2012nj, Aganagic:2002wv,Halmagyi:2003ze} for the $\mathrm{SU}(N)$ case,
and in \cite{Borot:2015fxa} for ADE types other than $\mathrm{E}_8$, it will
be proposed that the open and closed B-model theory on the
relativistic Toda spectral curves $\mathscr{S}_g$ with Dubrovin--Krichever
data specified by \eqref{eq:dk} give in one go the Seiberg--Witten solution of the
five-dimensional $\mathrm{E}_8$ gauge theory in a self-dual $\Omega$-background, the mirror theory of the A-model on
$(Y,L)$ and $(\cX,\LL)$, and a large-$N$ dual of
Chern--Simons theory on $\Sigma$. 

For definiteness, let's put again ourselves at a generic moduli point $(u,\aleph)$ .
The first step to define a prepotential from the assignment \eqref{eq:dk} to
$\mathscr{S}_{\mathfrak{g}}$ is to consider periods of $\rd\sigma=\log \mu \rd \log \lambda$
on $\Lambda_{\rm PT}$  \cite{Strominger:1990pd,Dubrovin:1992eu,Krichever:1992qe}. At genus zero, define
\beq
\Pi_{A_i}(\rd\sigma)=\frac{1}{2\pi \ri} \a_i= \frac{1}{2\pi \ri}\oint_{A_i}
\rd \sigma, \quad \Pi_{B_i}(\rd\sigma)=\frac{1}{2} \oint_{B_i}
\rd \sigma,
\label{eq:piAB}
\eeq
for the set of $(A_i,B_i)_{i=1}^8$ cycles generating the
$\mathscr{P}_\mathfrak{g}$-invariant part ot $H_1(\Gamma_{u, \aleph},
\bbZ)$. I am first of all going to fix $\Lambda^L_{\rm PT} \triangleq \bbZ\bra\{A_i\}_i \ket $;
what this means is that, locally around $a_i=\infty$, the A-periods
\eqref{eq:piAB} will define a map 
\bea
a_i : \mathscr{B}_{\mathfrak{g}} & \to & \bbC \nn \\
(u,\aleph) & \to & \Pi_{A_i}(\rd \sigma).
\eea
with the B-periods \eqref{eq:piAB} being further subject to the {\it rigid special K\"ahler
relations} \cite{Strominger:1990pd,Dubrovin:1992eu,Krichever:1992qe}
\beq
\Pi_{B_i}(\rd\sigma)= \frac{\de F^{\rm Toda}}{\de a_i}
\label{eq:preptoda}
\eeq
for a locally defined analytic function $F^{\rm Toda}(a)$ in a punctured
neighbourhood of $a_i=\infty$. 
\begin{conj}
We have 
\beq
F_0^{\rm Toda} = F_0^{\rm SYM} = F_0^{Y}
\eeq
locally around $a_i=\infty=t_i$, under the identifications of \eqref{eq:geoeng},
and after setting $\aleph=R =
\re^{-t_{\rm B}/4}$. Furthermore, let $\tilde A_i\triangleq -B_i$, $\tilde B_i
\triangleq A_i$ and define
\beq
\tilde a_i \triangleq \frac{1}{2\pi \ri } \Pi_{\tilde A_i }(\rd\sigma), \quad
\frac{\de \tilde F^{\rm Toda}_0}{\de a_i}\triangleq \frac{1}{2} \Pi_{\tilde B_i}(\rd\sigma)
\eeq
Then there exist linear change-of-variables $\tilde a = \mathscr{L}_1(\tau) =
\mathscr{L}_2(\mathtt{t})$ such that
\beq
\tilde F_0^{\rm Toda}(\tilde a) =
F_0^{X_{\mathbb{I}}}(\mathscr{L}_1^{-1}\tau)  = F_0^{\rm CS}(\mathscr{L}_2^{-1}\mathtt{t}).
\eeq
\label{cl:toda0}
\end{conj}
For the reader familiar with \cref{fig:dualities} in the $\mathrm{SU}(N)$ case, this
is all by and large expected provided we show that our choice of $A$ and $B$ cycles
in \eqref{eq:AB1}-\eqref{eq:AB2} reflects the corresponding choice of SW
cycles in the weakly coupled
(electric) duality frame in the gauge theory, and of mirror B-model cycles for the smooth chamber in the stringy K\"ahler moduli
space of $Y$: that would justify the first part of the claim, with
the second following by composing with the $S$-duality transformation
$(A_i,B_i) \to (\tilde A_i, \tilde B_i)$ to the orbifold/Chern--Simons
chamber. For the first bit, I re-introduce $\Lambda_4$ everywhere on the
gauge theory side by dimension counting and take the limit $\Lambda_4\to 0$
holding fixed $a_i$ and $R$, which corresponds to switching off the
non-perturbative part of \eqref{eq:gprep}. At the level of the Toda chain
variables this is $\aleph\to 0$ with $u_i$ kept fixed. Recall that the branch
points $b_i^\pm$ of $\lambda: \Gamma_{u, \aleph}$ come in pairs related
by
\beq
b_i^=\frac{\aleph}{b_i^+}.
\label{eq:bipm}
\eeq
In particular, in the degeneration limit $\aleph\to 0$, where $\Gamma_{u,
  0} \simeq \Gamma'_{u, 0}$, the branch points $b_i^-$ in \cref{fig:cuts}
all collapse to zero, and therefore, the contours $C_i^\alpha$ are given by
the difference of the lifts to
the sheet labelled by $\alpha$ and $\sigma_i(\alpha)$ of a simple loop around the origin in the
$\lambda$-plane. In other words, and in terms of the Cartan torus element
$\exp(l)$ in \cref{eq:xi3}, we find
\bea
\lim_{\aleph \to 0} \oint_{A_i}\rd\sigma &=& \lim_{\aleph \to 0} \frac{1}{2q_\mathfrak{g}} \sum_{\a \in \Delta^*}
\bra \alpha, \alpha_i\ket \oint_{C_i^\alpha} \log \mu \frac{\rd
  \lambda}{\lambda}, 
 \nn \\ &=&
 \frac{1}{2q_\mathfrak{g}} \sum_{\a \in \Delta^*}
\bra \alpha, \alpha_i\ket \oint_{\lambda=0}  \lim_{\aleph \to 0} \frac{(\sigma_i(\a)(l) - \a(l))\rd
  \lambda}{\lambda},
 \nn \\ &=&
 \frac{1}{2q_\mathfrak{g}} \sum_{\a \in \Delta^*} (\bra \alpha, \alpha_i\ket)^2
 \a_i(l)|_{\aleph=\lambda=0} =\a_i(l)|_{\aleph=\lambda=0},
\label{eq:clvev}
\eea
where we have used (see \cite{MR1013158,Longhi:2016rjt})
\beq
\frac{1}{2q_\mathfrak{g}} \sum_{\a \in \Delta^*} (\bra \alpha, \alpha_i\ket)^2 = 1.
\eeq
The r.h.s. of \eqref{eq:clvev} is just the semi-classical Higgs vev
$(a_i)_{\Lambda_4 =0}$ for the
complexified scalar $\phi=\varphi+\ri \vartheta$ \cite{Martinec:1995by}. This
pins down $A_i$ as the correct choice of an electric cycle for the $i^{\rm
  th}$ $U(1)$ factor in the IR theory, with logarithmic
monodromy around the weakly coupled/maximally unipotent monodromy point
$a_i=\infty$, and $B_i$ (up to monodromy) as their doubly-logarithmic
counterpart. 
%
%\end{rmk}

The identifications in \cref{cl:toda0} pave the way to an extension  to the higher genus theory upon appealing to the
remodelled-B-model recursive scheme of \cite{Bouchard:2007ys}. Let $\Psi$ be a
sub-lattice of $H_1(\Gamma_{u, \aleph},\bbZ)$ containing $\{A_i\}_i$
which is maximally isotropic w.r.t. the intersection pairing. 
%For $p,q \in
%\Gamma_{u, \aleph}$, $p\neq q$, denote
Denote by
%
%\beq
$B^{\rm Toda} \in H^0(\mathrm{Sym}^2 \Gamma_{u, \aleph}\setminus \Delta(\Gamma_{u,\aleph})), \cK^{\boxtimes 2}_{\Gamma_{u,
    \aleph}})$ 
%\eeq
% 
%as 
the unique (up to scale) meromorphic bidifferential on $\Gamma_{u, \aleph}$ with double
pole on the diagonal $\Delta(\Gamma_{u,\aleph})$, vanishing residues thereon, and vanishing periods on all
cycles $C \in \Psi$; we fix the scaling ambiguity by imposing the coefficient
of the double pole to be $1$ in the local coordinate patch given by the
$\lambda$ projection. I further write
\beq
B^{\rm Toda}(p,q) \triangleq \mathscr{P}^*_{\mathfrak{g}} E_\Psi(p,q)
\label{eq:BKMS}
\eeq
whose definition, by the nature of $(\mathscr{P}_{\mathfrak{g}})_*$ as a
projection on $\mathrm{PT}(\Gamma_{u,\aleph})$, is independent of the choice of the
particular Lagrangian extension $\Psi \supset \Lambda_{\rm PT}$. Further write, for $\lambda(q)$ locally near $b_i^\pm$, 
\beq
\cK_{0,2}^{\rm Toda}(p,q) \triangleq \frac{1}{2}\frac{\int_{q'=\bar q}^{q}
  B^{\rm Toda}(p,q')}{\log\mu (\bar q)-\log\mu(q)},
\label{eq:recker}
\eeq
where locally around each ramification point $\lambda^{-1}(b_i^\pm)$, 
$\bar q$ is the local deck transformation $\mu(q)=\a\cdot l \to
\a \cdot l + \bra \a_i, \a \ket \a_i \cdot l$. 
We call $B^{\rm Toda}$ and
$\cK^{\rm Toda}$ respectively the {\it
  symmetrised Bergmann kernel and recursion kernel for the DK data \eqref{eq:dk}.}
%$\mathscr{S}_{\mathfrak{g}}(\Gamma_{u,
%    \aleph}, \rd \sigma, \Lambda_{\rm PT})$}.
%
\begin{rmk} In terms of the
Dubrovin--Krichever data \eqref{eq:dk}, notice that the family of
differentials $B^{\rm Toda}$ is determined by $\mathscr{S}_{\mathfrak{g}}$ and
$\Lambda^L_{\rm PT}\subset \Lambda_{\rm PT}$ alone -- that is, by the curves themselves,
the invariant periods $\Lambda_{\rm PT}$, and the specific marking of the ``A''
cycles in $\Lambda^L_{\rm PT}$ to be those with vanishing periods for $B^{\rm
  Toda}$. 
% $\Lambda_{\rm   PT}$. 
On the other hand, $\cK^{\rm Toda}$ feels on top of that
the specific choice of relative differential $\mathscr{M}\leftrightarrow \rd\ln\mu$ in \eqref{eq:dk},
which is reflected by the presence of the logarithm of the universal map
$\mu$ to $\bbP^1$ of \eqref{eq:diagsc} in the denominator of \eqref{eq:recker}. The further
choice of $\mathscr{L}\leftrightarrow \rd \ln\lambda$ will play a role
momentarily in the definition of the topological recursion.
\end{rmk}
\begin{defn}
For $g,h\in \bbN$, $2g-2+h>0$, the Chekhov--Eynard--Orantin generating functions \cite{Chekhov:2005rr,Eynard:2007kz} for the Toda spectral curve
$\mathscr{S}_{\mathfrak{g}}$ with DK data \eqref{eq:dk}
%$(\Gamma_{u, \aleph},\rd\sigma, \Lambda_{\rm PT})$ 
are recursively defined as
\bea
\label{eq:w02rec}
W_{0,2}^{\rm Toda}(p,q) & \triangleq & \frac{B^{\rm Toda}(p,q)}{\rd p \rd q}-\frac{\lambda'(p)
  \lambda'(q)}{(\lambda(p)-\lambda(q))^2}, \\
W_{g,h+1}^{\rm Toda}(p_0, p_1 \ldots, p_h) &=& \sum_{b_i^\pm}  \underset{\lambda(p)=b_i^\pm}{\rm Res~} \cK_{0,2}^{\rm Toda}(p_0,p) \Big ( W^{\rm Toda}_{g-1,h+2} (p, \bar{p}, p_1, \ldots, p_{h} )\nn \\
&+& {\sum_{l=0}^g}  {\sum'_{J\subset H}} W^{\rm Toda}_{g-l,|J|+1}(p, p_J) W^{\rm Toda}_{(l),|H|-|J| +1} (\bar{p}, p_{H\backslash J}) \Big).
\label{eq:wghrec}
\eea
where $I \cup J=\{p_1,\dots, p_h\}$, $I \cap J=\emptyset$, and $\sum'$ denotes
omission of the terms $(h,I)=(0,\emptyset)$ and $(g,J)$. Furthermore, for
$g>0$ we define the higher genus free energies
\bea
F_1^{\rm Toda} & \triangleq & \frac{1}{2} \l[-\log \tau_{\rm KK}(\Xi_{\mathfrak{g},\rm red})
  +\frac{1}{12} \log \det \Omega\r], \nn \\
F_g^{\rm Toda} & \triangleq &\frac{1}{2-2g}  \sum_{b_i^\pm}
\underset{\lambda(p)=b_i^\pm}{\rm Res~} \sigma(p) W_{g,1}^{\rm Toda},
\label{eq:fgrec}
\eea
where $\tau_{\rm KK}$ is the Kokotov--Korotkin $\tau$-function of the branched
cover $\Xi_{\mathfrak{g},\rm red}$ \cite{MR2211143}, $\Omega$ is the
Jacobian matrix of angular frequencies \eqref{eq:angfreq}, and $\sigma$ is the
Poincar\'e action \eqref{eq:sigma}.
\label{defn:EO}
\end{defn}
\eqref{eq:wghrec} is the celebrated topological recursion of
\cite{Eynard:2007kz}, which inductively defines generating functions
$\{W_{g,h}^{\rm Toda}\}_{g,h}$ purely in terms of the Dubrovin--Krichever data \eqref{eq:dk}.
%$(\Gamma_{u, \aleph}, \rd\sigma \Lambda_{\rm PT})$. 
The root motivation of
\cref{defn:EO}, which arose in the formal study of random matrix models, is that the generating functions thus constructed provide a
solution of Virasoro constraints  whenever the
spectral curve setup arises as the genus zero solution of the planar loop
equation for the 1-point function; it was put forward in
\cite{Bouchard:2007ys}, and further elaborated upon in
\cite{Dijkgraaf:2007sx}, that the very same recursion solves $\cW$-algebra
constraints for the open/closed Kodaira--Spencer theory of
gravity/holomorphic Chern--Simons theory on local
Calabi--Yau threefolds of the form
$$\nu \xi = \Phi (\lambda, \mu),$$
with
B-branes wrapping either of the lines $\nu=0$ or $\xi=0$. We follow the same
path of \cite{Bouchard:2007ys,Borot:2015fxa} by setting
$\Phi=\Xi_{\mathfrak{g}, \rm red}$, taking
\eqref{eq:wghrec}-\eqref{eq:fgrec} as the {\it definition} of the higher genus/open string completion of
the Toda prepotential \eqref{eq:preptoda}, and submit the following 
\begin{conj}
We have 
\bea
F_{g}^{\rm Toda} =F_g^{\rm SYM} = F_g^{\rm GW}
\eea
locally around $a_i=\infty=t_i$ and under the same identifications of
\cref{cl:toda0}; here we defined the gravitational correction
\beq
F_g^{\rm SYM} = [\epsilon^{2g}]F^{\rm SYM}(\epsilon, -\epsilon),
\eeq
as the $\cO((\epsilon_1=-\epsilon_2)^{2g})$ coefficient in an expansion of the
$\Omega$ background around the flat space limit. Furthermore, denote by
$(\widetilde W_{g,h}, \widetilde F_g)$ the Toda/CEO generating functions obtained upon
applying \eqref{eq:wghrec}-\eqref{eq:fgrec} to the Toda spectral curves with
zero $\tilde A_i$-period normalisation for \eqref{eq:BKMS} and
\eqref{eq:recker}. Then, with the same notation as in \cref{cl:toda0}, we have that
\bea
\widetilde F_g^{\rm Toda}(\tilde a) &=&
F_g^{\rm GW}(\cX; \mathscr{L}_1^{-1}t)  = F_g^{\rm
  CS}(\mathscr{L}_2^{-1}\mathtt{t}), \nn \\
\widetilde W_{g,h}^{\rm Toda}(\tilde a, \lambda_1, \dots, \lambda_h) &=&
W_{g,h}^{\rm GW}(\cX, \LL; \mathscr{L}_1^{-1}t, \lambda_1, \dots, \lambda_h)  = W_{g,h}^{\rm
  CS}(\mathscr{L}_2^{-1}\mathtt{t}, \lambda_1, \dots, \lambda_h), \nn \\
\eea
where we have identified $\lambda_i=\lambda(p_i)$.
\label{cl:todag}
\end{conj}
As in \cref{cl:gva}, I will refer to the equality of Toda and Chern--Simons
generating functions as the {\it strong/weak B-model
  Gopakumar--Vafa correspondence for $\Sigma$}, according to whether the
restriction to the trivial connection $\mathtt{t}_i=t_0\delta_{i0}$ is taken or not.

\begin{rmk}
The two claims above are slightly asymmetrical between $Y$ and
$\cX$ in that they do not include the open string sector in the
latter. On the GW side, exactly by the same token as for the orbifold chamber
and in keeping with the toric cases \cite{Bouchard:2007ys}, the same type of
statement should hold, namely that the topological recursion potentials
$W_{g,h}^{\rm Toda}$ equate to $W_{g,h}^{Y,L}$; for the gauge theory, the
extension one is after requires the introduction to surface defects in the
gauge theory \cite{Alday:2009fs,Kozcaz:2010af}. I do not further discuss
these here, and refer the reader to \cite{Kozcaz:2010af,Borot:2015fxa} for
more details.
\end{rmk}

\subsection{On the Gopakumar--Vafa correspondence for the Poincar\'e sphere}
\label{sec:gvproof}

After much conjecturing I will prove at least one of the correspondences of
the previous section. In the next section, I will show that the weak version
of the 
B-model Gopakumar--Vafa correspondence holds for all genera, colourings, and
degrees of expansion in the 't Hooft parameter.

\subsubsection{LMO invariants and matrix models}
\label{sec:LMOmm}

I will set  
\bea
F^{\rm LMO}_{g}(\Sigma; \tau) &=& F^{\rm CS}_{g}(\Sigma; \mathtt{t}_i=\tau \delta_{i0},
x) \nn \\
W^{\rm LMO}_{g,h}(\Sigma, \cK; \tau, \lambda_1, \dots, \lambda_h) &=& W^{\rm CS}_{g,h}(\Sigma,\cK; \mathtt{t}_i=\tau \delta_{i0}, \lambda_1, \dots, \lambda_h) 
\eea
to designate the LMO contribution ($f_i=0$) to the Chern--Simons partition
function \eqref{eq:fcs2} of $\Sigma$, and quantum invariants of the fibre knot
$\cK$ respectively; similarly I will use $Z^{\rm LMO}$ for the restricted
partition function. The first step to relate the latter to spectral curves, as in \cite{Aganagic:2002wv}, is  
to re-write \eqref{eq:fcs2} as a matrix model
%
%enjoys an integral representation that
%is particularly suited for the study of its large $N$ limit, 
as first pointed
out in \cite{Marino:2002fk} (see also \cite{MR2054838,Blau:2013oha}): this
follows from %taking % Indeed,
taking a Gaussian integral representation of the exponential in \eqref{eq:fcs2} and
using Weyl's denominator formula. The upshot \cite{Marino:2002fk} is that the
restriction of \eqref{eq:fcs2} to its summand at $f_i=0$ is the
total mass of an eigenvalue model
\beq
Z_{\rm LMO}(\Sigma,k,N) =  \cN(\Sigma)
\mathbb{E}_{\rd\mu}(1)=\cN(\Sigma)\int_{\bbR^N} \rd \mu,
\label{eq:ZLMO}
\eeq
with measure given by a Gaussian 1-body potential, and a trigonometric Coulomb
2-body interaction,
\beq 
\rd\mu \triangleq 
\rd^N \kappa \,\prod_{i<j} \frac{\prod_{l = 1}^3
  \sinh\frac{\kappa_i-\kappa_j}{{2\mathsf{s}}(l)}}{\sinh \frac{\kappa_i-\kappa_j}{2}} 
\re^{-\frac{N \kappa\cdot \kappa}{2 \tau}}.
\label{eq:dmLMO}
\eeq
with $\tau = g_{\rm YM} N$, $g_{\rm YM}=2\pi \ri (k+N)^{-1}$. The integral of \eqref{eq:ZLMO} is by {\it fiat} a convergent matrix
(eigenvalue) model, and it takes the
form of a
perturbation of the ordinary (gauged) Gaussian matrix model by double-trace
insertions, owing to the sinh-type 2-body interaction of the eigenvalues (see
\cite[Sec.~6]{Aganagic:2002wv}). The Chern--Simons knot invariants
\eqref{eq:wcs} are similarly computed as
\beq
W^{\rm LMO}_h(\Sigma,\cK, k,N, \lambda_1, \dots, \lambda_h) =  
\mathbb{E}^{\rm conn}_{\rd\mu}\l( \prod_{i=1}^h \sum_{j=1}^N
\frac{x_i}{x_i-\re^{\kappa_i}}\r)
\label{eq:WLMO}
\eeq
where the coefficients of degree $k_i$ in $\lambda_i$, for $k_i = (30/{\mathsf{s}}(l))
j_i$ and $j_i \in \bbZ$, gives the perturbative quantum invariant (in
colouring given by the $j^{\rm th}$ connected power sum) 
 of the knot going along the fiber of order ${\mathsf{s}}(l)$ in ${\mathsf{s}}$, $l=1,2,3$.

This type of eigenvalue measures falls squarely under the class of
$N$-dimensional eigenvalue models considered in \cite{Borot:2013pda}, for
which the authors rigorously prove that a topological expansion of the form
\eqref{eq:fcs} and \eqref{eq:wcs} applies to the asymptotic expansion of
\eqref{eq:ZLMO} and \eqref{eq:WLMO} respectively. What is more, in
\cite{Borot:2013lpa} the authors prove that the
topological recursion \eqref{eq:wghrec}-\eqref{eq:fgrec}
% to \eqref{eq:ZLMO} as proved in
 with initial data for the induction given by
\bea
\label{eq:WCS01}
W^{\rm LMO}_{0,1}(x) & \triangleq & \lim_{N \rightarrow \infty}
\frac{1}{N}\mathbb{E}_{\rd \mu}\l( \sum_{i = 1}^N \frac{x}{x -
  \re^{\kappa_i}}\r),  \\
W^{\rm LMO}_{0,2}(x_1,x_2) & \triangleq & \lim_{N \rightarrow \infty} \Bigg[
  \mathbb{E}_{\rd \mu}\l( \sum_{i_1,i_2 = 1}^N \frac{x_1x_2}{(x_1 -
    \re^{\kappa_{i_1}})(x_2 - \re^{\kappa_{i_2}})}\r) \nn \\ &-&
  \mathbb{E}_{\rd \mu}\l(\sum_{i_1 = 1}^N \frac{x_1}{x_1 -
    \re^{\kappa_{i_1}}}\r) \mathbb{E}_{\rd \mu}\l( \sum_{i_2 = 1}^N
  \frac{x_2}{x_2 - \re^{\kappa_{i_2}}}\r)\Bigg].
\label{eq:WCS02}
\eea
computes the all-order, higher genus, all-colourings quantum invariants of
fibre knots $\cK$.  
As is typical in most settings where the topological
recursion applies, the planar two point function \eqref{eq:WCS02} can be
written as a section $W^{\rm CS}_{0,2}\in K^{\boxtimes 2}_{\Gamma_\tau^{\rm LMO}}(\mathrm{Sym}^2
\Gamma_\tau^{\rm LMO} \setminus \Delta(\Gamma_\tau^{\rm LMO}))$ on the double
symmetric product (minus the diagonal) of the smooth completion $\Gamma^{\rm LMO}_\tau$ of the
algebraic\footnote{From the discussion above this does not need to be more
  than just analytic; it turns however that $\re^{y}= \re^{W^{\rm
      CS}_{0,1}(x)}$ is algebraic, as follows from the proof of
  \cite[Prop.~1.1]{Borot:2014kda}, and as we will review in
  \cref{sec:LMOcurve}.} plane curve
$y=W_{0,1}^{\rm LMO}(x)$: the {\it LMO spectral curve}. A strategy to determine the
family of Riemann surfaces $\Gamma^{\rm LMO}_\tau$ as the
base parameter $\tau$ is varied was put forward in the extensive analysis
of Chern--Simons-type matrix models of \cite{Borot:2014kda}, and is summarised
in the next Section.

\subsubsection{The planar solution, after Borot--Eynard}
\label{sec:LMOcurve}

The LMO spectral curve can be expressed as the solution of the singular integral
equation describing the equilibrium density for the eigenvalues in
\eqref{eq:ZLMO} \cite{Borot:2014kda}. %First of all, as in the case of the
%Hermitian 1-matrix model, 
Introduce the density distribution
\beq
\varrho(x) \triangleq \lim_{N \rightarrow \infty} \frac{1}{N}\mathbb{E}_{\rd
  \mu}\l( \sum_{i = 1}^N \delta(x-\re^{\kappa_i})\r).
\label{eq:densev}
\eeq
As in the case of the Wigner distribution, Borot--Eynard in
\cite{Borot:2014kda} prove that, for $\tau \in \bbR^+$, the large $N$ eigenvalue
density $\varrho \in
C^0_{\rm c}(\bbR)$ is a continuous function with compact support $\cC_\varrho
= [-b(t),
  b(t)]$  given by a single segment, symmetric around
the origin, at whose ends $\pm b(\tau)$
$\varrho$ has square-root vanishing, $\varrho=\cO(\sqrt{x\pm b (\tau)}$. Furthermore, by \eqref{eq:ZLMO}, $\varrho$ satisfies the saddle-point equation
\beq
\label{eq:sp}
\frac{\kappa}{\tau} = \sum_{l = 1}^{3} \mathrm{pv} \int_\bbR
\varrho(\kappa')\l[  \coth{\frac{\kappa - \kappa'}{2{\mathsf{s}}(l)}} -\coth{\frac{\kappa - \kappa'}{2}} \r].
\eeq
By the Pl\'emely lemma, this is equivalent to a Riemann--Hilbert problem for the
planar 1-point function \eqref{eq:WCS01},
\beq
W_{0,1}^{\rm LMO}(x + {\rm i}0) + W_{0,1}^{\rm LMO}(x - {\rm i}0) -\sum_{\ell =
  1}^{{\mathsf{s}}} W_{0,1}^{\rm LMO}(\zeta^{\ell}x) + \sum_{m = 1}^3 \sum_{\ell_m = 1}^{{\mathsf{s}}/{\mathsf{s}}(m) - 1}
   W_{0,1}^{\rm LMO}(\zeta_{{\mathsf{s}}/{\mathsf{s}}(m)}^{\ell_m}x) = ({\mathsf{s}}^2/\kappa)\ln x + {\mathsf{s}}
\eeq
with $\zeta_k$ a primitive $k$-th root of
unity; note that $W^{\rm LMO}_{0,1}(x)$ has a cut for $x\in \cC_\varrho\triangleq \mathrm{supp}\varrho$, with
jump equal to $2\pi\ri \varrho$. Following \cite{Brini:2011wi}, and setting
\beq
c \triangleq \exp(\tau/2{\mathsf{s}}).
\eeq 
the exponentiated resolvent
\beq
\mathcal{Y}(x) \triangleq -cx\exp\l(\frac{\tau W_{0,1}^{\rm LMO}(x)}{\mathsf{s}^2}\r),
\label{eq:Ydef}
\eeq
is holomorphic on $\mathbb{C}\setminus \cC_\varrho$,  it asymptotes to 
\bea
\mathcal{Y}(x) & \sim & -cx, \quad x=0, \nn \\
\mathcal{Y}(x) & \sim &-c^{-1}x, \quad x=\infty,
\label{eq:asympY}
\eea 
and further satisfies
\beq
\qquad \mathcal{Y}(x + {\rm i}0)\mathcal{Y}(x -
      \ri 0) \l[\prod_{\ell = 1}^{{\mathsf{s}} - 1}
          \mathcal{Y}(\zeta_{{\mathsf{s}}}^{\ell}x)\r] ^{-1}
          \times \prod_{m = 1}^{3} \l[\prod_{\ell_m = 1}^{{\mathsf{s}}/{\mathsf{s}}(m) - 1}
            \mathcal{Y}(\zeta_{{\mathsf{s}}/{\mathsf{s}}(m)}^{\ell_m})\r] = 1.
\label{eq:mono}
\eeq
Furthermore, the $\bbZ_2$-symmetry $\{\kappa_i \rightarrow -\kappa_i\}$ of
\eqref{eq:ZLMO} entails that
\beq
\mathcal{Y}(x)\mathcal{Y}(1/x) = 1.
\eeq
Every time we cross the cut $\cC_\varrho$, the exponentiated resolvent is
subject to the monodromy transformation \eqref{eq:mono}. An approach to solve
the monodromy problem \eqref{eq:mono} together the asymptotic conditions at
$0$ and $\infty$ was systematically developed in \cite{Borot:2014kda}
following in the direction of \cite{Brini:2011wi}, and it goes as follows. Fix
$v \in \bbZ^{{\mathsf{s}}}$ and let
\beq
\label{Ydefff}\mathcal{Y}_{v}(x) \triangleq \prod_{j = 0}^{{\mathsf{s}} - 1} [\mathcal{Y}(\zeta_{{\mathsf{s}}}^{j}x)]^{v_j},
\eeq 
Here $\mathcal{Y}_{v}(x)$ inherits a cut on the rotation $\cC_\varrho^{(j)} =
\zeta_{{\mathsf{s}}}^{-j} \cC_{\varrho}$ for all $j$ such that $v_j \neq 0$; in
particular, the jump on each of these cuts returns the spectral density
$\varrho$, and thus $W^{\rm LMO}_{0,1} (x)$. 

By definition, $\mathcal{Y}_{v}(x)$ is a single-valued function on the
universal cover $\widehat{\Gamma}$ of $\bbP^1\setminus \{ \zeta_{\mathsf{s}}^j
b^\pm(\tau)\}_{j=1}^{\mathsf{s}}$. We want to ask whether there is a clever choice of
$v$ such that this factors through a {\it finite degree} covering map $\Gamma_{\rm LMO}
\to \bbP^1$ branched at $\{ \zeta_{\mathsf{s}}^j b^\pm(\tau)\}_{j=1}^{\mathsf{s}}$ such that
$\mathcal{Y}_{v}(x)$ is single-valued on $\Gamma_{\rm LMO}$. This was answered
in the affirmative in \cite{Borot:2014kda}, as follows. A direct consequence of \eqref{eq:mono}, as
in the study of the torus knots matrix model of \cite{Brini:2011wi}, is that
the change-of-sheet transition given by crossing the cut $\cC_\varrho^{(j)}$
results in a lattice automorphism $T_j \in {\rm GL}({\mathsf{s}},\mathbb{Z})$ such that
\beq
\mathcal{Y}_{v}(x + {\rm i}0) = \mathcal{Y}_{T_j(v)}(x - {\rm i}0)\,.
\eeq
The monodromy group of the local system determined by $\mathcal{Y}_{v}(x)$ is
then (a subgroup of) the group of lattice transformations $T_j$ for $j =
0,\ldots,({\mathsf{s}} - 1)$. This is beautifully characterised by the following
\begin{prop}[\cite{Borot:2014kda}]
There is a $\bbZ$-linear monomorphism 
\beq
\iota: \Lambda_r \to \bbZ^{\mathsf{s}}
\eeq
embedding $\Lambda_r(\mathfrak{e}_8)$ as a rank~8 sublattice of
$\bbZ^{\mathsf{s}}$. Its image $\iota(\Lambda_r)$ is invariant under the $\{T_j\}_j$-action,
and the pullback of the monodromy \eqref{eq:mono} to $\Lambda_r$
is isomorphic to the Coxeter action of $\cW=\mathrm{Weyl}(\mathfrak{e}_8)$.
\label{prop:mono}
\end{prop}
By \cref{prop:mono}, picking $v$ to lie in $\iota(\Lambda_r)$ does exactly the
trick of returning a finite degree covering of the complex line by the affine curve
\beq
y: \mathbb{V}\l[\prod_{\varpi \in\iota(\cW) v}
  \l(y-\mathcal{Y}_{\varpi}(x)\r)\r] \to \bbA^1,
\label{eq:LMOsc}
\eeq
with sheets labelled by elements of a $\cW$-orbit on $\Lambda_r$. Our freedom
in the choice of the initial element $v$ in the orbit is given by the number of
semi-simple, $7$-vertex Dynkin subdiagrams of the black part of \cref{fig:dynke8}
\cite{MR1778802}, which classify the stabilisers of any given element in the
orbit; in other words, by the choice of a fundamental weight $\omega_i$ of
$\mathfrak{g}$. The natural choice here is to pick the
minimal orbit, corresponding to the largest stabilising group, by choosing
to delete the node $\alpha_7$ in \cref{fig:dynke8}, so that
$v=\omega_7=\alpha_0$: in this case, obviously, $\cW v = \Delta^*$, the set of
non-zero roots. I refer the reader to \cref{sec:minorb} for further details
on the orbit, and give the following
\begin{defn}
We call the  normalisation of the closure in $\bbC\bbP^2$ of \eqref{eq:LMOsc}
 with $v=\iota(\alpha_0)$ the {\rm LMO curve of type $\mathrm{E}_8$}.
\end{defn}
This places us in the same setup of the Toda spectral curves of
\cref{sec:speccurve,sec:cameral} (see in particular \eqref{eq:xi3} and
\cref{def:sc}), by realising the LMO curve as a curve of eigenvalues for a
$\cG$-valued Lax operator with rational spectral parameter; at this stage, of course,
it is still unclear whether this rational dependence has anything to do with
that of \cref{eq:laxaff2}. The upshot of the discussion above is that that there exists a degree-240, monic
polynomial $\mathcal{P}_{\alpha_0} \in \bbC[x,y]$ with $y$-roots given exactly
by the branches of the $\bbZ_{\mathsf{s}}$-symmetrised, exponentiated resolvent $\cY(x)$:
\beq
\mathcal{P}_{\alpha_0}(x,y) =  \prod_{\a\in\Delta^*}^{240} \l(y - \mathcal{Y}_{\varpi_\alpha}(x)\r).
\label{eq:Pa0}
\eeq
where we wrote $\varpi_\alpha\triangleq \iota(\a)$. As we point out in \cref{sec:minorb}, the rescaling $x\to \zeta_{\mathsf{s}}^{-1} x$ corresponds to an
action on $\bbZ^{\mathsf{s}}$ given by the image of the action of the Coxeter element
on $\Lambda_r$, under which the orbit $\Delta^*$ is obviously
invariant. The resulting $\bbZ_{\mathsf{s}}$-symmetry implies that $\mathcal{P}_{\alpha_0}(x,y)$ is in fact a polynomial in
$\lambda=x^{\mathsf{s}}$, and we define 
\beq
\Xi^{\rm LMO}(\lambda,\mu) \triangleq
\mathcal{P}_{\alpha_0}(\lambda^{1/{\mathsf{s}}},\mu) \in \bbC[\lambda,\mu].
\eeq
Vanishing of $\Xi^{\rm LMO}$ defines a family $\pi: \mathscr{S}_{\rm LMO} \to
\mathscr{B}_{\rm LMO}\simeq \bbA^1$ algebraically varying over a
one-dimensional base $\mathscr{B}_{\rm LMO}$ parametrised by the 't~Hooft parameter
$\tau$; the same picture of \eqref{eq:diagsc} then holds over this smaller
dimensional base.

%\beq
%\tilde{\Omega} = \frac{a}{\chi_{{\rm orb}}\lambda}\ln(-Y/cX)\,\frac{\dd X}{X} = \frac{a}{\chi_{{\rm orb}}\lambda}\,\ln Y\,\frac{\dd X}{X} + \dd f(X)\,.
%\eeq
%Adding the differential of a rational function of $X$ does not change the free energies and correlators computed by the topological recursion, so we can equally choose the %$1$-form:
%$$
%\Omega = \frac{a}{\chi_{{\rm orb}}\lambda}\,\ln Y\,\frac{\dd X}{X}
%$$
%Besides, the only effect of the rescaling by $a/\chi_{{\rm orb}}\lambda$ is that $W_{g,n}$ are multiplied by $(\chi_{{\rm orb}}\lambda/a)^{2g - 2 + n}$. 

\subsubsection{Hunting down the Toda curves}

We are now ready to show the weak B-model Gopakumar--Vafa correspondence,
\cref{cl:todag}. This will follow
from establishing that the LMO spectral curves are a subfamily of Toda curves
with canonical Dubrovin--Krichever data matching with \eqref{eq:dk}.

\begin{thm}
\label{thm:gv}
There exists an embedding
\bea
\mathscr{B}_{\rm LMO} & \hookrightarrow & \mathscr{B}_{\mathfrak{g}} \nn \\
\tau & \longrightarrow & \l(u(c), \aleph(c)\r),
\eea
such that 
\beq
\mathscr{S}_{\rm LMO}=\mathscr{S}_{\mathfrak{g}}\times_{\mathscr{B}_{\mathfrak{g}}} \mathscr{B}_{\rm LMO}.
\eeq
Explicitly, this is realised by the existence of algebraic maps $u_i=u_i(c)$, $\aleph=c^{-q_{\mathfrak{g}}}$ such that
\beq
\Xi_{\rm LMO} = \Xi_{\mathfrak{g},\rm red}\big|_{u=u(c), \aleph=-c^{-q_{\mathfrak{g}}}}.
\label{eq:xixi}
\eeq
Furthermore, the full $1/N$ asymptotic expansion of \eqref{eq:WLMO}
is computed by the topological recursion \eqref{eq:wghrec}-\eqref{eq:fgrec} with
induction data \eqref{eq:WCS01}-\eqref{eq:WCS02}, and the $\cO(\prod_i x_i^{k_i})$
coefficients with $k_i =
({\mathsf{s}}/{\mathsf{s}}(m)) j_i$, $j_i \in \mathbb{Z}$, $m=1,2,3$,  return the $1/N$ expansion of the
perturbative quantum invariants of the knot $\cK_m$ going along the singular fibre of
order ${\mathsf{s}}(m)$ with colouring given by the virtual connected power
sum representation specified by $\{j_i\}$.
%\een
\end{thm}

\begin{proof}

The statement of the first part of the theorem condenses what were called ``Step~A'' and ``Step~B'' in the construction of
LMO spectral curves that was offered in our previous paper
\cite{Borot:2015fxa}, where we
stated that Step~B was too complex to be feasibly completed. I
am going to show how the stumbling blocks we found there can be overcome here.

Let me first recall the strategy of \cite{Borot:2015fxa}. As in \cite{Brini:2011wi}, the first thing we
do is to use the asymptotic conditions \eqref{eq:asympY} for the un-symmetrised
resolvent on the physical sheet (the eigenvalue plane),  to read off the
asymptotics of the symmetrised resolvent $\cY_{\iota(\alpha)}$ on the sheet
labelled by $\alpha$. Let $\varpi=(\varpi_j)_j=\iota(\alpha))$ as displayed in
\cref{tab:minorb}, and further write
\beq
n_0(\varpi) \triangleq \sum_{j = 1}^{{\mathsf{s}} } \varpi_{j},\qquad n_1(v) = \sum_{j = 1}^{{\mathsf{s}} - 1} j \varpi_{j}\,.
\eeq
Then, from \eqref{eq:asympY}, we have
\bea
\label{eq:slope1} x \rightarrow 0\,, & \qquad & \cY_{\varpi}(x) \sim (-cx)^{n_0(\varpi)}\,\zeta_{{\mathsf{s}}}^{n_{1}(\varpi)}\,,\\ 
\label{eq:slope2} x \rightarrow \infty\,, & \qquad & \cY_{\varpi}(x) \sim (-x/c)^{n_0(\varpi)}\,\zeta_{{\mathsf{s}}}^{n_{1}(\varpi)}\,.
\eea
which in one shot gives both the Puiseux slopes of the Newton polygon of
$\mathcal{P}_{\alpha_0}$ as $(\pm 1,n_{0}(\varpi))$, and the coefficients of its boundary lattice points
up to scale; in view of the comparison with $\Xi_{\mathfrak{g}, \rm red}$ we
set the normalisation for the latter by fixing the coefficient of $y^0$ to be
equal to one. Taking into account the symmetries of $\mathcal{P}_{\alpha_0}$
and plugging in the data of \cref{tab:minorb} on the minimal orbit, this is seen to return exactly the Newton polygon and the boundary coefficients
of $\Xi_{\mathfrak{g}, \rm red}$ (see \cref{fig:npol}). 

The remaining part is to prove the existence of the map $u_i(c)$ such that all
the interior coefficients match as well. As was done in \cite{Borot:2015fxa}, I set out to
prove it by working out the constraints due to the global nature of $\cY$ as a
meromorphic function on $\Gamma^{\rm LMO}_\tau$. Write
\beq
\frac{\tau}{{\mathsf{s}}^2}\, W(x) = \sum_{k \geq
  1} \mathfrak{m}_k\,x^{k + 1}\,,
\label{eq:momexp} 
\eeq
for the expansion of the 1-point function \eqref{eq:WCS01} in terms of
the planar moments 
\beq
\mathfrak{m}_k = \lim_{N\to \infty}\mathbb{E}_{\rd \mu} \l( \sum_{i=1}^N \re^{k\lambda_i}\r).
\label{eq:moments}
\eeq
Then, by \eqref{eq:Ydef} and \eqref{eq:mono}, we have that
\beq
\cY_{\varpi_\alpha}(x) = (-cx)^{n_0(\varpi_\alpha)}{\mathsf{s}}_{a}^{n_1(\varpi_\alpha)}\,\exp\l[\sum_{k>0}
\mathfrak{m}_{k - 1}\,(\widehat{\varpi_\alpha})_{k\,\,{\rm mod}\,\,{\mathsf{s}}}\,x^{k}\r]\,,
\label{eq:Yexp}
\eeq
where, as in \cite{Borot:2015fxa}, we wrote
\beq
(\widehat{\varpi_\alpha})_{k} \triangleq \sum_{j = 0}^{{\mathsf{s}} - 1} \zeta_{{\mathsf{s}}}^{jk}\,(\varpi_{\alpha})_i\,.
\eeq
for the discrete Fourier transform of $\varpi_{\alpha}$. There are only eight
Fourier modes that are non-vanishing: these are\footnote{In
  \cite{Borot:2015fxa}, versions 1 and 2, these were erroneously listed as being in the
  complement of the r.h.s. of \eqref{eq:vpnz}.} 
\beq
\exists \alpha | (\widehat{\varpi_\alpha})_{k} \neq 0 \Rightarrow k\in
\{6,10,12,15,18,20,24,30\} =: \mathfrak{k}.
\label{eq:vpnz}
\eeq
In particular, the only moments $\mathfrak{m}_{k}$ that may be found when Taylor-expanding $\cY$ at one of
the pre-images of $x=0$ satisfy $$(k+1)~\mathrm{mod}~{\mathsf{s}} \in
\mathfrak{k}.$$ 
Consider now inserting the Taylor expansion \eqref{eq:Yexp} into the r.h.s.
\eqref{eq:Pa0}.  Without any further constraints on the surviving momenta $\mathfrak{m}_k$, we have no
guarantee a priori that \eqref{eq:Yexp} is indeed (a) the Taylor expansion of a branch
of an algebraic function and (b) that it gives the roots of a polynomial $\cP_{\alpha_0}$ as
  presented in \eqref{eq:Pa0}. This means that if we expand up to power
  $\cO(x^{L+1})$ the product
\beq
\prod_{\a\in\Delta^*}^{240} \l(y - \mathcal{Y}_{\varpi_\a}(x)\r) =
\sum_{i=1}^{L} B_i(y) x^i + \cO\l(x^{L+1}\r)
\label{eq:LMOpolexp}
\eeq
then the polynomial $\sum_{i=1}^{L} B_i(y) x^i$ may well have non-vanishing coefficients
outside the Newton polygon of $\cP^{\rm LMO}$; imposing that these are zero,
and that those at the boundary return the slope coefficients of
\eqref{eq:slope1} and \eqref{eq:slope2}, gives a set of algebraic conditions
on $\{\mathfrak{m}_k\}_{k \mathrm{mod} {\mathsf{s}} \in \mathfrak{k}}$. In
\cite{Borot:2015fxa} we pointed out that the complexity
of the calculations to solve for these conditions is unworkable if taken at face
value, and refrained to pursue their solution; however I am going to show here that 
it is possible to carve out a sub-system of these equations which pins down uniquely an
$8$-parameter family of solutions, provides a solution to {\it
  all} these constraints for arbitrary $L$, and simultaneously leads exactly to the
full family of Toda spectral
curves \eqref{eq:xi2}-\eqref{eq:xi3}. Take $$L=540 = \deg_x \Xi_{\mathfrak{g},
  \rm red}(\mu,x^{q_\mathfrak{g}}) = q_\mathfrak{g} \deg_\lambda \Xi_{\mathfrak{g},   \rm
  red}(\mu,\lambda)$$ and expand \eqref{eq:Yexp} up to $L$. Plugging this
into \eqref{eq:Pa0} and equating to zero leads to an algebraic equation for
each coefficient of $x^{{\mathsf{s}} m}y^{n}$ with $(m,n)$ once we impose that $\Xi_{\rm
  LMO}(\lambda,\mu)= \lambda^{18}\Xi_{\rm LMO}(\lambda^{-1},\mu)$. Consider now the subsystem of
equations given by $(m,n)$ in the region
\beq
\mathfrak{v}\triangleq \l\{(m,n) \in \bbZ^2 \big| 1\leq n\leq 10,~ 1\leq m \leq [L-6 n,{\mathsf{s}}]   \r\}
\eeq
depicted in \cref{fig:npolcs}: imposing that $B_n(y)= B_{L-n}(y)$ for the list of
exponents in $\mathfrak{v}$ gives an algebraic system of 165 equations for the 144
moments $\mathfrak{m}_k$, $L>k \in \mathfrak{k}$. Recall that the planar moments $\mathfrak{m}_k(c)$ are analytic in
$c$ around $c=1$ \cite{Borot:2013pda,Marino:2002fk}  (i.e., small 't~Hooft
parameter),  and they vanish in that limit
\beq
\mathfrak{m}_k(c)=\sum_{n\geq 1}\mathfrak{m}^{(n)}_{k}(c-1)^n.
\eeq
Since the constraints on $\mathfrak{m}_k$ are
analytic, we can solve them order by order in a Taylor expansion around
$c=1$. It easy to figure out from \eqref{eq:Yexp} that, at $\cO(c-1)^n$, we
find a linear inhomogeneous system of the form 
\beq
\cA \mathfrak{m}^{(n)} = \cB_n
\label{eq:const}
\eeq
with $\cA$ independent of $n$, $\cB_1=0$, and $\cB_n$ a polynomial in
$\mathfrak{m}^{(m)}_k$ for $m<n$. From the explicit form of $\cA$ we calculate that
$\dim \mathrm{Ker} \cA=8$; the solutions of the system \eqref{eq:const} 
must then take the form
\beq
\mathfrak{m}^{(n)}_k(c)=\mathfrak{m}_k\l(\{ \mathfrak{m}^{(m)}_l \}_{l\in \mathfrak{k}, m\leq n} \r)
\label{eq:highermom}
\eeq
and they are {\it uniquely determined} by solving recursively \eqref{eq:const} order by order in
$n$. A priori, the Taylor coefficients $\mathfrak{m}^{(m)}_{l+1\in\mathfrak{k}}$ of the
basic moments $\mathfrak{m}_5$, $\mathfrak{m}_{9}$, $\mathfrak{m}_{11}$, $\mathfrak{m}_{14}$, $\mathfrak{m}_{17}$,
$\mathfrak{m}_{19}$, $\mathfrak{m}_{23}$ and $\mathfrak{m}_{29}$ are subject to two further sets of
closed conditions: the first stems from imposing that the inhomogeneous system \eqref{eq:const} has
solutions for $n>1$ (i.e. $\dim \mathrm{Ker} [\cA|\cB_n]=9$), and the second
from coefficients of the expansion \eqref{eq:LMOpolexp} which are outside of
the region $\mathfrak{v}$: this may lead to the solution manifold of
\eqref{eq:const} having positive codimension in $\bbC\bra
\mathfrak{m}^{(m)}_{l+1\in\mathfrak{k}} \ket$. This is however not the
case: let us {\it impose} that 
\beq
[\lambda^9 \mu^i] \Xi_{\rm
  LMO}(\lambda,\mu)  = \mathfrak{p}_i(u_1, \dots, u_8), \quad i=1,\dots, 120.
\label{eq:imptoda}
\eeq
where $\mathfrak{p}_i$ are the decomposition of the antisymmetric characters defined as in \eqref{eq:pk}, \eqref{eq:pkgen} and \cref{claim:E8}.
\eqref{eq:imptoda} gives invertible polynomial maps
$\mathfrak{m}_l \in c^{a_l} \bbC[u_1, \dots, u_8]$; for $l+1\in \mathfrak{k}$ we find
\bea
\mathfrak{m}_5 &=& \frac{c^6}{30}  (u_7+2), \quad \mathfrak{m}_9 = \frac{c^{10}}{30}  (u_1+2
u_7+3),\quad \mathfrak{m}_{11} = \frac{c^{12}}{300} \left(-3 u_7^2+38 u_7+20 u_1+10
u_6+38\right), \nn \\
\mathfrak{m}_{14} &=& \frac{c^{15}}{30}  (5
   u_1+2 u_6+8 u_7+u_8+7),\nn \\
\mathfrak{m}_{17} &=& \frac{c^{18}}{2250} \bigg(7 u_7^3+117 u_7^2+90 u_1
     u_7-30 u_6 u_7+1284 u_7+ 855 u_1+75 u_5+ 315 u_6\nn \\ &+&225
     u_8+1031\bigg),\nn \\
\mathfrak{m}_{19} &=& -\frac{c^{20}}{180}
   \left(u_1^2-26 u_7 u_1-114 u_1-26 u_7^2-6 u_2-12 u_5-48
   u_6-180 u_7-30 u_8-129\right),\nn \\
\mathfrak{m}_{23} &=&
\frac{c^{24}}{15000} \bigg(60 u_6
   u_7^2-11 u_7^4-88 u_7^3-80 u_1 u_7^2+12686 u_7^2+12280 u_1 u_7-100 u_5
   u_7+2240 u_6 u_7 \nn \\ &+& 1200 u_8 u_7+45448 u_7 + 1300 u_1^2-50
   u_6^2+27880 u_1+1000 u_2+500 u_4+2800 u_5+800 u_1 u_6 \nn \\ &+& 10740
   u_6+7400 u_8+28374\bigg),\nn \\
\mathfrak{m}_{29} &=&\frac{c^{30}}{30}  \bigg(14 u_7^3+16 u_1 u_7^2+233 u_7^2+3 u_1^2 u_7+238 u_1 u_7+2 u_2 u_7+7 u_5 u_7+65 u_6 u_7+35 u_8
   u_7\nn \\ &+& 499 u_7+44 u_1^2+3 u_6^2+287 u_1+9 u_2+u_3+2 u_4+2 u_1 u_5+23
   u_5+29 u_1 u_6+108 u_6\nn \\ &+& 11 u_1 u_8+3 u_6 u_8+65 u_8+259\bigg),
\label{eq:muu}
\eea
which are easily seen to have polynomial inverses $u_k \in
\bbC[\{\mathfrak{m}_l\}_l;c^{-1}]$. As we know that $\Xi_{\mathfrak{g}, \rm red}$
and $\Xi_{\rm LMO}$ share the same Newton polygon with the same boundary
coefficients by \eqref{eq:y9}, \eqref{eq:npolbdy} and \cref{tab:ptsinf}, postulating
\eqref{eq:imptoda} is the same as giving an eight parameter family of
polynomial solutions of the constraints \eqref{eq:imptoda} which furthermore
satisfies all our constraints $B_n(y)= B_{L-n}(y)$ for all $n \in [[0,18]]$.
The first part of the claim, \eqref{eq:xixi}, follows then from the uniqueness of the solution of
\eqref{eq:imptoda} above. 

To prove the second part, we show that the two-point functions \eqref{eq:w02rec} and
\eqref{eq:WCS02} coincide. We have just shown that $\Xi_{\mathfrak{g}, \rm red}=\Xi_{\rm LMO}$
under the change-of-variables \eqref{eq:muu}, and we know that the symmetrised
Bergmann kernel of \eqref{eq:BKMS} is completely determined by $\Gamma_{u, \aleph}$ and the choice of $\tilde A_i$ cycles in \cref{cl:todag}: by its
definition in \eqref{eq:BKMS}, it is the unique bidifferential on
$\Gamma_{u, \aleph}$ with vanishing $\tilde A$-periods and double poles
with zero residues at the $240\times 240$ components of
the image of the diagonal in
$\Gamma_{u, \aleph}^{[2]}$ under the correspondence
$\mathscr{P}_{\mathfrak{g}}\times \mathscr{P}_{\mathfrak{g}}$, the
coefficients of the double poles being specified by \eqref{eq:proj} in terms of a
$240 \times 240$ matrix of integers $\mathrm{B}^{\rm Toda}_{ij}$. As
was proved in \cite{Borot:2014kda}, the regularised two-point function,
\beq
\frac{B_{\rm LMO}(p,q)}{\rd p \rd q}  \triangleq  W_{0,2}^{\rm LMO}(p,q)
+\frac{\lambda'(p)
  \lambda'(q)}{(\lambda(p)-\lambda(q))^2},
\eeq
has precisely the same properties: its matrix of singularities in
\cite[Sec.~6.6.3]{Borot:2014kda} can be shown to coincide with
$\mathrm{B}_{ij}^{\rm Toda}$ above, and the vanishing of the $\tilde A$-periods can
be proven exactly as in the case of the ordinary Hermitian 1-matrix model 
to be a consequence of the planar loop equation for the 2-point function \cite{Eynard:2009phf};
we conclude by uniqueness that 
\beq
W_{0,2}^{\rm LMO}=W_{0,2}^{\rm Toda}
\label{eq:weq}
\eeq
under the identification \eqref{eq:muu}. This suffices to reach the conclusion
of the second part of the claim: on the Toda side, the higher generating functions
satisfy the topological recursion relations \eqref{eq:wghrec}-\eqref{eq:fgrec}
by definition\footnote{From a physics point of view, a first-principles heuristic argument to
  prove straight from the Kodaira--Spencer theory of gravity that these are genuine open/closed B-model
  amplitudes may be found in \cite{Dijkgraaf:2007sx}. Also note that, as $Y$ is
  non-toric, it momentarily lies outside the scope of existing proofs of the
  remodelling-the-B-model approach of \cite{Bouchard:2007ys}, which rely
  either on the existence of a topological vertex formalism
  \cite{Eynard:2012nj} or of a torus action on the target with
  zero-dimensional fixed loci \cite{Fang:2016svw}. I nonetheless believe these
  obstructions to be merely of a technical nature.}. On the LMO side, the
higher generating functions \eqref{eq:ZLMO}-\eqref{eq:WLMO} fall within the
class of integrals studied in \cite{Borot:2013lpa}, for which the authors
prove that the Chekhov--Eynard--Orantin recursion determines the
ribbon graph expansion in $1/N$ via \eqref{eq:wghrec}-\eqref{eq:fgrec}. Since
both sides satisfy the recursion, the
recursion kernels coincide from \eqref{eq:weq}, and so do the initial data
$W_{0,1}$ and $W_{0,2}$, the statement of the theorem follows by induction on $(g,h)$.
\end{proof}

\begin{figure}[t]
\centering
\includegraphics{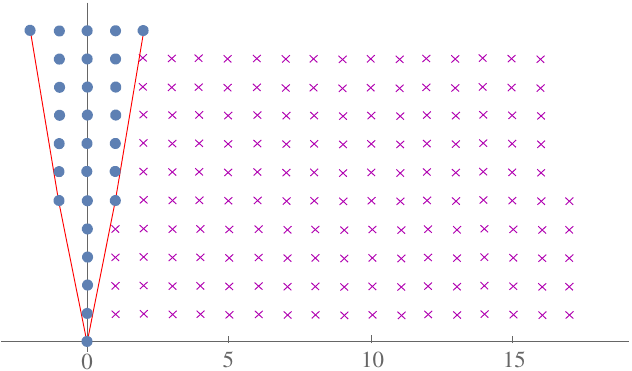}
\label{fig:npolcs}
\caption{Points in $\bbZ^2$ corresponding to the region $\mathfrak{v}$ of non-vanishing coefficients of the
  Taylor expansion of $\cY_{\varpi_\alpha}$ around zero; those indicated with
  a purple cross lie outside the Newton polygon of $\Xi_{\rm LMO}$.}
\end{figure}

\begin{rmk}
The proof of the existence of the embedding $\mathscr{B}_{\rm LMO}
\hookrightarrow \mathscr{B}_{\mathfrak{g}}$ is only constructive up to the point where
the fibres of the LMO family are shown to be determined by the planar moments,
and in turn by the Toda Hamiltonians and Casimirs via
\eqref{eq:muu}. Providing explicit algebraic equations for the restriction $u_i(c)$ 
to the codimension~8 locus $\mathscr{B}_{\rm LMO}$ is however a separate
problem. It is worth pointing out that a direct way of calculating the restriction exists in
perturbation theory around $c=1$ using the Gaussian perturbation theory
methods of \cite{Marino:2002fk}, which allow to determine
$\mathfrak{m}_k^{(n)}$ for arbitrary order in $n$; it would be however
desirable to present a closed-form algebraic
solution by altenative methods, such as the one provided for spherical
3-manifolds of type $\mathrm{D}$ in
\cite{Borot:2014kda} and type $\mathrm{E}_6$ in \cite{Borot:2015fxa}.
\end{rmk}

\subsection{Some degeneration limits}

To conclude this section, I will highlight three degeneration limits of the $\mathrm{E}_8$
relativistic Toda spectral curves which have a neat geometrical interpretation
on each of the other three corners of \cref{fig:dualities}. This is
summarised in the following table, and discussed in detail in the next three
Sections.

\begin{table}[h]
\begin{tabular}{|c|c|c|c|c|c|c|c|c|c|}
\hline
Limit & $\aleph$ & $u_1$ & $u_2$ & $u_3$ & $u_4$ & $u_5$ & $u_6$ & $u_7$ &
$u_8$\\
\hline
I & $1$ & $1$ & $3$ & $0$ & $3$ & $-3$ & $3$ & $-2$ & $-2$\\
\hline
II & 0 & $\gg 1$ & $\gg 1$ & $\gg 1$ & $\gg 1$ & $\gg 1$ & $\gg 1$ & $\gg 1$ & $\gg 1$ \\
\hline
III  & $ \cO(\epsilon) $ & $ \dim \rho_{\omega_1} $ & $ \dim
\rho_{\omega_2} $ & $ \dim \rho_{\omega_3} $ & $ \dim
\rho_{\omega_4} $ & $ \dim \rho_{\omega_5} $ & $ \dim
\rho_{\omega_6} $ & $ \dim \rho_{\omega_7} $ & $ \dim \rho_{\omega_8} $ \\
& & $+\cO(\epsilon^2)$ & $+\cO(\epsilon^2)$ & $+\cO(\epsilon^2)$ &
$+\cO(\epsilon^2)$ & $+\cO(\epsilon^2)$ & $+\cO(\epsilon^2)$ & $+\cO(\epsilon^2)$ &
$+\cO(\epsilon^2)$ \\
\hline
Limit & \multicolumn{3}{|c|}{5d gauge theory} & \multicolumn{3}{|c|}{GW
  target} & \multicolumn{3}{|c|}{CS}  \\ \hline
&   \multicolumn{3}{|c|}{maximal Argyres--} &
 \multicolumn{3}{|c|}{$\tilde \bbI$-orbifold of the}  &
 \multicolumn{3}{|c|}{zero 't Hooft} \\

I & 
\multicolumn{3}{|c|}{--Douglas SCFT} &  \multicolumn{3}{|c|}{singular conifold} & \multicolumn{3}{|c|}{limit}\\
%& \multicolumn{3}{|c|}{SCFT} &  \multicolumn{3}{|c|}{} &
%\multicolumn{3}{|c|}{limit}\\ 
\hline
II & 
 \multicolumn{3}{|c|}{perturbative limit} &
\multicolumn{3}{|c|}{$\widehat{\bbC^2/\tilde{\bbI}}\times \bbC $}  &
\multicolumn{3}{|c|}{?}\\ \hline
III & 
\multicolumn{3}{|c|}{4d limit} &
\multicolumn{3}{|c|}{?}  &
\multicolumn{3}{|c|}{?}\\
\hline
\end{tabular}
\caption{Notable degeneration limits of the Toda spectral curves.}
\label{tab:lim}
\end{table}

\subsubsection{Limits I: the maximal Argyres--Douglas point and geometry}

I am gonna start with the case in which we set
\beq
\aleph=1,~ u_1=1,~ u_2=3=u_4=u_6=-u_5,~ u_7=u_8=-2.
\label{eq:ssu}
\eeq
In terms of the LMO variables \eqref{eq:muu}, this corresponds to $c=1$,
$\mathfrak{m}_i=0$ for all $i$: this is the limit of zero 't~Hooft coupling  of the
measure \eqref{eq:dmLMO}, in which the support of the eigenvalue density
shrinks to a point. In this limit the $\cY$-branches of the LMO curve are
given by the slope asymptotics of \eqref{eq:slope1}, which is in turn entirely encoded by the
orbit data of \cref{tab:minorb}. From \eqref{eq:LMOpolexp} and \cref{tab:minorb}, we get
\bea
\Xi_{\mathfrak{g},\rm red}(\lambda,\mu)\big|_{u(\mathfrak{m}=0),\aleph=1} & = &  (\mu+1)^2 \left(\mu^2+\mu+1\right)^3 \left(\mu^4+\mu^3+\mu^2+\mu+1\right)^5
 \left(\lambda+\mu^5\right) \left(\lambda \mu^5+1\right)
 \nn \\
& \times & \left(\lambda \mu^6-1\right)
   \left(\lambda-\mu^{10}\right)^2   \left(\lambda \mu^{10}-1\right)^2  \left(\lambda^2-\mu^{15}\right)
   \left(\lambda+\mu^{15}\right)^2  \left(\lambda \mu^{15}+1\right)^2 \nn \\ &
   \times & \left(\lambda^2 \mu^{15}-1\right)
   \left(\lambda-\mu^{30}\right) \left(\lambda \mu^{30}-1\right)  \left(\lambda-\mu^6\right).
\label{eq:ss}
\eea
I'll call \eqref{eq:ss} the {\it super-singular limit} of the $\mathrm{E}_8$
Toda curves: in this limit, $\mathrm{Spec}\bbC[\lambda, \mu]/\bra
\Xi_{\mathfrak{g},\rm red}\ket$ is a reducible, non-reduced scheme with the
radicals of its 19 distinct non-reduced components given by lines or plane
cusps. In particular, denoting by$~^h\Xi$ the homogenisation of $\Xi$, the Picard group of the corresponding reduced scheme is trivial,
\beq
\mathrm{Pic}^{(0,\dots,0)} \l(\mathrm{Proj}\frac{\bbC[\lambda, \mu, \nu]}{\sqrt{\bra
^h\Xi_{\mathfrak{g},\rm red} \ket}}\r) \simeq 0,
\eeq
the resolution of singularities $\Gamma_{u(\mu=0), 1}$ is a disjoint
union of 19 $\bbP^1$'s, and the whole Prym--Tyurin
$\mathrm{PT}(\Gamma_{u_{\mu=0},1})$ collapses to a point
in the super-singular limit. This is more tangibly
visualised by what happens to \cref{fig:cuts} when we approach
\eqref{eq:ssu}: since $\aleph=1$, the branch points of the
$\lambda$-projection satisfy $b_i^+b_i^-=1$ from \eqref{eq:bipm}, and from
\cref{prop:isopar} and the discussion that follows it, they correspond to
$\alpha_i(l)=0$ for some simple root $\alpha \in \Pi$. The corresponding
ramification points on the curve are then at $\mu=\exp(\alpha_i(l))=1$, and
substituting into \eqref{eq:ss} we get
\beq
\Xi_{\mathfrak{g},\rm red}(\lambda,1)\big|_{u(\mathfrak{m}=0),\aleph=1} = 337500 (\lambda -1)^8 (\lambda +1)^{10},
\eeq
which means that the branch points collide together in four pairs with
$b_i^+=b_-=1$, and five with $b_i^+=b_-=-1$. It is immediate to see that the
$A/\tilde B$-periods of $\rd \sigma$ vanish in the limit (as the corresponding cycles
shrink), as do the $B/\tilde A$ periods upon performing the elementary cycle
integration explicitly. 

This degeneration limit should have a meaningful physical counterpart in the
dynamics of the corresponding compactified 5d theory at this particular point
on its Coulomb branch, and in particular it should correspond  to the UV fixed
point of \cite[Section~7-8]{Intriligator:1997pq} (see also the recent works
\cite{Xie:2017pfl,Jefferson:2017ahm}). I won't pursue the details here, but I
will give some comments on the resulting A- and B-model geometries, and
on the broad type of physics implications it might lead to. The first
comment is on the geometrical character of \eqref{eq:ss}: it is clearly expected
that singularities in the Wilsonian 4d action should arise from vanishing
cycles in the family of Seiberg--Witten curves \cite{Seiberg:1994rs}, and in turn from the
development of nodes as we approach its discriminant; and furthermore, more
exotic phenomena related typically related to superconformal symmetry arise
whenever these vanishing cycles have non-trivial intersection
\cite{Argyres:1995jj}, leading to the
appearance in the low energy spectrum of mutually non-local BPS
solitons, and cusp-like singularities in the SW geometry (see
\cite{SEO:2013mla} for a review). \eqref{eq:ss} provides a limiting version of
this phenomenon whereby {\it all SW periods} vanish\footnote{There is no room
  for cusp-like singularities like this in the simpler setting of pure ${\rm SU}(2)$ $\cN=2$ pure
  Yang--Mills with SW curve $y^2=(x^2-u)^2 - \Lambda_4^4$, unless
we put ourselves in the physically
  degenerate situation where we sit at the point of classically unbroken gauge
  symmetry $u=0$ and take the classical limit $\Lambda_4\to 0$: the theory
  is then {\it classical} pure $\cN=2$ gluodynamics, where we have essentially imposed by
  {\it fiat} to discard the quantum corrections that give a gapped vacuum and
  the breaking of superconformal symmetry.}. I will refer to \eqref{eq:ss}
as the {\it maximal Argyres--Douglas point} of the $\mathrm{E}_8$ gauge theory, and as
in the more classical cases of Argyres--Douglas theories, it presents several
hallmarks of a theory at a superconformal fixed point. Besides 
the vanishing of the central charges of its BPS saturated states, we see that
the way we reach the super-singular vacuum is akin to the mechanism of
\cite{Seiberg:1996bd,Intriligator:1997pq} to engineer fixed points from
five-dimensional gauge theories: since the engineering dimension of the
five-dimensional gauge coupling $1/g_{\rm YM}^{(5)}$ is that of mass, the
theory is non-renormalisable and quantising it requires a cutoff; in the
$M$-theoretic version \cite{Intriligator:1997pq,Lawrence:1997jr} of the
geometric engineering of \cite{Katz:1996fh}, this is naturally given in terms
of the inverse of the radius of the eleventh-dimensional circle $R$ in
\eqref{eq:geoeng}. Considerations about brane dynamics in \cite{Seiberg:1996bd} allow
to conclude that the limit in which the bare gauge coupling diverges leads to
a sensible quantum field theory at an RG fixed point with enhanced global
symmetry; and notice that, under the identifications
\eqref{eq:geoeng}, setting the Casimir $\aleph=1$ amounts to taking precisely
that limit. Indeed, upon reintroducing the four-dimensional scale $\Lambda_4$
and identifying $\Lambda_{\rm UV}=1/R$ as the cutoff scale,
the second equality in \eqref{eq:geoeng} reads
\beq
\aleph=\frac{\Lambda_4}{\Lambda_{\rm UV}}=\re^{-t_{\rm B}/4}.
\label{eq:altb}
\eeq
Recall that $\Lambda_4=\Lambda_{\rm UV}\re^{-\frac{1}{g_{\rm UV}}}$, the RG
invariant scale in four dimensions; the Seiberg limit $g_{\rm UV}\to \infty$
for the fixed point theory is given then by $\aleph=1$, with the vanishing of the
masses of BPS modes being realised by \eqref{eq:ssu}. 

In light of
\cref{thm:gv}, there is an A-model/Gromov--Witten take on this as well, which
also allows us to reconnect the above to the work of
\cite{Intriligator:1997pq,Xie:2017pfl}. Let us put ourselves in the
appropriate duality frame for \eqref{eq:ssu}-\eqref{eq:ss}, which corresponds
to the choice of $\tilde A_i$ as the cycles whose $\rd\sigma$-periods serve as
flat coordinates around \eqref{eq:ssu}. By \cref{cl:geoeng} 
%{\bf where the
%  part about $\cX$ must be added!!}, 
this corresponds to the maximally singular chamber in
the extended K\"ahler moduli space of $Y$ given by the orbifold GW theory
of $\cX$. Notice first that $\aleph=1$ corresponds to the
shrinking limit of  the K\"ahler volume of the base $\bbP^1$, $t_{\rm B}=0$. Furthermore, as remarked in our earlier paper
\cite{Borot:2015fxa}, the Bryan--Graber Crepant Resolution Conjecture \cite{MR2483931}
for the $\mathrm{E}_8$ singularity prescribes that the orbifold point in its stringy
K\"ahler moduli space should be given by a vector $\mathrm{OP} \in \mathfrak{h}^*\simeq
H_2\l(\widehat{\bbC^2/\tilde\bbI}, \bbZ\r)$ such that
\beq
\tau_i({\rm OP})= \l(\frac{2 \pi \ri \a_0}{|\tilde\bbI|}\r)_i=\frac{\pi \ri \mathfrak{d}_i}{
  15}
\label{eq:tauOP}
\eeq
the second equality being taken w.r.t. the root basis for
$\mathfrak{h}^*$. The values \eqref{eq:ssu} for the Toda Hamiltonians
correspond exactly to the values of the fundamental traces of a Cartan torus
element corresponding to \eqref{eq:tauOP}: \eqref{eq:ss} is then
the spectral curve mirror of the A-model at the $\mathrm{E}_8$-orbifold-of-the-conifold
singularity, that is, the tip of the K\"ahler cone of $Y$. 

The above, together with the constructions in \cref{sec:E8chain,sec:actangl} provides some
preliminary take, in this specific 
$\mathrm{E}_8$ case, to a few  
of the questions raised at the end of \cite{Xie:2017pfl} regarding
the Seiberg--Witten geometry, Coulomb branch and prepotential of 5d SCFT corresponding to Gorenstein singularities.
%provide ely  are determined. 
A detailed study and the determination of some of the relevant quantities for the 5d
SCFT (such as the superconformal index) is left for future study, and will be
pursued elsewhere.

\subsubsection{Limits II: orbifold quantum cohomology of the $\mathrm{E}_8$
  singularity}

\label{sec:limII}
Since the correspondence of the left vertical line of \cref{fig:dualities} was shown to hold in the
context of \cref{thm:gv}, I will offer here some calculations giving
plausibility (other than the expectation from the underlying physics) for the lower horizontal and
the diagonal arrow in the diagram. This will be done in a second interesting
limit, given by taking $\aleph\to 0$ while keeping all
the other parameters finite (but possibly large). By \eqref{eq:geoeng} and \cref{cl:toda0}, this
corresponds to a partial decompactification
limit in which we send the K\"ahler parameter of the base $\bbP^1$ in $Y\to
\bbP^1$ to infinity; the resulting A-model theory has thus the resolution of the threefold
transverse $\mathrm{E}_8$ singularity $\widehat{\bbC^2/\tilde \bbI} \times \bbC$ as its
target, or equivalently, by \cite{MR2483931}, the orbifold
$[\bbC^2/\tilde \bbI \times \bbC]$ upon analytic continuation in the
K\"ahler parameters. Accordingly, on the gauge theory side, this corresponds to sending
$\Lambda_4 \to 0$ while keeping the classical order parameters $u_i$ constant,
and it singles out the perturbative part in the prepotential
\eqref{eq:gprep}. And finally, in the Toda context, this type of limit was considered in
\cite{Braden:2000he,Krichever:2000uw} for the non-relativistic type $A$ chain,
where it was shown to recover, after a suitable change-of-variables, the
non-periodic Toda chain.

To bolster the claim, let me show that special geometry on the space of
$\mathrm{E}_8$ Toda curves does indeed reproduce correctly the degree-zero part of the genus zero GW potential\footnote{Physically, this is $g_s\to 0$, $\a'\to 0$.} of
$\widehat{\bbC^2/\tilde \bbI} \times \bbC$ in the sector where we have at least
one insertion of $\mathbf{1}_Y$: by the string equation, this is the $tt$-metric on
the Frobenius manifold $QH(\widehat{\bbC^2/\tilde \bbI} \times \bbC)\simeq
QH(\widehat{\bbC^2/\tilde \bbI})$ (see \cref{sec:frob} for more details on this). As vector spaces, we have $$QH(\widehat{\bbC^2/\tilde
  I})=H(\widehat{\bbC^2/\tilde \bbI})=H^0(\widehat{\bbC^2/\tilde \bbI})\oplus
H^2(\widehat{\bbC^2/\tilde \bbI})\simeq \bbC \oplus \mathfrak{h}.$$  
Let us use linear coordinates $\{l_i\}_{i=0}^8$ for the decomposition in the last two
equalities, where we write $H(\widehat{\bbC^2/\tilde \bbI}) \ni v = l_0
\mathbf{1}_Y \oplus_i l_i [E_i]$, with $E_i$ the $i^{\rm th}$ exceptional
curve in the canonical resolution of singularities $\pi:\widehat{\bbC^2/\tilde
  \bbI} \to \bbC^2/\tilde \bbI$, and likewise $\{l_i\}_{i=1}^8$ in the second
isomorphism are taken w.r.t. the $\a$-basis of $\mathfrak{h}^*$. On the GW side, the McKay correspondence
implies that 
\beq
\eta_{ij}= (E_i,E_j)_Y = -\mathscr{C}^{\mathfrak{g}}_{ij}
\label{eq:intmat}
\eeq
On the other hand, by \eqref{eq:piAB}-\eqref{eq:preptoda} (see
\cite[Lecture~5]{Dubrovin:1994hc}, and \cref{sec:frob} below),
the $tt$-metric on the Frobenius manifold on the base of the family of Toda
spectral curves is
\beq
\eta_{ij}=-\sum_{{\rd\mu(p)=0}} \mathrm{Res}_{p}\frac{\de_{l_i} \mu \de_{l_j} \mu}{ \mu \de_\lambda \mu}\frac{\rd \lambda}{q_\mathfrak{g}\lambda^2}.
\eeq
where, in the language of \cite[Lecture~5]{Dubrovin:1994hc} and
\cite{Dunin-Barkowski:2015caa} and as will be reviewed more in detail in \cref{sec:frob,sec:hur}, we view the family of Toda spectral curves as
a closed set in a Hurwitz space with $\mu$, $\ln \mu$ and $\rd \ln \lambda$ identified with the covering map, the
superpotential, and the prime form respectively (see \cref{sec:hur}); this identification follows
straight from the special K\"ahler relations \eqref{eq:piAB}. The argument of the residue has poles at $\de_\lambda \mu =0$,
$\lambda,\mu=0,\infty$. Swapping sign and orientation in the
contour integral we pick up the residues at the poles and zeroes of $\lambda$
and $\mu$. Let me start from the zeroes of $\lambda$. Note that
\bea
\Xi_{\mathfrak{g},\rm red}(0,\mu)\Big|_{\aleph=0} &=& \Xi'_{\mathfrak{g},\rm
  red}(0,\mu)=\prod_{\a\in\Delta^*}(\mu-\re^{\a \cdot l}) \nn \\
&=& \prod_{\a\in\Delta^*}(\mu-\re^{\sum_j \a^{[j]} l_j}),
\label{eq:xil0}
\eea
so that $\lambda=0$ amounts to $\mu= \re^{\sum_j \a^{[j]} l_j}$ for some non-zero
root $\a$. Then,
\bea
\mathrm{Res}_{\mu=\re^{\sum_j \a^{[j]} l_j}}\frac{\de_{l_i} \mu ~ \de_{l_j} \mu}{
  \mu \de_\lambda \mu}\frac{\rd \lambda}{q_\mathfrak{g}\lambda} &=& -\mathrm{Res}_{\mu=\re^{\sum_j \a^{[j]} l_j}}\frac{\de_{l_i} \lambda \de_{l_j} \lambda}{
  \lambda \de_{\mu} \lambda}\frac{\rd \mu}{q_\mathfrak{g}\mu^2} \nn \\
&=& 
-\frac{\de_{l_i} \Xi'_{\mathfrak{g},\rm
  red}\de_{l_j} \Xi'_{\mathfrak{g},\rm
  red} }{q_\mathfrak{g}\mu^2 (\de_\mu\Xi'_{\mathfrak{g},\rm
  red})^2}\bigg|_{\mu=\re^{\sum_j \a^{[j]} l_j}} \nn \\
&=& -\frac{\re^{2\sum_j \a^{[j]} l_j} \a^{[i]} \a^{[j]} \prod_{\beta,\gamma\neq \alpha}\l(\re^{\sum_j \a^{[j]} l_j}-\re^{\sum_j \b_j l_j}\r)\l(\re^{\sum_j
    \a^{[j]} l_j}-\re^{\sum_j \gamma_j l_j}\r)}{q_\mathfrak{g}\re^{2\sum_j \a^{[j]} l_j}
  \prod_{\beta,\gamma\neq \alpha}\l(\re^{\sum_j \a^{[j]}
    l_j}-\re^{\sum_j \b_j l_j}\r)\l(\re^{\sum_j \a^{[j]} l_j}-\re^{\sum_j
    \gamma_j l_j}\r)} \nn \\
&=& -\frac{\a^{[i]} \a^{[j]}}{q_\mathfrak{g}}
\label{eq:reseta}
\eea
where we have used the ``thermodynamic identity''\footnote{Namely, the fact that the exchange $\mu
\leftrightarrow \lambda$ is an anti-canonical transformation of the
symplectic algebraic torus $((\bbC^\star)^2, \rd \ln\mu \wedge \rd\ln\lambda)$ the
curve $\bbV(\Xi_{\mathfrak{g},\rm red})$ embeds into, leading to $F \to
-F$ in the expression of the Frobenius prepotential, and thus $\eta\to-\eta$.} of
\cite[Lemma~4.6]{Dubrovin:1994hc} to switch $\mu\leftrightarrow \lambda$ at
the cost of a swap of sign in the first line, the implicit function
theorem for the derivatives $\de_\bullet \lambda$ in the second line, and
finally \eqref{eq:xil0}. It is easy to see that the poles at $\mu=0,\infty$
have vanishing residues; summing over the pre-images of $\lambda=0$ then gives
\beq
\eta_{ij} = \sum_{\lambda(p)=0} \mathrm{Res}_{p}\frac{\de_{l_i} \mu~
  \de_{l_j} \mu}{ \mu \de_\lambda \mu}\frac{\rd \lambda}{q_\mathfrak{g}\lambda^2} =
-\sum_{\a \in \Delta^*} \frac{\a^{[i]}\a^{[j]}}{q_\mathfrak{g}} = -\mathscr{C}^{\mathfrak{g}}_{ij},
\label{eq:etacart}
\eeq
where we used \cite[Appendix~E]{Longhi:2016rjt}
\beq
\sum_{\a \in \Delta^*} \bra \a, \a_i\ket \bra \a, \a_j \ket = q_\mathfrak{g} \mathscr{C}^{\mathfrak{g}}_{ij},
\eeq
and we find precise agreement with \eqref{eq:intmat}. The calculation of the
Frobenius product (namely, the 3-point function $\de^3_{ijk} F\eta^{il}$)
\beq
c_{ijk}=\sum_{\stackrel{p\in \Gamma'_{u(l),0}}{\rd\mu(p)=0}}
\mathrm{Res}_{p}\frac{\de_{l_i} \mu ~ \de_{l_j} \mu  ~\de_{l_k} \mu}{ \mu
  \de_\lambda \mu}\frac{\rd \lambda}{q_\mathfrak{g}\lambda^2},
\label{eq:cijk}
\eeq
is slightly more involved due to the necessity to expand the integrand in
\eqref{eq:reseta} to higher order at $\lambda=0$; in other words, and unsurprisingly, the product does depend on the expression of the higher order terms in
$\lambda$ of $\Xi_{\mathfrak{g},\rm red}$, unlike $\eta_{ij}$ for which all we
needed to know was $\Xi_{\mathfrak{g},\rm red}(\lambda=0,\mu)$ in
\eqref{eq:xil0}. Let us content ourselves with noticing, however, that by the
same token of the preceding calculation for $\eta_{ij}$, the r.h.s. of
\eqref{eq:cijk} is necessarily a rational function in exponentiated flat
variables $t_j$: this is in keeping with the trilogarithmic nature of the
1-loop correction \eqref{eq:Fgauge}, whose triple derivatives have precisely
such functional dependence on the flat variables $a_j$.

\subsubsection{Limits III: the 4d/non-relativistic limit}
\label{sec:limIII}

The last limit we consider involves the fibres of $\pi:\mathscr{S}_{\mathfrak{g}} \to
\mathscr{B}_{\mathfrak{g}}$. We take
\beq
\mu=\re^{\epsilon \chi}, \quad \mathfrak{l}(\lambda) \to \epsilon
\mathfrak{l}(\lambda)
\eeq
and take the $\epsilon\to 0$ limit while holding $\chi$, $\lambda$ and
$\mathfrak{l}$ fixed; note that rescaling the Cartan torus representative
$\mathfrak{l}(\lambda)$ of the conjugacy class of $\widehat{L}_{x, 
  y}$ and taking $\epsilon \to 0$ corresponds to the limits in row~III of
\cref{tab:lim} at the level of $u_i$ and $\aleph$. Then \eqref{eq:xi3} becomes
\bea
\epsilon^{-d_\mathfrak{g}} \Xi_{\mathfrak{g}}(\lambda, \mu)
%=\det_{\mathfrak{g}}(\widehat{L_{x, 
%    y}}-\mu 1)
&=&\epsilon^{-d_\mathfrak{g}}(\mu-1)^8\prod_{\a\in\Delta^*}(\re^{\a \cdot \mathfrak{l}}-\mu)
= \chi^8 \prod_{\a\in\Delta^*}(\a \cdot \mathfrak{l}-\chi)+\cO(\epsilon),\nn \\
&=& \det_{\mathfrak{g}}\l(\log \widehat{L_{x, y}}-\chi 1\r) +
\cO(\epsilon)
\label{eq:nonrel}
\eea
so in this limit the curve $\bbV( \Xi_{\mathfrak{g}}(\lambda, \mu))$
degenerates to the spectral curve of the family of Lie-algebra elements
$\log\widehat{L}_{x, y}$. These coincide with the spectral-parameter
dependent Lax operators of the $\widehat{\mathrm{E}}_8$ {\it non-relativistic} Toda
chain \cite{MR0436212}, to which \eqref{eq:lax} reduce upon taking
$\epsilon\to 0$. As the picture of \eqref{eq:nonrel} as a curve-of-eigenvalues
carries through to this setting\footnote{A more accurate way of putting it
  would be to point out that the original setting of
  \cite{MR1013158,MR1401779,MR1668594,MR1397059} dealt precisely with
  Lie algebra-valued systems of this type; since $\cG$ is simply-laced, the construction of the PT variety
  dates back to \cite{MR1013158}; and since $\log\widehat{L}_{x, y}$
  depends rationally on $\lambda$, Theorem~29 of \cite{MR1397059} applies despite
  $\mathfrak{g}$ not being minuscule.}, so does the construction of the
preferred Prym--Tyurin; on the other hand, the $\epsilon\to 0$ degenerate
limit of \cref{thm:kp}, which amounts in its proof to pick up the
Lie-algebraic Krichever--Poisson Poisson bracket $\omega_{\rm KP}^{(1)}$,
leads to a non-relativistic spectral differential of the form
\beq
\rd\sigma_{\epsilon \to 0} \to \chi \frac{\rd \lambda}{q_\mathfrak{g} \lambda}.
\label{eq:dseps}
\eeq
As the non-relativistic limit is equivalent to the shrinking limit of the
five-dimensional circle in $\bbR^4 \times S^1$, the corresponding limit on the
gauge theory side leads to pure ${\rm E}_8$ $\cN=2$ super Yang--Mills theory in
four dimensions, with \eqref{eq:dseps} being the appropriate Seiberg--Witten
differential in that limit. Then \cref{claim:E8} solves the problem of giving
an explicit Seiberg--Witten curve for this theory; it is instructive to
present what the polynomial \eqref{eq:nonrel} looks like more in detail. We
have
\beq
\lim_{\epsilon \to 0}\epsilon^{-d_\mathfrak{g}} \Xi_{\mathfrak{g}}(\lambda,
\mu) = \chi^8 \sum_{i=0}^{120} \mathfrak{q}_{120-k}(v_1, \dots, v_8) \chi^{2k}
\eeq
where the $\chi \to -\chi$ parity operation reflects the reality of
$\mathfrak{g}$, and $v_1, \dots v_8$ is a set of generators of
$\bbC[\mathfrak{h}]^\cW$. Taking the power sum basis
$v_1=\mathrm{Tr}_\mathfrak{g}(\mathfrak{l}^{2})$, $v_i=\mathrm{Tr}_\mathfrak{g}(\mathfrak{l}^{2i+6})$, we get
\beq
\mathfrak{q}_0=1, \quad \mathfrak{q}_1= -\frac{v_1}{2}, \quad  \mathfrak{q}_2=
-\frac{7}{40} v_1, , \quad  \mathfrak{q}_3=
-\frac{49}{240} v_1, \quad  \mathfrak{q}_4=
-\frac{1697}{9600} v_1  - \frac{v_2}{8}, \dots~.
\eeq

\section{Application II: the $\widehat{\mathrm{E}_8}$ Frobenius manifold}
\label{sec:applII}

\subsection{Dubrovin--Zhang Frobenius manifolds and Hurwitz
  spaces}

\subsubsection{Generalities on Frobenius manifolds}
\label{sec:frob}

I gather here the basic definitions about Frobenius manifolds for the
appropriate degree of generality that is needed here. The
reader is referred to the classical monograph \cite{Dubrovin:1994hc} for more details.

\begin{defn} An $n$-dimensional complex manifold $X$ is a {\rm semi-simple Frobenius manifold}
  if it supports a pair $(\eta, \star)$, with $\eta$ a non-degenerate,
  holomorphic symmetric $(0, 2)$-tensor with flat Levi--Civita connection
  $\nabla$, and a commutative, associative, unital, fibrewise $\cO_X$-algebra
  structure on $TX$ satisfying

\begin{description}
\item[Compatibility]
\beq
\eta(A\star B, C)=\eta(A, B \star C)\qquad \forall A, B, C \in \cX(X);
\eeq
\item[Flatness]
the 1-parameter family of connections
\beq
\label{eq:defEC}
\nabla^{(\hbar)}_A B\triangleq \nabla_A B+\hbar A\star B\qquad \hbar \in \bbC
\eeq
is flat identically in $\hbar \in \bbC$;
\item[String equation]
the unit vector field $e \in \cX(X)$  for the product $\star$ is
$\nabla$-parallel, $$\nabla e=0~;$$
\item[Conformality]
there exists a vector field $E\in\cX(X)$ such that
$\nabla E \in \Gamma(\mathrm{End}(TX))$ is diagonalisable, $\nabla$-parallel,
and the family of connections \cref{eq:defEC} extends
to a holomorphic connection $\nabla^{(\hbar)}$ on $X\times \bbC^\star$ by
\bea
\label{eq:defECint}
\nabla^{(\hbar)}\frac{\partial}{\partial \hbar}&=& 0  \\
\nabla^{(\hbar)}_{\partial/\partial \hbar} A &= & \frac{\partial}{\partial \hbar}A+E\star
A-\frac{1}{\hbar}\widehat{\mu}A
\label{eq:defECz}
\eea
where $\widehat{\mu}$ is the traceless part of $-\nabla E$;
\item[Semi-simplicity]
the product law $\star|_{x}$ on the tangent fibres $T_xX$ has no nilpotent elements for generic $x\in X$.
\end{description}
\label{def:frob1}
\end{defn}

\begin{defn} $X$ is a semi-simple Frobenius manifold iff there exists an open
  set $X_0$, a coordinate chart $t^1, \dots, t^n$ on $X_0$, and a regular
  function $F\in \cO(X_0)$ called the {\rm Frobenius prepotential}
  such that, defining $c_{ijk}\triangleq \de^3_{ijk}F$, we have
\ben
\item $\de^3_{1jk}F= \eta_{jk}=\mathrm{const},~\det\eta\neq 0$
\item Letting $\eta^{ij}=(\eta^{-1})_{ij}$ and summing over repeated indices,
  the Witten--Dijkgraaf--Verlinde--Verlinde equations hold
\beq
c_{ijk}\eta^{kl}c_{lmn} = c_{imk}\eta^{kl}c_{ljn} \quad \forall i,j,m,n
\eeq
\item there exists a linear vector field and numbers $d_i$, $r_i$, $d_F$ 
\beq 
E=\sum_i d_i t^i \de_i + \sum_{i| d_i=0} r^i \de_i \in\cX(X_0)
\eeq
such that
\beq
\LL_E F= d_F F + \text{\rm quadratic in}~t
\eeq
\item there is a positive co-dimension subset $X_{0}^* \subset X_0$ and 
coordinates $u^1,
  \dots, u^n$ on $X_0 \setminus X_0^*$ such that for all $m$
\beq
\de_i u^m \eta^{ij}c_{jkl} = \de_k u^m \de_l u^m.
\eeq
\een
\end{defn}
Upon defining $\de_{t_i} \star \de_{t_j}=\eta^{kl}c_{lij}\de_{t_k}$, $e=\de_{t_1}$
the latter definition is easily seen to be equivalent to the previous
one. Point 1) ensures non-degeneracy of the metric\footnote{As is customary in
  the subject, I use the word ``metric'' without assuming any positivity of
  the symmetric bilinear form $\eta$.} $\eta$, its flatness and the
String Equation; Point 2) and the fact that the structure constants come from
a potential function implies the restricted flatness condition, with the
extension due to conformality coming from Point 3); and Point 4) establishes that
$\de_{u_i}$ are idempotents of the $\star$ product on $X_0\setminus X_{0}^*$; the reverse
implications can be worked out similarly \cite{Dubrovin:1998fe}. 

The Conformality property has an important consequence, related to the
existence of a bihamiltonian structure on the loop space of the $X$. Define a
second metric $g$ by
\beq
g(E \star A,B) = \eta(A,B)
\eeq
which makes sense on all tangent fibres $T_pX$ where $E$ is in the group of
units of $\star|_p$. In flat coordinates $t_i$, this reads
\beq
g_{ij}=E^k c_{kij}.
\eeq
A central result in the theory of Frobenius manifolds is that this second
metric is flat, and that it forms a non-trivial\footnote{I.e., it does not
  share a flat co-ordinate frame with $\eta$.} flat pencil of metrics with
$\eta$, namely $g+\lambda \eta$ is a flat metric $\forall \lambda \in
\bbC$. Knowledge of the second metric in flat coordinates for the first is
sufficient to reconstruct the full prepotential: indeed, the induced metric on
the cotangent bundle (the {\it intersection form}) reads
\beq
g^{\a\b}=\l(2-d_F+d_\a+d_\b\r)\eta^{\a\lambda}\eta^{\b\mu}\de^2_{\lambda,\mu}F
\label{eq:getaF}
\eeq
from which the Hessian of the prepotential can be read off.

\subsubsection{Extended affine Weyl groups and Frobenius manifolds}
\label{sec:dzfrob}

A classical construction of Dubrovin \cite[Lecture~4]{Dubrovin:1994hc}, proved to be complete in \cite{MR1924259}, gives a classification of all Frobenius manifolds
with {\it polynomial} prepotential: these are in bijection with the finite Euclidean
reflection groups (Coxeter groups). I'll recall briefly here their construction in the case
in which the group is a Weyl group $\cW$ of a simple Lie algebra
$\mathfrak{g}$ of dimension $d_{\mathfrak{g}}$. Let
$(\mathfrak{h}, \bra, \ket)$ be the Cartan subalgebra with $\bra,
\ket$ being the $\bbC$-linear extension of the 
Euclidean inner product given by the Cartan--Killing form, and let $\{x_i\}_i$
be orthonormal coordinates on $(\mathfrak{h}^*, \bra, \ket)$. It is well-known
\cite{MR0573068} that the
$\cW$-invariant part $S(\mathfrak{h^*})^\cW$ of the polynomial algebra
$S(\mathfrak{h^*})=H^0(\mathfrak{h}, \cO)$ is a
graded polynomial ring in $r_{\mathfrak{g}}=\mathrm{dim}_\bbC(\mathfrak{h})$ homogeneous variables $y_1,
\dots, y_{r_{\mathfrak{g}}}$; the degrees of the basic invariants $d_i\triangleq \deg_x y_i$,
which are distinct and ordered so that $d_i>d_{i+1}$, are
the {\it Coxeter exponents} of the Weyl group\footnote{A parallel and somewhat more common
  convention is to call $d_i-1$ the exponents of the group (the eigenvalue of
  a Coxeter element), rather than the degrees $d_i$ themselves.}; also 
$d_1=h(\mathfrak{g})=\frac{\mathrm{dim}\mathfrak{g}}{\mathrm{rank}\mathfrak{g}}-1$, the
Coxeter number. Let now
\beq
\mathrm{Discr}_\cW(\mathfrak{h})=\mathrm{Spec}\frac{\bbC[x_1, \dots, x_{r_{\mathfrak{g}}}]}{\bra \{\a_i \cdot x\}_{i=1}^{{r_{\mathfrak{g}}}} \ket }= \bigcup_i H_i
\eeq
where $H_i$ are root hyperplanes in $\mathfrak{h}$: the open set
\beq
\mathfrak{h}^{\rm reg} \triangleq \mathfrak{h}\setminus \mathrm{Discr}_\cW(\mathfrak{h})
\eeq
is the set of regular Cartan algebra elements (i.e. $\mathrm{Stab}_\cW(h)=e$
for $h\in \mathfrak{h}^{\rm reg}$). We will be interested in the unstable and
stable quotients 
\bea
X_{\mathfrak{g}}^{\rm us} & \triangleq & \mathfrak{h}/\cW = \mathrm{Spec}\bbC[x_1, \dots, x_{r_{\mathfrak{g}}}]^\cW = \mathrm{Spec}\bbC[y_1, \dots,
  y_{{r_{\mathfrak{g}}}}],  \nn \\
X_{\mathfrak{g}}^{\rm st} & \triangleq &  \mathfrak{h} /\!\!/ \cW = \mathfrak{h}^{\rm
  reg}/\cW = \mathrm{Spec}\l(\cO_{\mathfrak{h}}(\mathfrak{h}^{\rm reg})^\cW\r) 
\label{eq:coxeter}
\eea
Notice that $\pi:\mathfrak{h}^{\rm reg} \to X_{\mathfrak{g}}^{\rm st}$ is a regular cover (a
principal $\cW$-bundle) of $X_{\mathfrak{g}}^{\rm st}$, and linear co-ordinates on
$\mathfrak{h}^{\rm reg}$ can serve as a set of local co-ordinates on $X^{\rm st}$.

Dubrovin constructs a polynomial Frobenius structure on $X_{\mathfrak{g}}^{\rm st}$ as
follows. First off, the Coxeter exponents are used to define a vector field
\beq
E \triangleq \frac{1}{d_1} \sum_i x^i \de_{x_i} = \de_{y_1}+\sum_{i=1}^{r_{\mathfrak{g}}} \frac{d_i}{d_1} \de_{y_i}.
\eeq
Also, view the Cartan--Killing pairing on $\mathfrak{h}$ as giving a flat
metric $\xi$ on $T\mathfrak{h}$, i.e. $\xi(\de_{x_i}, \de_{x_j})=\delta_{ij}$. If $V=\pi^{-1}(U)= V_1 \sqcup \dots \sqcup V_{r_{\mathfrak{g}}}$ for $U \subset
X_{\mathfrak{g}}^{\rm st}$, and for $i=1,\dots, |\cW|$, let $\sigma_i:U \to \mathfrak{h}^{\rm reg}$ be a
section of $\pi:V\to U$ lifting $U$ isomorphically to the $i^{\rm
  th}$ sheet of the cover, so that $\sigma_i(U) \simeq V_i$, and define
\beq
g \triangleq (\sigma_i)^* \xi.
\eeq
By the Weyl invariance of $\xi$ and $\{y_j\}_j$, it is immediately seen that $g$
defines a well-defined pairing on $T^*X_{\mathfrak{g}}^{\rm st}$ (i.e., the r.h.s. is invariant
under deck transformations of the cover $\mathfrak{h}^{\rm reg}$, see \cite[Lemma~4.1]{Dubrovin:1994hc}). Armed with
this, a Frobenius structure with unit $\de_{y_1}$, Euler vector field $E$,
intersection form $g$ and flat pairing $\eta=\LL_{\de_{y_1}} g$ is defined on
$X_{\mathfrak{g}}^{\rm st}$ upon proving that $g+\lambda \eta$ thus defined give a flat
pencil of metrics on $T^*X_{\mathfrak{g}}^{\rm st}$
\cite[Theorem~4.1]{Dubrovin:1994hc}. In the same paper,  it
is further proved that such Frobenius structure is polynomial in flat
co-ordinates for $\eta$, semi-simple, and unique given $(e, E, g)$.

In a subsequent paper \cite{MR1606165}, Dubrovin--Zhang consider a group
theory version of the above construction, as follows. Fix a node $\bar i \in
\{1,\dots,{r_{\mathfrak{g}}}\}$ in the Dynkin diagram of $\mathfrak{g}$, and let
$\a_{\bar i}$, $\omega_{\bar i}$ be the corresponding simple root and
fundamental weight. The $\cW$-action on
$\mathfrak{h}$ 
%gives a $\cW$ action on the Cartan torus
%$\cT=\exp(\mathfrak{h})$, as in \cref{sec:notation}, which 
can be lifted to
an action of the affine Weyl group $\widehat{\cW}\simeq \cW \rtimes
\Lambda_r(\cG)$ by affine transformations on $\mathfrak{h}$, 
\bea
\widehat w ~ : ~\widehat{\cW} \times \mathfrak{h} & \longrightarrow &  \mathfrak{h} \nn \\
((w,\a),l) & \longrightarrow &  w(l) + \a,
\eea
which is further covered by a $\widetilde{\cW}_{\bar i}\triangleq \widehat{\cW}\rtimes \bbZ$-action on $\mathfrak{h}\times
\bbC$ given by
\bea
\widetilde w ~ : ~\widetilde{\cW}_{\bar i} \times \mathfrak{h} \times \bbC &
\longrightarrow &  \mathfrak{h} \times \bbC \nn \\
((w,\a,{r_{\mathfrak{g}}}),(l, v)) & \longrightarrow &  \l(w(l) + \a+{r_{\mathfrak{g}}} \omega_{\bar i}, x_{r_{\mathfrak{g}}+1}-{r_{\mathfrak{g}}}\r).
\eea
$\widetilde{\cW}_{\bar i}$ is called {\it the extended affine Weyl group} with
marked root $\a_{\bar i}$. In \cite{MR1606165}, the authors give
a characterisation of the ring of invariants of $\widetilde{\cW}_{\bar i}$,
which may be reformulated as follows. Set $\mathsf{g}=\re^{2 \pi \ri l}\in\cG$ and let $u_i=\chi_{\rho_i}(\mathsf{g})$ be as in \eqref{eq:hi}
the regular fundamental characters of $\mathsf{g}$; also define $d_j \triangleq \bra
\omega_j, \omega_{\bar i}\ket$.  Then\footnote{%\begin{rmk}
The reader familiar with \cite{MR1606165} will notice the slight difference
between what we call $u_i$ here and the basic Laurent polynomial
  invariants $\tilde y_i$ in \cite{MR1606165}, the latter being defined as the
  Weyl-orbit sums
\beq
\tilde y_i(t) \triangleq  \re^{2 \pi \ri d_1 t_{{r_{\mathfrak{g}}}+1}} \sum_{w\in \cW} \re^{2 \pi \ri
  \bra w(\omega_i), t\ket}.
\eeq
It is immediate from the definition that there exists a linear, triangular change-of-variables
with rational coefficients
\beq
u_i=\sum_j \frac{\mathrm{Mult}_{\rho_{\omega_i}}(\omega_j)}{|\cW_{\omega_j}|}
\tilde y_j(t) + \mathrm{Mult}_{\rho_{\omega_i}}(0).
\eeq
with $\mathrm{Mult}_{\rho}(\omega)$ being the multiplicity of $\omega\in
\Lambda_w$ in the weight system of $\rho\in R(\cG)$, so that \cite[Theorem~1.1]{MR1606165} holds as in \eqref{eq:invext}.
%\end{rmk}
} \cite[Theorem~1.1]{MR1606165},
\beq
\bbC[t_1, \dots, t_{r_{\mathfrak{g}}}; t_{{r_{\mathfrak{g}}}+1}]^{\widetilde{\cW}_i} \simeq \bbC[\re^{2 \pi \ri d_1
    t_{{r_{\mathfrak{g}}}+1}} u_1, \dots, \re^{2 \pi \ri d_{r_{\mathfrak{g}}} t_{{r_{\mathfrak{g}}}+1}} u_{r_{\mathfrak{g}}};
  u_{{r_{\mathfrak{g}}}+1} \re^{2\pi \ri t_{{r_{\mathfrak{g}}}+1}}]
\label{eq:invext}
\eeq
As before, define
\bea
\widetilde{X}_{\mathfrak{g}, \bar i}^{\rm us} & \triangleq & \mathrm{Spec}\bbC[u_1, \dots, u_{{r_{\mathfrak{g}}}+1}] \simeq
\l(\mathfrak{h}\times \bbC\r)/\widetilde{\cW}_i \simeq \cT/\cW \times
\bbC^\star \\
\widetilde{X}_{\mathfrak{g},\bar i}^{\rm st} & \triangleq & \l(\mathfrak{h}^{\rm reg} \times \bbC\r) /\!\!/
\widetilde{\cW}_i = 
\mathrm{Spec}\l(\cO_{\mathfrak{h}\times \bbC}(\mathfrak{h}^{\rm reg} \times
\bbC)\r)^{\widetilde{\cW_i}} \simeq \cT^{\rm reg}/\cW \times
\bbC^\star
\eea
with $\cT^{\rm reg}=\exp(\mathfrak{h}^{\rm reg})$ and $\cT^{\rm
  reg}/\cW$ being the set of regular elements of $\cT$ and 
regular conjugacy classes of $\cG$ respectively. A Frobenius structure
polynomial in $u_1, \dots, u_{{r_{\mathfrak{g}}}+1}$ can be constructed along the same lines as
for the classical case of finite Coxeter groups: adding a further linear
coordinate $x_{{r_{\mathfrak{g}}}+1}$ for the right summand in $\mathfrak{h} \oplus \bbC$, 
we define a metric $\xi$ with signature $({r_{\mathfrak{g}}},1)$ on
$\mathfrak{h} \times \bbC$ by orthogonal extension of $4\pi^2$ times the Cartan--Killing pairing on
$\mathfrak{h}$, and normalising $\|\de_{x_{{r_{\mathfrak{g}}}+1}}\|^2=-4\pi^2$:
\beq
\xi(\de_{x_i}, \de_{x_j})=\l\{
\bary{cc}
4\pi^2 \delta_{ij}, & i,j<{r_{\mathfrak{g}}}+1,\\
-4\pi^2, & i=j={r_{\mathfrak{g}}}+1,\\
0 & {\rm else}.
\eary
\r.
\eeq
Exactly as in the
previous discussion of the finite Weyl groups, we have a $\cW$-principal bundle 
\beq
\xymatrix{\cT^{\rm reg} \times \bbC^\star  \ar[d]^u   \\
\widetilde{X}_{\mathfrak{g},\bar i}^{\rm st}   \ar@/^1pc/[u]^{\widetilde{\sigma}_i} 
}
\eeq
with sections $\widetilde{\sigma}_i$, $i=1, \dots, |\cW|$ defined as before. Then the following theorem holds \cite[Theorem~2.1]{MR1606165}:
\begin{thm}
\label{thm:uniqdz}
There is a unique semi-simple Frobenius manifold structure
$\l(\widetilde{X}_{\mathfrak{g}, \bar i}^{\rm st}, e,E,\xi,
g, \star\r)$ on
$\widetilde{X}_{\mathfrak{g}, \bar i}^{\rm st}$ such that
\ben
\item in flat co-ordinates $t^1, \dots, t^{r_{\mathfrak{g}}}, t^{{r_{\mathfrak{g}}}+1}$ for $\xi$,
  the prepotential is polynomial in $t^1, \dots, t^{r_{\mathfrak{g}}}$ and $\re^{t^{{r_{\mathfrak{g}}}+1}}$;
\item $e=\de_{u_{\bar i}}=\de_{t^{\bar i}}$;
\item $E=\frac{1}{2\pi\ri d_{\bar i}}\de_{x_{r_{\mathfrak{g}}+1}}=\sum_j \frac{d_j}{d_{\bar
    i}}t^j\de_{t_j}+\frac{1}{d_{\bar i}}\de_{t^{{r_{\mathfrak{g}}}+1}}$;
\item $g=\widetilde{\sigma}_i^*\xi$.
\een
\end{thm}

\subsubsection{Hurwitz spaces and Frobenius manifolds}
\label{sec:hur}

As was already hinted at in \cref{sec:limII}, a further source of semi-simple
Frobenius manifolds is given by Hurwitz spaces
\cite[Lecture~5]{Dubrovin:1994hc}. For $r \in \bbN_0$, $\mathtt{m} \in \bbN_0^r$, these are moduli spaces
$\HH_{g,\mathtt{m}}=\cM_{g}(\bbP^1, \mathtt{m})$ of isomorphism classes of
degree $|\mathtt{m}|$ covers $\lambda$ of the complex projective line by a smooth genus-$g$ curve
$C_g$, with marked ramification profile over $\infty$ specified by
$\mathtt{m}$; in other words, $\lambda$ is a meromorphic function on $C_g$
with pole divisor
$(\lambda)_-= -\sum_i \mathtt{m}_i P_i$ for points $P_i \in C_g$, $i=1, \dots,
r$. Denoting as in \cref{def:sc} by $\pi$, $\lambda$ and $\Sigma_i$ respectively the universal
family, the universal map, and the sections marking the $i$-th point in
$(\lambda)_-$, this is
\beq
\xymatrix{ C_g \ar[d]  \ar@{^{(}->}[r]& \mathcal{\cC}\ar[d]^\pi  \ar[r]^{\lambda}  &  \bbP^1 \\
                     [\lambda]  \ar@{^{(}->}[r]^{\rm pt}   \ar@/^1pc/[u]^{P_i}&                                             \HH  \ar@/^1pc/[u]^{\Sigma_i}& 
}
\label{eq:hurdiag}
\eeq
As a result, $\HH_{g,\mathtt{m}}$ is a reduced, irreducible complex variety with $\dim_\bbC
\HH_{g,\mathtt{m}}=2g+ \sum_i \mathtt{m}_i+r-1$, which is typically smooth
(i.e. so long as the ramification profile is incompatible with automorphisms of the
cover). 

Dubrovin provides in \cite{Dubrovin:1994hc} a systematic way of
constructing a semi-simple Frobenius manifold structure on
$\HH_{g,\mathtt{m}}$, for which I here provide a simplified account. As in
\cref{sec:speccurve}, let $\rd=\rd_\pi$ denote the relative differential with respect to the
universal family (namely, the differential in the fibre direction), and let
$p_i^{\rm cr}\in C_g\simeq \pi^{-1}([\lambda])$ be the critical points $\rd
\lambda =0$ of the universal map (i.e., the ramification points of the
cover). By the Riemann existence theorem, the critical values
\beq
\mathfrak{u}^i=\lambda\l(p_i^{\rm cr}\r)
\eeq
are local co-ordinates on $\HH_{g,\mathtt{m}}$ away from the discriminant
$\mathfrak{u}^i=\mathfrak{u}^j$. We then locally define an $\cO_{\HH_{g,\mathtt{m}}}-$al\-gebra structure
on the space of vector fields $\cX(\HH_{g,\mathtt{m}})$ by imposing that the
co-ordinate vector fields $\de_{\mathfrak{u}^i}$ are idempotents for it:
\beq
\de_{\mathfrak{u}^i}\star \de_{\mathfrak{u}^j}=\delta_{ij}\de_{\mathfrak{u}^i}.
\eeq
The algebra is obviously unital with unit $e=\sum_i \de_{\mathfrak{u}^i}$; a linear (in
these co-ordinates) vector field $E$ is further defined as $\sum_i
\mathfrak{u}^i \de_{\mathfrak{u}^i}$. The one missing ingredient in the
definition of a Frobenius manifold is a flat pairing of the vector fields,
which is provided by specifying some auxiliary data. Let then
$\phi\in \Omega^1_{C}(\log (\lambda))$ be an {\it exact}
 meromorphic one form having {\it simple} poles\footnote{Exactness
     and simplicity of the poles  can be disposed of by looking instead at suitably normalised
  Abelian differentials w.r.t. a chosen symplectic basis of 1-homology cicles
  on $C_g$; a fuller discussion, with a classification of the five types of differentials
  that are compatible with the existence of flat structures on the resulting
  Frobenius manifold, is given in the discussion preceding
  \cite[Theorem~5.1]{Dubrovin:1994hc}. The generality considered here however
  suits our purposes in the next section.}
 at the support of $(\lambda)_-$ with
constant residues; the pair $(\lambda, \phi)$ are called respectively {\it
  superpotential} and the {\it primitive differential} of $\HH_{g,\mathtt{m}}$. A
non-degenerate symmetric pairing $\eta(X,Y)$ for vector fields $X,Y \in
\cX(\HH_{g,\mathtt{m}})$ is defined by
\beq
\label{eq:hurxi}
\eta(X,Y) \triangleq  \sum_{i}\Res_{p_i^{\rm cr}}\frac{X(\lambda)
  Y(\lambda)}{\rd \lambda}\phi^2,
\eeq
where, for $p$ locally around $p_i^{\rm cr}$, the Lie derivatives
$X(\lambda)$, $Y(\lambda)$ are taken at constant $\mu(p)=\int^p\phi$. It turns
out that $\eta$ thus defined is flat, compatible with $\star$, with $E$ being
linear in flat co-ordinates, and it further satisfies
\bea
\label{eq:hurprod}
\eta(X,Y \star Z) &=&   \sum_{i}\Res_{p_i^{\rm cr}}\frac{X(\lambda)
  Y(\lambda) Z(\lambda)}{\rd \mu \rd \lambda}\phi^2,  \\
g(X,Y) &=&  \sum_{i}\Res_{p_i^{\rm cr}}\frac{X(\log \lambda)
  Y(\log \lambda) }{\rd \log \lambda}\phi^2.
\label{eq:hurintf}
\eea
\begin{rmk}
There is a direct link between the prepotential of the Frobenius manifold
structure above on $\HH_{g,\mathtt{m}}$ and the special K\"ahler prepotential
of familes of spectral curves (see \eqref{eq:preptoda} in the Toda case),
whenever the latter is given by moduli of a generic cover of the line with
ramification profile $\mathtt{m}$: 
the two things coincide upon identifying the superpotential and
primitive Abelian integral $(\lambda, \mu)$ on the Hurwitz space side with the marked meromorphic
functions $(\lambda, \mu)$ on the spectral curve end \cite{Dubrovin:1994hc,Krichever:1992qe}.
It is a common situation, however, that the $\lambda$-projection is highly
non-generic: the Toda spectral curves of \cref{sec:curvdet} are an obvious
example in this sense. One might still ask, however, what type of geometric
conditions ensure that a semi-simple, conformal Frobenius manifold structure
exists on the base of the family $\cB \stackrel{\iota}{\hookrightarrow}
\HH_{g,\mathtt{m}}$: an obvious sufficient condition is that, away 
from the discriminant and locally on an open set $\Omega \subset
\HH_{g,\mathtt{m}}$ with a chart $\mathsf{t}:\Omega\to \bbC^{\dim \HH_{g,\mathtt{m}}}$
given by flat co-ordinates for $\eta$, 
\ben
\item $\cB$ embeds
as a linear subspace of $H\subset \bbC^{\dim \HH_{g,\mathtt{m}}}$;
\item $H\simeq T_0H \simeq \bbC\bra e \ket \oplus H'$ contains the line
  through $e$;
% (i.e. $H \ni
%  \mathsf{t} \Leftrightarrow \mathsf{t}+\{\mathtt{s},0,\dots, 0\} \in H
%  \forall s \in \bbC$);
\item the minor correponding to the restriction to $H$ of the Gram matrix of $\eta$ is non-vanishing.
\een
In this case, \eqref{eq:hurxi}-\eqref{eq:hurprod} define a semi-simple,
conformal Frobenius manifold structure with flat identity on the base $\cB$ of
the family of spectral curves, with all ingredients obtained being projected
down from the parent Frobenius manifold. We will see in the next section that
the family of $\widehat{\mathrm{E}}_8$ Toda spectral curves falls precisely within this class.
\label{rmk:frobsub}
\end{rmk}

\subsection{A 1-dimensional LG mirror theorem}

\subsubsection{Saito co-ordinates}
\label{sec:saito}

I will now elaborate on the previous \cref{rmk:frobsub} in the case of the
degenerate limit $\aleph \to 0$ of the family of Toda curves over $\cU\times
\bbC$. Recall from \cref{sec:curvdet} that there is an intermediate branched
double cover $\Gamma'_{u}\simeq \Gamma_{u,0}$  of the base curve
$\Gamma''_{u}$, defined as
\beq
\Gamma'_{u}=\overline{\bbV\l[\Xi''_{\mathfrak{g}, \rm red}\l(\mu+\frac{1}{\mu}, \lambda\r)\r]}
\eeq
For future convenience, rescale $\lambda \to \frac{\lambda}{u_0}$ in the
following. Looking at $\lambda$ as our marked covering map gives, by
\eqref{eq:genred} and \cref{tab:ptsinf}, we have an embedding of 
\beq
\iota : X_{\mathfrak{g}}^{\rm Toda} \hookrightarrow \HH_{g,\mathtt{m}}
\eeq
of $X_{\mathfrak{g}}^{\rm Toda} \simeq \bbC_{u_0} \times \cU$ into the Hurwitz
space $\HH_{g,\mathtt{m}}$ with $g=128$ and, letting $\varepsilon_k=\re^{2\pi \ri/k}$,
\beq
\mathtt{m}=\l(\stackrel{\mu=-1}{\overbrace{2}},\stackrel{\mu=\varepsilon_3^j\neq
1}{\overbrace{3,3}},\stackrel{\mu=\varepsilon_5^j\neq
1}{\overbrace{5,5,5,5}},~\stackrel{\mu=0}{\overbrace{5,6,10,10,15,15,15,30}},~\stackrel{\mu=\infty}{\overbrace{5,6,10,10,15,15,15,30}}\r).
\label{eq:ramprof}
\eeq
Mindful of \cref{rmk:frobsub} I am going to declare $(\lambda, \phi)$ with $\phi\triangleq\rd\ln \mu$
to be the superpotential and primary differential and proceed to examine the
pull-back of $\eta$ to $\bbC_{u_0}\times \cU$. An important point to stress
here is that {\it this will not be a repetition of what was done in \cref{sec:limII}}: in
that case, we were looking at \eqref{eq:hurxi} with $\log\lambda$ as the
superpotential (up to $\mu\leftrightarrow \lambda$, $F\leftrightarrow -F$);
this means that the computation leading up to the flat metric \eqref{eq:etacart} was rather
computing {\it the intersection form} $g$ of $X^{\rm Toda}_{\mathfrak{g}}$, by
\eqref{eq:hurintf}. The relation between the Frobenius manifold structure on
$X^{\rm Toda}$ defined by \eqref{eq:hurxi}-\eqref{eq:hurintf} and $QH(\widehat{\bbC/\tilde{I}})$
is indeed a non-trivial instance of Dubrovin's notion of almost-duality of Frobenius
manifolds \cite{MR2070050}, with the almost-dual product being given by \eqref{eq:cijk}.

\begin{lem}
Let $X,Y \in \cX(X_{\mathfrak{g}}^{\rm Toda})$ %, TX_{\mathfrak{g}}^{\rm
                                %Toda})$.
be holomorphic vector fields on $X_{\mathfrak{g}}^{\rm Toda}$.
 Then \eqref{eq:hurxi} defines a flat non-degenerate pairing on $TX_{\mathfrak{g}}^{\rm
  Toda}$, with flat co-ordinates given by 
\bea
\mathsf{t}_0(u) &=& \frac{\ln u_0}{30}, \nn \\
\mathsf{t}_i(u) &=& M_i \mathfrak{m}_{\mathfrak{l}_i}(c=u_0^{1/30},u_1,\dots,
u_8), \quad i>0.
\label{eq:tflat}
\eea
where $\{\mathfrak{m}_{\mathfrak{l}_i}\}_i$ are the planar moments
\eqref{eq:muu} and $M_i\in \bbC$. Furthermore, the metric has constant
anti-diagonal form in these coordinates.
\label{lem:flatc}
\end{lem}

\begin{proof}
As in \cref{sec:limII}, let's reverse orientation in the residue formula
\eqref{eq:hurxi} and pick up residues on $\Gamma'_{u}\setminus \{b_i^{\rm
  cr}\}_i$; these are all located at the poles $P_i$ of $\lambda$.
As in \cite{Dubrovin:1994hc}, I define local
coordinates $\nu_i$ centred around $P_i$ such that
$\lambda=\nu_i^{-\mathtt{m}_i}+\cO(1)$, as well as functions $r_i^j$ with 
$(i,j) \in \mathtt{R}=\{(k,l) | 1\leq k \leq 23=l(\mathtt{m}), 1\leq j\leq \mathtt{m}_i\}$ by
\beq
\mathtt{r}_i^j \triangleq \l\{ 
\bary{ccc} 
\mathrm{pv} \int_{P_8}^{P_i} \frac{\rd \mu}{\mu} & \mathrm{for} & j=0\\
\Res_{P_i} \nu_i^{j} \mathrm{pv}\ln\mu \rd\lambda & \mathrm{for} & j=1, \dots,
\mathtt{m}_i-1,\\
\Res_{P_i} \lambda \frac{\rd\mu}{\mu} & \mathrm{for} & j=\mathtt{m}_i,
\eary\r.
\label{eq:resmom}
\eeq
and where 
$$\mathrm{pv}\ln\mu(P_i)=\mathrm{pv}\int_{P_8}^{P_i} \frac{\rd \mu}{\mu}.$$
Here $P_8$, in the numbering of marked points of \eqref{eq:ramprof}, is the
lowest (fifth) order pole of $\lambda$ at $\mu=0$. A remarkable fact, that can
be proven straightforwardly from the Puiseux expansion of $\lambda$ near $P_i$
using \eqref{eq:xi2}-\eqref{eq:pkgen} and \cref{claim:E8}, is that
$\mathtt{r}_i^j$ is in all cases a multiple of one of the planar moments \eqref{eq:muu}:
\beq
\mathtt{r}_i^j(u) = \cN_i^j \mathsf{t}_{\ell(i,j)}(u)
\label{eq:rij}
\eeq
for some map of finite sets $\ell: \mathtt{R} \to [[0,8]]$ and complex constants $\cN_i^j \in \bbC$. For
$0<j<\mathtt{m}_i$ the result is
collected in \cref{tab:resmom}, where we denoted  
$$\upsilon_1=5-\sqrt{5}-\ri \sqrt{10 +2\sqrt{5}}, \quad
\upsilon_2=-5-\sqrt{5}+i \sqrt{10-2 \sqrt{5}},$$ $$\upsilon_3=\l(1+i
\sqrt{5+2 \sqrt{5}}\r) \upsilon_1/2, \quad\upsilon_4=-(-1)^{3/5}
\upsilon_1~.$$ 
We furtermore have $r_i^0=0$ and
$r_i^{\mathtt{m}_i}=0$ for $\mu(P_i)\neq 0,\infty$, and
$r_i^0=1=-r_i^{\mathtt{m_i}}=-r_{i+8}^0=r_{i+8}^{\mathtt{m_i}}$ for
$\mu(P_i)=0$. It turns out that \eqref{eq:rij} suffices to prove constancy of
$\eta$ in the coordinate chart given by $\mathsf{t}_{i}$ above; indeed, from
\eqref{eq:resmom} and \cref{tab:resmom} we obtain
\bea
\eta(\de_{\mathsf{t}_i},\de_{\mathsf{t}_j}) &=&  
-\sum_{l}\Res_{P_l}\frac{\de_{\mathsf{t}_i}\lambda 
\de_{\mathsf{t}_j}\lambda}{\mu \de_\mu \lambda} \rd\mu,
=-\sum_{l}\Res_{P_l}\frac{\de_{\mathsf{t}_i}\ln \mu \rd \lambda
\de_{\mathsf{t}_j}\ln\mu \rd\lambda
}{\rd \lambda} 
%\l(\frac{\rd\mu}{\mu}\r)^2 
\nn \\
&=& -\sum_{l}\sum_{k=0}^{\mathtt{m}_l} \de_{\mathsf{t}_i} \mathtt{r}_l^k
\de_{\mathsf{t}_j} \mathtt{r}_l^{\mathtt{m}_l-k}=-\sum_i
\sum_{k=0}^{\mathtt{m}_i} \delta_{i,\ell(l,k)}
\delta_{j,\ell(l,\mathtt{m}_l-k)} \cN_l^k \cN_l^{\mathtt{m}_l-k} \nn \\
&=&
\delta_{i,8-j} \mathfrak{y}_{i},
\eea
for numbers $\mathfrak{y}_{i}$, so that the Gram matrix of $\eta$ is constant
and anti-diagonal in these coordinates. These can be scaled away by an
appropriate rescaling of $M_i$ in the definition of $\mathsf{t}_i$ in \eqref{eq:tflat}.

\end{proof}

\begin{table}[h]
\begin{tabular}{|c|c|c|c||c|c|c|c||c|c|c|c||c|c|c|c|}
\hline
$i$ & $j$ & $\ell(i,j)$ & $\cN_i^j$ & $i$ & $j$ & $\ell(i,j)$ & $\cN_i^j$ & $i$ & $j$ & $\ell(i,j)$ & $\cN_i^j$ \\
\hline
 $1$ & $1$ & $5$ & $\ri$ & $10$ & $5$ & $3$ & $-\frac{5 \sqrt[3]{-1}}{2}$ & $13$ & $6$ & $4$ & $\frac{3 \phi }{10}$ \\ \hline
 $2$ & $1$ & $3$ & $-\frac{\sqrt[6]{-1}}{\sqrt{3}}$ & $10$ & $6$ & $4$ & $\frac{3}{20} (-1)^{2/5}$ & $13$ & $9$ & $6$ & $-\frac{\phi}{25}$ \\ \hline
 $2$ & $2$ & $7$ & $\frac{(-1)^{5/6}}{6 \sqrt{3}}$ & $10$ & $9$ & $6$ & $\frac{1}{50} (-1)^{3/5}$ & $13$ & $10$ & $7$ & $-\frac{5}{6}$ \\ \hline
 $4$ & $1$ & $2$ & $\frac{\upsilon _1}{20}$ & $10$ & $10$ & $7$ & $\frac{5}{12} (-1)^{2/3}$ & $13$ & $12$ & $8$ & $-\frac{3}{500 \phi }$ \\ \hline
 $4$ & $2$ & $4$ & $\frac{\upsilon _2}{200}$ & $10$ & $12$ & $8$ & $\frac{3 (-1)^{4/5}}{1000}$ & $13$ & else & $-$ & $0$ \\ \hline
 $4$ & $3$ & $6$ & $\frac{\upsilon _3}{1500}$ & $10$ & else & $-$ & $0$ & $14$ & $3$ & $2$ & $-3 \phi $ \\ \hline
 $4$ & $4$ & $8$ & $\frac{u_4}{10000}$ & $11$ & $2$ & $2$ & $-2 \sqrt[5]{-1} \phi $ & $14$ & $5$ & $3$ & $-5$ \\ \hline
 $5$ & $1$ & $2$ & $-\frac{\upsilon _3}{20}$ & $11$ & $4$ & $4$ & $\frac{(-1)^{2/5}}{5 \phi }$ & $14$ & $6$ & $4$ & $-\frac{3}{10 \phi }$ \\ \hline
 $5$ & $2$ & $4$ & $\frac{\upsilon _1}{200}$ & $11$ & $5$ & $5$ & $5 \ri$ & $14$ & $9$ & $6$ & $\frac{1}{25 \phi }$ \\ \hline
 $5$ & $3$ & $6$ & $\frac{u_4}{1500}$ & $11$ & $6$ & $6$ & $\frac{2 (-1)^{3/5}}{75 \phi }$ & $14$ & $10$ & $7$ & $-\frac{5}{6}$ \\ \hline
 $5$ & $4$ & $8$ & $-\frac{\upsilon _2}{10000}$ & $11$ & $8$ & $8$ & $-\frac{(-1)^{4/5} \phi}{250}  $ & $14$ & $12$ & $8$ & $\frac{3 \phi }{500}$ \\ \hline
 $8$ & $1$ & $2$ & $1$ & $11$ & else & $-$ & $0$ & $14$ & else & $-$ & $0$ \\ \hline
 $8$ & $2$ & $4$ & $-\frac{1}{10}$ & $12$ & $2$ & $2$ & $\frac{2 \sqrt[5]{-1}}{\phi }$ & $15$ & $6$ & $2$ & $-6 \sqrt[5]{-1}$ \\ \hline
 $8$ & $3$ & $6$ & $\frac{1}{75}$ & $12$ & $4$ & $4$ & $-\frac{(-1)^{2/5} \phi}{5}  $ & $15$ & $10$ & $3$ & $-10 \sqrt[3]{-1}$ \\ \hline
 $8$ & $4$ & $8$ & $-\frac{1}{500}$ & $12$ & $5$ & $5$ & $5 \ri$ & $15$ & $12$ & $4$ & $-\frac{3(-1)^{2/5}}{5} $ \\ \hline
 $9$ & $2$ & $3$ & $2 \sqrt[3]{-1}$ & $12$ & $6$ & $6$ & $-\frac{2(-1)^{3/5}}{75}  \phi $ & $15$ & $15$ & $5$ & $-15 \ri$ \\ \hline
 $9$ & $3$ & $5$ & $-3 \ri$ & $12$ & $8$ & $8$ & $\frac{(-1)^{4/5}}{250 \phi }$ & $15$ & $18$ & $6$ & $-\frac{2(-1)^{3/5}}{25} $ \\ \hline
 $9$ & $4$ & $7$ & $-\frac{(-1)^{2/3}}{3} $ & $12$ & else & $-$ & $0$ & $15$ &
$20$ & $7$ & $\frac{5 (-1)^{2/3}}{3} $ \\ \hline
 $9$ & else & $-$ & $0$ & $13$ & $3$ & $2$ & $\frac{3}{\phi }$ & $15$ & $24$ &
$8$ & $-\frac{3  (-1)^{4/5}}{250}$ \\ \hline
 $10$ & $3$ & $2$ & $\frac{3 \sqrt[5]{-1}}{2}$ & $13$ & $5$ & $3$ & $-5$ &
 $15$ & else & $-$ & $0$ \\
\hline
\end{tabular}
\caption{Residues \eqref{eq:resmom} in terms of the planar moments
  \eqref{eq:muu}; poles of $\lambda$ not appearing in the table are related to
those listed above by $\mu\to 1/\mu$ and a sign-flip in the residue.}
\label{tab:resmom} 
\end{table}

Before we carry on to examine the product structure on $X_{\mathfrak{g}}^{\rm
  Toda}$ let us first check that the unit $e$ and Euler vector field $E$
satisfy indeed the String Equation and (part of the) Conformality properties of
\cref{def:frob1} by verifing that $\nabla e=0$, $\nabla \nabla E=0$ with
$\nabla=\nabla^{(\eta)}$. It is easy to verify the following 
\begin{prop}
\beq
e=\frac{\de}{\de\mathsf{t}_8}, \quad E=\sum_j \frac{d_j}{d_8}t^j \de_{\mathsf{t}_j}+\frac{1}{d_8}\frac{\de}{\de\mathsf{t}_0},
\eeq
with 
\beq
d_1=6,~ d_2= 10,~ d_3 = 12,~ d_4 =15,~ d_5 = 18,~ d_6=20,~ d_7=24,~ d_8=30.
\label{eq:di}
\eeq
\label{prop:eE}
\end{prop}
\begin{proof}
The easiest way to see this is to realise that, by their definition in the
Hurwitz space setting, $e$ and $E$
generate an affine subgroup of $\mathrm{PSL}(2,\bbZ)$ on the target $\bbP^1$
in \eqref{eq:hurdiag} by 
\beq
\LL_e \lambda= \mathrm{const}, \quad \LL_E \lambda= \lambda, 
\eeq
Now, recall first of all that in the natural Hurwitz (B-model) co-ordinates $u_i$, a
shift in $u_3$ keeping all other variables fixed gives a constant shift in $\lambda/u_0$ by
\eqref{lem:uispec} with $\lambda \to \lambda/u_0$ and $\aleph \to 0$, as we are considering here. Moreover,
from \eqref{eq:muu} we see that
in flat co-ordinates for $\eta$, and since $\de_{\mathsf{t}_i} u_j=
\delta_{j3} \re^{-\mathsf{t}_0} \mathsf{t}_8$, a constant shift in $\mathsf{t}_8$
leaving all other flat co-ordinates constant gives a constant shift in the
rescaled $\lambda \to \lambda/u_0$:
\beq
\frac{\de \lambda}{\de \mathsf{t}_8} = \mathrm{const},
\eeq
Then $L_e\lambda \propto L_{\de_{\mathsf{t}_8}} \lambda$, from which we deduce
$e\propto \de_{\mathsf{t}_8}$ as there is no continous symmetry on $\lambda$
that holds up identically in $\mu$; the proportionality can be turned into an
equality upon appropriate choice of $M_8$, which gives in any case an
isomorphism of Frobenius manifolds. As far as the Euler vector field is
concerned, recall similarly that a rescaling in $u_0$ at constant $u_i$ gives a
rescaling of $\lambda$ with the same scaling factor, so that $E= u_0 \de_{u_0}$. Writing down $u_0
\de_{u_0}$ in flat co-ordinates using \eqref{eq:resmom} and \eqref{eq:muu} concludes the proof.
\end{proof}

\subsubsection{The mirror theorem}

Let us now dig deeper into the Frobenius manifold structure of
$X_{\mathfrak{g}}^{\rm Toda}$. By \eqref{eq:getaF}, the
$\star$-structure on $TX_{\mathfrak{g}}^{\rm Toda}$ can be retrieved from
knowledge of the intersection form $g$ in flat co-ordinates for $\eta$; whilst
theoretically this would only require the calculation of a Jacobian from the flat
co-ordinates for $g$ in \cref{sec:limII} to \eqref{eq:tflat}, such calculation
is however unviable due to the difficulty in inverting the Laurent polynomials
$u_i(t)$. We proceed instead from an analysis of \eqref{eq:hurprod}, and prove
the following

\begin{thm}
There is an isomorphism of Frobenius manifolds
\beq
\widetilde{X}_{\mathfrak{g}, 3}^{\rm us} \simeq X_{\mathfrak{g}}^{\rm Toda}
\label{eq:isofrob}
\eeq
\label{thm:frob}
\end{thm}

\begin{proof}
We have already come a long way proving \eqref{eq:isofrob}: by \cref{prop:eE}
the pairs $(e,E)$ match on the nose already, since $d_i = \bra \omega_j,
\omega_{3}\ket$ by \eqref{eq:di}. Also, by \eqref{eq:intmat} and since
\bea
g_{k0} &=& -\sum_{\lambda(p)=\infty}\Res_{p}\frac{\de_{x_k}\lambda}{q_\mathfrak{g}\mu^2 \de_\mu
  \lambda} \rd\mu=0, \nn \\
g_{00} &=& -\sum_{\lambda(p)=\infty}\Res_{p}\frac{\lambda}{q_\mathfrak{g}\mu^2 \de_\mu
  \lambda}\rd\mu= \frac{1}{q_\mathfrak{g}}\sum_{\lambda(p)=\infty}\mathrm{ord}_p \lambda = \frac{107}{900},
\eea
so do the two intersection forms up to a linear change-of-variables. By \cref{thm:uniqdz}, the remaining and largest bit in the
proof resides then in the proof of the polynomiality of the $\star$-product in
flat co-ordinates \eqref{eq:tflat}, which I am now going to show. This would be achieved
once we show that
\beq
c_{i,j,k}=  \sum_{i}\Res_{b_i^{\pm}}\frac{\de_{\mathsf{t}_i} \lambda \rd
  \mu \de_{\mathsf{t}_j} \lambda \rd \mu \de_{\mathsf{t}_k} \lambda \rd
  \mu}{q_\mathfrak{g} \mu^2 \rd \mu \rd \lambda} \in \bbC[\mathsf{t}_1, \dots, \mathsf{t}_8,
  \re^{\mathsf{t}_0}]
\label{eq:cijkf}
\eeq
for all $i,j,k$. As in the proof of \cref{lem:flatc}, because of the difficulty in
controlling the moduli dependence of $b_i^\pm$ in either $u_i$ or
$\mathsf{t}_i$, we turn the contour around and pick up residues in the
complement of $\{b_i^\pm\}$. However, one major difference here with the case of the
calculation of $\eta$ is that not only do poles of $\lambda$ contribute, but
also the 240 ramification points of the
$\mu$-projection $q_i$ (counted with multiplicity) satisfying
either $\mu(q_i)=-1$ or $\{\mu(q_i)+\mu^{-1}(q_i)=r_1^k\}_{k=1}^{133}$ (see
\eqref{eq:discry}), whose dependence on $u_i$ is even more involved. Indeed, near one of
those points, the superpotential behaves like
\beq
\lambda(p) = \lambda_0(\mathsf{t}) + \lambda_1(\mathsf{t}))\sqrt{\mu-q_i}+
\cO\l(\mu-q_i\r), 
\eeq
and therefore the moduli derivatives of $\lambda$ at $\mu=\mathrm{const}$ which appear in
\eqref{eq:cijkf} develop a simple pole as soon as $\de_{\mathsf{t}_i}\lambda_1
\neq 0$. This leads to a non-vanishing contribution to the residue as the triple
pole resulting from them in \eqref{eq:cijkf} is now, unlike for $\eta$, only partially offset by the
vanishing of $\rd \mu$ and $1/\de_\mu \lambda$ at the branch points. It is a
straightforward calculation to check that these contributions do
  contribute, but are best avoided calculating directly as their moduli dependence is intractable. Luckily, there's a workaround to do precisely so, as
follows. Instead of \eqref{eq:cijkf}, consider the
3-point function in {\it B-model} co-ordinates $\tilde u_0=\ln u_0, \tilde
u_1=u_1,..., \tilde u_8=u_8$,
\beq
\tilde c_{i,j,k}=  -\sum_{\rd \mu(p)=0 \mathrm{~or~} \lambda(p)=\infty}\Res_{p}\frac{\de_{u_i} \lambda \rd
  \mu \de_{u_j} \lambda \rd \mu \de_{u_k} \lambda \rd
  \mu}{q_{\mathfrak{g}} \mu^2 \rd \mu \rd \lambda} 
\label{eq:cijk2}
\eeq
sticking to the case $i=0$ to begin with. Now, $\de_{\tilde u_0}$ is the
Euler vector field, $\de_{\tilde u_0} \lambda=\lambda$, and we have $\de_{\tilde u_0}\lambda_1
= 0$ at all ramification points of $\mu$: this means that the problematic residues at
$\rd\mu=0$ give individually vanishing contributions
to \eqref{eq:cijk2}, unlike for the flat 3-point functions $c_{0,i,j}$. For
this restricted set of correlators and in this particular set of co-ordinates,
the only contribution to the LG formula \eqref{eq:cijk2} may come from the
poles $P_i$: here, a direct calculation
from the Puiseux expansion of $\lambda$ at its pole divisor immediately shows that the
Puiseux coefficients of $\lambda$ are polynomial in $u_0, u_1, \dots, u_8$ at
$\mu=0,\infty$. Furthermore, while the Puiseux coefficients at $\mu=-1$,
$\mu^3=1$ and $\mu^5=1$ are only Laurent polynomials in $\mathsf{t}_i$ with denominators given
by powers of $\mathsf{t}_4$, $\mathsf{t}_2$ and $\mathsf{t}_1$ respectively,
these powers turn out to delicately cancel from the final answer in
\eqref{eq:cijk2}. All in all, we find
\beq
\tilde c_{0,j,k}= \in \bbQ[\re^{u_0/30}, u_1, \dots, u_8].
\eeq
To consider the case $i>0$ in \eqref{eq:cijk2}, we use the WDVV equation in
these co-ordinates:
\beq
\tilde c_{ijk} \tilde \eta^{kl} \tilde c_{lmn} = \tilde c_{imk} \tilde
\eta^{kl} \tilde c_{ljn}.
\label{eq:cwdvv}
\eeq
Setting $n=0$, and letting $i,j,m$ go for the ride, \eqref{eq:cwdvv} gives a linear
inhomogeneous system with unknowns $\tilde c_{ijk}$, $i>0$  with coefficients 
being given by (complicated) polynomials in $\re^{u_0/30}, u_1, \dots, u_8$
with rational coefficients. One way to circumvent the complexity of solving it
explicitly is as follows: firstly, it is immediate to prove that the system
has maximal rank, which is an open condition, by evaluating the
coefficients at a generic moduli point, so that
$\tilde c_{ijk}$ are uniquely determined rational functions in $\re^{u_0/30}, u_1, \dots, u_8$. To check
that the solution is indeed polynomial, we just plug a general polynomial ansatz into
\eqref{eq:cwdvv} satisfying the degree conditions of \cref{prop:eE} and solve
for its coefficients, and find that such an ansatz does indeed solve \eqref{eq:cwdvv}. The claim
follows by uniqueness, the polynomiality of the inverse of \eqref{eq:muu}, and \cref{thm:uniqdz}.

\end{proof}

%\begin{rmk}
One immediate bonus of \cref{thm:frob}, and a further vindication of taking
great pains to give a closed-form calculation of the mirror in
\cref{claim:E8}, is that both the Saito--Sekiguchi--Yano coordinates
\eqref{eq:muu} and the
prepotential of $\widetilde{X}_{\mathfrak{g},3}$, for which an
explicit form was unavailable to date\footnote{This is a private communication from
  Boris Dubrovin and Youjin Zhang.}, can now be computed
straightforwardly: the reader may find an expression for the latter in \cref{sec:E8prep}. %\end{rmk} 
 A further bonus is a mirror theorem for the Gromov--Witten theory of the
polynomial $\bbP^1$-orbifold of type $\mathrm{E}_8$ \cite{MR2672302,Zaslow:1993fa}:
\begin{cor}
Let $C_{\mathfrak{g}}\simeq \bbP_{2,3,5}$ denote the orbifold base of the
Seifert fibration of the Poincar\'e sphere $\Sigma$ (see \cref{sec:cs}). Then, 
\beq
QH_{\rm orb}\l(C_{\mathfrak{g}}\r) \simeq X_{\mathfrak{g}}^{\rm Toda}.
\eeq
as Frobenius manifolds.
\label{cor:p1orb}
\end{cor}
This follows from composing the isomorphism $QH_{\rm
  orb}\l(C_{\mathfrak{g}}\r) \simeq \widetilde{X}_{\mathfrak{g}, 3}$ (see \cite{MR2672302}) with \cref{thm:frob}. \\

\begin{rmk} In \cite{MR2672302}, a different type of mirror theorem was proved in terms of
a polynomial three-dimensional Landau--Ginzburg model; it would be interesting
to deduce directly a relation between the two mirror pictures, along the lines
of what was done in a related context in \cite{Lerche:1996an}. As in \cite{Lerche:1996an}, the two mirror
pictures have complementary virtues: the threefold mirror of \cite{MR2672302}
has a considerably simpler form than the Toda/spectral curve mirror. On the
other hand, having a spectral curve mirror pays off two important dividends:
firstly, at genus zero, the calculation of flat co-ordinates for the Dubrovin
connection \eqref{eq:defEC} is simplified down to 1-dimensional (as opposed to
3-dimensional) oscillating integrals. Furthermore, and more remarkably, Givental's formalism
and the topological recursion might allow one to foray into the higher genus
theory, recursively to all
genera. This second aspect of the story will indeed be the subject of \cref{sec:p1orb}.
\end{rmk}

\subsection{General mirrors for Dubrovin--Zhang Frobenius manifolds}

There's a fairly compelling picture emerging from \cref{thm:frob} and the
constructions of \cref{sec:E8chain,sec:actangl} relating the Dubrovin--Zhang Frobenius
manifolds of \cref{sec:dzfrob} to relativistic Toda spectral curves. I am
going to propose here what the most general form of the conjecture should
be. For this section only, the symbols $\mathfrak{g}$, $\mathfrak{h}$,
$\cG=\exp(\mathfrak{g})$, $\cT=\exp(\mathfrak{t})$, $\cW$ will refer to an arbitrary simple, not
necessarily simply-laced, complex
Lie algebra, the corresponding Cartan subalgebra, simple simply-connected
complex Lie group, Cartan torus and Weyl group. 
%I also denote $r_{\mathfrak{g}}=\dim_\bbC\mathfrak{h}$ for the rank.  
As in  \cref{sec:E8chain}, let $\rho\in R(\cG)$ be an
irreducible representation of $\cG$ and for $\mathsf{g}\in \cG$ consider the
characteristic polynomial
\bea
\Xi_{\rho}(\theta_1, \dots, \theta_{r_\mathfrak{g}};\mu) &=& \det_{\rho}\l(\mu
\mathbf{1}-\mathsf{g}\r) = \sum_{k=0}^{\dim \rho} \mu^{\dim \rho-k}
\mathfrak{p}_k(\theta_1, \dots, \theta_{r_\mathfrak{g}}) \in \bbZ[\theta, \mu],
\eea
with $\mathfrak{p}_k\in \bbZ[\theta]$ and $\theta_i$ defined as in
\cref{sec:speccurve}. Recall that $\Xi_{\rho}$ reduces to a
product over Weyl orbits $W_i^\rho$
\beq
\Xi_{\rho}(\theta;\mu) = \prod_k \Xi^{(k)}_{\rho}(\theta;\mu)= \prod_k
\prod_{j=1}^{|W_k^\rho|}\l(\mu-\re^{\omega^{(k)}_j \cdot l}\r)
\label{eq:xiweyl}
\eeq
where $\re^{l}$ with $[\re^{l}]=[\mathsf{g}]$ is any conjugacy class representative in the
Cartan torus, and $l\in \mathfrak{h}$; for example, when $\rho=\mathfrak{g}$, we have two factors $\Xi^{(0)}_\rho
=(\mu-1)^{r_\mathfrak{g}}$ ($W^{\mathfrak{g}}_0=\Delta^{(0)}$)
and $\Xi_{\mathfrak{g},\rm red}\triangleq \Xi^{(1)}_\rho$ irreducible of
degree $d_{\mathfrak{g}}-r_{\mathfrak{g}}$ 
($W^{\mathfrak{g}}_1=\Delta^+ \cup \Delta^-$); this
was the case we considered for $\cG={\rm E}_8$. In general, let $\bar k$
be any integer in the product over $k$ in \eqref{eq:xiweyl} such that $W_{\bar
  k}^\rho$ is non-trivial and define $\Xi_{\rho,\rm red}\triangleq
\Xi_{\rho}^{(\bar k)}$.
% to be any one factor in \eqref{eq:xiweyl} corresponding to
%a non-trivial Weyl orbit $W_k^\rho$, 
Fixing $\alpha_{\bar i}\in \Pi$ a simple root, write
\beq
\Gamma_u^{(\bar i)} = \overline{\bbV\l(
  \Xi_{\rho,\mathrm{red}}\l(\theta_j=u_j-\delta_{\bar ij}\frac{\lambda}{u_0} \r)\r)},
\eeq
where the overline sign once again indicates taking the normalisation of the
projective closure, and letting $\omega^{(\bar k)}_1$ be the dominant
weight in $W_{\bar k}^\rho$, denote 
\beq
q_{\rho, \bar k} = \frac{1}{2} \sum_{j=1}^{|W_{\bar k}^\rho|} \bra \omega^{(\bar i)}_j,
\omega^{(\bar i)}_1 \ket^2.
\eeq
Finally, writing $X_{\mathfrak{g}, \bar i}^{\rm Toda} = \bbC^\star \times
(\bbC^\star)^{r_{\mathfrak{g}}}$ for the $r_{\mathfrak{g}}+1$ dimensional
torus with co-ordinates $(u_0; u_1, \dots,
u_{r_\mathfrak{g}})$, define pairings $(\eta,g)$ and product
structure $\de_{u_i} \star \de_{u_j}$ on $TX_{\mathfrak{g}, \bar i}^{\rm Toda}$ by 
\bea
\label{eq:todaeta}
\eta(\de_{u_i}, \de_{u_j}) &=&   \sum_{l}\Res_{p_l^{\rm cr}}\frac{\de_{u_i}\lambda
  \de_{u_j}\lambda}{q_{\rho, \bar k} \mu^2 \de_\mu \lambda}\rd\mu,  \\
\label{eq:todac}
\eta(\de_{u_i}, \de_{u_j}\star \de_{u_k}) &=&  
 \sum_{l}\Res_{p_l^{\rm cr}}\frac{\de_{u_i}\lambda \de_{u_j}\lambda
  \de_{u_k}\lambda}{q_{\rho, \bar k} \mu^2 \de_\mu \lambda}\rd\mu, \\
\label{eq:todag}
g(\de_{u_i}, \de_{u_j}) &=&  \sum_{l}\Res_{p_l^{\rm cr}}
\frac{\de_{u_i}\lambda \de_{u_j}\lambda
  \de_{u_k}\lambda}{q_{\rho, \bar k} \mu^2 \de_\mu \lambda}\rd\mu, 
\eea
where $\{p_l^{\rm cr}\}_l$ are the ramification points of $\lambda :
\Gamma_u^{(\bar i)} \to \bbP^1$.

\begin{conj}[Mirror symmetry for DZ Frobenius manifolds]
The Landau--Ginzburg formulas \eqref{eq:todaeta}--\eqref{eq:todag} define a
semi-simple, conformal Frobenius manifold $(X_{\mathfrak g,\bar i}^{\rm
  Toda}, \eta, e, E, \star)$,  which is independent of the choice of irreducible
representation $\rho$ and non-trivial Weyl orbit $W_{\bar k}^\rho$. In particular,
\eqref{eq:todaeta} and \eqref{eq:todag} define flat non-degenerate metrics on $TX_{\mathfrak g,i}^{\rm  Toda}$, and the identity and Euler vector
fields read, in curved coordinates $u_0, \dots, u_{r_\mathfrak{g}}$,
\beq
e= u_0^{-1} \de_{u_{\bar i}}, \qquad E= u_0 \de_{u_0}.
\eeq
Moreover,
\beq
X_{\mathfrak g,\bar i}^{\rm Toda} \simeq \widetilde{X}_{\mathfrak{g},\bar i}.
\label{eq:geniso}
\eeq
\label{conj:mirror}
\end{conj}

There is a fair amount of circumstantial evidence in favour of the validity of
the mirror conjecture in the form and generality proposed.

\ben
\item Firstly, the independence on the choice of representation should be a consequence of
the work of \cite{MR1182413,MR1401779,MR1668594} on the ``hierarchy'' of Jacobians of spectral curves for the
periodic Toda lattice and associated, isomorphic preferred Prym--Tyurins. The
very same calculation of \cref{sec:limII} of the intersection form
\eqref{eq:todag} in this case does indeed show that the sum-of-residues in
different representations and Weyl orbits $(\rho, \bar k)$, $(\rho', \bar k')$ coincide up to an overall factor of
$q_{\rho, \bar k}/q_{\rho', \bar k'}$, which in \eqref{eq:todaeta}--\eqref{eq:todag} is
accounted for by the explicit inclusion of $q_{\rho, \bar k}$ at the denominator.
\item Isomorphisms of the type \eqref{eq:geniso} have already appeared in the
  literature, and they all fit in the framework of \cref{conj:mirror}. In their
  original paper \cite{MR1606165}, Dubrovin--Zhang formulate a mirror theorem
  for the $A$-series which is indeed the specialisation of \cref{conj:mirror} to
  $\mathfrak{g}=\mathfrak{sl}_N$ and $\rho=\square=\rho_{\omega_1}$ the fundamental representation. Their mirror theorem was
  extended to the other classical Lie algebras $\mathfrak{b}_N$, $\mathfrak{c}_N$ and $\mathfrak{d}_N$ by the same
  authors with Strachan and Zuo in \cite{Dubrovin:2015wdx}: the Toda mirrors
  of \cref{conj:mirror} specialise to their LG models for
  $\mathfrak{g}=\mathfrak{so}_{2N+1}$, $\mathfrak{sp}_N$ and
  $\mathfrak{so}_{2N}$ with $\rho$ being in all cases the defining vector
  representation $\rho=\rho_{\omega_1}$. 
\item At the opposite end of the simple Lie
  algebra spectrum, \cref{thm:frob}
  gives an affirmative answer to \cref{conj:mirror} for the most exceptional
  example of $\cG={\rm E}_8$; it is only natural to speculate that the missing
  exceptional cases should fit in as well.
\item Some further indication that \cref{conj:mirror} should hold true comes from the
  study of Seiberg--Witten curves in the same limit considered for
  \cref{sec:limIII}, together with $\Lambda_4\to 0$. It was speculated already
  in \cite{Lerche:1996an} that the perturbative limit of 4d SW curves with ADE
  gauge group should be related to ADE topological Landau--Ginzburg models
  (and hence the finite Coxeter Frobenius manifolds of \eqref{eq:coxeter})
  via an operation foreshadowing the notion of almost-duality in
  \cite{MR2070050}; this was further elaborated upon in
  \cite{Eguchi:1996nh} for $\mathfrak{g}=\mathfrak{e}_6$, and
  \cite{Eguchi:2001fm} for $\mathfrak{g}=\mathfrak{e}_7$.
\item Finally, our way of accommodating the extra datum of the choice of
  simple root $\alpha_{\bar i}$ is not only consistent with the results
  \cite{Dubrovin:2015wdx}, but also with the general idea that these Frobenius
  manifolds should be related to each other by a Type I symmetry of WDVV (a
  Legendre-type transformation) in the language of
  \cite[Appendix~B]{Dubrovin:1994hc}. Indeed, different choices of fundamental
  characters $u_{\bar i}$ shifting the value of the superpotential correspond precisely
  to a symmetry of WDVV where the new unit vector field is one of the old non-unital coordinate
  vector fields. This parallels precisely the general construction of
  \cite{MR1606165, Dubrovin:2015wdx}.
\een

\begin{table}[t]
\begin{tabular}{|c|c|c|c|c|c|c|c|}
\hline
$\cG$ & $\rho$ & $\bar k$ & $q_{\rho, \bar k}$ & $\dim\rho$ & $\deg_\mu \Xi_{\rho, \rm red}$ & $\deg_{u_j} \Xi_{\rho, \rm red}$ & $g\big(\Gamma_u^{(\bar i)}\big)$ \\
\hline
$\mathrm{E}_6$ & $\rho_{\omega_1}$ & $1$ & $6$ & $27$ &  $27$ & $5, 3, 2, 3, 5, 3$ & $5$ \\
\hline
$\mathrm{E}_7$ & $\rho_{\omega_6}$ & $6$ & $12$ & $56$ & $56$ & $6, 4, 3, 4, 5, 10, 6$ & $33$ \\
\hline
$\mathrm{E}_8$ & $\rho_{\omega_7}$ & $7$ & $60$ & $248$ & $240$ & $23, 13, 9, 11, 14, 19, 29, 17$ &
$128$ \\
\hline
${\rm F}_4$ & $\rho_{\omega_4}$ & $4$ & $6$ & $26$ & $24$ & $3, 2, 3, 5$ & $4$ \\
\hline
${\rm G}_2$ & $\rho_{\omega_1}$ & $1$ & $6$ & $7$ & $7$ & $2, 1$ & $0$ \\
\hline
\end{tabular}
\caption{Degree and genera of minimal spectral curve (putative) mirrors of
  $\widetilde{X}_{\mathfrak{g},\bar i}$ for the exceptional series EFG; for
  simplicity I only indicate the genus for the original choice of marked node
  $\bar i$ in \cite[Table I]{MR1606165}.}
\label{tab:EFG}
\end{table}

It should be noticed that, away from the classical ABCD series and the
exceptional case $\mathrm{G}_2$,
$\Gamma'^{(\bar i)}_{u}$ is typically not a rational curve, not even for the
``minimal'' case in which $\alpha_{\bar i}$ is chosen as the root
corresponding to the attaching node of the external root in the Dynkin
diagram, and $\rho$ is a minimal non-trivial irreducible representation. For the time being, I'll content myself to
provide some data on the
exceptional cases in \cref{tab:EFG}, and defer a proof of \cref{conj:mirror} to a separate publication.

\subsection{Polynomial $\bbP^1$ orbifolds at higher genus}
\label{sec:p1orb}

As a final application, I restrict my attention to $\cG$ being
simply-laced. In this case, \cref{conj:mirror} and \cite{MR2672302} would imply the following
\begin{conj}
With notation as in \cref{conj:mirror}, let $\bar i$ be an arbitrary node of
the Dynkin diagram for $\mathfrak{g}$ of type $A$, or the node corresponding to the
highest dimensional fundamental representation\footnote{Equivalently, this is the attaching node of the external root(s) in the
Dynkin diagram.} for type $D$ and $E$:
\beq
\bar i =
\l\{
\bary{lcrcl}
i=1,\dots, n, & ~ & \mathfrak{g} & = & \mathrm{A}_n \\
 n-2, & ~ & \mathfrak{g} & = & \mathrm{D}_n \\
3, & ~ & \mathfrak{g} & = & \mathrm{E}_n \\
\eary
\right.
\eeq
Then,
\beq
X_{\mathfrak g,\bar i}^{\rm Toda} \simeq QH_{\rm orb}(C_{\mathfrak{g}})
%\label{eq:geniso2}
\eeq
where $C_{\mathfrak{g}}$ is the polynomial $\bbP^1$-orbifold of type
$\mathfrak{g}$:
\beq
C_{\mathfrak{g}} = 
\l\{
\bary{lcrcl}
\bbP(\bar i,n-\bar i+1), & ~ & \mathfrak{g} & = & \mathrm{A}_n \\
\bbP(2,2,n-2) & ~ & \mathfrak{g} & = & \mathrm{D}_n \\
\bbP(2,3,n-3) & ~ & \mathfrak{g} & = & \mathrm{E}_n \\
\eary
\right.
\eeq
\label{conj:mirrorp1}
\end{conj}
There are two noteworthy implications of such a statement. The first is that
the LG model of the previous section would provide a dispersionless Lax
formalism for the integrable hierarchy of topological type on the loop space
of the Frobenius manifold $QH_{\rm orb}(C_{\mathfrak{g}})$ \cite[Lecture~6]{Dubrovin:1994hc}; for type~A, this is well-known to be the extended
bi-graded Toda hierarchy of \cite{MR2246697} (see also \cite{MR2433616}), and
for all ADE types, a construction was put forward in \cite{Milanov:2014pma} for these
hierarchies in the form of Hirota quadratic equations. The zero-dispersion
Lax formulation of the hierarchy could be a key to relate such remarkable, yet obscure hierarchy
to a well-understood parent 2+1 hierarchy such as 2D-Toda, as was done in a
closely related context in
\cite{Brini:2014mha}. A more direct consequence is a Givental-style,
genus-zero-controls-higher-genus statement, as follows. On the Gromov--Witten
side, and as a vector space, the Chen--Ruan co-homology of $C_{\mathfrak{g}}$
is the co-homology of the inertia stack $IC_{\mathfrak{g}}$
\cite{MR1950941,Zaslow:1993fa}, which is generated by the identity class
$\phi_{r_{\mathfrak{g}}}\triangleq \mathbf{1}_0$, the K\"ahler class
$\phi_{0}\triangleq p$, and twisted cohomology
classes concentrated at the stacky points of $C_{\mathfrak{g}}$,
\beq
\phi_{\nu(i,r)} \triangleq \mathbf{1}_{\l(\frac{i}{s_r},r\r)} \in H^{(\frac{i}{s_r},r)}(C_{\mathfrak{g}}) \simeq H(B \bbZ_{s_r}), \quad i=1, \dots, r-1
\eeq
where $r=1,2$ for type A and $r=1,2,3$ for type D and E label the orbifold
points of $C_{\mathfrak{g}}$, $s_r$ is the order of the respective isotropy
groups, we label components of the $IC_{\mathfrak{g}}$ by
$(\frac{i}{r},s_r)$, and $\nu(i,r)$ is a choice of a map to $[[1, r_{\mathfrak{g}}]]$
increasingly sorting the sets of pairs $(i,r)$ by the value of
$i/s_r$.\footnote{For the case $\mathfrak{g}=\mathfrak{e}_8$, since ${\rm gcd}(2,3,5)=1$, there is no
  ambiguity in the choice of $\nu$, and the choice of labelling of $\phi_\a$ here
  was made to match  that of the 
  Saito vector fields $\de_{\mathtt{t_\a}}$ of \cref{lem:flatc}: up to scale,
  we have $\phi_\a=\de_{\mathsf{t}_\a}$, $\de_{\mathsf{t}_0}=p$, $\de_{\mathsf{t}_8}=\mathbf{1}_0$,
  $\de_{\mathsf{t}_1}=\mathbf{1}_{\l(\frac{1}{5},3\r)}$,
  $\de_{\mathsf{t}_2}=\mathbf{1}_{\l(\frac{1}{3},2\r)}$,
  $\de_{\mathsf{t}_3}=\mathbf{1}_{\l(\frac{2}{5},3\r)}$,   
  $\de_{\mathsf{t}_4}=\mathbf{1}_{\l(\frac{1}{2},1\r)}$,
  $\de_{\mathsf{t}_5}=\mathbf{1}_{\l(\frac{3}{5},3\r)}$,
  $\de_{\mathsf{t}_6}=\mathbf{1}_{\l(\frac{2}{3},2\r)}$,
  $\de_{\mathsf{t}_7}=\mathbf{1}_{\l(\frac{4}{5},2\r)}$.}
Define now the genus-$g$ full-descendent Gromov--Witten potential of
$C_{\mathfrak{g}}$ as the formal power series
\beq
\cF_g^{C_{\mathfrak{g}}} = \sum_{n\geq 0} \sum_{d\in
  \mathrm{Eff}(C_{\mathfrak{g}})} \sum_{\stackrel{\a_1, \dots, \a_n}{k_1,
    \dots, k_n}}\frac{\prod_{i=1}^n t_{\a_i, k_i}}{n!} \bra
\tau_{k_1}(\phi_{\a_1}) \dots \tau_{k_n}(\phi_{\a_n}) \ket^{C_{\mathfrak{g}}}_{g,n,d},
\eeq
where $\mathrm{Eff}(C_{\mathfrak{g}})\subset
H^2(C_{\mathfrak{g}},\bbZ)/H^2_{\rm tor}(C_{\mathfrak{g}},\bbZ)$ is the set of
degrees of twisted stable maps to $C_{\mathfrak{g}}$, and the usual
correlator notation for multi-point descendent Gromov--Witten invariants was employed,
\beq
\bra
\tau_{k_1}(\phi_{\a_1}) \dots \tau_{k_n}(\phi_{\a_n})
\ket^{C_{\mathfrak{g}}}_{g,n,d} \triangleq
\int_{[\cM_{g,n}(C_{\mathfrak{g}},d)]^{\rm vir}}\prod_{i=1}^n \mathrm{ev}_i^*
\phi_{\a_i} \psi_i^{k_i}.
\eeq
Since $QH(C_{\mathfrak{g}})$ is semi-simple, the Givental--Teleman
Reconstruction theorem applies \cite{MR2917177}. I will refer the reader to
\cite{MR1901075,lee2004frobenius,Brini:2013zsa} for the relevant background
material, context, and detailed explanations of origin and inner workings of the formula; symbolically and somewhat
crudely, this is, for a general target $X$ with semi-simple quantum cohomology,
\beq
\exp\l(\sum_{g}\epsilon^{2g-2} \cF_g^{X}(t)\r) =
\widehat{S_{{\rm GW}, X}^{-1}} \widehat{\psi_{{\rm GW}, X}} \widehat{R_{{\rm GW}, X}}
\prod_{i=1}^{r_{\mathfrak{g}}+1} \tau_{\rm KdV}(u)
\label{eq:giv}
\eeq
where the {\it calibrations} $S_{{\rm GW}, X}$ and $R_{{\rm GW}, X}$ are elements of the linear symplectic loop group
of $QH(X)\otimes \bbC[\hbar,\hbar^{-1}]]$ given by flat coordinate
 frames for the restricted Dubrovin connection to the internal direction of
 the Frobenius manifold \eqref{eq:defEC}, which are respectively analytic in
 $\hbar$ and formal in $1/\hbar$. The hat symbol signifies normal-ordered
 quantisation of the corresponding linear symplectomorphism (namely, an
 exponentiated quantised quadratic Hamiltonian), $\psi_{{\rm GW}, X}$ is the Jacobian matrix
 of the change-of-variables from flat to normalised canonical frame, and
 $\tau_{\rm KdV}$ is the Witten--Kontsevich Kortweg--de-Vries $\tau$-function,
 that is, the exponentiated
 generating function of GW invariants of the point. The essence of \eqref{eq:giv}
 is that there exists a judicious composition of explicit, exponentiated quadratic differential operators in
 $t_{\a,k}$ and changes of variables $u_k^{(i)} \to t_{\a, k}$ from the
 $k^{\rm th}$ KdV time of the $i^{\rm th}$ $\tau$-function in \eqref{eq:giv}
 which returns the full-descendent, all-genus GW partition function of
 $X$. In our specific case $X=C_{\mathfrak{g}}$ (and in general, whenever we
 consider non-equivariant GW invariants), by the Conformality axiom of
 \cref{def:frob1}, both $\widehat{S}_{{\rm GW}, X}$ and $\widehat{R}_{{\rm GW}, X}$ are determined by 
 the Frobenius manifold structure of $QH(X)$ alone, without any
 further input \cite{MR2917177}: the grading condition given by the flatness
 in the $\bbC^\star_{\hbar}$ direction of the Dubrovin connection fixes uniquely
 the normalisation of the canonical flat frames $S$ and $R$ at $\hbar=0,
 \infty$ respectively. For reference, the $R$-action on the Witten--Kontsevich
 $\tau$-functions gives the {\it ancestor potential} in the normalised canonical frame
\beq
\exp\l(\sum_{g}\epsilon^{2g-2} \cA_g^{X}(t)\r) \triangleq \widehat{R_{{\rm GW}, X}}
\prod_{i=1}^{r_{\mathfrak{g}}+1} \tau_{\rm KdV}(u)
\eeq
to which the descendent generating function \eqref{eq:giv} is related by a
linear change of variables (via $\psi$) and a triangular transformation of the
full set of time variables (via $S^{-1}$); see
\cite[Chapter~2]{lee2004frobenius}.

On the Toda/spectral curve side, a similar higher genus reconstruction theorem exists in light
of its realisation as a Frobenius submanifold of a Hurwitz space: this is, as
in \cref{sec:toprec}, the Chekhov--Eynard--Orantin (CEO) topological recursion procedure, giving a 
sequence $(F_g^{\rm CEO}(\mathscr{S}), W_{g,h}^{\rm CEO}(\mathscr{S}))$ of generating
functions \eqref{eq:wghrec}-\eqref{eq:fgrec} specified by the
Dubrovin--Krichever data of \cref{defn:dk}.
%datum $\mathscr{S}=(\mathscr{G}^{\rm Toda}_\mathfrak{g}, \rd \lambda,
%\rd\ln\mu, \Lambda^-)$. 
Having proved, or taking for granted the isomorphism of the underlying Frobenius manifolds as in \cref{conj:mirror}, it is natural to ask
whether the two higher genus theories are related at all. A precise answer
comes from the work of \cite{DuninBarkowski:2012bw}, where the authors
show that there exists an explicit change of variables $t_{\a,k} \to v_{i,j}$
and an $R$-calibration of the Hurwitz space Frobenius (sub)manifold associated
to $\mathscr{S}_{\mathfrak{g}}$ such that
\beq
\exp\l(\epsilon^{2g-2} \sum_{g,d}W^{\rm CEO}(\mathscr{S})_{g,d}(v)\r) = \widehat{R_{\rm CEO}}(\mathscr{S})
\prod_{i=1}^{r_{\mathfrak{g}}+1} \tau_{\rm KdV}(u)
\label{eq:giv2}
\eeq
where the independent variables $v_{i,j}$ on the l.h.s. are obtained from the
arguments of the CEO multi-differentials upon expansion
around the $i^{\rm th}$ branch point of the spectral curve (see
\cite[Theorem~4.1]{DuninBarkowski:2012bw} and the discussion preceding it for
the exact details). In other words, the topological recursion reconstructs the {\it
  ancestor} potential of a two-dimensional semi-simple cohomological field
theory, with $R$-calibration $R_{\rm CEO}(\mathscr{S})$ entirely specified by the spectral
curve geometry via a suitable Laplace transform of the Bergman kernel.  One upshot of this is that, up to a further
change-of-variables and a (non-trivial) shift by a quadratic term, \eqref{eq:giv2} can be
put in the form of \eqref{eq:giv}.

So, in a situation where $\mathscr{S}_X$ is a spectral curve mirror to $X$, 
we have two identical reconstruction theorems for the higher genus ancestor potential
starting from genus zero CohFT data, both being unambiguosly specified in terms
of $R$-actions $R_{\rm GW, X}$ and $R_{\rm CEO}(\mathscr{S}_X)$. If
these agree, then the full higher genus potentials agree, and the higher genus ancestor invariants of $X$
are computed by the topological recursion on $\mathscr{S}_X$ by
\eqref{eq:giv2}. Happily, it is a result of Shramchenko that in
non-equivariant GW theory this is always precisely
the case \cite{MR2429320} (see also \cite[Theorem~7]{dunin2016primary}): 
\beq
R_{\rm GW, X}=R_{\rm CEO}(\mathscr{S}_X).
\eeq
In other words, the $R$-calibration $R_{\rm CEO}(\mathscr{S}_X)$, which is uniquely specified by the
Bergmann kernel of a family of spectral
curves $\mathscr{S}_X$ whose prepotential coincides with the genus zero GW
potential of a projective variety\footnote{More generally, a Gorenstein
  orbifold with projective coarse moduli space.} $X$, coincides with the
$R$-calibration $R_{{\rm GW}, X}$ uniquely picked by the de Rham grading in the
(non-equivariant) quantum cohomology of $X$. We get to the following
\begin{cor}
Suppose that \cref{conj:mirror} holds. Then the ancestor higher genus
potential of $C_{\mathfrak{g}}$ equates to the
higher genus topological recursion potential
\beq
\cA_g^{C_{\mathfrak{g}}}= \sum_{h}W^{\mathscr{S}_{\mathfrak{g}}}_{g,h}
\eeq
up to the change-of-variables of \cite[Theorem~4.1]{DuninBarkowski:2012bw}.
\end{cor}
In particular, such all-genus full-ancestor statements hold in type A by
\cite[Theorem~3.1]{MR1606165}, type D by \cite[Theorem~5.6]{Dubrovin:2015wdx} and type
$\mathrm{E}_8$ by \cref{thm:frob}. The two remaining exceptional cases can be
treated along the same lines of \cref{thm:frob}, and far more easily than the
case of $\mathrm{E}_8$, and are left as an exercise to the reader.

\begin{rmk}[On an ADE Norbury--Scott theorem]

For the case of the Gromov--Witten theory of $\bbP^1$, it was proposed by
Norbury--Scott in \cite{MR3268770}, supported by a low-genus proof and a heuristic all-genus argument, and
later proved in full generality by the authors of
\cite{DuninBarkowski:2012bw} using \eqref{eq:giv2}, that the residue at infinity of the CEO
differentials $W^{\mathscr{S}_{\mathfrak{g}}}_{g,n}$ gives the $n$-point,
genus-$g$ stationary GW invariant of $\bbP^1$,
\beq
\prod_{j=1}^n \Res_{z_j=\infty}\frac{z_j^{m_j+1}}{(m_j+1)!}W^{\rm Toda}_{g,n}(z_1,
\dots, z_n)  = (-)^n \bra\prod_{i=1}^n \tau_{m_j}(p) \ket.
\label{eq:ns}
\eeq
I fully expect that a completely analogous ADE orbifold version of \eqref{eq:ns}, which allows for a very efficient way
to compute GW invariants at higher genera, would hold for
all polynomial $\bbP^1$-orbifolds. For type A and D, where the curve is
rational, the statement of \eqref{eq:ns} would probably carry forth verbatim,
with the r.h.s. being given by $n$-pointed, non-stationary, untwisted GW
invariants. For type E, it will perhaps be necessary to sum over all
branches above $\infty$ (8, in the case of $\mathrm{E}_8$) to obtain the
desired result. I also expect that in type D and E, poles of $\lambda$ at
finite $\mu$ will presumably compute twisted invariants, with twisted insertions 
being labelled by the location of the poles. In particular, in type $\mathrm{E}_8$, the poles
at $\frac{\ln \mu}{2\pi \ri}\in \{1/5,1/3,2/5,1/2,3/5,1/3,4/5\}$ should
correspond to insertions of $\mathbf{1}_{i/s_r,r}$ for the corresponding value
of $i/s_r$.
\end{rmk}

\begin{appendix}

\section{Proof of \cref{prop:pg}}
\label{sec:proofpg}

I am going to prove \cref{prop:pg} by first establishing the following
\begin{lem}
The number of double cosets of $\cW$ by $\cW_{\alpha_0}$ is 
\beq
|\cW_{\alpha_0} \backslash \cW / \cW_{\alpha_0}|=5.
\eeq
\label{lem:numcosets}
\end{lem}
\begin{proof}
The order of the double coset space $\cW_{\alpha_0} \backslash \cW / \cW_{\alpha_0}$
is the square norm of the character of the trivial
representation of $\cW_{\alpha_0}$, induced up to $\cW$ \cite[Ex.~7.77a]{MR1676282},
\beq
|\cW_{\alpha_0} \backslash \cW / \cW_{\alpha_0}|=\bra \mathrm{ind}_{\cW_{\alpha_0}}^\cW
1, \mathrm{ind}_{\cW_{\alpha_0}}^\cW 1\ket.
\label{eq:orddoubcs}
\eeq
Now, $\mathrm{ind}_{\cW_{\alpha_0}}^\cW 1$ is just the permutation representation
$\bbC\bra \Delta^* \ket$ on the free vector space on the set of non-zero roots
$\Delta^*\simeq \cW/\cW_{\alpha_0}$. Suppose that
\beq
\bbC\bra \Delta^* \ket = \bigoplus m_i R_i
\eeq
for irreducible representations $R_i \in R(\cW)$ and $m_i \in \bbZ$. Then, by \eqref{eq:orddoubcs}, 
\beq
|\cW_{\alpha_0} \backslash \cW / \cW_{\alpha_0}|=\sum_i m_i^2.
\eeq
The multiplicity of $R_i$ in $\bbC\bra \Delta^* \ket$ is easily computed as
follows. Let $c \in \cW$ and $[c]$ its conjugacy class. Then its $\bbC\bra\Delta^*\ket$-character 
\beq
\chi_{\bbC\bra \Delta^* \ket}([c]) = \mathrm{dim}_\bbC\{v \in \mathfrak{h}^* |
cv=v\} 
\label{eq:chiCDdim}
\eeq
is equal to the dimension of the eigenspace of fixed points of $c$. In the
standard labelling \cite{MR1266626,MR1285208} of conjugacy classes of
$\cW=\mathrm{Weyl}(\mathfrak{e}_8)$, we compute the r.h.s. of
\eqref{eq:chiCDdim} to be
\beq
\chi_{\bbC\bra \Delta^* \ket}([c])=
\l\{
\bary{crcl}
2 & [c] & \in &
\{\mathtt{2b},\mathtt{4c},\mathtt{6a},\mathtt{12a},\mathtt{4e},\mathtt{4g},\mathtt{30b},\mathtt{10d},\mathtt{6m},\mathtt{3d},\mathtt{24c},\nn \\ & & & ~\mathtt{6s},
\mathtt{18c},\mathtt{6x},\mathtt{12n},
\mathtt{7a},\mathtt{
   14a},\mathtt{6y},\mathtt{6z},\mathtt{6ab}\}, \\
4 & [c] & \in &
\{\mathtt{4h},\mathtt{2e},\mathtt{10b},\mathtt{30c},\mathtt{6i},\mathtt{30d},\mathtt{12p},\mathtt{12r}\}, \\
6 & [c] & \in &
\{\mathtt{8c},\mathtt{4b},\mathtt{8e},\mathtt{6c},\mathtt{2c},\mathtt{12h},\mathtt{8h},\mathtt{6q},\mathtt{14b},\mathtt{12q}\}, \\
8 & [c] & \in &
\{\mathtt{4d},\mathtt{4m}\}, \\
12 & [c] & \in &
\{\mathtt{12f},\mathtt{2f},\mathtt{5a},\mathtt{6g},\mathtt{10e}\}, \\
14 & [c] & = & 
\mathtt{12m}, \\
20 & [c] & = & 
\mathtt{18d}, \\
24 & [c] & \in &
\{\mathtt{8a},\mathtt{24b},\mathtt{20b}\}, \\
26 & [c] & = & 
\mathtt{6b}, \\
30 & [c] & = & 
\mathtt{12d}, \\
40 & [c] & = & 
\mathtt{6h}, \\
60 & [c] & = & 
\mathtt{12e}, \\
72 & [c] & = & 
\mathtt{4f}, \\
126 & [c] & = & 
\mathtt{6e}, \\
240 & [c] & = & 
\mathtt{1a}, \\
0 & \mathrm{else.} & &
\eary
\r.
\label{eq:charCD}
\eeq
From \eqref{eq:charCD}, we obtain\footnote{I have labelled here irreps of $\cW$ as
  $\mathrm{dim} R_{\mathrm{sgn}\chi_R(\mathtt{8d})}$. A rather standard use in
  the literature is to label irreps of exceptional Weyl groups by how they are
  stored in the
    \texttt{GAP} library; for reference, here the summands in the decomposition
    \eqref{eq:charCD} would be called $\mathtt{X.1}$ (trivial),
    $\mathtt{X.3}$ (Coxeter), $\mathtt{X.8}$, $\mathtt{X.15}$ and $\mathtt{X.16}$.}
\bea
m_i & = &\bra \chi_{R_i}, \chi_{\bbC\bra \Delta^* \ket }\ket \nn \\
& = & \l\{
\bary{cl}
1 & R_i \in \{\mathbf{1}_+\simeq 1, \mathbf{8}_+\simeq \mathfrak{h}^*, \mathbf{35}_-,
\mathbf{84}_-, \mathbf{112}_+\}, \\
0 & \mathrm{else},
\eary
\r.
\label{eq:deccharCD}
\eea
from which the claim follows.
\end{proof}

\cref{prop:pg} is an easy consequence of the \cref{lem:numcosets}: by
$\cW_{\alpha_0}$-invariance, \eqref{eq:proj} defines an element of the Hecke
ring, and it is immediately seen to assume exactly five constant values
$-2$,$-1$, $0$, $1$, $2$ on the hyperplanes $\mathsf{H}_i$, which are then in bijection
with the elements of $H(\cW,\cW_{\alpha_0})$.

\section{Some formulas for the $\mathfrak{e}_8$ and $\mathfrak{e}_8^{(1)}$ root system}

I gather here some reference material for the finite and affine $\mathrm{E}_8$
root systems. Let $\{e_i\}$, $i=1, \dots, 8$ be an orthonormal basis for
$\bbR^8$. The simple roots $\{\a_i\}$, $i=1,\dots, 8$ have components in this
basis given by
%
%\subsection{Simple roots in the orthogonal basis}
%
\beq
%\Delta_{{\rm E}_8}=\left\{
\begin{array}{rcl}\a_1 &=& (\frac{1}{2}, -\frac{1}{2}, -\frac{1}{2}, -\frac{1}{2}, -\frac{1}{2},
-\frac{1}{2}, -\frac{1}{2}, \frac{1}{2}), \\ \a_2 &=& (0, 0, 0, 0, 0, -1, 1, 0), \\ \a_3 &=& (0, 0, 0, 0, -1,
1, 0, 0), \\ \a_4 &=& (0, 0, 0, -1,
   1, 0, 0, 0), \\ \a_5 &=& (0, 0, -1, 1, 0, 0, 0, 0), \\ \a_6 &=& (1, 1, 0, 0, 0, 0,
   0, 0), \\ \a_7 &=& (0, -1, 1, 0, 0, 0, 0, 0), \\ \a_8 &=& (-1, 1, 0, 0, 0, 0, 0,
   0)\end{array}
%\right\}
.
\label{eq:roots}
\eeq
The affine root system is obtained from \eqref{eq:roots} upon adding the
affine root
\beq
%\Delta_{{\rm E}_8^{(1)}}= \Delta_{{\rm E}_8} \cup \{
\a_0 = (0,0,0,0,0,0,-1,-1).
%\}
\label{eq:affroot}
\eeq
The respective Cartan matrices are given, from
\eqref{eq:roots}--\eqref{eq:affroot}, by
\beq
\mathscr{C}^{\mathfrak{e}_8}=
\left(
\begin{array}{cccccccc}
 2 & 0 & 0 & 0 & 0 & 0 & 0 & -1 \\
 0 & 2 & -1 & 0 & 0 & 0 & 0 & 0 \\
 0 & -1 & 2 & -1 & 0 & 0 & 0 & 0 \\
 0 & 0 & -1 & 2 & -1 & 0 & 0 & 0 \\
 0 & 0 & 0 & -1 & 2 & 0 & -1 & 0 \\
 0 & 0 & 0 & 0 & 0 & 2 & -1 & 0 \\
 0 & 0 & 0 & 0 & -1 & -1 & 2 & -1 \\
 -1 & 0 & 0 & 0 & 0 & 0 & -1 & 2 \\
\end{array}
\right),
\label{eq:cartE8}
\eeq
\beq
\mathscr{C}^{\mathfrak{e}_8^{(1)}}=
\left(
\begin{array}{cccccccccc}
2 & 0 & 1 & 0 & 0 & 0 & 0 & 0 & 0 \\
0 & 2 & 0 & 0 & 0 & 0 & 0 & 0 & -1 \\
1 & 0 & 2 & -1 & 0 & 0 & 0 & 0 & 0 \\
0 &  0 & -1 & 2 & -1 & 0 & 0 & 0 & 0 \\
0 & 0 & 0 & -1 & 2 & -1 & 0 & 0 & 0 \\
0 & 0 & 0 & 0 & -1 & 2 & 0 & -1 & 0 \\
0 & 0 & 0 & 0 & 0 & 0 & 2 & -1 & 0 \\
0 & 0 & 0 & 0 & 0 & -1 & -1 & 2 & -1 \\
0 & -1 & 0 & 0 & 0 & 0 & 0 & -1 & 2 \\
\end{array}
\right).
\label{eq:cartE81}
\eeq
The resulting simple Lie algebra for $\mathfrak{e}_8$ has rank 8 and dimension
240. In the $\alpha$-basis, \eqref{eq:roots}, the affine root
\eqref{eq:affroot} reads
\beq
\a_0 = \sum_i \mathfrak{d}_i \a_i= 2 \a_1 + 4 \a_2 +6\a_3 +5 \a_4 +
4\a_5+3\a_6 +2 \a_7+3\a_8.
\eeq
Since $\det\mathscr{C}^{\mathfrak{e}_8}=1$, we have
$\Lambda_r(\mathfrak{e}_8)\simeq \Lambda_w(\mathfrak{e}_8)$.

\subsection{On the minimal orbit of $\cW$}
\label{sec:minorb}
I group here the details of the minimal orbit of $\cW$ in $\Lambda_r \subset
\bbZ^{\mathsf{s}}$ generated by the adjoint weight $\omega_7$, in terms of $240$
vectors in a ${\mathsf{s}}=30$-dimensional lattice. Since the orbit is in bijection
with the set of non-zero roots of $\mathfrak{g}$, $\omega$ is in the orbit iff
$-\omega$ is; also the cyclic shift of the components in $\bbZ^{\mathsf{s}}$
corresponds to the action of the Coxeter element on the orbit, which is thus
preserved if we send $\omega \to (\omega_{j + 1\,\,{\rm
    mod}\,\,{\mathsf{s}}})_j$. The resulting $\bbZ_2 \times \bbZ_{30}$ action breaks
up the orbit into sub-orbits, representatives for which are
displayed\footnote{The entries in the $7^{\rm th}$ and $8^{\rm th}$ column of
  row $(\omega)_7$, as well as those of the $26^{\rm th}$ and $27^{\rm th}$
  column of row $(\omega)_{11}$ correct four typos in the table of \cite[Appendix~F.1]{Borot:2015fxa}.} in \cref{tab:minorb}.

\begin{table}[h]
\begin{tabular}{|c|c|c|c|c|c|c|c|c|c|c|c|c|}
\hline
$n_0$ & $\pm 6$ & $\pm 5$ & $\pm 4$ & $\pm 3$ & $\pm 3$ & $\pm 2$ & $\pm 2$ &
$\pm 1$ & $0$ & $0$ & $0$ & $0$ \\ \hline
card & $10$ & $12$ & $30$ & $20$ & $20$ & $30$ & $30$ & $60$ & $10$ & $10$ &
$6$ & $2$ \\ \hline
$(\omega)_1$ & $1$ & $-1$ & $-1$ & $-1$ & $1$ & $-1$ & $1$ & $2$ & $-1$ & $1$ & $1$ & $1$ \\ \hline
$(\omega)_2$ &  $0$ & $1$ & $1$ & $1$ & $-1$ & $1$ & $-1$ & $-1$ & $0$ & $0$ & $-1$ & $-1$ \\ \hline
$(\omega)_3$ &  $0$ & $0$ & $0$ & $0$ & $1$ & $1$ & $0$ & $0$ & $0$ & $-1$ & $0$ & $1$ \\ \hline
$(\omega)_4$ &  $0$ & $0$ & $0$ & $1$ & $-1$ & $-1$ & $1$ & $0$ & $0$ & $0$ & $1$ & $-1$ \\ \hline
$(\omega)_5$ &  $0$ & $0$ & $1$ & $-1$ & $1$ & $0$ & $-1$ & $0$ & $1$ & $0$ & $-1$ & $1$ \\ \hline
$(\omega)_6$ &  $1$ & $1$ & $-1$ & $0$ & $0$ & $0$ & $1$ & $0$ & $-1$ & $1$ & $0$ & $-1$ \\ \hline
$(\omega)_7$ &  $0$ & $-1$ & $0$ & $0$ & $0$ & $0$ & $0$ & $1$ & $0$ & $0$ & $1$ & $1$ \\ \hline
$(\omega)_8$ &  $0$ & $1$ & $1$ & $1$ & $0$ & $1$ & $-1$ & $-1$ & $0$ & $-1$ & $-1$ & $-1$ \\ \hline
$(\omega)_9$ &  $0$ & $0$ & $0$ & $0$ & $0$ & $0$ & $1$ & $0$ & $0$ & $0$ & $0$ & $1$ \\ \hline
$(\omega)_{10}$ &  $0$ & $0$ & $0$ & $0$ & $0$ & $-1$ & $0$ & $0$ & $1$ & $0$ & $1$ & $-1$ \\ \hline
$(\omega)_{11}$ &  $1$ & $0$ & $0$ & $-1$ & $1$ & $0$ & $0$ & $1$ & $-1$ & $1$ & $-1$ & $1$ \\ \hline
$(\omega)_{12}$ &  $0$ & $1$ & $0$ & $1$ & $-1$ & $1$ & $0$ & $-1$ & $0$ & $0$ & $0$ & $-1$ \\ \hline
$(\omega)_{13}$ &  $0$ & $-1$ & $0$ & $0$ & $1$ & $0$ & $0$ & $1$ & $0$ & $-1$ & $1$ & $1$ \\ \hline
$(\omega)_{14}$ &  $0$ & $1$ & $1$ & $1$ & $-1$ & $0$ & $0$ & $-1$ & $0$ & $0$ & $-1$ & $-1$ \\ \hline
$(\omega)_{15}$ &  $0$ & $0$ & $0$ & $-1$ & $1$ & $0$ & $0$ & $0$ & $1$ & $0$ & $0$ & $1$ \\ \hline
$(\omega)_{16}$ &  $1$ & $0$ & $-1$ & $0$ & $0$ & $-1$ & $1$ & $1$ & $-1$ & $1$ & $1$ & $-1$ \\ \hline
$(\omega)_{17}$ &  $0$ & $0$ & $1$ & $0$ & $0$ & $1$ & $-1$ & $0$ & $0$ & $0$ & $-1$ & $1$ \\ \hline
$(\omega)_{18}$ &  $0$ & $1$ & $0$ & $1$ & $0$ & $1$ & $0$ & $-1$ & $0$ & $-1$ & $0$ & $-1$ \\ \hline
$(\omega)_{19}$ &  $0$ & $-1$ & $0$ & $0$ & $0$ & $-1$ & $1$ & $1$ & $0$ & $0$ & $1$ & $1$ \\ \hline
$(\omega)_{20}$ &  $0$ & $1$ & $1$ & $0$ & $0$ & $0$ & $-1$ & $-1$ & $1$ & $0$ & $-1$ & $-1$ \\ \hline
$(\omega)_{21}$ &  $1$ & $0$ & $-1$ & $-1$ & $1$ & $0$ & $1$ & $1$ & $-1$ & $1$ & $0$ & $1$ \\ \hline
$(\omega)_{22}$ &  $0$ & $0$ & $0$ & $1$ & $-1$ & $0$ & $0$ & $0$ & $0$ & $0$ & $1$ & $-1$ \\ \hline
$(\omega)_{23}$ &  $0$ & $0$ & $1$ & $0$ & $1$ & $1$ & $-1$ & $0$ & $0$ & $-1$ & $-1$ & $1$ \\ \hline
$(\omega)_{24}$ &  $0$ & $1$ & $0$ & $1$ & $-1$ & $0$ & $1$ & $-1$ & $0$ & $0$ & $0$ & $-1$ \\ \hline
$(\omega)_{25}$ &  $0$ & $-1$ & $0$ & $-1$ & $1$ & $-1$ & $0$ & $1$ & $1$ & $0$ & $1$ & $1$ \\ \hline
$(\omega)_{26}$ &  $1$ & $1$ & $0$ & $0$ & $0$ & $0$ & $0$ & $0$ & $-1$ & $1$ & $-1$ & $-1$ \\ \hline
$(\omega)_{27}$ &  $0$ & $0$ & $0$ & $0$ & $0$ & $1$ & $0$ & $0$ & $0$ & $0$ & $0$ & $1$ \\ \hline
$(\omega)_{28}$ &  $0$ & $0$ & $0$ & $1$ & $0$ & $0$ & $0$ & $0$ & $0$ & $-1$ & $1$ & $-1$ \\ \hline
$(\omega)_{29}$ &  $0$ & $0$ & $1$ & $0$ & $0$ & $0$ & $0$ & $0$ & $0$ & $0$ & $-1$ & $1$ \\ \hline
$(\omega)_{30}$ &  $0$ & $1$ & $0$ & $0$ & $0$ & $0$ & $0$ & $-1$ & $1$ & $0$ & $0$ & $-1$ \\ \hline
\end{tabular}
\caption{$\bbZ_2\times \bbZ_{30}$ sub-orbits of the minimal orbit of $\cW$.}
\label{tab:minorb}
\end{table}

\subsection{The binary icosahedral group $\tilde{\mathbb{I}}$}
\label{sec:I120}

The binary icosahedral group $\tilde{\mathbb{I}}$ is the preimage of the symmetry group of a
regular icosahedron in $\mathbb{E}^3$ by the degree two covering map ${\rm SU}(2)\to
SO(3)$. It has a presentation as the group generated by the unit quaternions
\beq
s=\frac{1}{2}\l(1+i+j+k\r), \quad t= \frac{1}{2}\l(\phi+\phi^{-1}+i+j\r),
\eeq
whose full set of relations is $s^3=t^5=(st)^2$. The resulting group
$\tilde{\mathbb{I}}$ has order 120, exponent 60, and class order 9. Its
character table is given in \cref{tab:ctI120}.

\begin{table}
\begin{center}
\begin{tabular}{|c|c|c|c|c|c|c|c|c|c|}
\hline
 & $1_{1}$ & $1_{-1}$ & $30_1$ & $20_{1}$ & $20_{-1}$ & $12_{\phi}$ & $12_{\phi^{-1}}$ & $12_{-\phi}$ & $12_{-\phi^{-1}}$ \\
\hline
$\chi_{\rm id}$ & $1$& $1$ & $1$ & $1$ & $1$ & $1$ & $1$ & $1$ & $1$\\
\hline
$\chi_1$ & $2$& $-2$ & $0$ & $1$ & $-1$ & $\phi$ & $\phi^{-1}$ & $-\phi$ & $-\phi^{-1}$\\
\hline
$\chi_2$ & $4$& $-4$ & $0$ & $-1$ & $1$ & $1$ & $-1$ & $-1$ & $1$\\
\hline
$\chi_3$ & $6$& $-6$ & $0$ & $0$ & $0$ & $-1$ & $1$ & $1$ & $-1$\\
\hline
$\chi_5$ & $4$& $4$ & $0$ & $1$ & $1$ & $-1$ & $-1$ & $-1$ & $-1$\\
\hline
$\chi_6$ & $3$& $3$ & $-1$ & $0$ & $0$ & $-\phi^{-1}$ & $\phi$ & $-\phi^{-1}$ & $\phi$\\
\hline
$\chi_4$ & $5$& $5$ & $1$ & $-1$ & $-1$ & $0$ & $0$ & $0$ & $0$\\
\hline
$\chi_7$ & $2$& $-2$ & $0$ & $1$ & $-1$ & $-\phi^{-1}$ & $-\phi$ & $\phi^{-1}$
& $\phi$\\
\hline
$\chi_8$ & $3$& $3$ & $-1$ & $0$ & $0$ & $\phi$ & $-\phi^{-1}$ & $\phi$ & $-\phi^{-1}$\\
\hline
\end{tabular}
\end{center}
\caption{The character table of $\tilde{\mathbb{I}}$.  Conjugacy classes $[g]$
are denoted by their order and sign of the ${\rm SU}(2)$ character $\chi_1$ as $|[g]|_{\chi_1([g])}$.
}
\label{tab:ctI120}
\end{table}

\subsection{The prepotential of $\widetilde{X}_{\mathfrak{e}_8,3}$}
\label{sec:E8prep}

\bea
F_{\widetilde{X}_{\mathfrak{e}_8,3}} &=&
-\frac{7 \mathsf{t}_1^{10}}{135000000}+\frac{521 \re^{6 \mathsf{t}_0} \mathsf{t}_1^9}{2700000}+\frac{5117 \re^{12 \mathsf{t}_0} \mathsf{t}_1^8}{100000}+\frac{\re^{2 \mathsf{t}_0} \mathsf{t}_2 \mathsf{t}
   _1^8}{300000}+\frac{7 \mathsf{t}_3 \mathsf{t}_1^8}{15000000}+\frac{243}{250} \re^{18
  \mathsf{t}_0} \mathsf{t}_1^7 \nn \\ &+& \frac{136 \re^{8 \mathsf{t}_0} \mathsf{t}_2 \mathsf{t}_1^7}{1875} \nn
+ \frac{67 \re^{6 \mathsf{t}_0} \mathsf{t}_3
   \mathsf{t}_1^7}{75000}+\frac{\re^{3 \mathsf{t}_0} \mathsf{t}_4 \mathsf{t}_1^7}{1875}-\frac{\mathsf{t}_5 \mathsf{t}_1^7}{3937500}+\frac{1954}{375} \re^{24 \mathsf{t}_0} \mathsf{t}_1^6+\frac{7}{600} \re^{4 \mathsf{t}_0} \mathsf{t}
   _2^2 \mathsf{t}_1^6 \nn \\ &-& \frac{13 \mathsf{t}_3^2 \mathsf{t}_1^6}{7500000}+\frac{2401}{750}
\re^{14 \mathsf{t}_0} \mathsf{t}_2 \mathsf{t}_1^6 + \frac{2151 \re^{12 \mathsf{t}_0} \mathsf{t}_3 \mathsf{t}_1^6}{25000}+\frac{\re^{2 \mathsf{t}
   _0} \mathsf{t}_2 \mathsf{t}_3 \mathsf{t}_1^6}{25000}+\frac{243}{500} \re^{9
  \mathsf{t}_0} \mathsf{t}_4 \mathsf{t}_1^6+\frac{19 \re^{6 \mathsf{t}_0}
  \mathsf{t}_5 \mathsf{t}_1^6}{56250}\nn \\ &-& \frac{\re^{4 \mathsf{t}_0} \mathsf{t}_6 \mathsf{t}
   _1^6}{1000}+\frac{\mathsf{t}_7 \mathsf{t}_1^6}{18750000}+12 \re^{30 \mathsf{t}_0} \mathsf{t}_1^5
+ \frac{43}{12} \re^{10 \mathsf{t}_0} \mathsf{t}_2^2 \mathsf{t}_1^5+\frac{71 \re^{6 \mathsf{t}_0} \mathsf{t}_3^2 \mathsf{t}
   _1^5}{50000} \nn \\ &+& \frac{43}{2} \re^{20 \mathsf{t}_0} \mathsf{t}_2 \mathsf{t}_1^5+\frac{189}{250} \re^{18 \mathsf{t}_0} \mathsf{t}_3 \mathsf{t}_1^5+\frac{4}{25} \re^{8 \mathsf{t}_0} \mathsf{t}_2 \mathsf{t}_3 \mathsf{t}_1^5+7 \re^{15 \mathsf{t}
   _0} \mathsf{t}_4 \mathsf{t}_1^5 \nn \\ &+& \frac{1}{2} \re^{5 \mathsf{t}_0} \mathsf{t}_2 \mathsf{t}_4 \mathsf{t}_1^5+\frac{1}{250} \re^{3 \mathsf{t}_0} \mathsf{t}_3 \mathsf{t}_4 \mathsf{t}_1^5+\frac{67 \re^{12 \mathsf{t}_0} \mathsf{t}_5 \mathsf{t}
   _1^5}{6250}+\frac{\re^{2 \mathsf{t}_0} \mathsf{t}_2 \mathsf{t}_5 \mathsf{t}_1^5}{18750}+\frac{11
  \mathsf{t}_3 \mathsf{t}_5 \mathsf{t}_1^5}{4687500}\nn \\ &-& \frac{1}{12} \re^{10 \mathsf{t}_0} \mathsf{t}_6 \mathsf{t}_1^5+\frac{\re^{6 \mathsf{t}
   _0} \mathsf{t}_7 \mathsf{t}_1^5}{25000}+\frac{159}{10} \re^{36 \mathsf{t}_0} \mathsf{t}_1^4+\frac{117}{100} \re^{6 \mathsf{t}_0} \mathsf{t}_2^3 \mathsf{t}_1^4+\frac{7 \mathsf{t}_3^3 \mathsf{t}_1^4}{2500000}+\frac{592}{15}
   \re^{16 \mathsf{t}_0} \mathsf{t}_2^2 \mathsf{t}_1^4 \nn \\ &+& \frac{2557 \re^{12 \mathsf{t}_0} \mathsf{t}_3^2 \mathsf{t}_1^4}{50000}+\frac{7 \re^{2 \mathsf{t}_0} \mathsf{t}_2 \mathsf{t}_3^2 \mathsf{t}_1^4}{50000}+\frac{88}{25} \re^{6 \mathsf{t}
   _0} \mathsf{t}_4^2 \mathsf{t}_1^4-\frac{\mathsf{t}_5^2 \mathsf{t}_1^4}{703125}+\frac{507}{10}
   \re^{26 \mathsf{t}_0} \mathsf{t}_2 \mathsf{t}_1^4+\frac{229}{125} \re^{24 \mathsf{t}_0} \mathsf{t}_3
   \mathsf{t}_1^4 \nn \\ &+& \frac{23}{600}
   \re^{4 \mathsf{t}_0} \mathsf{t}_2^2 \mathsf{t}_3 \mathsf{t}_1^4+\frac{343}{125} \re^{14 \mathsf{t}_0} \mathsf{t}_2 \mathsf{t}_3 \mathsf{t}_1^4+\frac{98}{5} \re^{21 \mathsf{t}_0} \mathsf{t}_4 \mathsf{t}_1^4+\frac{1}{300} \re^{\mathsf{t}_0}
   \mathsf{t}_2^2 \mathsf{t}_4 \mathsf{t}_1^4+\frac{1331}{50} \re^{11 \mathsf{t}_0} \mathsf{t}_2 \mathsf{t}_4
   \mathsf{t}_1^4 \nn \\ &+& \frac{459}{500} \re^{9 \mathsf{t}_0} \mathsf{t}_3 \mathsf{t}_4 \mathsf{t}_1^4+\frac{9}{250} \re^{18 \mathsf{t}_0}
   \mathsf{t}_5 \mathsf{t}_1^4+\frac{76 \re^{8 \mathsf{t}_0} \mathsf{t}_2 \mathsf{t}_5 \mathsf{t}_1^4}{1875}+\frac{11 \re^{6 \mathsf{t}_0} \mathsf{t}_3 \mathsf{t}_5 \mathsf{t}_1^4}{9375}+\frac{17 \re^{3 \mathsf{t}_0} \mathsf{t}_4 \mathsf{t}_5 \mathsf{t}
   _1^4}{3750}\nn \\ &-&\frac{4}{15} \re^{16 \mathsf{t}_0} \mathsf{t}_6 \mathsf{t}_1^4 - \frac{27}{100} \re^{6 \mathsf{t}_0} \mathsf{t}_2 \mathsf{t}_6 \mathsf{t}_1^4-\frac{19 \re^{4 \mathsf{t}_0} \mathsf{t}_3 \mathsf{t}_6 \mathsf{t}
   _1^4}{3000}-\frac{1}{300} \re^{\mathsf{t}_0} \mathsf{t}_4 \mathsf{t}_6 \mathsf{t}_1^4+\frac{\re^{12 \mathsf{t}_0} \mathsf{t}_7 \mathsf{t}_1^4}{25000}+\frac{\re^{2 \mathsf{t}_0} \mathsf{t}_2 \mathsf{t}_7 \mathsf{t}
   _1^4}{25000}\nn \\ &-& \frac{\mathsf{t}_3 \mathsf{t}_7 \mathsf{t}_1^4}{1250000} + 10 \re^{42 \mathsf{t}_0} \mathsf{t}_1^3+\frac{1}{45} \re^{2 \mathsf{t}_0} \mathsf{t}_2^4 \mathsf{t}_1^3+\frac{159}{5} \re^{12 \mathsf{t}_0} \mathsf{t}_2^3 \mathsf{t}
   _1^3+\frac{19 \re^{6 \mathsf{t}_0} \mathsf{t}_3^3 \mathsf{t}_1^3}{25000}+\frac{484}{5} \re^{22
     \mathsf{t}_0} \mathsf{t}_2^2 \mathsf{t}_1^3  \nn \\ &+& \frac{9}{50} \re^{18 \mathsf{t}_0} \mathsf{t}_3^2 \mathsf{t}
   _1^3+ \frac{54}{625} \re^{8
   \mathsf{t}_0} \mathsf{t}_2 \mathsf{t}_3^2 \mathsf{t}_1^3+\frac{248}{5} \re^{12 \mathsf{t}_0} \mathsf{t}_4^2 \mathsf{t}_1^3+\frac{4}{5} \re^{2 \mathsf{t}_0} \mathsf{t}_2 \mathsf{t}_4^2 \mathsf{t}_1^3+\frac{2 \re^{6 \mathsf{t}_0} \mathsf{t}_5^2
   \mathsf{t}_1^3}{5625}+\frac{1}{45} \re^{2 \mathsf{t}_0} \mathsf{t}_6^2
   \mathsf{t}_1^3  \nn \\ &+& 48 \re^{32
     \mathsf{t}_0} \mathsf{t}_2 \mathsf{t}_1^3 + 2 \re^{30 \mathsf{t}_0} \mathsf{t}_3 \mathsf{t}_1^3+\frac{19}{6} \re^{10 \mathsf{t}_0} \mathsf{t}_2^2
   \mathsf{t}_3 \mathsf{t}_1^3+9 \re^{20 \mathsf{t}_0} \mathsf{t}_2 \mathsf{t}_3
   \mathsf{t}_1^3+18 \re^{27 \mathsf{t}_0} \mathsf{t}_4
   \mathsf{t}_1^3+\frac{343}{15} \re^{7 \mathsf{t}_0} \mathsf{t}_2^2
   \mathsf{t}_4 \mathsf{t}_1^3  \nn \\ &+& \frac{4}{625} \re^{3
   \mathsf{t}_0} \mathsf{t}_3^2 \mathsf{t}_4 \mathsf{t}_1^3 + \frac{578}{5} \re^{17 \mathsf{t}_0} \mathsf{t}_2 \mathsf{t}_4 \mathsf{t}_1^3+6 \re^{15 \mathsf{t}_0} \mathsf{t}_3 \mathsf{t}_4 \mathsf{t}_1^3+\re^{5 \mathsf{t}_0} \mathsf{t}_2 \mathsf{t}_3 \mathsf{t}_4
   \mathsf{t}_1^3+\frac{4}{375} \re^{24 \mathsf{t}_0} \mathsf{t}_5
   \mathsf{t}_1^3 \nn \\ &+& \frac{7}{375} \re^{4
     \mathsf{t}_0} \mathsf{t}_2^2 \mathsf{t}_5 \mathsf{t}_1^3 - \frac{\mathsf{t}_3^2 \mathsf{t}_5 \mathsf{t}
     _1^3}{187500} + \frac{98}{375}
   \re^{14 \mathsf{t}_0} \mathsf{t}_2 \mathsf{t}_5 \mathsf{t}_1^3+\frac{51 \re^{12
       \mathsf{t}_0} \mathsf{t}_3 \mathsf{t}_5
     \mathsf{t}_1^3}{3125}+\frac{\re^{2 \mathsf{t}_0} \mathsf{t}_2 \mathsf{t}_3
     \mathsf{t}_5 \mathsf{t}_1^3}{3125} \nn \\ &+& \frac{36}{125}
   \re^{9 \mathsf{t}_0} \mathsf{t}_4 \mathsf{t}_5 \mathsf{t}_1^3-\frac{2}{45} \re^{2 \mathsf{t}_0} \mathsf{t}
   _2^2 \mathsf{t}_6 \mathsf{t}_1^3 - 3 \re^{12 \mathsf{t}_0} \mathsf{t}_2 \mathsf{t}_6 \mathsf{t}_1^3-\frac{1}{6} \re^{10 \mathsf{t}_0} \mathsf{t}_3
   \mathsf{t}_6 \mathsf{t}_1^3-\frac{49}{15} \re^{7 \mathsf{t}_0} \mathsf{t}_4
   \mathsf{t}_6 \mathsf{t}_1^3 \nn \\ &-& \frac{1}{125} \re^{4 \mathsf{t}_0} \mathsf{t}_5 \mathsf{t}_6 \mathsf{t}_1^3+\frac{2}{625} \re^{8 \mathsf{t}_0} \mathsf{t}_2 \mathsf{t}_7
   \mathsf{t}_1^3+\frac{3 \re^{6 \mathsf{t}_0} \mathsf{t}_3 \mathsf{t}_7 \mathsf{t}_1^3}{12500}+\frac{53 \re^{12 \mathsf{t}_0} \mathsf{t}_3^4}{100000}+90 \re^{9
     \mathsf{t}_0} \mathsf{t}_2^3 \mathsf{t}_4 \mathsf{t}_1 \nn
   \\ &+& \frac{2}{625} \re^{3 \mathsf{t}_0} \mathsf{t}_4 \mathsf{t}_7 \mathsf{t}_1^3+\frac{\mathsf{t}_5 \mathsf{t}_7 \mathsf{t}_1^3}{468750}+\frac{15}{2}
   \re^{48 \mathsf{t}_0} \mathsf{t}_1^2+\frac{34}{3} \re^{8 \mathsf{t}_0} \mathsf{t}_2^4 \mathsf{t}
   _1^2-\frac{9 \mathsf{t}_3^4 \mathsf{t}_1^2}{5000000}+99 \re^{18 \mathsf{t}_0} \mathsf{t}_2^3 \mathsf{t}
   _1^2 \nn \\ &+& \frac{153 \re^{12 \mathsf{t}_0}
   \mathsf{t}_3^3 \mathsf{t}_1^2}{25000}+\frac{3 \re^{2 \mathsf{t}_0} \mathsf{t}_2 \mathsf{t}_3^3 \mathsf{t}_1^2}{25000}+6 \re^{3 \mathsf{t}_0} \mathsf{t}_4^3 \mathsf{t}_1^2+105 \re^{28 \mathsf{t}_0} \mathsf{t}_2^2 \mathsf{t}
   _1^2+\frac{81}{250} \re^{24 \mathsf{t}_0} \mathsf{t}_3^2 \mathsf{t}_1^2+\frac{29 \re^{4 \mathsf{t}
       _0} \mathsf{t}_2^2 \mathsf{t}_3^2 \mathsf{t}_1^2}{1000} \nn \\ &+& \frac{147}{250} \re^{14 \mathsf{t}_0} \mathsf{t}_2 \mathsf{t}_3^2 \mathsf{t}
   _1^2+108 \re^{18 \mathsf{t}_0} \mathsf{t}_4^2 \mathsf{t}_1^2+96 \re^{8 \mathsf{t}_0} \mathsf{t}_2 \mathsf{t}_4^2 \mathsf{t}_1^2+\frac{78}{25} \re^{6 \mathsf{t}_0} \mathsf{t}_3 \mathsf{t}_4^2 \mathsf{t}_1^2+\frac{9 \re^{12 \mathsf{t}_0} \mathsf{t}
   _5^2 \mathsf{t}_1^2}{3125} \nn \\ &+&\frac{2 \re^{2 \mathsf{t}_0} \mathsf{t}_2 \mathsf{t}_5^2 \mathsf{t}
     _1^2}{9375}+ \frac{\mathsf{t}_3 \mathsf{t}_5^2 \mathsf{t}_1^2}{234375}+\frac{5}{6} \re^{8 \mathsf{t}_0} \mathsf{t}_6^2 \mathsf{t}
   _1^2-\frac{3 \mathsf{t}_7^2 \mathsf{t}_1^2}{1250000}+30 \re^{38 \mathsf{t}_0} \mathsf{t}_2 \mathsf{t}_1^2+\frac{3}{5} \re^{36 \mathsf{t}_0} \mathsf{t}_3 \mathsf{t}_1^2 \nn \\ &+& \frac{56}{5} \re^{16 \mathsf{t}_0} \mathsf{t}_2^2
   \mathsf{t}_3 \mathsf{t}_1^2 + \frac{39}{5} \re^{26 \mathsf{t}_0} \mathsf{t}_2 \mathsf{t}_3 \mathsf{t}_1^2+3 \re^{3 \mathsf{t}_0} \mathsf{t}_2^3 \mathsf{t}_4 \mathsf{t}_1^2+169 \re^{13
   \mathsf{t}_0} \mathsf{t}_2^2 \mathsf{t}_4 \mathsf{t}_1^2+\frac{27}{100} \re^{9 \mathsf{t}_0} \mathsf{t}_3^2
   \mathsf{t}_4 \mathsf{t}_1^2 \nn \\ &+& 138 \re^{23 \mathsf{t}_0} \mathsf{t}_2 \mathsf{t}_4 \mathsf{t}_1^2 + \frac{42}{5} \re^{21 \mathsf{t}_0} \mathsf{t}
   _3 \mathsf{t}_4 \mathsf{t}_1^2+\frac{1}{50} \re^{\mathsf{t}_0} \mathsf{t}_2^2 \mathsf{t}_3 \mathsf{t}_4 \mathsf{t}_1^2+\frac{363}{25} \re^{11 \mathsf{t}_0} \mathsf{t}_2 \mathsf{t}_3 \mathsf{t}_4 \mathsf{t}_1^2\nn \\ &+&\frac{7}{15} \re^{10 \mathsf{t}_0}
   \mathsf{t}_2^2 \mathsf{t}_5 \mathsf{t}_1^2+\frac{\re^{6 \mathsf{t}_0} \mathsf{t}_3^2 \mathsf{t}_5 \mathsf{t}
     _1^2}{1250}+ \frac{2}{5} \re^{20 \mathsf{t}_0} \mathsf{t}_2 \mathsf{t}_5 \mathsf{t}_1^2+\frac{3}{125} \re^{18 \mathsf{t}_0} \mathsf{t}_3
   \mathsf{t}_5 \mathsf{t}_1^2+\frac{4}{125} \re^{8 \mathsf{t}_0} \mathsf{t}_2
   \mathsf{t}_3 \mathsf{t}_5 \mathsf{t}_1^2 \nn \\ &+& \frac{4}{5} \re^{15 \mathsf{t}_0}
   \mathsf{t}_4 \mathsf{t}_5 \mathsf{t}_1^2 + \frac{2}{5} \re^{5 \mathsf{t}_0} \mathsf{t}_2 \mathsf{t}
   _4 \mathsf{t}_5 \mathsf{t}_1^2+ \frac{1}{125} \re^{3 \mathsf{t}_0} \mathsf{t}_3 \mathsf{t}_4 \mathsf{t}_5 \mathsf{t}_1^2-\frac{14}{3} \re^{8 \mathsf{t}_0} \mathsf{t}_2^2 \mathsf{t}_6 \mathsf{t}_1^2-\frac{1}{200} \re^{4 \mathsf{t}_0} \mathsf{t}
   _3^2 \mathsf{t}_6 \mathsf{t}_1^2\nn \\ &-& 3 \re^{18 \mathsf{t}_0} \mathsf{t}_2 \mathsf{t}_6 \mathsf{t}
   _1^2-\frac{2}{5} \re^{16 \mathsf{t}_0} \mathsf{t}_3 \mathsf{t}_6 \mathsf{t}_1^2- \frac{21}{50} \re^{6 \mathsf{t}_0} \mathsf{t}_2 \mathsf{t}_3 \mathsf{t}_6
   \mathsf{t}_1^2-13 \re^{13 \mathsf{t}_0} \mathsf{t}_4 \mathsf{t}_6 \mathsf{t}_1^2-3 \re^{3 \mathsf{t}_0} \mathsf{t}_2 \mathsf{t}_4 \mathsf{t}_6 \mathsf{t}_1^2\nn \\ &-&\frac{1}{50} \re^{\mathsf{t}_0} \mathsf{t}_3 \mathsf{t}_4 \mathsf{t}_6 \mathsf{t}
   _1^2-\frac{1}{15} \re^{10 \mathsf{t}_0} \mathsf{t}_5 \mathsf{t}_6
   \mathsf{t}_1^2 + \frac{3}{500} \re^{4 \mathsf{t}_0} \mathsf{t}_2^2
   \mathsf{t}_7 \mathsf{t}_1^2+\frac{3 \mathsf{t}_3^2 \mathsf{t}_7
     \mathsf{t}_1^2}{1250000}-\frac{1}{50}
   \re^{6 \mathsf{t}_0} \mathsf{t}_2 \mathsf{t}_6 \mathsf{t}_7\nn \\ &+& \frac{3
   \re^{12 \mathsf{t}_0} \mathsf{t}_3 \mathsf{t}_7 \mathsf{t}_1^2}{12500}+\frac{3 \re^{2 \mathsf{t}_0} \mathsf{t}_2 \mathsf{t}_3 \mathsf{t}_7 \mathsf{t}_1^2}{12500}+\frac{\re^{12 \mathsf{t}_0} \mathsf{t}_3^2 \mathsf{t}_5 \mathsf{t}_1}{6250}+\frac{51}{50} \re^{6 \mathsf{t}_0} \mathsf{t}_2^3 \mathsf{t}_3
   \mathsf{t}_1^2+\frac{1}{50} \re^{\mathsf{t}_0}
   \mathsf{t}_2^2 \mathsf{t}_4 \mathsf{t}_7\nn \\ &+&\frac{9}{250} \re^{9 \mathsf{t}_0} \mathsf{t}_4 \mathsf{t}_7 \mathsf{t}
   _1^2+\frac{\re^{6 \mathsf{t}_0} \mathsf{t}_5 \mathsf{t}_7 \mathsf{t}_1^2}{3125} - \frac{3}{500} \re^{4 \mathsf{t}_0} \mathsf{t}_6 \mathsf{t}_7 \mathsf{t}_1^2+\frac{5}{6} \re^{4 \mathsf{t}_0} \mathsf{t}_2^5 \mathsf{t}_1+\frac{245}{6}
   \re^{14 \mathsf{t}_0} \mathsf{t}_2^4 \mathsf{t}_1+\frac{3 \re^{6
       \mathsf{t}_0} \mathsf{t}_3^4 \mathsf{t}_1}{100000} \nn \\ &+& 60 \re^{24 \mathsf{t}_0} \mathsf{t}_2^3 \mathsf{t}_1+\frac{3}{250} \re^{18 \mathsf{t}_0} \mathsf{t}_3^3 \mathsf{t}
   _1 + \frac{9}{625} \re^{8 \mathsf{t}_0} \mathsf{t}_2 \mathsf{t}_3^3 \mathsf{t}_1+60 \re^{9 \mathsf{t}_0} \mathsf{t}_4^3 \mathsf{t}_1-\frac{2 \mathsf{t}_5^3 \mathsf{t}_1}{2109375}+30 \re^{34 \mathsf{t}_0} \mathsf{t}_2^2 \mathsf{t}
   _1 \nn \\ &+& \frac{7}{20} \re^{10 \mathsf{t}_0} \mathsf{t}_2^2 \mathsf{t}_3^2 \mathsf{t}_1+\frac{3}{10}
   \re^{20 \mathsf{t}_0} \mathsf{t}_2 \mathsf{t}_3^2 \mathsf{t}_1 + 60 \re^{24 \mathsf{t}_0} \mathsf{t}_4^2 \mathsf{t}_1+40 \re^{4 \mathsf{t}_0} \mathsf{t}
   _2^2 \mathsf{t}_4^2 \mathsf{t}_1+210 \re^{14 \mathsf{t}_0} \mathsf{t}_2
   \mathsf{t}_4^2 \mathsf{t}_1 \nn \\ &+& \frac{42}{5} \re^{12 \mathsf{t}_0} \mathsf{t}_3 \mathsf{t}_4^2 \mathsf{t}_1+\frac{6}{5} \re^{2 \mathsf{t}_0} \mathsf{t}_2 \mathsf{t}_3 \mathsf{t}_4^2
   \mathsf{t}_1 + \frac{4 \re^{8 \mathsf{t}_0} \mathsf{t}_2 \mathsf{t}_5^2
     \mathsf{t}_1}{1875}+\frac{2 \re^{6 \mathsf{t}_0} \mathsf{t}_3
     \mathsf{t}_5^2 \mathsf{t}_1}{9375} +6 \re^{27 \mathsf{t}_0} \mathsf{t}_3 \mathsf{t}_4 \mathsf{t}_1\nn \\ &+& \frac{4 \re^{3 \mathsf{t}_0} \mathsf{t}_4 \mathsf{t}_5^2 \mathsf{t}
   _1}{1875}+\frac{5}{6} \re^{14 \mathsf{t}_0} \mathsf{t}_6^2 \mathsf{t}_1+\frac{5}{6} \re^{4
     \mathsf{t}_0} \mathsf{t}_2 \mathsf{t}_6^2 \mathsf{t}_1+\frac{1}{30} \re^{2 \mathsf{t}_0} \mathsf{t}_3
   \mathsf{t}_6^2 \mathsf{t}_1 + +80 \re^{15 \mathsf{t}_0} \mathsf{t}_2^3
   \mathsf{t}_4 \nn \\ &+& \frac{3 \re^{6
   \mathsf{t}_0} \mathsf{t}_7^2 \mathsf{t}_1}{25000}+\frac{1}{30} \re^{2 \mathsf{t}_0} \mathsf{t}_2^4 \mathsf{t}_3 \mathsf{t}_1+\frac{41}{5} \re^{12 \mathsf{t}_0} \mathsf{t}_2^3 \mathsf{t}_3 \mathsf{t}_1+\frac{66}{5} \re^{22 \mathsf{t}_0}
   \mathsf{t}_2^2 \mathsf{t}_3 \mathsf{t}_1+6 \re^{32 \mathsf{t}_0} \mathsf{t}_2 \mathsf{t}_3 \mathsf{t}_1 \nn \\ &+& \frac{3 \re^{3 \mathsf{t}_0} \mathsf{t}_3^3 \mathsf{t}_4 \mathsf{t}_1}{1250}+190
   \re^{19 \mathsf{t}_0} \mathsf{t}_2^2 \mathsf{t}_4 \mathsf{t}_1+\frac{3}{5} \re^{15 \mathsf{t}_0} \mathsf{t}_3^2 \mathsf{t}_4 \mathsf{t}_1+\frac{3}{10} \re^{5 \mathsf{t}_0} \mathsf{t}_2 \mathsf{t}_3^2 \mathsf{t}_4 \mathsf{t}_1+60 \re^{29 \mathsf{t}_0}
   \mathsf{t}_2 \mathsf{t}_4 \mathsf{t}_1 \nn \\ &+&
   \frac{49}{5} \re^{7 \mathsf{t}_0} \mathsf{t}_2^2 \mathsf{t}_3 \mathsf{t}_4 \mathsf{t}_1+\frac{102}{5} \re^{17 \mathsf{t}_0} \mathsf{t}_2 \mathsf{t}_3
   \mathsf{t}_4 \mathsf{t}_1+\frac{4}{25} \re^{6 \mathsf{t}_0} \mathsf{t}_2^3 \mathsf{t}_5 \mathsf{t}_1+\frac{\mathsf{t}_3^3 \mathsf{t}_5 \mathsf{t}_1}{312500}+\frac{8}{15} \re^{16 \mathsf{t}_0} \mathsf{t}_2^2 \mathsf{t}_5 \mathsf{t}
   _1 \nn \\ &+& \frac{\re^{2 \mathsf{t}_0} \mathsf{t}_2 \mathsf{t}_3^2 \mathsf{t}_5 \mathsf{t}_1}{6250}+\frac{14}{25} \re^{6 \mathsf{t}_0} \mathsf{t}_4^2 \mathsf{t}_5 \mathsf{t}
   _1+\frac{2}{125} \re^{24 \mathsf{t}_0} \mathsf{t}_3 \mathsf{t}_5 \mathsf{t}_1+\frac{7}{375} \re^{4 \mathsf{t}_0} \mathsf{t}_2^2 \mathsf{t}_3 \mathsf{t}_5 \mathsf{t}_1+\frac{14}{125} \re^{14 \mathsf{t}_0} \mathsf{t}_2 \mathsf{t}_3 \mathsf{t}_5
   \mathsf{t}_1 \nn \\ &+& \frac{4}{5} \re^{21 \mathsf{t}_0} \mathsf{t}_4 \mathsf{t}_5 \mathsf{t}_1+\frac{2}{75} \re^{\mathsf{t}_0} \mathsf{t}_2^2 \mathsf{t}_4 \mathsf{t}_5 \mathsf{t}_1+\frac{44}{25} \re^{11 \mathsf{t}_0} \mathsf{t}_2 \mathsf{t}_4 \mathsf{t}_5
   \mathsf{t}_1+\frac{12}{125} \re^{9 \mathsf{t}_0} \mathsf{t}_3 \mathsf{t}_4 \mathsf{t}_5 \mathsf{t}
   _1-\frac{5}{3} \re^{4 \mathsf{t}_0} \mathsf{t}_2^3 \mathsf{t}_6 \mathsf{t}_1 \nn \\ &-& \frac{35}{3} \re^{14 \mathsf{t}_0} \mathsf{t}_2^2 \mathsf{t}_6 \mathsf{t}
   _1-\frac{1}{20} \re^{10 \mathsf{t}_0} \mathsf{t}_3^2 \mathsf{t}_6 \mathsf{t}_1-10 \re^{4 \mathsf{t}_0} \mathsf{t}_4^2 \mathsf{t}_6 \mathsf{t}_1-\frac{1}{15} \re^{2 \mathsf{t}_0} \mathsf{t}_2^2 \mathsf{t}_3 \mathsf{t}_6 \mathsf{t}_1-\re^{12 \mathsf{t}
   _0} \mathsf{t}_2 \mathsf{t}_3 \mathsf{t}_6 \mathsf{t}_1 \nn \\ &-& 10 \re^{19 \mathsf{t}_0} \mathsf{t}_4 \mathsf{t}_6 \mathsf{t}_1-30 \re^{9 \mathsf{t}_0} \mathsf{t}_2 \mathsf{t}_4 \mathsf{t}_6 \mathsf{t}_1-\frac{7}{5} \re^{7 \mathsf{t}_0} \mathsf{t}_3 \mathsf{t}_4 \mathsf{t}
   _6 \mathsf{t}_1-\frac{2}{15} \re^{16 \mathsf{t}_0} \mathsf{t}_5 \mathsf{t}_6 \mathsf{t}_1-\frac{4}{25}
   \re^{6 \mathsf{t}_0} \mathsf{t}_2 \mathsf{t}_5 \mathsf{t}_6 \mathsf{t}_1 \nn \\ &-& \frac{1}{375} \re^{4 \mathsf{t}_0} \mathsf{t}_3 \mathsf{t}_5 \mathsf{t}_6
   \mathsf{t}_1-\frac{2}{75} \re^{\mathsf{t}_0} \mathsf{t}_4 \mathsf{t}_5 \mathsf{t}_6 \mathsf{t}_1+\frac{1}{50} \re^{10 \mathsf{t}_0} \mathsf{t}_2^2 \mathsf{t}_7 \mathsf{t}_1+\frac{3 \re^{6 \mathsf{t}_0} \mathsf{t}_3^2 \mathsf{t}_7 \mathsf{t}
   _1}{25000}+\frac{3}{625} \re^{8 \mathsf{t}_0} \mathsf{t}_2 \mathsf{t}_3 \mathsf{t}_7 \mathsf{t}_1 \nn
   \\ &+& \frac{3}{25} \re^{5 \mathsf{t}_0} \mathsf{t}_2 \mathsf{t}_4 \mathsf{t}_7 \mathsf{t}_1+\frac{3}{625} \re^{3 \mathsf{t}_0} \mathsf{t}_3 \mathsf{t}_4
   \mathsf{t}_7 \mathsf{t}_1+\frac{\re^{12 \mathsf{t}_0} \mathsf{t}_5 \mathsf{t}_7 \mathsf{t}_1}{3125}+\frac{\re^{2 \mathsf{t}_0} \mathsf{t}_2 \mathsf{t}_5 \mathsf{t}_7 \mathsf{t}_1}{3125}-\frac{\mathsf{t}_3 \mathsf{t}_5 \mathsf{t}_7 \mathsf{t}
   _1}{156250} \nn \\ &-& \frac{1}{50} \re^{10 \mathsf{t}_0} \mathsf{t}_6 \mathsf{t}_7 \mathsf{t}_1 + \frac{3}{125} \mathsf{t}_7 \mathsf{t}_8 \mathsf{t}_1+\re^{60 \mathsf{t}_0}-\frac{\mathsf{t}_2^6}{324}+\frac{23}{3} \re^{10 \mathsf{t}_0} \mathsf{t}
   _2^5+\frac{\mathsf{t}_3^5}{5000000}+\frac{185}{6} \re^{20 \mathsf{t}_0} \mathsf{t}
   _2^4  \nn \\ &+& \frac{3 \re^{2 \mathsf{t}_0} \mathsf{t}
     _2 \mathsf{t}_3^4}{100000}- \frac{5
   \mathsf{t}_4^4}{8}+20 \re^{30 \mathsf{t}_0} \mathsf{t}_2^3+\frac{\re^{4 \mathsf{t}_0} \mathsf{t}_2^2 \mathsf{t}_3^3}{1000}+20 \re^{15 \mathsf{t}_0} \mathsf{t}_4^3+60 \re^{5 \mathsf{t}_0} \mathsf{t}_2 \mathsf{t}_4^3+2 \re^{3 \mathsf{t}_0} \mathsf{t}
   _3 \mathsf{t}_4^3\nn \\ &-&\frac{5 \mathsf{t}
     _6^3}{324}  15 \re^{40 \mathsf{t}_0} \mathsf{t}_2^2+\frac{1}{10} \re^{36 \mathsf{t}_0} \mathsf{t}_3^2+\frac{11}{100} \re^{6 \mathsf{t}_0}
   \mathsf{t}_2^3 \mathsf{t}_3^2+\frac{3}{5} \re^{16 \mathsf{t}_0}
   \mathsf{t}_2^2 \mathsf{t}_3^2+\frac{3}{10} \re^{26 \mathsf{t}_0}
   \mathsf{t}_2 \mathsf{t}_3^2+30 \re^{30 \mathsf{t}_0} \mathsf{t}_4^2 \nn
   \\ &+& 130 \re^{10 \mathsf{t}_0} \mathsf{t}_2^2
   \mathsf{t}_4^2  \frac{9}{25} \re^{6 \mathsf{t}_0} \mathsf{t}_3^2 \mathsf{t}_4^2+60 \re^{20 \mathsf{t}_0} \mathsf{t}_2 \mathsf{t}_4^2+6 \re^{18 \mathsf{t}_0} \mathsf{t}_3 \mathsf{t}_4^2+12 \re^{8 \mathsf{t}_0} \mathsf{t}_2 \mathsf{t}_3 \mathsf{t}
   _4^2+\frac{1}{375} \re^{24 \mathsf{t}_0} \mathsf{t}_5^2 \nn \\ &-& \frac{\mathsf{t}_3^2 \mathsf{t}_5^2}{468750} + \frac{2}{375} \re^{14 \mathsf{t}_0} \mathsf{t}_2 \mathsf{t}
   _5^2+\frac{5}{6} \re^{20 \mathsf{t}_0} \mathsf{t}_6^2-\frac{5}{108} \mathsf{t}_2^2 \mathsf{t}_6^2+\frac{5}{3} \re^{10 \mathsf{t}_0} \mathsf{t}_2 \mathsf{t}_6^2+\frac{5}{3} \re^{5 \mathsf{t}_0} \mathsf{t}_4 \mathsf{t}
   _6^2+\frac{1}{90} \re^{2 \mathsf{t}_0} \mathsf{t}_5 \mathsf{t}_6^2 \nn \\ &+& \frac{3 \re^{12 \mathsf{t}_0} \mathsf{t}_7^2}{25000}+\frac{3 \re^{2 \mathsf{t}_0} \mathsf{t}_2 \mathsf{t}_7^2}{25000}+\frac{3 \mathsf{t}_3 \mathsf{t}
   _7^2}{1250000}+30 \mathsf{t}_0 \mathsf{t}_8^2+\re^{8 \mathsf{t}_0} \mathsf{t}_2^4 \mathsf{t}_3+3 \re^{18 \mathsf{t}_0} \mathsf{t}_2^3 \mathsf{t}_3+\frac{35}{3} \re^{5 \mathsf{t}_0} \mathsf{t}_2^4
   \mathsf{t}_4 \nn \\ &+& \frac{3}{500} \re^{9 \mathsf{t}_0} \mathsf{t}_3^3 \mathsf{t}_4+60 \re^{25 \mathsf{t}_0} \mathsf{t}_2^2 \mathsf{t}_4+\frac{1}{100} \re^{\mathsf{t}_0} \mathsf{t}_2^2 \mathsf{t}_3^2
   \mathsf{t}_4+\frac{33}{50} \re^{11 \mathsf{t}_0} \mathsf{t}_2 \mathsf{t}_3^2 \mathsf{t}_4+\re^{3 \mathsf{t}
     _0} \mathsf{t}_2^3 \mathsf{t}_3 \mathsf{t}_4
   \nn \\ &+& 6 \re^{23 \mathsf{t}_0} \mathsf{t}_2
   \mathsf{t}_3 \mathsf{t}_4+\frac{1}{90} \re^{2 \mathsf{t}_0} \mathsf{t}_2^4 \mathsf{t}_5+\frac{2}{5} \re^{12 \mathsf{t}_0} \mathsf{t}_2^3 \mathsf{t}_5+\frac{2}{5} \re^{22 \mathsf{t}_0} \mathsf{t}_2^2 \mathsf{t}_5+\frac{1}{250} \re^{18
   \mathsf{t}_0} \mathsf{t}_3^2 \mathsf{t}_5 \nn \\ &+& \frac{4}{5} \re^{12 \mathsf{t}_0} \mathsf{t}_4^2 \mathsf{t}_5+\frac{2}{5} \re^{2 \mathsf{t}_0} \mathsf{t}_2 \mathsf{t}_4^2
   \mathsf{t}_5+\frac{1}{75} \re^{10 \mathsf{t}_0} \mathsf{t}_2^2 \mathsf{t}_3 \mathsf{t}_5+\frac{14}{15} \re^{7 \mathsf{t}_0} \mathsf{t}_2^2 \mathsf{t}_4 \mathsf{t}_5+\frac{\re^{3 \mathsf{t}_0} \mathsf{t}_3^2 \mathsf{t}_4 \mathsf{t}
   _5}{1250} \nn \\ &+& \frac{2}{25} \re^{5 \mathsf{t}_0} \mathsf{t}_2 \mathsf{t}_3 \mathsf{t}_4 \mathsf{t}_5-\frac{5}{324} \mathsf{t}_2^4 \mathsf{t}_6-\frac{10}{3} \re^{10
   \mathsf{t}_0} \mathsf{t}_2^3 \mathsf{t}_6-\frac{\re^{4 \mathsf{t}_0} \mathsf{t}_3^3 \mathsf{t}
     _6}{1000}-\frac{5}{3} \re^{20 \mathsf{t}_0} \mathsf{t}_2^2 \mathsf{t}_6-\frac{1}{100} \re^{6
     \mathsf{t}_0} \mathsf{t}_2 \mathsf{t}_3^2 \mathsf{t}_6 \nn \\ &-& 10
   \re^{10 \mathsf{t}_0} \mathsf{t}_4^2 \mathsf{t}_6-\re^{8 \mathsf{t}_0} \mathsf{t}_2^2 \mathsf{t}_3 \mathsf{t}_6-\re^{18 \mathsf{t}_0} \mathsf{t}_2 \mathsf{t}_3 \mathsf{t}_6-\frac{40}{3} \re^{5 \mathsf{t}_0} \mathsf{t}_2^2 \mathsf{t}_4 \mathsf{t}
   _6-\frac{1}{100} \re^{\mathsf{t}_0} \mathsf{t}_3^2 \mathsf{t}_4 \mathsf{t}_6-20 \re^{15 \mathsf{t}_0}
   \mathsf{t}_2 \mathsf{t}_4 \mathsf{t}_6 \nn \\ &-& \re^{13 \mathsf{t}_0} \mathsf{t}_3 \mathsf{t}_4 \mathsf{t}_6-\re^{3 \mathsf{t}_0} \mathsf{t}_2 \mathsf{t}_3 \mathsf{t}_4
   \mathsf{t}_6-\frac{1}{45} \re^{2 \mathsf{t}_0} \mathsf{t}_2^2 \mathsf{t}_5 \mathsf{t}_6-\frac{1}{75} \re^{10 \mathsf{t}_0} \mathsf{t}_3 \mathsf{t}_5 \mathsf{t}_6-\frac{2}{15} \re^{7 \mathsf{t}_0} \mathsf{t}_4 \mathsf{t}_5 \mathsf{t}
   _6 \nn \\ &-& \frac{\mathsf{t}_3^3 \mathsf{t}_7}{1250000}+\frac{3 \re^{12 \mathsf{t}_0} \mathsf{t}_3^2 \mathsf{t}_7}{25000}+\frac{3 \re^{2 \mathsf{t}_0} \mathsf{t}_2 \mathsf{t}_3^2
   \mathsf{t}_7}{25000}+\frac{3}{25} \re^{6 \mathsf{t}_0} \mathsf{t}_4^2 \mathsf{t}_7+\frac{\mathsf{t}_5^2 \mathsf{t}_7}{468750}+\frac{1}{500} \re^{4 \mathsf{t}_0} \mathsf{t}_2^2 \mathsf{t}_3 \mathsf{t}_7 \nn \\ &+& \frac{3}{25} \re^{11 \mathsf{t}_0} \mathsf{t}_2 \mathsf{t}_4 \mathsf{t}_7+\frac{3}{250} \re^{9 \mathsf{t}_0} \mathsf{t}_3 \mathsf{t}_4 \mathsf{t}_7+\frac{1}{625} \re^{8 \mathsf{t}_0} \mathsf{t}_2 \mathsf{t}_5
   \mathsf{t}_7+\frac{1}{625} \re^{3 \mathsf{t}_0} \mathsf{t}_4 \mathsf{t}_5 \mathsf{t}_7 \nn \\ &-& \frac{1}{500} \re^{4 \mathsf{t}_0} \mathsf{t}_3 \mathsf{t}_6 \mathsf{t}
   _7-\frac{1}{50} \re^{\mathsf{t}_0} \mathsf{t}_4 \mathsf{t}_6 \mathsf{t}_7+15 \mathsf{t}_4^2 \mathsf{t}
   _8+\frac{2}{125} \mathsf{t}_3 \mathsf{t}_5 \mathsf{t}_8-\frac{10}{3} \mathsf{t}_2 \mathsf{t}_6 \mathsf{t}
   _8 +\frac{1}{375} \re^{4 \mathsf{t}_0}
   \mathsf{t}_2^2 \mathsf{t}_5^2 \nn \\
&+& \frac{2 \re^{6 \mathsf{t}_0} \mathsf{t}_5^3}{84375}+\frac{1}{50} \re^{6 \mathsf{t}_0} \mathsf{t}_2^3 \mathsf{t}_7+\frac{3}{625} \re^{8 \mathsf{t}_0} \mathsf{t}_2 \mathsf{t}_3^2
   \mathsf{t}_5+13 \re^{13 \mathsf{t}_0} \mathsf{t}_2^2 \mathsf{t}_3 \mathsf{t}_4+\frac{4}{5} \re^{17 \mathsf{t}_0} \mathsf{t}_2 \mathsf{t}_4 \mathsf{t}_5~.
\eea

\section{$\wedge^\bullet\mathfrak{e}_8$ and relations in $R(\mathrm{E}_8)$: an
  overview of the results of \cite{E8comp}}
\label{sec:charE8}

I provide here a summary of the computer-aided proof of \cref{claim:E8}. In
principle, a direct approach to the determination of $\{\mathfrak{p}_k\}$ is
as follows
\ben
\item decompose $\wedge^k \mathfrak{g} = \oplus R^{(k)}_i$ into
  irreducibles;
\item letting
$\lambda_j=\sum_{i=1}^8 m_{i,j} \omega_j$ be the highest weight of $R^{(k)}_j$,
consider the tensor product decomposition of $\otimes_i \rho( m_{i,j}
\omega_j)$: this will contain $R^{(k)}_j$ as a summand (with coefficient 1),
plus extra terms;
\item iterate the operation until all virtual summands
(possibly with negative coefficients) have been replaced by tensor products of
fundamental representations; taking the character and summing over $j$ gives
$\mathfrak{p}_k$. 
\een
This however turns out to be computationally unfeasible already for $k\sim 7$; not only is the
decomposition of $\wedge^k \mathfrak{g}$ a daunting (and still unsolved, see
however \cite{MR1437204, MR1822681}) task; the decomposition into irreducibles
of even simple products
such as $\rho_{\omega_3}\otimes \rho_{\omega_3}$ would take hundreds of
Gigabytes of RAM to
compute. Instead, we proceed as follows:
\ben
\item For a given Cartan torus element $\exp(l)\in \cT$ with $l=\sum_i l_i
  \alpha^*_i \in \mathfrak{h}$, we can compute explicit Laurent polynomials
$\theta_j(\exp(l))\in \bbZ[(\re^{l_i})_i, (\re^{-l_i})_i]$, $\phi_k(l)\in
  \bbZ[(\re^{l_i})_i, (\re^{-l_i})_i]$ for $i=1,\dots, 8$ and $k=1,\dots, 120$ from the sole knowledge of the
  weight systems of $\rho_{\omega_i}$ with $i=1,7,8$, which have manageably
  small cardinality  2401, 241 and 26401 respectively (not taking into account
  weight multiplicities): for $\theta_{i}$ with $3\leq i\leq  7$ and all $\phi_k$ this follows from using Newton's identities applied
  to the power sums $\theta_7(\exp(k l))$, and then use of
  \eqref{eq:char37}; $\theta_{2}$ is similarly computed from
  $\theta_j$ with $j=1,7,8$ by the Adams operation on $\theta_1$:
\beq
\theta_1(\exp(2l))=\theta_2(\exp(l))-\theta_1(\exp l)+ \theta_7(\exp l) \theta_1(\exp l) - \theta_8(\exp l).
\eeq
\item We may then impose {\it a priori} constraints on the exponents
  $d_j^{(I)}$ appearing by inspection of the weight systems as well as on the
  dimensions of the tensor powers appearing on the monomials of the r.h.s. of \eqref{eq:pkgen}. We find
\beq
\{\max_{I\in M}d_j^{(I)}\}_{j=1}^8=  \{23,13,9,11,14,19,29,17\}, \quad
\max_{I\in M}\prod_{j=1}^8 \l(\mathrm{dim}\rho_{\omega_j}\r)^{d_j^{(I)}} =1.25366 \times 10^{96}.
\label{eq:djdeg}
\eeq
This truncates the sum on the r.h.s. to a finite, if large ($|M|=\cO(10^6)$),
number of monomials. 
\item In principle, imposing the identity of polynomials $\phi_k =
  \mathfrak{p}_k(\theta_1, \dots, \theta_8)$ determines uniquely all
  $n_{I,k}$; however the range of sum and the complexity of the polynomials involved
  renders this entirely unwieldy. A more sensible alternative is to solve
  \eqref{eq:pkgen} for $n_{I,k}$ by sampling the relation \eqref{eq:pkgen} at $|M|$
  random generic points $\exp(l)\in \bbQ^8$ of the torus\footnote{One might in
    principle pick generic, numerical random values with fixed precision and
    then hope to get an accurate integer truncation for the resulting
    $n_{I,k}\in \bbZ$. Such hope is however misplaced, as the resulting
    numerical matrix of monomial values (a matrix minor of a multi-variate
    Vandermonde matrix) is extremely ill-conditioned and leads to
    uncontrollable numerical errors.}, leading to a generically
  non-singular $|M|\times |M|$ linear system with rational entries for
  $n_{I,k}$, which can be solved exactly. Due to the sheer
  size of the monomial set and the density of the resulting linear system, however, this is unfeasible both for memory and
  time constraints.
\item There is however a non-generic choice of sampling points which, with
  some preparation, does the trick of reducing the problem to a large number
  of smaller problems of manageable size. First we subdivide the monomial set $M$ into slices $M_n=h^{-1}(n)$,
  $n=0, \dots, 8$  given by level
  sets of the function
\beq
h(d^{(I)} \in  M)= \sum_i \zeta(d_j^{(I)}), \quad
\zeta(x)=\l\{\bary{lcl}0 & \mathrm{if} & x=0,\\ 1 & \mathrm{if} & x>0. \eary \r.
\eeq
It is clear then that sampling $\exp(l)$ at values such that
$\theta_i(\exp(l))=0$ except for $n$ values of $i$ truncates the r.h.s. of
\eqref{eq:pkgen} to one of $\binom{8}{n}$ subsets of $M_n$. The strategy here
is to solve numerically for $\{\theta_i = u_i\}$ with some $u_i \in
\bbQ+I\bbQ$; there is a clever choice of the sampling set here such that with
sufficient floating point precision, we obtain a reliable -- and
in fact exact, with suitable analytical bounds -- rounding to rational
expressions for both the l.h.s. and the r.h.s. of \eqref{eq:pkgen}, for all
values of $k$. This simultaneously bypasses the ill-conditioning problem for
\eqref{eq:pkgen}, since we can then use exact arithmetic methods to solve it,
and moreover breaks it up into subsystems of size in the range $\cO(10)-\cO(10^4)$.
\item The latter point is not satisfactory yet since in the worst case
  scenario we deal with dense rational matrices of rank in the tens of
  thousands. However there's a refinement of the sets $M_k$ by
  considering a further slicing by one (or more) of the $d_j^{(I)}$ (i.e. level
  sets of the projections $p_j(d^{(I)} \in  M)=d_j^{(I)}$); these refined
  monomial sets are just selected by taking derivatives of \eqref{eq:pkgen}
  w.r.t. $\theta_j$ of order $d_j^{(I)}$. Since we have closed-form expressions for $\phi_k$
  and $\theta_j$ as functions of $\re^{l_i}$, these can be computed using
  Faa'~di~Bruno type formulas; while the complexity of the latter grows
  factorially, it turns out that derivative slicings of order up to five are
  both computable in finite time, compatible with the rounding of $\phi_k$
  with $8~\times$~machine
  precision, and they allow to break up the size of the
  resulting linear systems down to a maximum\footnote{A linear system of this
    rank, for the type of Vandermonde-type matrices we consider, required
    typically around 90Gb of RAM and half a day to terminate when solved using
    exact arithmetics (in our case, $p$-adic expansions and Dixon's method).}
 of $\cO(4\times 10^3)$. 
\item We are then left with a large number $(\cO(3.10^3))$ of relatively small linear systems and a large
  ($\cO(3.10^5)$) number of sampling points to evaluate $\phi_k$, $\theta_i$,
  and their derivatives in $l_j$; this would lead to a total runtime in the
  hundreds of months (about 120). However the numerical inversion, evaluation and calculation
  of derivatives at one sampling point is independent from that at another;
  this means that the calculation can be easily distributed over several CPU
  cores just by segmentation of the sampling set. Similar considerations
  apply, {\it mutatis mutandis}, to the solution of the linear
  subsystems. With $N~\simeq 75$ processor cores\footnote{In my specific case,
    this involved an average of about nine entry-level 64-bit cluster machines with dual 4-core CPUs.}, the absolute runtime gets reduced to
  about six weeks.
\een

The full result of the calculation is available at {\tt
  http://tiny.cc/E8SpecCurve}, and the original C~source code is available
upon request. 

It should be noted that, despite the innocent-looking appearance of \eqref{eq:p6}-\eqref{eq:p9}, both
the number of terms and the size of the coefficients grow extremely quickly with
$k$. The monomial set $M$ turns out to have cardinality
$|M|=949468$, with the matrix $n_{I,k}$ growing more and more dense for high
$k$ up to a maximum of 949256 non-zero coefficients for $k=118$, and $\max_{I,k}
n_{I,k} \simeq 1.7025\times 10^{10}$, $\min_{I,k} n_{I,k}\simeq -1.5403\times 10^{10}$.

\end{appendix}
\bibliography{miabiblio}
\end{document}

%% file: dualities.pdf_t
\begin{picture}(0,0)%
\includegraphics{dualities.pdf}%
\end{picture}%
\setlength{\unitlength}{4144sp}%
\begingroup\makeatletter\ifx\SetFigFont\undefined%
\gdef\SetFigFont#1#2#3#4#5{%
  \reset@font\fontsize{#1}{#2pt}%
  \fontfamily{#3}\fontseries{#4}\fontshape{#5}%
  \selectfont}%
\fi\endgroup%
\begin{picture}(4458,2484)(1606,-8230)
\put(5401,-7936){\makebox(0,0)[lb]{\smash{{\SetFigFont{12}{14.4}{\familydefault}{\mddefault}{\updefault}{\color[rgb]{0,0,0}5d $E_8$ SW theory }%
}}}}
\put(1846,-5911){\makebox(0,0)[lb]{\smash{{\SetFigFont{12}{14.4}{\familydefault}{\mddefault}{\updefault}{\color[rgb]{0,0,0}Chern-Simons/WRT}%
}}}}
\put(1981,-6136){\makebox(0,0)[lb]{\smash{{\SetFigFont{12}{14.4}{\familydefault}{\mddefault}{\updefault}{\color[rgb]{0,0,0}invariant of $S^3/\widetilde{\mathbb{I}}$}%
}}}}
\put(1621,-7936){\makebox(0,0)[lb]{\smash{{\SetFigFont{12}{14.4}{\familydefault}{\mddefault}{\updefault}{\color[rgb]{0,0,0}B-model: spectral curve}%
}}}}
\put(1666,-8161){\makebox(0,0)[lb]{\smash{{\SetFigFont{12}{14.4}{\familydefault}{\mddefault}{\updefault}{\color[rgb]{0,0,0}of $E_8$ relativistic Toda}%
}}}}
\put(5356,-5956){\makebox(0,0)[lb]{\smash{{\SetFigFont{12}{14.4}{\familydefault}{\mddefault}{\updefault}{\color[rgb]{0,0,0}A-model on $Y/\widetilde{\mathbb{I}}$}%
}}}}
\end{picture}%

%% file: dynke8.pdf_t
\begin{picture}(0,0)%
\includegraphics{dynke8.pdf}%
\end{picture}%
\setlength{\unitlength}{4144sp}%
\begingroup\makeatletter\ifx\SetFigFont\undefined%
\gdef\SetFigFont#1#2#3#4#5{%
  \reset@font\fontsize{#1}{#2pt}%
  \fontfamily{#3}\fontseries{#4}\fontshape{#5}%
  \selectfont}%
\fi\endgroup%
\begin{picture}(6517,1446)(-65,-2152)
\put(6437,-2078){\rotatebox{360.0}{\makebox(0,0)[rb]{\smash{{\SetFigFont{9}{10.8}{\rmdefault}{\mddefault}{\updefault}{\color[rgb]{1,0,0}$\alpha_0$}%
}}}}}
\put(1642,-1092){\makebox(0,0)[lb]{\smash{{\SetFigFont{9}{10.8}{\rmdefault}{\mddefault}{\updefault}{\color[rgb]{0,0,0}$\alpha_8$}%
}}}}
\put(850,-2092){\makebox(0,0)[lb]{\smash{{\SetFigFont{9}{10.8}{\rmdefault}{\mddefault}{\updefault}{\color[rgb]{0,0,0}$\alpha_2$}%
}}}}
\put(2650,-2092){\makebox(0,0)[lb]{\smash{{\SetFigFont{9}{10.8}{\rmdefault}{\mddefault}{\updefault}{\color[rgb]{0,0,0}$\alpha_4$}%
}}}}
\put(3550,-2092){\makebox(0,0)[lb]{\smash{{\SetFigFont{9}{10.8}{\rmdefault}{\mddefault}{\updefault}{\color[rgb]{0,0,0}$\alpha_5$}%
}}}}
\put(5319,-2097){\makebox(0,0)[lb]{\smash{{\SetFigFont{9}{10.8}{\rmdefault}{\mddefault}{\updefault}{\color[rgb]{0,0,0}$\alpha_7$}%
}}}}
\put(1731,-2092){\makebox(0,0)[lb]{\smash{{\SetFigFont{9}{10.8}{\rmdefault}{\mddefault}{\updefault}{\color[rgb]{0,0,0}$\alpha_3$}%
}}}}
\put(4451,-2095){\makebox(0,0)[lb]{\smash{{\SetFigFont{9}{10.8}{\rmdefault}{\mddefault}{\updefault}{\color[rgb]{0,0,0}$\alpha_6$}%
}}}}
\put(-50,-2092){\makebox(0,0)[lb]{\smash{{\SetFigFont{9}{10.8}{\rmdefault}{\mddefault}{\updefault}{\color[rgb]{0,0,0}$\alpha_1$}%
}}}}
\end{picture}%

%% file: npolred.pdf_t
\begin{picture}(0,0)%
\includegraphics{npolred.pdf}%
\end{picture}%
\setlength{\unitlength}{3947sp}%
\begingroup\makeatletter\ifx\SetFigFont\undefined%
\gdef\SetFigFont#1#2#3#4#5{%
  \reset@font\fontsize{#1}{#2pt}%
  \fontfamily{#3}\fontseries{#4}\fontshape{#5}%
  \selectfont}%
\fi\endgroup%
\begin{picture}(5223,3246)(-245,-12526)
\put(4832,-12361){\makebox(0,0)[lb]{\smash{{\SetFigFont{7}{8.4}{\rmdefault}{\mddefault}{\updefault}{\color[rgb]{0,0,0}$x$}%
}}}}
\put(-230,-9421){\makebox(0,0)[lb]{\smash{{\SetFigFont{7}{8.4}{\familydefault}{\mddefault}{\updefault}{\color[rgb]{0,0,0}$y$}%
}}}}
\end{picture}%

%% file: cuts.pdf_t
\begin{picture}(0,0)%
\includegraphics{cuts.pdf}%
\end{picture}%
\setlength{\unitlength}{4144sp}%
\begingroup\makeatletter\ifx\SetFigFont\undefined%
\gdef\SetFigFont#1#2#3#4#5{%
  \reset@font\fontsize{#1}{#2pt}%
  \fontfamily{#3}\fontseries{#4}\fontshape{#5}%
  \selectfont}%
\fi\endgroup%
\begin{picture}(4066,4140)(3900,-8026)
\put(7676,-4275){\makebox(0,0)[lb]{\smash{{\SetFigFont{9}{10.8}{\rmdefault}{\mddefault}{\updefault}{\color[rgb]{0,0,0}$\infty$}%
}}}}
\put(7426,-5911){\makebox(0,0)[lb]{\smash{{\SetFigFont{9}{10.8}{\rmdefault}{\mddefault}{\updefault}{\color[rgb]{0,0,0}$b_{3}^+$}%
}}}}
\put(6463,-5371){\makebox(0,0)[lb]{\smash{{\SetFigFont{9}{10.8}{\rmdefault}{\mddefault}{\updefault}{\color[rgb]{0,0,0}$b_{2}^-$}%
}}}}
\put(6850,-4921){\makebox(0,0)[lb]{\smash{{\SetFigFont{9}{10.8}{\rmdefault}{\mddefault}{\updefault}{\color[rgb]{0,0,0}$b_{2}^+$}%
}}}}
\put(6855,-6033){\makebox(0,0)[lb]{\smash{{\SetFigFont{9}{10.8}{\rmdefault}{\mddefault}{\updefault}{\color[rgb]{0,0,0}$b_{3}^-$}%
}}}}
\put(7201,-7010){\makebox(0,0)[lb]{\smash{{\SetFigFont{9}{10.8}{\rmdefault}{\mddefault}{\updefault}{\color[rgb]{0,0,0}$b_{4}^+$}%
}}}}
\put(6701,-6726){\makebox(0,0)[lb]{\smash{{\SetFigFont{9}{10.8}{\rmdefault}{\mddefault}{\updefault}{\color[rgb]{0,0,0}$b_{4}^-$}%
}}}}
\put(6125,-7216){\makebox(0,0)[lb]{\smash{{\SetFigFont{9}{10.8}{\rmdefault}{\mddefault}{\updefault}{\color[rgb]{0,0,0}$b_{5}^-$}%
}}}}
\put(6341,-7788){\makebox(0,0)[lb]{\smash{{\SetFigFont{9}{10.8}{\rmdefault}{\mddefault}{\updefault}{\color[rgb]{0,0,0}$b_{5}^+$}%
}}}}
\put(5221,-7797){\makebox(0,0)[lb]{\smash{{\SetFigFont{9}{10.8}{\rmdefault}{\mddefault}{\updefault}{\color[rgb]{0,0,0}$b_{6}^+$}%
}}}}
\put(5432,-7242){\makebox(0,0)[lb]{\smash{{\SetFigFont{9}{10.8}{\rmdefault}{\mddefault}{\updefault}{\color[rgb]{0,0,0}$b_{6}^-$}%
}}}}
\put(4329,-7032){\makebox(0,0)[lb]{\smash{{\SetFigFont{9}{10.8}{\rmdefault}{\mddefault}{\updefault}{\color[rgb]{0,0,0}$b_{7}^+$}%
}}}}
\put(4833,-6735){\makebox(0,0)[lb]{\smash{{\SetFigFont{9}{10.8}{\rmdefault}{\mddefault}{\updefault}{\color[rgb]{0,0,0}$b_{7}^-$}%
}}}}
\put(4122,-5919){\makebox(0,0)[lb]{\smash{{\SetFigFont{9}{10.8}{\rmdefault}{\mddefault}{\updefault}{\color[rgb]{0,0,0}$b_{8}^+$}%
}}}}
\put(4663,-6006){\makebox(0,0)[lb]{\smash{{\SetFigFont{9}{10.8}{\rmdefault}{\mddefault}{\updefault}{\color[rgb]{0,0,0}$b_{8}^-$}%
}}}}
\put(5081,-5358){\makebox(0,0)[lb]{\smash{{\SetFigFont{9}{10.8}{\rmdefault}{\mddefault}{\updefault}{\color[rgb]{0,0,0}$b_{9}^-$}%
}}}}
\put(4685,-4885){\makebox(0,0)[lb]{\smash{{\SetFigFont{9}{10.8}{\rmdefault}{\mddefault}{\updefault}{\color[rgb]{0,0,0}$b_{9}^+$}%
}}}}
\put(5792,-4544){\makebox(0,0)[lb]{\smash{{\SetFigFont{9}{10.8}{\rmdefault}{\mddefault}{\updefault}{\color[rgb]{0,0,0}$b_{1}^+$}%
}}}}
\put(5797,-5079){\makebox(0,0)[lb]{\smash{{\SetFigFont{9}{10.8}{\rmdefault}{\mddefault}{\updefault}{\color[rgb]{0,0,0}$b_{1}^-$}%
}}}}
\put(5817,-6037){\makebox(0,0)[lb]{\smash{{\SetFigFont{9}{10.8}{\rmdefault}{\mddefault}{\updefault}{\color[rgb]{0,0,0}$0$}%
}}}}
\end{picture}%